\documentclass{amsart}
\usepackage{fdsymbol} %
\usepackage{bbm}
\usepackage[all,cmtip]{xy}
\usepackage{stmaryrd} %

\usepackage{graphicx}

\usepackage{fonttable}

\makeatletter
\def\l@section{\@tocline{2}{4pt}{1pc}{5pc}{}}
\def\l@subsection{\@tocline{2}{0pt}{2.5pc}{4pc}{}}
\makeatother

\makeindex \topmargin=0in \headheight=8pt \headsep=0.5in \textheight=8.9in
\usepackage{etoolbox}
\patchcmd{\section}{\scshape}{\large\bfseries}{}{}
\makeatletter
\renewcommand{\@secnumfont}{\bfseries}
\makeatother
    
\usepackage[style=alphabetic]{biblatex}
\usepackage{booktabs}
\bibliography{references}

\numberwithin{equation}{section}

\newtheorem{thma}{Theorem}
\newtheorem{thmaa}{Theorem}
\newtheorem{thmas}{Theorem}[section]

\newtheorem{lemma}{Lemma}[section]
\newtheorem{cor}{Corollary}[section]
\newtheorem{prop}{Proposition}[section]

\theoremstyle{remark}
\newtheorem{rem}[lemma]{Remark}
\newtheorem{defin}{Definition}[section]

\usepackage{eucal}
\usepackage{enumitem}

\usepackage{hyperref}
\hypersetup{
    colorlinks=true,
    linkcolor=blue,
    filecolor=magenta,      
    urlcolor=cyan,
    citecolor=blue,
}

\newcommand{\nc}{\newcommand}

\nc{\Aa}{{\mathcal{A}}}
\nc{\Bb}{{\mathcal{B}}}
\nc{\Cc}{{\mathcal{C}}}
\nc{\Dd}{{\mathcal{D}}}
\nc{\Ee}{{\mathcal{E}}}
\nc{\Ff}{{\mathcal{F}}}
\nc{\Gg}{{\mathcal{G}}}
\nc{\Hh}{{\mathcal{H}}}
\nc{\Ii}{{\mathcal{I}}}
\nc{\Jj}{{\mathcal{J}}}
\nc{\Kk}{{\mathcal{K}}}
\nc{\Ll}{{\mathcal{L}}}
\nc{\Mm}{{\mathcal{M}}}
\nc{\Nn}{{\mathcal{N}}}
\nc{\Oo}{{\mathcal{O}}}
\nc{\Pp}{{\mathcal{P}}}
\nc{\Qq}{{\mathcal{Q}}}
\nc{\Rr}{{\mathcal{R}}}
\nc{\Ss}{{\mathcal{S}}}
\nc{\Tt}{{\mathcal{T}}}
\nc{\Uu}{{\mathcal{U}}}
\nc{\Vv}{{\mathcal{V}}}
\nc{\Ww}{{\mathcal{W}}}
\nc{\Xx}{{\mathcal{X}}}
\nc{\Yy}{{\mathcal{Y}}}
\nc{\Zz}{{\mathcal{Z}}}

\nc{\mA}{{\mathrm{A}}}
\nc{\mB}{{\mathrm{B}}}
\nc{\mC}{{\mathrm{C}}}
\nc{\mD}{{\mathrm{D}}}
\nc{\mE}{{\mathrm{E}}}
\nc{\mF}{{\mathrm{F}}}
\nc{\mG}{{\mathrm{G}}}
\nc{\mH}{{\mathrm{H}}}
\nc{\mI}{{\mathrm{I}}}
\nc{\mJ}{{\mathrm{J}}}
\nc{\mK}{{\mathrm{K}}}
\nc{\mL}{{\mathrm{L}}}
\nc{\mM}{{\mathrm{M}}}
\nc{\mN}{{\mathrm{N}}}
\nc{\mO}{{\mathrm{O}}}
\nc{\mP}{{\mathrm{P}}}
\nc{\mQ}{{\mathrm{Q}}}
\nc{\mR}{{\mathrm{R}}}
\nc{\mS}{{\mathrm{S}}}
\nc{\mT}{{\mathrm{T}}}
\nc{\mU}{{\mathrm{U}}}
\nc{\mV}{{\mathrm{V}}}
\nc{\mW}{{\mathrm{W}}}
\nc{\mX}{{\mathrm{X}}}
\nc{\mY}{{\mathrm{Y}}}
\nc{\mZ}{{\mathrm{Z}}}

\nc{\BB}{{\mathbb{B}}}
\nc{\CC}{{\mathbb{C}}}
\nc{\DD}{{\mathbb{D}}}
\DeclareMathOperator{\EE}{{\mathbb{E}}}
\nc{\FF}{{\mathbb{F}}}
\nc{\GG}{{\mathbb{G}}}
\nc{\HH}{{\mathbb{H}}}
\nc{\II}{{\mathbb{I}}}
\nc{\JJ}{{\mathbb{J}}}
\nc{\KK}{{\mathbb{K}}}
\nc{\LL}{{\mathbb{L}}}
\nc{\MM}{{\mathbb{M}}}
\nc{\NN}{{\mathbb{N}}}
\nc{\OO}{{\mathbb{O}}}
\nc{\PP}{{\mathbb{P}}}
\nc{\QQ}{{\mathbb{Q}}}
\nc{\RR}{{\mathbb{R}}}

\nc{\TT}{{\mathbb{T}}}
\nc{\UU}{{\mathbb{U}}}
\nc{\VV}{{\mathbb{V}}}
\nc{\WW}{{\mathbb{W}}}
\nc{\XX}{{\mathbb{X}}}
\nc{\YY}{{\mathbb{Y}}}
\nc{\ZZ}{{\mathbb{Z}}}

\DeclareMathOperator{\Aut}{\mathrm{Aut}}

\DeclareMathOperator{\Pic}{\mathrm{Pic}}
\nc{\Div}{D}

\DeclareMathOperator{\Tr}{\mathrm{Tr}}

\DeclareMathOperator{\Id}{\mathrm{Id}}

\DeclareMathOperator{\ord}{\mathrm{ord}}

\def\Arf{\mathrm{Arf}}
\nc{\Fun}{\mathrm{Fun}}
\def\Jac{\mathrm{Jac}}

\def\Tch{\mathrm{Teich}}
\def\Mod{\mathrm{Mod}}
\def\Diff{\mathrm{Diff}}

\def\length{\mathrm{length}}
\def\dist{\mathrm{dist}}
\def\Area{\mathrm{Area}}

\DeclareMathOperator{\wind}{\mathrm{wind}}

\def\exact{\mathrm{exact}}
\def\coexact{\mathrm{coexact}}
\def\harm{\mathrm{harm}}
\def\antiholom{\mathrm{antiholom}}

\DeclareMathOperator{\supp}{\mathrm{supp}}
\renewcommand{\div}{\mathrm{div}}

\let\Re\relax
\let\Im\relax
\DeclareMathOperator{\Re}{\mathrm{Re}}
\DeclareMathOperator{\Im}{\mathrm{Im}}

\DeclareMathOperator{\Res}{\mathrm{Res}}
\DeclareMathOperator{\dbar}{\bar{\partial}}
\nc\chr[2]{\begin{bmatrix}#1 \\ #2\end{bmatrix}}

\def\eps{\varepsilon}   
\def\vphi{\varphi}
\def\cst{\mathrm{cst}}
\DeclareMathOperator{\osc}{\mathrm{osc}}

\let\div\relax
\DeclareMathOperator{\div}{\mathrm{div}}

\def\harm{\mathrm{harm}}
\nc{\indic}{1\!\!1}
\DeclareMathOperator{\Var}{\mathrm{Var}}

\def\smm{\smallsetminus}
\def\op{\mathrm{op}}

\def\out{\mathrm{out}}
\def\inn{\mathrm{in}}

\def\rf{\mathrm{ref}}

\def\zero{\mathrm{zero}}

\def\GFF{\mathrm{GFF}}

\def\unif{\mathrm{Unif}}
\def\Lip{\mathrm{Lip}}

\def\spl{\mathrm{spl}}

\def\Bws{B^\circ_\spl}
\def\Wws{W^\circ_\spl}
\def\Ttws{\Tt_\spl^\circ}

\def\Bbs{B^\bullet_\spl}
\def\Wbs{W^\bullet_\spl}

\def\Gbs{G_\spl^\bullet}

\def\fluct{\mathrm{fluct}}

\nc{\w}{\mathbf{w}}
\def\M{\mathrm{M}}
\def\m{\mathfrak{m}}

\nc{\pf}{\mathrm{E}}
\def\o{\rotatebox[origin=c]{180}{$\omega$}}

\begin{document}

\title[Dimers on Riemann surfaces and compactified free field]{Dimers on Riemann surfaces and compactified free field}

\author[Mikhail Basok]{Mikhail Basok$^\mathrm{a}$}

\thanks{\textsc{${}^\mathrm{A}$ University of Helsinki, m.k.basok@gmail.com.}}

\begin{abstract}
We consider the dimer model on a bipartite graph embedded into a locally flat Riemann surface with conical singularities and satisfying certain geometric conditions in the spirit of~\cite{CLR1}. Following the approach developed by Dub\'edat~\cite{DubedatFamiliesOfCR} we establish the convergence of dimer height fluctuations to the compactified free field in the small mesh size limit. This work is inspired by the series of works~\cite{BerestyckiLaslierRayI, BerestyckiLaslierRayII} of Berestycki, Laslier, and Ray, where a similar problem is addressed, and the convergence to a conformally invariant limit is established in the Temperlian setup but the identification of the limit as the compactified free field is missing. This identification is the main result of our paper.
\end{abstract}

\maketitle

\tableofcontents

\section{Introduction}
\label{sec:introduction}

The celebrated result of Kenyon~\cite{KenyonGFF, KenyonConfInvOfDominoTilings} asserts that for each simply-connected domain $\Sigma\subset \CC$ and sequence of Temperleyan polygons $G^\delta\subset \delta\ZZ^2$ approximating $\Sigma$ as the mesh size $\delta$ tends to zero, the corresponding sequence of centered dimer height functions $h_\delta - \EE h_\delta$ converges to the Gaussian free field in $\Sigma$. During the last two decades Kenyon's approach was generalized to various other setups including other boundary conditions (e.g.~\cite{russkikh2020dominos}) and discrete approximations by more general graphs~\cite{CLR1,CLR2}. Another direction in which the results of Kenyon are generalized is related to more general Temperleyan domains where the graphs are not assumed to have a rigid geometric structure but rather satisfy some soft probabilistic conditions~\cite{berestycki2020dimers}. The techniques developed in~\cite{berestycki2020dimers} allowed to analyze the dimer model sampled on Temperleyan graphs on a general Riemann surface~\cite{BerestyckiLaslierRayI, BerestyckiLaslierRayII}. This led to proving the convergence of the dimer height function to a universal conformally invariant limit in this setup, which is perhaps the first known result of this type when the Riemann surface is general. While the soft probabilistic methods of~\cite{BerestyckiLaslierRayI, BerestyckiLaslierRayII} allowed to prove the convergence under rather mild assumptions (with the only significant limitation being the Temperleyan combinatorics), these methods did not allow to reconstruct the structure of the limit.

The goal of the present work is to build another approach for studying the dimer model on a general Riemann surface based upon the discrete complex analysis techniques developed in~\cite{CLR1,CLR2,DubedatFamiliesOfCR}. While this approach inevitably requires more `rigid' assumptions compared to the soft probabilistic approach of~\cite{BerestyckiLaslierRayI, BerestyckiLaslierRayII}, it also has a significant advantage of revealing exact relations between observables of the model and intrinsic objects on the Riemann surface. In particular, our results complement the results of~\cite{BerestyckiLaslierRayI, BerestyckiLaslierRayII} and allow to determine an exact structure of the limit of the dimer height function and identify it with a version of a compactified free field, see Corollary~\ref{cor:identification_for_BLR} and the remark afterwards.

Let us now discuss the setup in more details. Given a bipartite graph $G$ embedded into a Riemann surface $\Sigma$, we can define the dimer height function $h$ following the standard local rules (see, e.g., Section~\ref{subsec:intro_height_function}). This definition is consistent locally but not globally, which results in $h$ having an additive monodromy. This monodromy is random and depends only on the homology class of a loop along which it is evaluated; in other words, each dimer cover of $G$ defines a random cohomology class in $H^1(\Sigma, \RR)$. It is easy to see that the monodromy is integer \emph{up to a deterministic shift}, that is, there exists a deterministic $[u_0]\in H_1(\Sigma,\RR)$ such that the monodromy class almost surely belongs to $[u_0] + H_1(\Sigma,\ZZ)$. When $\Sigma$ has a boundary and $h$ is defined in such a way that it is constant along each boundary component (for example, if $\EE h = 0$), then the height difference between boundary components can be also included in the notion of the monodromy. In this case the corresponding cohomology class should be considered in the relative cohomology group $H^1(\Sigma; \partial \Sigma, \RR)$. 

Dimer model sampled on a topologically non-trivial surface has been considered already in works of Kasteleyn~\cite{kasteleyn1963dimer, kasteleyn1967graph}, where he computed the partition function on a torus and suggested the (proven later) formula in the general case, see~\cite{Cimasoni} for a comprehensive review and a modern approach to this problem. In the pioneering work of Kenyon~\cite{KenyonConfInvOfDominoTilings} the dimer height function was studied in multiply connected domains approximated by Temperleyan polygons. It was shown that the distribution of the height differences between boundary components becomes conformally invariant in the scaling limit, without identifying it explicitly. Later it was heuristically argued by Gorin in~\cite[Lecture~24]{gorin2021lectures} that the distribution becomes discrete Gaussian in the limit. In other words, Gorin argued that, for some deterministic \emph{shift} $x_0\in \RR^n$, the limit of the vector of height differences must belong to $x_0+\ZZ^n$ almost surely and the probability of observing a given vector must be proportional to the Gaussian of it. In the recent work~\cite{ahn2022lozenge} the height jump between two boundary components of a cylinder was proven to be approximately discrete Gaussian when the dimer model is sampled with respect to the $q$-volume measure.

In the work of Boutillier and de Tili\`ere~\cite{boutillier2009loop} the dimer model was studied on the hexagonal lattice projected to a torus. It was proven that the monodromies of the dimer height function have discrete Gaussian distribution in the scaling limit. %
The full convergence of the centered height function on a torus was proven by Dub\'edat~\cite{DubedatFamiliesOfCR} for the dimer model sampled on a sequence of arbitrary Temperleyan isoradial graphs. It was proven that the limit coincides with the \emph{compactified free field} on the torus, that is, can be written as a sum of the Gaussian free field ($\GFF$) and a linear function with a random integer slope, the latter being independent of the $\GFF$ and having discrete Gaussian distribution. A separate convergence of the monodromies of the toroidal dimer height function to the discrete Gaussian random vector was later proven in the case of much more general Temperleyan approximations by Dub\'edat and Gheissari~\cite{dubedat2015asymptotics}.

All these results support a widely accepted prediction that the scaling limit of the dimer height function on a general Riemann surface must be described by the compactified free field~\cite{Bosonization, francesco2012conformal, DubedatFamiliesOfCR}. To make this prediction more precise we need to introduce some notation. We prefer the language of differential forms when dealing with additively multivalued functions. To this end we can view the centered dimer height function $h - \EE h$ as a function on $\Sigma$ that is locally constant on faces of the bipartite graph $G$ and consider its exterior derivative $d(h - \EE h)$. This exterior derivative is a 1-form with coefficients being generalized functions, globally well-defined and closed (i.e. the exterior derivative of $d(h - \EE h)$ vanishes in the weak sense). Since $d(h - \EE h)$ is closed, the Hodge decomposition applies and we can write
\[
  d(h - \EE h) = d\Phi + \Psi
\]
where $\Phi$ is a function on $\Sigma$ and $\Psi$ is a harmonic differential. If the boundary of $\Sigma$ is non-empty, we can additionally require that $\Phi$ and $\Psi$ vanish along the boundary. (In the case of $\Psi$ the latter means that $\int\Psi$ is a harmonic function, multivalued if the genus of $\Sigma$ is positive, and constant on each boundary component of $\Sigma$.) Roughly speaking $\Phi$ is responsible for the local fluctuations of $h$, while $\Psi$ encodes its global behaviour, namely, the monodromy and how $h-\EE h$ changes between boundary components: the integral of $\Psi$ along a loop on $\Sigma$ (resp. a path between boundary components) is equal to the monodromy of $h - \EE h$ along this loop (resp. the jump of $h-\EE h$ between these components). In other words, the cohomology class in $H^1(\Sigma, \RR)$ (or in the relative cohomology group $H^1(\Sigma; \partial \Sigma,\RR)$ if $\partial\Sigma\neq \varnothing$) represented by $\Psi$ is equal to the cohomology class of the monodromy.

Using this terminology we can define the compactified free field to be the 1-form:
\[
  \m = d\phi + \psi
\]
where $\phi$ is the Gaussian free field (with zero boundary conditions if $\partial\Sigma\neq\varnothing$) and $\psi$ is a random harmonic differential vanishing along $\partial\Sigma$ and having the following two properties:
\begin{enumerate}
  \item There exists a deterministic harmonic differential $\psi_0$ (vanishing along $\partial\Sigma$) such that the periods and integrals between boundary components of $\psi - \psi_0$ are integer almost surely.
  \item $\psi$ is independent of $\phi$ and has discrete Gaussian distribution: the probability to observe a given differential is proportional to $\exp\left(-\frac{\pi}{2}\int_\Sigma \psi\wedge \ast \psi\right)$ where $\ast$ is the Hodge star (if $f = \int\psi$, then the integral above is nothing but the Dirichlet energy of $f$).
\end{enumerate}
We address the reader to Section~\ref{subsec:intro_freefield} for more details. Note that as soon as a basis in the first homology group $H_1(\Sigma;\partial\Sigma,\RR)$ is specified we can represent $\psi$ as a tuple of real numbers of length $2g(\Sigma) + \max(n-1,0)$ where $n$ is the number of boundary components. This tuple must almost surely be integer plus the corresponding tuple for $\psi_0$, and the quadratic form $\int_\Sigma \psi\wedge \ast \psi$ becomes explicitly expressed via this tuple and the matrix of $B$-periods of $\Sigma$ (the Schottky double of $\Sigma$ if the boundary is present), see Lemma~\ref{lemma:inner_product_via}.

Following the terminology adopted from physics we call $\Phi$ (resp. $\phi$) the scalar component of $h - \EE h$ (resp., of the compactified free field), and $\Psi$ (resp. $\psi$) the instanton component of $h - \EE h$ (resp., of the compactified free field). Note that it would be more accurate to define the compactified free field as the primitive of $\m$, but we prefer this abuse of the terminology for simplicity. Note also that in the case of the standard definition of the compactified free field the instanton component is assumed to have integer periods (and so $\int\m$ is a ``random map'' between $\Sigma$ and $S^1$). However, as we will see below, the appearance of the shift $\psi_0$ is very important and cannot be avoided in the dimer model context.

Consider now a family of graphs $G^\delta$ approximating the conformal structure of $\Sigma$ in a suitable sense. Let $h_\delta$ be the dimer height function on $G^\delta$ and $\Phi_\delta,\Psi_\delta$ be the scalar and instanton components of $h_\delta - \EE h_\delta$. The aforementioned prediction asserts that the pair $(\Phi_\delta,\Psi_\delta)$ must converge to the pair $(\phi,\psi)$ of the scalar and instanton components of the compactified free field, given that the latter is properly normalized and the shift $\psi_0$ is properly chosen.

\subsection*{Main results} The main goal of this paper is to give a mathematical proof of this convergence for a certain class of discrete approximations $G^\delta$ of a general Riemann surface specified below. The primal motivation for considering this setup comes from the works of Berestycki, Laslier, and Ray~\cite{BerestyckiLaslierRayI, BerestyckiLaslierRayII} that we mentioned in the beginning. Let us briefly discuss the details of this work. 
Recall that a graph is called Temperleyan if it obtained as a superposition of a graph $\Gamma$ and its dual $\Gamma^\dagger$. The superposition construction makes sense on any surface and produces a bipartite graph (see Figure~\ref{fig:construction_of_Gamma}), which does not however possess a dimer cover unless the Euler characteristics $\chi(\Sigma)$ of the surface is zero: indeed, the Euler formula implies that the mismatch between the number of black and the number of white vertices of the superposition graph is given by $\chi(\Sigma)$. In order obtain a dimerable graph in the general case the authors of~\cite{BerestyckiLaslierRayI, BerestyckiLaslierRayII} remove $-\chi(\Sigma)$ white vertices from the superposition graph and study the dimer model on the punctured graph $G$. %
Generalizing on the classical Temperleyan bijection, Berestycki, Laslier, and Ray provide a weight-preserving correspondence between dimer covers of $G$ and Temperleyan cycle rooted spanning forests (t-CRSF) of the initial graph $\Gamma$. (The weights on $G$ are defined out of the weights on $\Gamma$ is a suitable way.) Assuming that a sequence of graphs $\Gamma_k$ satisfy certain metric regularity conditions, the random walks on $\Gamma_k$ converge to the Brownian motion in the Skorokhod topology, and that the white vertices removed from the superposition graphs $G^k$ converge to points $p_1,\dots, p_{-\chi(\Sigma)}\in \Sigma$, Berestycki, Laslier, and Ray prove that t-CRSFs have a scaling limit in the Schramm topology. Moreover, they prove that the limit depends \emph{only} on the conformal type of the marked surface $(\Sigma,p_1,\dots, p_{-\chi(\Sigma)})$. Further, expressing the dimer height function $h_k$ on $G^k$ via the winding of branches of the t-CRSF on $\Gamma_k$ the authors of~\cite{BerestyckiLaslierRayI, BerestyckiLaslierRayII} prove that $h_k - \EE h_k$ converge to a certain \emph{universal} limit depending on $(\Sigma,p_1,\dots, p_{-\chi(\Sigma)})$ and nothing else.

As we already emphasized, the approach of~\cite{BerestyckiLaslierRayI,BerestyckiLaslierRayII} does not allow to identify the limit of $h_k - \EE h_k$ with the compactified free field (unlike in the simply connected case, where the identification with the $\GFF$ can be made by means of imaginary geometry coupling~\cite{berestycki2020dimers}). Our main goal therefore is to develop a setup that includes Temperleyan graphs such as considered by Berestycki, Laslier, and Ray, and, on the other hand, allows to apply the discrete complex analysis technique and use it to get a structural understanding of the limit. Note that since the limit of $h_k - \EE h_k$ constructed by Berestycki, Laslier, and Ray does not depend on the particular sequence of graphs, we can restrict our attention only to graphs with special geometric properties; however, such graphs must exist on each marked Riemann surface.

To describe our setup we begin with a choice of a \emph{uniformization} of a marked surface $(\Sigma, p_1,\dots, p_{-\chi(\Sigma)})$. As the number of marked points is conveniently equal to $-\chi(\Sigma)$, it is natural to fix a locally flat metric on $\Sigma$ with conical singularities at $p_1,\dots, p_{-\chi(\Sigma)}$ with cone angles $4\pi$. We also assume that each boundary component of $\Sigma$ has a neighborhood isometric to a vertical cylinder with horizontal boundary; such a metric always exists and is unique up to a global rescaling, see Proposition~\ref{prop:existence_of_metric_on_Sigma}.
Using that the metric is locally Euclidean, we extend the notion of a t-embedding introduced in~\cite{CLR1} by Chelkak, Laslier and Russkikh to this setup and consider bipartite graphs on $\Sigma$ whose dual possess a `regular enough' t-embedding into $\Sigma$ (see Sections~\ref{subsubsec:intro_t-embedding_prelim} and~\ref{subsubsec:intro_graph_assumptions}). If $\partial \Sigma\neq \varnothing$, then we rather consider t-embeddings on the double (see Section~\ref{subsec:intro_flat_metric}) of $\Sigma$ (which also is equipped with a locally flat metric due to our assumption about the boundary) that are symmetric under the natural involution. A t-embedding of the dual graph provides a natural choice for the dimer weights.

Let us now formulate our main results. We begin by discussing the special case of Temperleyan graphs considered by Berestycki, Laslier and Ray. In this case the shift of the compactified free field turns out to have a deep relation with the holonomy of the aforementioned locally flat metric with conical singularities. If $\partial\Sigma\neq \varnothing$, replace $\Sigma$ with its double and complement the set of $p_j$'s with its symmetric image (see Section~\ref{subsec:intro_flat_metric} for details, note that we again have $-\chi(\Sigma)$ marked points). Recall that the holonomy of the metric on $\Sigma$ along a smooth oriented loop $\gamma\subset \Sigma\smm\{ p_1,\dots,p_{-\chi(\Sigma)} \}$ is the operator $\chi(\gamma)\in \mathrm{SO}(2,\RR)\cong \TT = \{ z\in \CC\ \mid\ |z| = 1 \}$ such that, for a given point $p\in \gamma$, the parallel transportation of every tangent vector $v\in T_p\Sigma$ along $\gamma$ results in replacing $v$ with $\chi(\gamma)\cdot v$. The properties of the metric on $\Sigma$ imply that the holonomy depends only on the homology class of $\gamma$ in $\Sigma$, and so it induces a homomorphism $\chi: H^1(\Sigma, \ZZ)\to \TT$. There exists a harmonic differential $u_0$ such that
\[
  \chi(\gamma) = \exp(i\int_\gamma u_0),\qquad \forall [\gamma]\in H^1(\Sigma,\ZZ).
\]

The metric on $\Sigma$ also allows to define the (intrinsic) winding of smooth oriented curves. Given such a curve $\gamma:[0,1]\to \Sigma\smm\{ p_1,\dots, p_{-\chi(\Sigma)} \}$ we compute the winding of it by evaluating the total change of the oriented angle between $\gamma'(t)$ and the vector field obtained by the parallel transportation of $\gamma'(0)$ along $\gamma$. For example, the winding of a small simple loop oriented counterclockwise is equal to $2\pi$. If $\gamma$ is a general simple loop, then its winding is not necessary an integer times $2\pi$: if the holonomy of the metric is non-trivial, then the winding is equal to $-\int_\gamma u_0$ modulo $2\pi \ZZ$. One can deduce from the results of~\cite[Section~4]{BerestyckiLaslierRayI} that, given that the reference flow (see Section~\ref{subsec:intro_height_function} and~\eqref{eq:def_of_fA}) determining the dimer height function on $\Sigma$ is chosen properly, the increments of the height can be expressed via the winding of branches of t-CRSF computed with respect to the above defined metric on $\Sigma$ and normalized by $2\pi$. Thus, if we take this choice of the metric for granted, it is plausible to expect that the shift of the compactified free field arising as the limit of Temperleyan dimer height functions on $\Sigma$ will be $-(2\pi)^{-1}u_0$.

When $\partial\Sigma\neq\varnothing$ the above definition of $u_0$ does not specify it precise enough and we need to fix an additional normalization. In this case we can choose $u_0$ to be anti-invariant under the natural involution of the double, and such that for every simple oriented loop $\gamma$ whose cohomology class is \emph{anti-symmetric} under the involution on the double its winding is equal to $2\pi-\int_\gamma u_0$ modulo $4\pi \ZZ$. Note that in this way we fix $u_0$ up to adding a harmonic differential representing an integer element in the \emph{relative} cohomology group, that is, in $H^1(\Sigma; \partial 2\pi\Sigma, \ZZ)$. In other words, if $u_1$ is another choice, then both periods and \emph{integrals between boundary components} of $u_0 - u_1$ belong to $2\pi\ZZ$. We are now ready to present our results for Temperleyan graphs. We address the reader to Section~\ref{subsec:intro_relation_to_BLR} for more details.

\renewcommand{\thethmaa}{A}
\begin{thmaa}
  \label{thmaa:informal_temp}
  Given $(\Sigma, p_1,\dots, p_{-\chi(\Sigma)})$ there exists a sequence of graphs $G^k$ on $\Sigma$ with the following properties:

-- Each $G^k$ is of the form $\Gamma^k\cup (\Gamma^k)^\dagger\smm \{ p_1,\dots, p_{-\chi(\Sigma)} \}$ where $\Gamma^k\cup (\Gamma^k)^\dagger$ is the superposition of a weighted graph $\Gamma$ and its dual, and $p_1,\dots, p_{-\chi(\Sigma)}$ are positions of some white vertices of $G^k$.

-- Graphs $G^k$ satisfy all the assumptions imposed in~\cite{BerestyckiLaslierRayI, BerestyckiLaslierRayII} if the latter are stated with respect to the locally flat metric \emph{with conical singularities} chosen on $\Sigma$.

-- The sequence $G^k$ satisfies the assumptions imposed in Section~\ref{subsubsec:intro_graph_assumptions} and the dimer weights on $G^k$ defined in this section are gauge equivalent to the dimer weights considered in~\cite{BerestyckiLaslierRayI, BerestyckiLaslierRayII}. %

Assume moreover that the sequence $h_k - \EE h_k$ is \emph{tight} and that for any smooth 1-form $u$ the first moments of random variables $\int_\Sigma d(h_k - \EE h_k) \wedge u $ are bounded in $k$. Then $h_k - \EE h_k$ converge to $\int (\m - \EE\m)$ where $\m$ is the compactified free field normalized as in Section~\ref{subsec:intro_freefield} and with the shift $\psi_0 = -(2\pi)^{-1}u_0$ where $u_0$ is the holonomy 1-form chosen as above.
\end{thmaa}

Theorem~\ref{thmaa:informal_temp} is a particular case of Theorem~\ref{thma:main1}, see Theorem~\ref{thma:for_BLR}. Let us emphasize that the main result in~\cite{BerestyckiLaslierRayI, BerestyckiLaslierRayII} is stated when the graphs are satisfying aforementioned assumptions with respect to a \emph{smooth} metric on $\Sigma$. We expect the necessary generalization to appear in following works. At the moment however, combining the results of~\cite{BerestyckiLaslierRayI, BerestyckiLaslierRayII} with the theorem above we arrive to the following statement:

\begin{cor}
  \label{cor:identification_for_BLR}
  Assume that the conclusion of~\cite[Theorem~1.1]{BerestyckiLaslierRayI} can be generalized to the case when the metric on the surface $\Sigma$ has finitely many conical singularities. Then the limit of height fluctuations that appears in~\cite[Theorem~1.1]{BerestyckiLaslierRayI} is equal (in distribution) to the compactified free field normalized as in Section~\ref{subsec:intro_freefield} and with the shift $\psi_0 = -(2\pi)^{-1}u_0$ where $u_0$ is the holonomy 1-form chosen as above.
\end{cor}

Let us briefly argue why the choice of the above defined metric on $\Sigma$ fits naturally with the Temperleyan setup. Note that a locally flat metric $ds^2$ has a conical signularities of cone angle $4\pi$ at $p\in \Sigma$ if it is locally given by the pullback of the Euclidean metric along a double cover ramified at $p$. Assume now that we have a Temperleyan graph on the plane and look at its preimage under this branched cover. As we discuss in Section~\ref{subsec:intro_relation_to_BLR} (see in particular Figure~\ref{fig:adding_edges}), the combinatorics of the preimage at $p$ will be the one of a Temperleyan graph with a white vertex removed. This observation suggests that to define `nice' (say, isoradial, or by circle patterns~\cite{KLRR}) embeddings of Temperleyan graphs with removed white vertices it is natural to introduce conical singularities of cone angles $4\pi$ at these vertices. Finally, let us notice that in the isoradial setup the winding reference flow miraculously coincides with the reference flow defined by the infinite volume Gibbs measure with a maximal entropy, see Remark~\ref{rem:flows} for more detailed discussion.

\begin{rem}
  \label{rem:cylinder}
  Note that the normalization of $u_0$ might force it to be non-zero even when the metric has trivial holonomy. This happens already if $\Sigma$ is a vertical cylinder, which results the height jump between two boundary components to be half-integer in the scaling limit if one approximates the cylinder by Temperleyan graphs. To argue that this must be the case one can consider the dimer model on a square grid and the height function defined with respect to the usual $\tfrac14$ reference flow.
\end{rem}

We further investigate the nature of the shift of the compactified free field by including graphs with more general combinatorics in our general identification Theorem~\ref{thma:main1}. Namely, we allow the graphs to not have Temperleyan structure while keeping the boundary conditions intact. The relevant choice of the shift for the compactified free field is then determined by the geometry of the t-embedding of the graph, in particular, by the gauge class of a real valued Kasteleyn operator on it (see Section~\ref{subsec:intro_Kasteleyn_operator} and~\eqref{eq:def_of_K}). As a simple example of a graph on a torus demonstrates (see Example~\ref{intro_example:torus}), this shift may \emph{not} coincide with $-(2\pi)^{-1}u_0$ in the general case.

To support our identification result for non-Temperleyan graphs we also prove the following theorem asserting that a sequence of centered dimer height functions is tight provided that graphs satisfy our regularity conditions (See Theorem~\ref{thma:main2} for the detailed statement):

\renewcommand{\thethmaa}{B}
\begin{thmaa}
  \label{thmaa:tightness}
  Assume that $(\Sigma,p_1,\dots, p_{-\chi(\Sigma)})$ is generic and graphs $G^\delta$ satisfy all the conditions from Section~\ref{subsubsec:intro_graph_assumptions}. Then the family of centered height function $h^\delta - \EE h^\delta$ on $G^\delta$ is tight, and for any smooth 1-form $u$ second moments of $\int_\Sigma d(h_\delta - \EE h_\delta)\wedge u$ are bounded in $\delta$ uniformly in $u$ taken from any set bounded with respect to the $\mC^1$-norm.
\end{thmaa}

The marked surface $(\Sigma,p_1,\dots, p_{-\chi(\Sigma)})$ is ``generic'' if, in particular, $\Sigma$ is arbitrary and $p_1,\dots, p_{-\chi(\Sigma)}$ are chosen generically, or if $\Sigma$ has the topology of a multiply connected domain and $p_1,\dots, p_{-\chi(\Sigma)}$ are arbitrary points in the bulk.

\subsection*{Discussion of the proof}

Let us now give a short overview of the strategy that we use to obtain our results. As a starting point we take the approach developed by Dub\'edat~\cite{DubedatFamiliesOfCR} to analyse the dimer height function on a torus. Working with isoradial graphs, Dub\'edat introduces a discrete version $K_\alpha$ of a perturbed Dirac operator of the form $\begin{pmatrix} 0 & \dbar + \alpha \\ \partial - \bar{\alpha} & 0 \end{pmatrix}$, where $\alpha$ is an arbitrary $(0,1)$ form. A generalization of the Kasteleyn theorem allows to link $\det K_\alpha$ and the observable $\EE \exp(i \int \Im \alpha \wedge dh)$ where $h$ is a properly defined dimer height function. To derive the asymptotics of this observable Dub\'edat considers the logarithmic variation of $\det K_\alpha$ with respect to $\alpha$. The resulting object is shown to resemble the Quillen's variation identity~\cite{Quillen} in the scaling limit, which allows to relate the limit of the observable $\EE \exp(i \int \Im \alpha \wedge dh)$ with the compactified free field on the torus.

To generalize this to an arbitrary Riemann surface, we mimic the same construction for the family of perturbed Dirac operators $K_\alpha$ and study the logarithmic variation of $\det K_\alpha$. %
There are various complications comparable to the setup of Dub\'edat: we have to deal with non-isoradial graphs, to take care of the boundary and of the conical singularities while constructing and estimating $K_\alpha^{-1}$, to deal with more involved algebraic and geometric structures arising when the genus of the surface is bigger than 1. This requires a use of various techniques including the recently developed regularity theory for discrete holomorphic functions on t-embeddings~\cite{CLR1}, and extending Dub\'edat's results on multivalued discrete holomorphic functions~\cite{DubedatFamiliesOfCR} to study $K_\alpha^{-1}$ at conical singularities. Similarly as in the case of a torus, $\det K_\alpha$ is naturally related with a certain observable of the dimer height function. The asymptotics of the logarithmic variation of $\det K_\alpha$ allows us to control this observable and, similarly to the torus case, to relate it with the Quillen's determinant. In order to apply this to relating the scaling limit of the dimer height function with the compactified free field we generalized a `bosonization identity' of Luis Alvarez-Gaum\'e et al.~\cite{Bosonization} (see Lemma~\ref{lemma:bosonization_identity}).

It is worth noting that various versions of the aforementioned perturbation technique have been used to analyze the dimer model in several setups in which a non-trivial topology plays a role. Besides the torus case, in~\cite{DubedatFamiliesOfCR} the asymptotics of monomer correlators on the plane was treated; in this case the non-trivial topology appears when one considers monomers as punctures. In the work~\cite{dubedat2018double} the Dub\'edat--Kenyon `topological observables' (see also~\cite{kenyon2014conformal}) of the double-dimer model are shown to converge to an instance of the isomonodromic tau function in the scaling limit. The work~\cite{dubedat2018double} has been developed further by Basok and Chelkak~\cite{basok2021tau} and Bai and Wan~\cite{bai2023crossing}.%
In~\cite{basok2024nesting} a similar approach is used to study nesting of double-dimer loops.

\subsection*{Relation to the works of Cimasoni and Costa-Santos and McCoy} Dimer model sampled on a Riemann surface equipped with a locally flat metric with conical singularities was earlier considered by Cimasoni~\cite{cimasoni2012discrete}. Given an isoradial (not necessary Temperleyan) bipartite graph embedded into such a surface, Cimasoni introduced a construction of $2^{2g}$ discrete Dirac operators associated with spin structures on the surface. This begins with a discrete $\dbar$ operator that mimics the standard construction of~\cite{KenyonCriticalPlanarGraphs}. By its definition, this operator acts on (discrete) functions and returns (discrete) $(0,1)$-forms. This discrete $\dbar$ is then twisted by a spin structure, which makes it acting on sections of a spin line bundle. It is proved in~\cite{cimasoni2012discrete} that whenever discrete holomorphic spinors converge to a continuous section of the corresponding spin line bundle, the limit is actually holomorphic. In addition to this, the work~\cite{cimasoni2012discrete} suggests geometric conditions that are necessary and sufficient for these discrete Dirac operators to be gauge equivalent to real Kasteleyn matrices on the corresponding graph. Notably, these conditions do not necessary imply that the metric on the surface has trivial holonomy, rather that the cone angles are odd multiple of $2\pi$ and a certain angle condition. An analogy of the last condition in our work would be the condition $\alpha_G = 0$, where $\alpha_G$ is the $(0,1)$-form introduced in Section~\ref{subsec:intro_Kasteleyn_operator}. Note that if this condition is satisfied, then the scaling limit of dimer height fluctuations is the compactified free field whose instanton component is integer.

We conclude the introduction by indicating the following corollary of our result. Kasteleyn's formula expresses the partition function of the dimer model on a closed Riemann surface $\Sigma$ as a linear combination of $2^{2g}$ signed partition functions enumerated by spin structures on $\Sigma$, where $g$ is the genus~\cite{Cimasoni}. On the continuous side, such a decomposition appears naturally in the context of conformal field theory~\cite{Bosonization}. Computations made from the conformal field theory perspective suggest that the ratio of any two aforementioned signed partition functions should converge in the scaling limit to a concrete expression involving theta constants. This hypothesis was verified numerically by Costa-Santos and McCoy~\cite{costa2002dimers} in the case of genus 2 surfaces approximated by `lattices' (cf.~Example~\ref{intro_example:pillow_surface}). Our result implies the aforementioned convergence for all $g$ when the Riemann surface $(\Sigma,p_1,\dots,p_{2g-2})$ is generic; the genericity assumption can be dropped given that the result of Berestycki, Laslier and Ray~\cite{BerestyckiLaslierRayI, BerestyckiLaslierRayII} is extended to metrics with conical singularities. See Section~\ref{subsec:intro_ratio_of_partition_functions} for details.

We address the reader to Section~\ref{sec:background_and_formulations} for more detailed background and strict formulation of main results. Section~\ref{subsec:organization} contains the organization of the rest of the paper.

\subsection*{Acknowledgments} The author thanks Beno\^it Laslier for introducing him to the problem, and Nathana{\"e}l Berestycki, Beno{\^i}t Laslier and Gourab Ray for discussions on the subject. The author is grateful to Dmitry Chelkak and Konstantin Izyurov for many enlightening advises and for making the work on this paper possible. The author thanks David Cimasoni for his valuable comments on the first version of the paper. The work on this paper was finalized while the author was visiting the Institute for Pure and Applied Mathematics (IPAM), which is supported by the National Science Foundation (Grant No. DMS-1925919). The author thanks the anonymous referees for their valuable comments and suggestions. This project has been started at the ENS Paris with a partial support of the ENS-MHI Chair in Mathematics and the ANR-18-CE40-0033 project DIMERS, the author is grateful to the ENS for the hospitality. The author was supported by Academy of Finland via the project ``Critical phenomena in dimension two'' and by Research Council of Finland via the project ``Lattice models and conformal field theory''.

\section{Background and formulation of main results}
\label{sec:background_and_formulations}

In this section we introduce main notions and concepts we need to formulate and discuss our main results. We formulate our results in Section~\ref{subsec:intro_main_results}.

\subsection{Locally flat metrics with conical singularities on a Riemann surface}
\label{subsec:intro_flat_metric}

Let $\Sigma_0$ be a Riemann surface of genus $g_0$ with $n\geq 0$ boundary components $B_0,\dots, B_{n-1}$ (if $n=0$, then we assume that $\partial\Sigma_0 = \varnothing$). We will always assume that $-\chi(\Sigma_0) = 2g_0 - 2 + n \geq 0$, that is, $\Sigma_0$ has non-positive Euler characteristics. If $2g_0 - 2+n>0$, then we assume that $2g_0-2+n$ distinct points $p_1,\dots, p_{2g_0 -2+n}$ in the interior of $\Sigma_0$ are given. If $\partial \Sigma_0\neq \varnothing$, then it is useful to consider the \emph{double} of $\Sigma_0$. Is is constructed by gluing $\Sigma_0$ with $\Sigma_0^\op$ along the boundary, where $\Sigma_0^\op$ is a copy of $\Sigma_0$ with the reversed orientation. The double comes with a natural conformal structure and an anti-holomorphic involution $\sigma$. We can also double the collection $p_1,\dots, p_{2g_0 -2+n}$ by putting $p_{2g_0 -2+n+i} = \sigma(p_i)$, $i = 1,\dots, 2g_0-2+n$. Note that $g = 2g_0-1+n$ is the genus of the double, and the total amount of marked points on the double is $2g-2$.

Let $\Sigma$ denote the Riemann surface equal to $\Sigma_0$ if $\partial \Sigma_0 = \varnothing$ or to the double of $\Sigma_0$ if $\partial \Sigma_0 \neq \varnothing$. Let $g$ denote the genus of $\Sigma$ in both cases. In what follows we will usually be working with $\Sigma$ rather than with $\Sigma_0$, extending all the other objects to the double if $\Sigma\neq \Sigma_0$. The presence the boundary of $\Sigma_0$ is indicated by the presence of the involution $\sigma$; if the latter is given, then we always keep track of how the objects change under its action.

It can be proven (see Proposition~\ref{prop:existence_of_metric_on_Sigma}) that there exists a unique singular Riemannian metric $ds^2$ on $\Sigma$ with the following properties:
\begin{enumerate}
  \item The area of $\Sigma$ is 1.
  \item $ds^2$ is smooth outside $p_1,\dots, p_{2g-2}$. Any point $p\in \Sigma_0\smm\{ p_1,\dots, p_{2g-2} \}\cup \partial\Sigma_0$ has a neighborhood isometric to an open subset of the Euclidean plane, that is, $ds^2$ is locally flat outside $p_1,\dots, p_{2g-2}$.
  \item For any $j = 1,\dots, 2g-2$ there is a neighborhood of $p_j$ isometric to a neighborhood of the origin of $\CC$ with the metric $|d(z^2)|^2$. In other words, $ds^2$ has conical singularities of cone angles $4\pi$ at $p_1,\dots, p_{2g-2}$.
  \item The metric $ds^2$ is symmetric with respect to $\sigma$ if the latter is given.
\end{enumerate}
We will often be calling a pair $(\Sigma, ds^2)$ a \emph{locally flat surface} for simplicity, not mentioning conical singularities. The notation $\dist(x,y)$ will stand for the distance between $x,y\in \Sigma$ measured with respect to the inner metric induced by $ds^2$.

Note that $ds^2$ may have non-trivial holonomy: a parallel transport along a non-contractible loop may create a non-trivial turn. This defines a cohomology class on $\Sigma$ which can be represented using an anti-holomorphic 1-form $\alpha_0$ such that the holonomy map along a loop $\gamma$ is given by the multiplication by $\exp\left( 2i \int_\gamma \Im \alpha_0 \right)$. Note that $\alpha_0$ can be chosen such that
\[
  \sigma^*\alpha_0 = \bar{\alpha}_0
\]
if $\sigma$ is present. Moreover, according to Proposition~\ref{prop:existence_of_metric_on_Sigma}, the metric $ds^2$ has the form $|\omega_0|^2$ where $\omega_0$ is a smooth $(1,0)$-form on $\Sigma_0$ satisfying $(\dbar - \alpha_0)\omega_0 = 0$ (the latter condition means that the multivalued $(1,0)$-form $\exp\left( -2i\int^p\Im\alpha_0 \right)\omega_0(p)$ is holomorphic). If $\sigma$ is given, then we can choose $\omega_0$ such that $\sigma^*\omega_0 = \bar{\omega}_0$.

\begin{rem}
  \label{rem:why_01_forms}
  It is more common to represent cohomology classes with harmonic 1-forms. There is however a linear isomorphism between the space of anti-holomorphic $(0,1)$-forms and the space of real-valued harmonic 1-forms given by $\alpha\mapsto 2\Im \alpha$ (see Section~\ref{subsec:cff_formal_def}). For a certain technical reason using $(0,1)$-forms is more convenient for us; it does not create any difference due to the aforementioned isomorphism. 
\end{rem}

Let us choose a base point $p_0\in \Sigma$. If $\sigma$ is present, then we assume also that $\sigma(p_0) = p_0$. We define
\begin{equation}
  \label{eq:def_of_omega}
  \omega(p) = \exp\left( -2i\int_{p_0}^p\Im\alpha_0 \right)\omega_0(p)
\end{equation}
and
\begin{equation}
  \label{eq:def_of_Tt}
  \Tt(p) = \int_{p_0}^p\omega.
\end{equation}
The holomorphic $(1,0)$-form $\omega$ and the function $\Tt$ are defined consistently only in a simply-connected vicinity of the point $p_0$; extending them analytically along a loop $\gamma$ amounts in replacing $\omega$ with $a\omega$ and $\Tt$ with $a\Tt + b$, where $a,b\in \CC$ and $|a| = 1$. We thus regard $\omega$ and $\Tt$ as a multivalued $(1,0)$-form and a multivalued function on $\Sigma$ respectively. We intentionally keep a little ambiguity in these definitions: in what follows $\omega$ and $\Tt$ will be used only as building blocks for defining other objects, and in all the cases the replacements of $\omega$ and $\Tt$ with $a\omega$ and $a\Tt + b$ will not affect the constructions. So, for example, we can freely write $ds^2 = |\omega|^2 = |d\Tt|^2$. Note that $\Tt$ is locally one-to-one outside conical singularities, and a branched double cover at the conical singularities.

Using the differential $\omega_0$ we can define the (intrinsic, cf.~\cite{berestycki2020dimers}) winding of a smooth path $\gamma:[0,1]\to \Sigma\smm\{ p_1,\dots, p_{2g-2} \}$
\begin{equation}
  \label{eq:def_of_wind_w0}
  \wind(\gamma, \omega_0) = \Im\int_0^1 \frac{d}{dt} \log \omega_0(\gamma'(t))\,dt.
\end{equation}
Note that if $\gamma$ is a loop, then $\wind(\gamma, \omega_0)\in 2\pi \ZZ$. Note also that 
\begin{equation}
  \label{eq:wind_w}
  \wind(\gamma, \omega) = \wind(\gamma, \omega_0) - 2\Im \int_{\gamma([0,1])}\alpha_0
\end{equation}
is the winding defined by the metric connection of the metric $ds^2$. In particular, such winding of a loop may not be integer times $2\pi$.

Note that $\omega_0$ and $\alpha_0$ are not determined uniquely by the metric $ds^2$. We can replace $\alpha_0$ with $\alpha_1$ if $\int_\gamma\Im(\alpha_1 - \alpha_0)\in \pi \ZZ$ for each loop $\gamma$; in this case we can also put $\omega_1(p) = \exp(2i\int^p \Im (\alpha_1 - \alpha_0))\omega_0(p)$. This in particular allows us to tune $\alpha_0$ and $\omega_0$ in such a way that for every simple smooth loop $\gamma\subset \Sigma\smm\{ p_1,\dots, p_{2g-2} \}$ such that $\sigma(\gamma) = \gamma$ we have
\begin{equation}
  \label{eq:wind_condition}
  \wind(\gamma,\omega_0) \in 2\pi + 4\pi \ZZ,
\end{equation}
see Corollary~\ref{cor:wind_condition}.

\subsection{Assumptions on the graphs embedded into a surface}
\label{subsec:intro_graphs_on_Sigma0}

Given a bipartite graph $G$, we will be calling its vertices ``black'' and ``white'' to distinguish its bipartite classes as usual. We use the notation $B$ and $W$ for the set of black and white vertices of $G$, and keep the same notation for the vertices of $G$ and for the corresponding faces of the dual graph $G^\ast$. Thus, for example, $b\in B$ may denote a point on a surface, corresponding to a vertex of $G$, and a polygon corresponding to a face of $G^\ast$ at the same time. We say that a graph is embedded into $\Sigma$ if it is drawn on $\Sigma$ in such a way that all its vertices are mapped onto distinct points, all its edges are simple curves which can intersect each other only at the vertices, and the complement to the union of all the edges is a union of topological discs, which we call faces of the graph. 

To be able to apply methods from discrete complex analysis we impose a number of assumptions on the graphs we consider. We first formulate them informally:

(i) We need the graph to possess a good notion of discrete holomorphicity, with an appropriate regularity theory behind. For this we assume that the dual graph $G^\ast$ is t-embedded into $\Sigma$, and the corresponding origami map is small, so that the t-embedding is approximating the complex structure of the surface (see below for the definitions of a t-embedding and an origami map). This is formulated in Assumptions~\ref{item:intro_G_well_embedded}--~\ref{item:intro_reg_part} below.

(ii) We need the graph to be very regular near conical singularities, because the local analysis is more delicate there. For this we assume that the graph is locally a double cover of a Temperleyan isoradial graph at conical singularities, see Assumption~\ref{item:intro_conical_sing}.

(iii) We need the discrete Cauchy--Riemann operator acting on functions on black vertices of $G$ to approximate the Cauchy--Riemann operator acting on smooth functions well enough. This is fixed by Assumption~\ref{item:intro_dbar_adapted}, and by choosing positions of black vertices in Assumption~\ref{item:intro_Gast_t-embedding}.

Before we dive further into the details, let us mention that the sequence of graphs constructed in Example~\ref{intro_example:triangular_graphs} is already sufficient for the reconstruction result needed for Berestycki, Laslier, and Ray work~\cite{BerestyckiLaslierRayI, BerestyckiLaslierRayII}, which is the main motivation of this paper.

\subsubsection{T-embedding, its origami map and a circle pattern}
\label{subsubsec:intro_t-embedding_prelim}
We now dive into details. Before we begin listing the assumptions, we need to recall the notion of a \emph{t-embedding} introduced in~\cite{CLR1}. Assume that $G$ is a bipartite planar graph. An embedding of its dual $G^\ast$ (endowed with the natural planar structure) into the plane is called a t-embedding if it is a proper embedding (i.e. there is no self-intersections and self-overlappings), all edges are mapped onto straight segments and the following \emph{angle condition} is satisfied: for each vertex $v$ of $G^\ast$, the sum of ``black'' angles (i.e. those contained in black faces) at this vertex is equal to the sum of ``white'' angles. Note that the embedding of $G^\ast$ does not need to have any geometric relation to any embedding of $G$; we only require that the planar structure of $G^\ast$ inherited from the embedding coincide with the planar structure on $G^\ast$ as of the dual graph to $G$. We address the reader to Section~\ref{subsec:t-embedding_def} for more details.

A t-embedding of $G^\ast$ into $\CC$ gives a rise to an \emph{origami map} $\Oo:\CC\to \CC$ introduced in~\cite{CLR1}. The map $\Oo$ is defined inductively: one first declare it to be an identity on some white face, and then, each time one crosses an edge, one applies a reflection with respect to the image of this edge under the origami map, which is already defined up to this edge. In other words, origami map corresponds to the procedure of folding the plane along edges of $G^\ast$. The angle condition is necessary and sufficient for this procedure to be consistent; we refer the reader to Definition~\ref{defin:of_origami} for the details. Note that the origami map is defined up to a translation and a rotation.

Following~\cite{KLRR} we can associate a \emph{circle pattern} with any t-embedding of $G^\ast$. To this end, let us pick an arbitrary vertex $b_0$ from $G$ and identify it with a point in $\CC$. The angle condition ensures then that we can map all other vertices of $G$ into $\CC$ inductively by requiring that for any pair of incident vertices $b$ and $w$ the corresponding points on the plane are symmetric with respect to the line passing through the edge of the t-embedding dual to $wb$. It is easy to verify that for any vertex $v$ of the t-embedding all the vertices of $G$ incident to the corresponding face of $G$ lie on a circle with the center at $v$. Note also that, if we connect vertices of the circle pattern by straight segment according to the combinatorics of $G$, than we do not necessary get a proper embedding of $G$ into plane.

We now introduce the class of t-embeddings we will be working with. Let $0<\lambda<1$ and $\delta>0$ be given. We call a t-embedding of $G^\ast$ \textbf{weakly uniform} (with respect to these parameters) if the following conditions are satisfied:
\begin{enumerate}
  \item\label{item:full-plane_bdd_density} For any $r>0$ and $z\in \CC$ the number of vertices of $G^\ast$ in the disc $B(z,r)$ is at most $\lambda^{-1}r\delta^{-2}$.

  \item Each black face of the t-embedding is $\lambda\delta$-fat, i.e. contains a disc of radius $\lambda\delta$.
  \item\label{item:full-plane_white_angles} Each white face of the t-embedding has all its angles bounded from $0$ and $\pi$ by $\lambda$.
    
  \item\label{item:full-plane_unif_black_face} Each white face of the t-embedding has a black neighbor with all its edges of length at least $\lambda\delta$.
\end{enumerate}

Note that all these assumptions would follow from a stronger assumption $\unif$, which forces all the edges of t-embedding to have comparable lengths and the angles to be uniformly bounded from $0$ and $\pi$. We, however, cannot afford having $\unif$, as will be clear from the crucial Example~\ref{intro_example:triangular_graphs}.

Besides these ``soft'' geometric properties needed for the application of the regularity theory from~\cite{CLR1}, there is another important characteristics of a t-embedding, namely the asymptotic behaviour of the origami map $\Oo$ in the small mesh size limit. This asymptotics determines the asymptotic properties of discrete holomorphic functions on t-embeddings: depending on the small mesh size limit of origami maps discrete holomorphic functions may converge to functions which are or are not holomorphic.
Thus, if we want our graphs to approximate the conformal structure of the plane, we need to take care of the limit origami maps. To this end we impose the very strong \textbf{$O(\delta)$-small origami} assumption: it requires the origami map $\Oo$ to satisfy
\begin{equation}
  \label{eq:intro_small_origami_assumption}
  |\Oo(z)| \leq \lambda^{-1}\delta
\end{equation}
holds for all $z$. Note that this assumption implies that the t-embedding has a circle pattern such that all the radii of the corresponding circles are at most $2\lambda^{-1}\delta$, see~\cite{KLRR}.

Finally, let us say that a t-embedding of a graph $G^\ast$ into $\CC$ is a \emph{full-plane} t-embedding if the union of its bounded faces exhausts $\CC$.

\subsubsection{Assumptions on the graph embedded into a Riemann surface}
\label{subsubsec:intro_graph_assumptions}

The definition of the notion of a t-embedding is purely local, hence it can be translated to the case when $G$ is embedded into a locally flat surface with conical singularities verbatim. In this case we additionally require that some of the vertices of $G^\ast$ are mapped to conical singularities, and that all face angles at these vertices are less than $\pi$. The procedure of constructing the origami map is still locally defined, but may not have a global extension. Fix a white face of $G^\ast$ and identify it with a polygon on the Euclidean plane isometrically and preserving the orientation. Repeating the inductive construction of the origami starting from this face we get a multivalued function $\Oo$ on $\Sigma$, having the same monodromy properties as $\Tt$ introduced above (one can make $\Oo$ to be a function of $\Tt$ actually). We call this function an origami map, keeping the ambiguity in its definition for the same reason we did it for $\Tt$.

Assume now that we are given a locally flat Riemann surface $(\Sigma, ds^2)$ with conical singularities and possibly with an involution $\sigma$ as constructed in Section~\ref{subsec:intro_flat_metric}, and a bipartite graph $G$ embedded into $\Sigma$. Let $0<\lambda<1$ and $0<\delta < \lambda^2$ be fixed, let $\Tt$ be defined by~\eqref{eq:def_of_Tt}. We say that $G$ is \textbf{$(\lambda, \delta)$-adapted} if the following conditions are satisfied:

\begin{enumerate}
  \item \label{item:intro_metric_assumptions} For any two distinct $i,j = 1,\dots, 2g-2$ we have $\dist(p_i,p_j)\geq 4\lambda$ and each loop on $\Sigma$ has the length at least $4\lambda$. Note in particular that for any point $p\in \Sigma$ situated at the distance at least $\lambda$ from conical singularities the restriction of $\Tt$ to the corresponding ball $B_\Sigma(p, \lambda)$ isometrically identifies it with a disc in $\CC$.
  \item \label{item:intro_G_well_embedded} All edges of $G$ are smooth curves of length at most $\lambda^{-1}\delta$.
  \item \label{item:intro_Gast_t-embedding} The graph $G^\ast$ is t-embedded into $\Sigma$. Each vertex of $G$ is at the distance at most $\lambda^{-1}\delta$ from the corresponding face of $G^\ast$ (note that we do not require the vertices of $G^\ast$ to belong to faces of $G$). 
  \item \label{item:intro_reg_part} For each point $p\in \Sigma$ at the distance at least $\lambda$ from the conical singularities, the map $\Tt$ restricted to the open disc $B_\Sigma(p, \lambda)$ identifies $G^\ast\cap B_\Sigma(p, \lambda)$ with a subgraph of a full-plane weakly uniform t-embedding having $O(\delta)$-small origami. Moreover, the image of black vertices $G\cap B_\Sigma(p, \lambda)$ under $\Tt$ are mapped to the vertices of a circle pattern associated with this t-embedding.
  \item \label{item:intro_conical_sing} For each $j = 1,\dots, 2g-2$ the graphs $B_\Sigma(p_j, 2\lambda)\cap G$ and $B_\Sigma(p_j, 2\lambda)\cap G^\ast$ are invariant under the rotation by $2\pi$ around $p_j$. The map $\Tt$ restricted to $B_\Sigma(p_j, 2\lambda)$ maps $B_\Sigma(p_j, 2\lambda)\cap G$ onto a subgraph of a full-plane Temperleyan isoradial graph (i.e. a superposition of an isoradial graph and its dual, see Figure~\ref{fig:graph_near_conical}), and each vertex of $B_\Sigma(p_j, 2\lambda)\cap G^\ast$ to the circumcenter of the corresponding face of $\Tt(G)$ (cf.~\cite[Section~8.3]{CLR1}).

  \item \label{item:intro_boundary_part} If the involution $\sigma$ is present, then both $G$ and $G^\ast$ are invariant under it. Moreover, the boundary arcs of $\Sigma_0$ are composed of edges of $G$ (the corresponding cycles are called \emph{boundary cycles}), do not contain vertices of $G^\ast$ and cross edges of $G^\ast$ perpendicularly.

  \item \label{item:intro_dbar_adapted} This is the most technical and restrictive condition, which can be thought of as a `good approximation' condition for the discrete Cauchy--Riemann operator. Let $w$ be an arbitrary white face of $G^\ast$ at the distance at least $\lambda$ from conical singularities. Let $b_1,\dots, b_k$ be the black vertices of $G$ inticident to $w$ and listed in the counterclockwise order, and let $v_i$ be the vertex of $G^\ast$ (a face of $G$) incident to $b_{i+1}$,$w$ and $b_i$. We require that
    \[
      \sum_{i = 1}^k (\Tt(v_i)-\Tt(v_{i-1}))\Tt(b_i) = 0.
    \]
\end{enumerate}

\begin{figure}
  \centering
  \includegraphics[clip, trim=0cm 0cm 0cm 0cm, width = 0.66\textwidth]{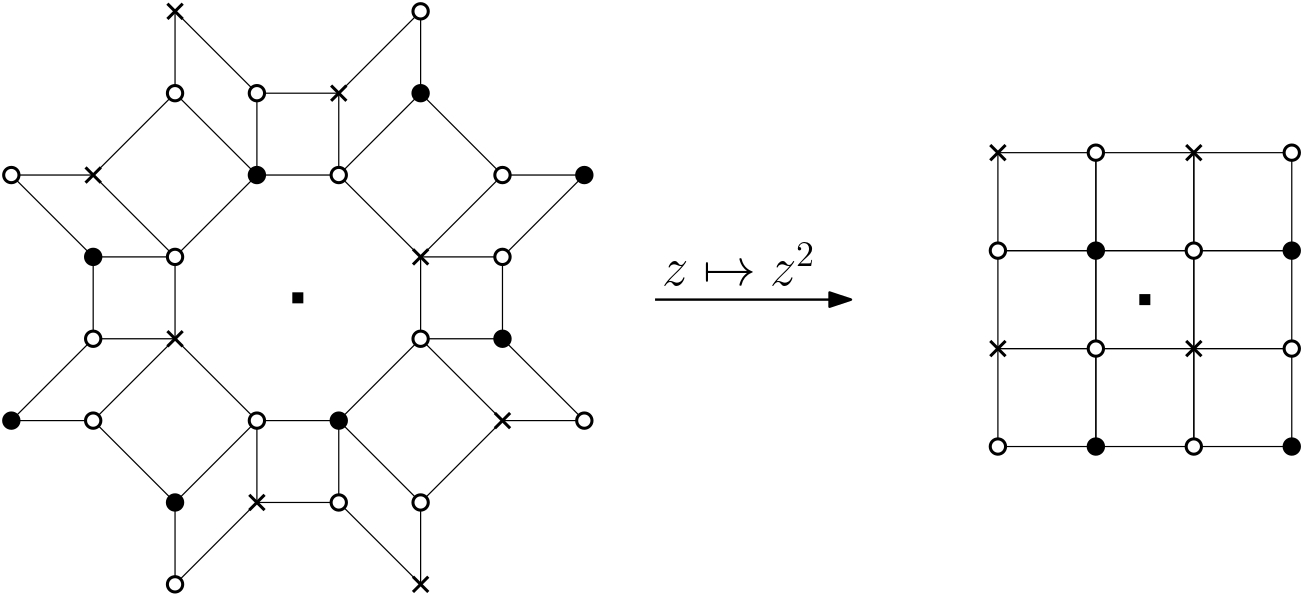}
  
  \vskip0.03\textwidth
  \includegraphics[clip, trim=0cm 0cm 0cm 0cm, width = 0.86\textwidth]{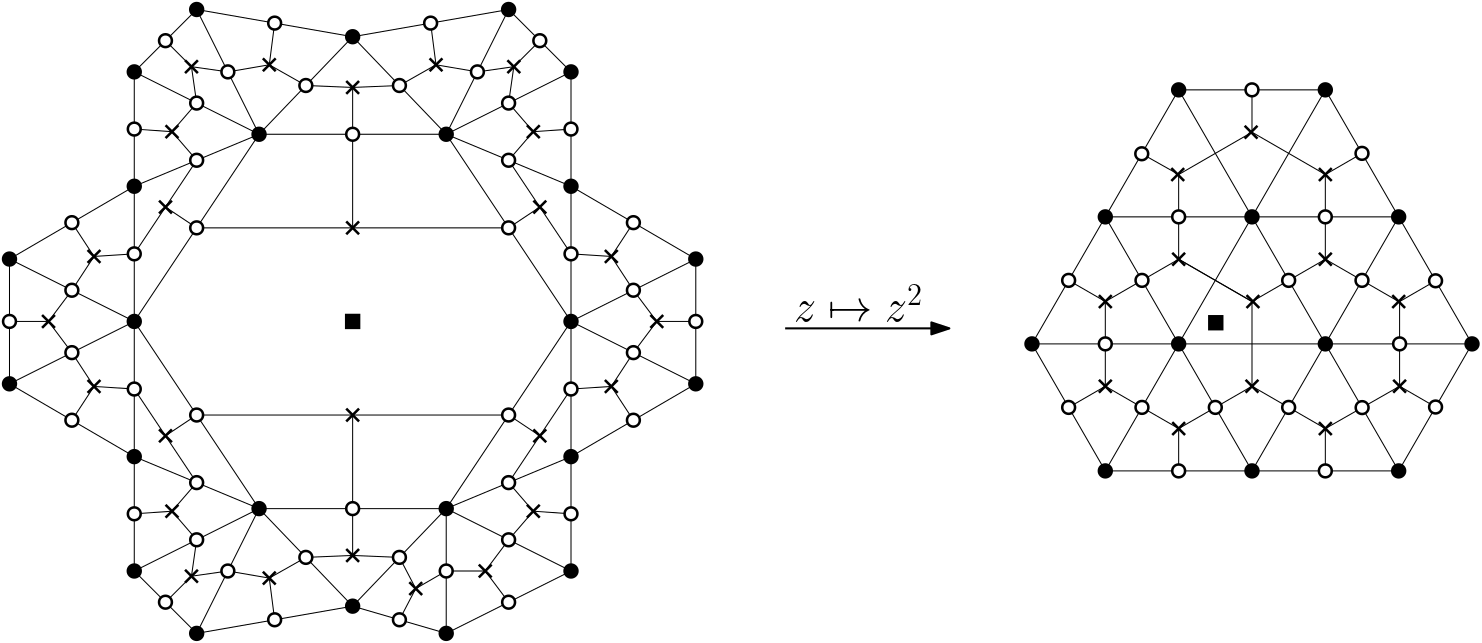}

  \caption{Two examples of a double cover of a Temperleyan isoradial graph branched around a circumcenter of a face. On the fist picture we have the superposition of two square lattices (a square lattice again), the second one corresponds to the superposition of a triangular lattice an its dual hexagonal lattice. In both cases the double cover is drawn only schematically.}
  \label{fig:graph_near_conical}
\end{figure}

\noindent If the assumptions above are satisfied, then we assign $\delta$ to be the mesh size of $G$. We put
\[
  G_0 = G\cap \Sigma_0\smm\partial \Sigma_0.
\]
We say that $G_0$ \emph{is obtained from} $G$.

We conclude this section by specifying edge weights of $G$ given it satisfies assumptions above. Let $wb$ be an edge of $G$ and let $v_1v_2$ be the dual edge of $G^\ast$. Then we set
\begin{equation}
  \label{eq:def_of_w}
  \w(wb) = \dist(v_1,v_2).
\end{equation}

\subsection{Dimer model on a graph and its double}
\label{subsec:intro_dimer_model}

Let $G$ be a finite bipartite graph, $B,W$ denote the bipartite classes of its vertices. A dimer cover of $G$ is a subset of edges of $G$ covering each vertex exactly ones --- in other words, dimer cover is a perfect matching of $G$. Let $\w: \mathrm{Edges}(G)\to \RR_{>0}$ be a weight function and assume that $G$ admits at least one dimer cover. The dimer model on $G$ assigns a probability measure on the set of all dimer covers defined by
\[
  \PP[D] = \frac{1}{\Zz_G} \prod_{wb\in D}\w(wb).
\]
where
\begin{equation}
  \label{eq:def_of_Zz}
  \Zz_G = \sum_{D\text{ - dimer cover of }G}\prod_{wb\in D}\w(wb).
\end{equation}

Assume now that $G$ is embedded into a surface $\Sigma$ and is $(\lambda,\delta)$ adapted, i.e. satisfies all the assumptions from Section~\ref{subsec:intro_graphs_on_Sigma0}. Let $G_0$ be the graph on $\Sigma_0$ obtained from $G$. The main objective of our work is to analyze the dimer model on $G_0$ with respect to the weight function~\eqref{eq:def_of_w}. However, we prefer to work with dimer covers of $G$, even if $\partial \Sigma_0\neq \varnothing$ and so $G\neq G_0$. In the latter case the invariance of $G$ under the involution $\sigma$ allows us to extend dimer covers on $G_0$ to dimer covers on $G$. For this, fix a dimer cover $E$ of the union of boundary cycles and for each dimer cover $D_0$ of $G_0$ define $D = D_0\cup E\cup \sigma(D_0)$. This gives a bijection between dimer covers of $G_0$ and symmetric dimer covers of $G$ containing $E$. Given $E$, we define the $\sigma$-invariant dimer model on $G$ to be the probability measure on dimer covers $D$ of the form above which comes from the dimer model on $G_0$ via the aforementioned bijection. In what follows we will be working with the dimer model on $G$ which is $\sigma$-invariant if $\sigma$ is present. We will specify $E$
 later, when it will be convenient for us to make a particular choice.

\subsection{Kasteleyn operator for \texorpdfstring{$G$}{G}}
\label{subsec:intro_Kasteleyn_operator}

An important technical role in analyzing the dimer model is played by a Kasteleyn operator $K$ associated with the graph $G$. Recall that a Kasteleyn operator of a weighted bipartite graph embedded into a surface is any matrix $(K(w,b))_{w\in W, b\in B}$ such that 
\[
  |K(w,b)| = \begin{cases}
    \w(wb),\quad w\sim b,\\
    0,\quad \text{else,}
  \end{cases}
\]
and such that for any face $v$ of $G$ with the boundary vertices $b_1,w_1,b_2,\dots,b_k,w_k$ listed consequently we have
\begin{equation}
  \label{eq:Kasteleyn_condition}
  \prod_{j = 1}^k \frac{K(w_j,b_j)}{K(w_j, b_{j+1})}\in (-1)^{k+1}\RR_{>0};
\end{equation}
this relation is called the \emph{Kasteleyn condition}. Note that the definition of $K$ implies that for any dimer cover $D_0$ of $G$ and a random dimer cover $D$ sampled with respect to the dimer model associated with weights on $G$ we have
\[
  \PP[D = D_0] \sim \prod_{wb\in D_0} |K(w,b)|
\]
where $\sim$ means equality up to a multiplicative constant which does not depend on $D_0$.

Note that any t-embedding of the dual of a planar bipartite graph induces a natural choice of the Kasteleyn operator on it according to~\cite{CLR1} (generalizing the classical `critical' Kasteleyn operator on isoradial graphs~\cite{KenyonCriticalPlanarGraphs}). This construction however makes use of a local coordinate on the plane that create a certain difficulty. 

Assume first that the metric $ds^2$ has trivial holonomy. Let $G$ be a $(\lambda,\delta)$-adapted graph on $\Sigma$. Given an edge $wb$ of $G$ let $v_1v_2$ be the dual edge of $G^\ast$ oriented such that $w$ is on the left. Then the recipe in~\cite{CLR1} suggests us to define $K(w,b) = \Tt(v_2) - \Tt(v_1)$, and as soon as the holonomy of the metric on $\Sigma$ is trivial this definition is consistent and $K$ satisfies Kasteleyn condition around each face of $G$ which does not contain a conical singularity, while at conical singularities the sign of the alternating product is opposite to what is required. To get rid of the latter inconsistency we modify the definition of $K$ as follows. Let $\gamma_1,\dots, \gamma_{g-1}$ be simple non-intersecting paths on $G^\ast$ connecting $p_1,\dots, p_{2g-2}$ pairwise. If $\partial \Sigma_0\neq \varnothing$, then we assume that $\sigma(\gamma_i) = \gamma_i$. We now replace $K(w,b)$ with $-K(w,b)$ if $wb$ intersects $\gamma_i$ for some $i = 1,\dots, 2g-2$. In this way $K$ will satisfy the Kasteleyn condition everywhere.

Assume now that the holonomy of $ds^2$ is non-trivial. In this case the definition of $K$ introduced above is not consistent along non-contractible loops as $\Tt$ may have a multiplicative monodromy along them. To compensate this monodromy we can modify the definition of $K$ as follows. Recall the $(0,1)$-form $\alpha_0$ introduces in Section~\ref{subsec:intro_flat_metric}. We put
\[
  \tilde{K}(w,b) = \begin{cases}
    \exp(i\Im (\int_{p_0}^b\alpha_0 + \int_{p_0}^w\alpha_0))(\Tt(v_2) - \Tt(v_1)),\quad bw\text{ not intersect }\cup_{i=1}^{2g-2}\gamma_i,\\
    -\exp(i\Im (\int_{p_0}^b\alpha_0 + \int_{p_0}^w\alpha_0))(\Tt(v_2) - \Tt(v_1)),\quad \text{else,}
  \end{cases}
\]
where the paths of integration in the exponent and in the definition of $\Tt$ are chosen to be the same (in particular, $p_0$ is the base point used to define $\Tt$). It is easy to verify that the operator $\tilde{K}$ is a well-defined and satisfies the Kasteleyn condition around each face. To motivate this definition (which might look unnecessary bulky on a first glance), let us mention that a straightforward computation (see Lemma~\ref{lemma:Kalpha_and_dbar} and also~\cite[eq.~(3.15)]{DubedatFamiliesOfCR}) can relate $\tilde{K}$ with the perturbed Cauchy--Riemann operator $\dbar + \frac{\alpha_0}{2}$ acting on smooth functions on $\Sigma$ with square root singularities at $p_1,\dots, p_{2g-2}$. This will allow us later to interpret functions $f$ on black vertices of $G$ as spinors of the form $f\sqrt{\omega_0}$.
Since $(\dbar -\alpha_0)\omega_0 = 0$ (see Section~\ref{subsec:intro_flat_metric}), $f\sqrt{\omega_0}$ is holomorphic if and only if $(\dbar + \frac{\alpha_0}{2})f = 0$. 

Yet, we need to modify the definition of $\tilde{K}$ further. Recall that two Kasteleyn operators $K_1,K_2$ are said to be \emph{gauge equivalent} if there exist functions $U_1,V_1$ such that $K_1(w,b) = U_1(w)K_2(w,b)V_1(b)$ for each $w,b$ (see Section~\ref{subsec:intro_ratio_of_partition_functions} for more comments). We want our Kasteleyn operator $K$ to be gauge equivalent to a real one. The reason for this is the interpretation of $K$ as a \emph{Dirac} operator with holomorphic and anti-holomorphic parts. Having $K$ to be gauge equivalent to a real operator is the same as to have the holomorphic and anti-holomorphic parts of $K$ to act on the same function space (sections of the same bundle).

Requiring $K$ to be gauge equivalent to a real operator is the same as requiring that for any loop $b_1w_1b_2\ldots b_kw_kb_1$ we have
\[
  \left(\prod_{j = 1}^k \frac{K(w_j, b_j) }{|K(w_j, b_j)|}\cdot \frac{|K(w_j,b_{j+1})|}{K(w_j,b_{j+1})}\right)^2 = 1.
\]
If $K$ satisfies the Kasteleyn condition, then the expression above does depend on the homology class of the loop only, hence it defines a cohomology class in $H^1(\Sigma, \TT^*)$. It follows that, given $\tilde{K}$ defined as above, we can find an anti-holomorphic $(0,1)$-form $\alpha_G$ such that the operator
\begin{equation}
  \label{eq:def_of_K}
  K(w,b) = \exp\left( 2i\int_w^b \Im \alpha_G \right)\tilde{K}(w,b)
\end{equation}
has this cohomology class equal to zero, which makes it to be gauge equivalent to a real operator. In the case when $\partial \Sigma_0\neq \varnothing$ we can choose $\alpha_G$ to satisfy $\sigma^*\alpha_G = \bar{\alpha}_G$. This finally fixes our choice of $K$.

Note that $\alpha_G$ is not unique: we can replace it with any other anti-holomorphic $(0,1)$-form $\alpha_G^{(1)}$ such that $2\pi^{-1}\Im (\alpha_G - \alpha_G^{(1)})$ has integer cohomologies (we also have to keep symmetry with respect to $\sigma$ is the latter is present). The way equivalence class of $\alpha_G$ behaves asymptotically will determine the instanton component of the compactified free field that occurs in the scaling limit.

Let $K$ be defined by~\eqref{eq:def_of_K}, then we can find a function $\eta$ on vertices of $G$ with values in $\TT$ such that:
\begin{enumerate}
  \item For each edge $b\sim w$ we have $(\bar{\eta}_b\bar{\eta}_w)^2 = \frac{K(w,b)^2}{|K(w,b)|^2}$;
  \item If $\partial \Sigma_0\neq \varnothing$, then for any $b\in B$ we have $\eta_{\sigma(b)} = \bar{\eta}_b$ and for any $w\in W$ we have $\eta_{\sigma(w)} = - \bar{\eta}_w$.
\end{enumerate}
Following~\cite{CLR1} we call $\eta$ an \emph{origami square root function}.

\subsection{Examples of \texorpdfstring{$\Sigma$}{Sigma} and \texorpdfstring{$G$}{G} satisfying the assumptions of main results}
\label{subsec:intro_examples}

Before we proceed we would like to present several examples of $(\lambda,\delta)$-adapted graphs. In particular, Example~\ref{intro_example:triangular_graphs} will be crucial for the relation between our results and~\cite{BerestyckiLaslierRayI, BerestyckiLaslierRayII}.

\subsubsection{Torus}
\label{intro_example:torus}

Let $\Sigma_\Lambda = \CC/\Lambda$ where $\Lambda = a\ZZ + b\ZZ$ for some $a,b\in \CC$ such that $\Im (b\bar{a}) > 0$. Note that the holonomy of the natural flat metric on $\Sigma$ is trivial, so that we may take $\alpha_0 = 0$; however $\alpha_G$ introduced in Section~\ref{subsec:intro_Kasteleyn_operator} might be non-zero. To demonstrate this we consider the example of a hexagonal lattice. Put
\[
  \Lambda = \ZZ + e^{\pi i/3}\ZZ,
\]
let $N>0$ be an integer and let $G^\ast$ be the graph whose vertex set is $N^{-1}\Lambda$ and $v_1,v_2$ are connected by an edge if and only if $v_1-v_2\in \{ \pm N^{-1}, \pm N^{-1}e^{\pi i/3}, \pm N^{-1}e^{2\pi i/3} \}$; in this way $G^\ast$ becomes the triangular lattice. Let $G$ be the hexagonal lattice dual to $G$; declare the vertex $\frac{e^{\pi i/6}}{N\sqrt{3}}$ of $G$ to be white for definiteness. In this case we can put 
\[
  \alpha_G = -N\left( \frac{\pi}{6\sqrt{3}} - \frac{i\pi}{6} \right)\,d\bar{z}
\]
Note that the integrals of $2\pi^{-1}\Im \alpha_G$ along basis cycles of $\Sigma_{\Lambda}$ are $\frac{N}{3}$ and $\frac{2N}{3}$; in particular, $\alpha_G$ cannot be replaced by zero if $N$ is not divisible by 3.

\subsubsection{Pillow surface}
\label{intro_example:pillow_surface}

\begin{figure}
  \centering
  \includegraphics[clip, trim=0cm 0cm 0cm 0cm, width = 0.56\textwidth]{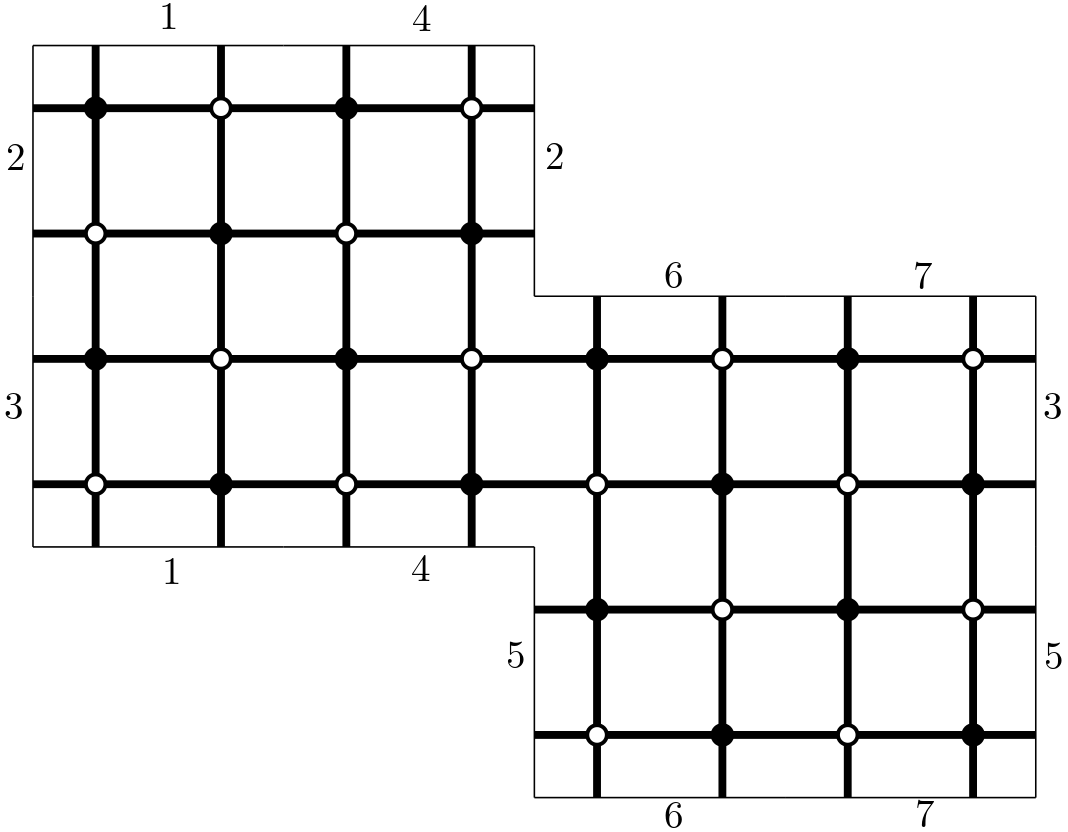}
  \caption{Pillow surface of genus 2 with two conical singularities glued from 8 squares, and the square lattice on it. Labels on outer sides of squares indicate pairs of sides glued together.}
  \label{fig:pillow_surface}
\end{figure}

Let $T = \CC/\Lambda$ be the corresponding torus. Assume that $f:\Sigma\to T$ is a ramified cover branching over one point only (say, over $0 + \Lambda\in T$) only. A Riemann surface $\Sigma$ endowed with such a cover is called a ``pillow surfaces'' (see e.g.~\cite{eskin2010lyapunov}). The flat metric on $T$ pulls back to a locally flat metric on $\Sigma$ with conical singularities at ramification points. Assuming that all the ramifications are simple we ensure that all cone angles at the conical singularities are equal to $4\pi$; dividing the metric by $\deg f$ we make it to have a unit area. Note that the holonomy of the metric on $\Sigma$ is trivial in this case. Given $n>0$ we put
\[
  G_{\text{base}}^n = \frac{1+i}{4n} + \frac{1}{2n}\ZZ^2
\]
and let $G^n = f^{-1}(G_{\text{base}}^n)$. Clearly, this choice makes $G^n$ to be $(\lambda, \frac{1}{n\sqrt{\deg f}})$-adapted, where $\lambda$ is a small enough constant depending at the distances between the conical singularities only.

Note that any pillow surface can be obtained by gluing squares $[0,1]^2$ along sides in any such a way that gives an oriented surface with trivial holonomy. Then $G_{\text{base}}^n$ is obtained by taking a piece of the square grid in each square and gluing them together, see Figure~\ref{fig:pillow_surface} for an example.

One can show that any locally flat surface with conical singularities of cone angles multiple of $2\pi$ and \emph{trivial} holonomy can be approximated by a sequence of pillow surfaces in the topology of the corresponding moduli space. One can also generalize this example to the case of a surface with a boundary. However, we can only produce examples of surfaces with trivial holonomy in this way.

\subsubsection{Temperleyan graphs corresponding to Delaunay triangulations}
\label{intro_example:triangular_graphs}

In this example we aim to construct a sequence of adapted graphs on \emph{any given} locally flat surface. Let $(\Sigma, ds^2)$ be as in Section~\ref{subsec:intro_flat_metric}. We choose a small $0<\lambda<1$ small enough, in particular, such that $(\Sigma, ds^2)$ satisfies Assumption~\ref{item:intro_metric_assumptions} from Section~\ref{subsec:intro_graphs_on_Sigma0}, and $0 < \delta < \lambda^2$. We also assume that the involution $\sigma$ is present; if it is not, then one can apply the same construction as below just omitting the boundary part.

We begin by fixing the graph $G$ and the embedding of its dual near the conical singularities and near the boundary of $\Sigma_0$. Replacing $\lambda$ with $\lambda/3$ if necessary, we may assume that distances between conical singularities and boundary components are at least $5\lambda$. For each $j = 1,\dots, 2g-2$ we choose a local coordinate $z_j$ at $p_j$ such that $z_j(p_j) = 0$ and $ds^2 = |d(z_j^2)|$. We think of $z_j$ as of an isometry between $B_\Sigma(p_j, 2\lambda)$ and the corresponding cone. For each $j = 2g-1,\dots, 2g-2+n$ we choose a holomorphic isometry $z_j$ between a $2\lambda$-neighborhood of the $(j+1-2g)$-th boundary component and the cylinder $\{ z\in \CC\ \mid\ |\Im z|<2\lambda \}/_{z\mapsto z+l_j}$, where $l_j$ is the length of the component. We define graphs $\Gamma_j, \Gamma_j^\dagger, G_j$ and $G_j^\ast$ as follows.

-- For $j = 1,\dots, 2g-2$ we put $\Gamma_{j,0}$ to be the properly shifted and rescaled triangular lattice:
\[
  \Gamma_{j,0} = \delta\ZZ + \delta e^{\pi i/3}\ZZ - \frac{\delta e^{\pi i/6}}{\sqrt{3}},
\]
and $\Gamma_{j,0}^\dagger$ to be the dual hexagonal lattice. Let $G_{j,0}$ be the superposition graph of $\Gamma_{j,0}$ and $\Gamma_{j,0}^\dagger$. Note that $G_{j,0}$ is Temperleyan isoradial and the origin is the circumcenter of one of its faces. We put $G_{j,0}^\ast$ to be the dual to $G_{j,0}$ embedded by circumcenters of the faces. Finally, we put
\[
  \Gamma_j = (z_j^2)^{-1}(\Gamma_{j,0}),\quad \Gamma_j^\dagger = (z_j^2)^{-1}(\Gamma_{j,0}^\dagger) \qquad G_j = (z_j^2)^{-1}(G_{j,0}),\quad G_j^\ast = (z_j^2)^{-1}(G_{j,0}^\ast).
\]

-- For $j = 2g-1,\dots, 2g-2+n$ we put 
\[
  \Gamma_{j,0} = l_j\lfloor \delta^{-1}\rfloor^{-1}\ZZ + l_j\lfloor \delta^{-1}\rfloor^{-1} e^{\pi i/3}\ZZ,
\]
and then define $\Gamma_{j,0}^\dagger, G_{j,0}, G_{j,0}^\ast$ as above. We put
\[
  \Gamma_j = z_j^{-1}(\Gamma_{j,0}),\qquad \Gamma_j^\dagger = z_j^{-1}(\Gamma_{j,0}^\dagger) \qquad G_j = z_j^{-1}(G_{j,0}),\qquad G_j^\ast = z_j^{-1}(G_{j,0}^\ast).
\]

Now, let us construct the graphs $\Gamma, \Gamma^\dagger$ and $G$ on $\Sigma$. The construction is drawn schematically on Figure~\ref{fig:construction_of_Gamma}. It can be described as follows.

First, let us complete the vertices of $\Gamma_1,\dots, \Gamma_{2g-2+n}$ to a $\lambda^{-1}\delta$-net $\Gamma$ on $\Sigma$. This can be done such that the following conditions are satisfied:

-- No points inside the $2\lambda$-neighborhood of the conical singularities and boundary components were added.

-- Making $\lambda$ small enough we can ensure that $\Gamma$ is a $\lambda\delta$-separated set (first complete $\Gamma_1,\dots, \Gamma_{2g-2+n}$ to an arbitrary $\lambda^{-1}\delta$-net $\Gamma$; if there is a point $v$ in $\Gamma\smm(\Gamma_1\cup\dots\cup \Gamma_{2g-2+n})$ that has another point of $\Gamma$ in its $\lambda^{-1}\delta$ neighborhood, then remove $v$ from $\Gamma$; repeat this procedure until one ends up with a $\lambda\delta$-separated set).

-- The Voronoi diagram associated with $\Gamma$ has all its vertices of degree 3 (indeed, note that the Voronoi diagram of a generic set on the plane has this property; since the definition of the Voronoi diagram is local it is enough to choose the points complementing $\Gamma_1,\dots, \Gamma_{2g-2+n}$ generically).

-- We have $\sigma(\Gamma) = \Gamma$ if $\partial\Sigma_0\neq \varnothing$ and $\sigma$ is present (e.g. choose $\Gamma\cap\Sigma_0$ first and then reflect it to $\Sigma_0^\op$).

Now, let us extend the graph $\Gamma_1^\dagger\cup\ldots\cup \Gamma_{2g-2+n}^\dagger$ to the whole $\Sigma$. To this end, lets define $\Gamma^\dagger$ to coincide with $\Gamma_1^\dagger\cup\ldots\cup \Gamma_{2g-2+n}^\dagger$ in the $\lambda$-neighborhood of the conical singularities and boundary components, and to be the Voronoi diagram of $\Gamma$ elsewhere. Note that $\Gamma^\dagger$ is in fact the Voronoi diagram of $\Gamma$ everywhere outside an $O(\delta)$-neighborhood of conical singularities.

Note that each face of $\Gamma^\dagger$ that does not contain a conical singularity has exactly one point from $\Gamma$ inside. In particular, we can endow $\Gamma$ with a graph structure by declaring it to be dual to $\Gamma^\dagger$ outside of the conical singularities; this definition agrees with $\Gamma_1\cup\ldots\cup \Gamma_{2g-2+n}$. The embedding of $\Gamma$ into $\Sigma$ is specified by mapping edges to geodesic. In this way $\Gamma$ becomes the Delaunay triangulation of the corresponding vertex set (outside the conical singularities), and the vertices of $\Gamma^\dagger$ become the circumcenters of the corresponding triangles. It might happen that some of the triangles are obscure, and the corresponding vertices of $\Gamma^\dagger$ are not lying inside the triangles. Despite this, the ``superposition'' graph $G$ of $\Gamma$ and $\Gamma^\dagger$ can be properly defined as follows. Let the black vertices of $G$ be the vertices of $\Gamma$ and $\Gamma^\dagger$, and the white vertices be the midpoints of the edges of $\Gamma^\dagger$. The edges of $G$ include

-- in the $\lambda$-neighborhood of conical singularities and boundary components: all edges of $G_1\cup\ldots\cup G_{2g-2+n}$ .

-- outside the $\lambda$-neighborhood of conical singularities and boundary components: half-edges of $\Gamma^\dagger$ and straight segments connecting the midpoints of edges of $\Gamma^\dagger$ and the incident dual vertices of $\Gamma$.

Let us emphasize that $G$ coincides with $G_1\cup\ldots\cup G_{2g-2+n}$ in the $\lambda$-neighborhood of conical singularities and boundary components. 

Finally, let us define the t-embedding of $G^\ast$. This definition is schematically shown on Figure~\ref{fig:construction_of_t-embedding}. In a $\lambda$-neighborhood of conical singularities and the boundary of $\Sigma_0$ it is already given by $G_1^\ast\cup\ldots\cup G_{2g-2+n}^\ast$. It can be extended outside as follows. Note that each face of $G$ outside a conical singularity has degree 4, and there is a correspondence between such faces and `corners' of $\Gamma$, that is, pairs $b\sim b^\dagger$ where $b$ is a vertex of $\Gamma$ and $b^\dagger$ is a vertex of $\Gamma^\dagger$. For each pair $b,b^\dagger$ of incident vertices of $\Gamma$ and $\Gamma^\dagger$, we put a vertex of $G^\ast$ at the middle of the geodesic $bb^\dagger$. Then, we draw a straight segment per each edge of $G^\ast$.

\begin{figure}
  \centering
  \begin{minipage}{.3\textwidth}
    \includegraphics[clip, trim=0cm 0cm 0cm 0cm, width = 0.96\textwidth]{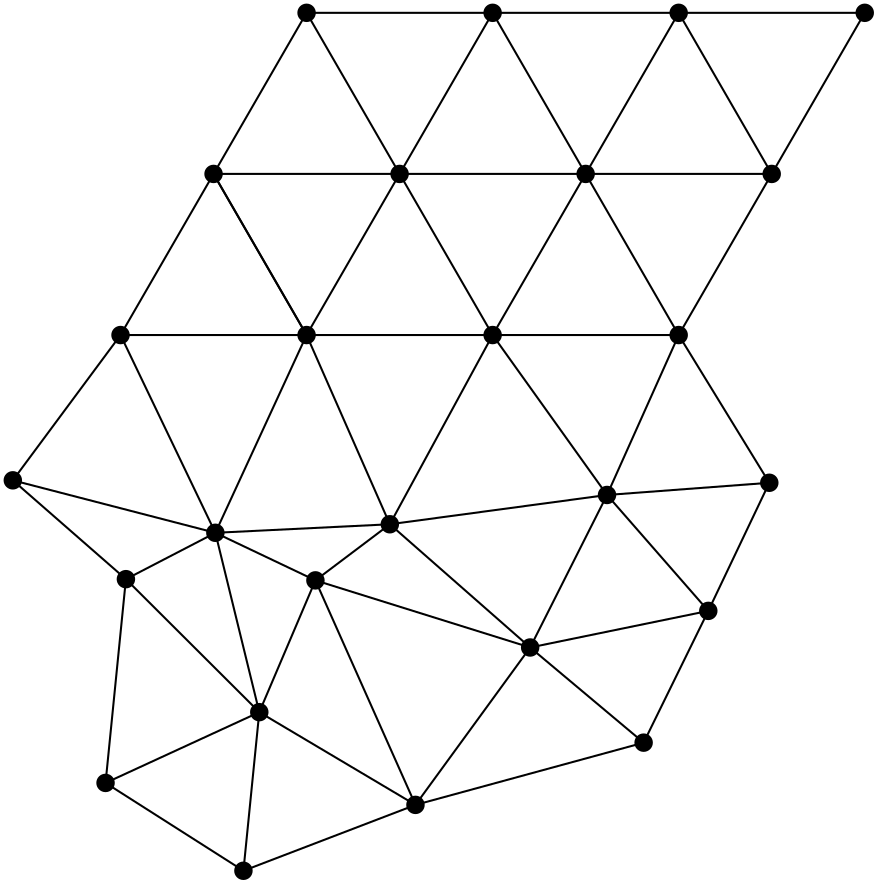}

    To get $\Gamma$: extend the triangular lattice pattern by a Delaunay triangulation
  \end{minipage}
  \hskip 0.03\textwidth\begin{minipage}{.3\textwidth}
    \includegraphics[clip, trim=0cm 0cm 0cm 0cm, width = 0.96\textwidth]{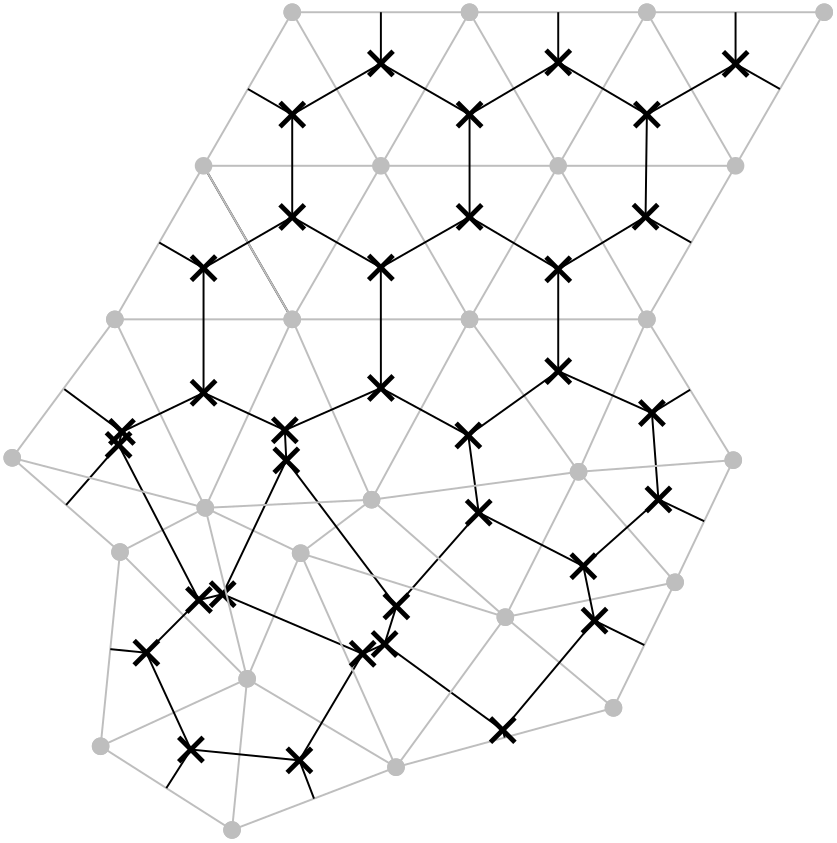}

    To get $\Gamma^\dagger$: draw the Voronoi diagram for $\Gamma$ outside the conical singularities
  \end{minipage}
  \hskip 0.03\textwidth\begin{minipage}{.3\textwidth}
    \includegraphics[clip, trim=0cm 0cm 0cm 0cm, width = 0.96\textwidth]{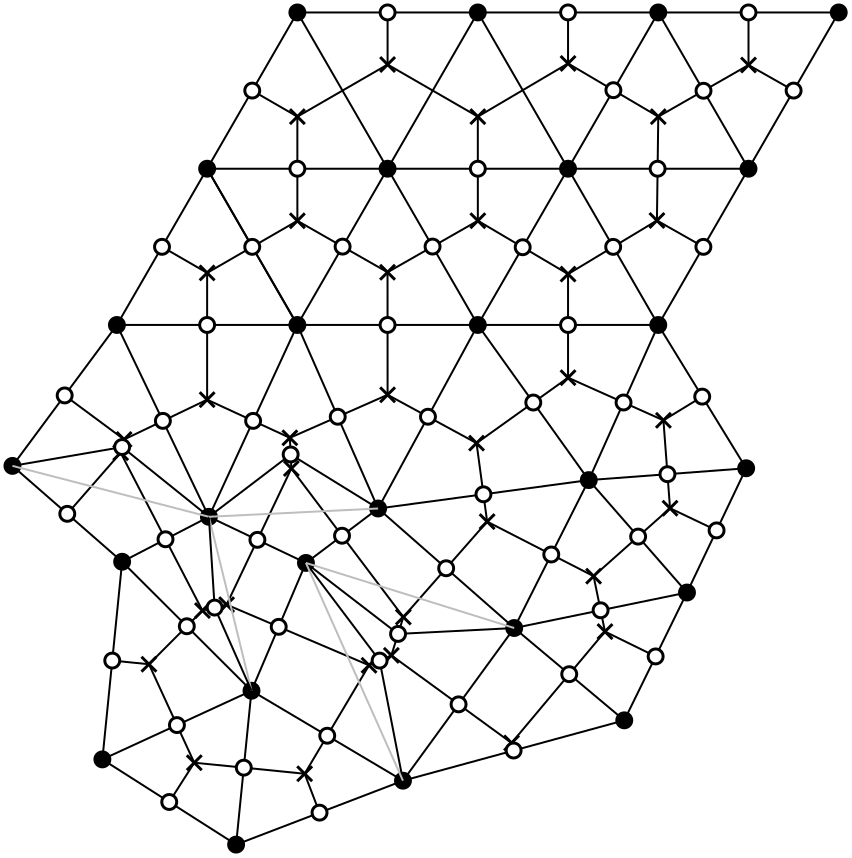}

    To get $G$: put a white vertex per edge of $\Gamma^\dagger$, connect blacks with whites by straight lines
  \end{minipage}
  \caption{The process of extending $\Gamma$, $\Gamma^\dagger$ and $G$ outside conical singularities. Note that dual edges of $\Gamma$ and $\Gamma^\dagger$ dot not always intersect.}
  \label{fig:construction_of_Gamma}
\end{figure}

\begin{figure}
  \centering
  \begin{minipage}{.3\textwidth}
    \includegraphics[clip, trim=0cm 0cm 0cm 0cm, width = 0.96\textwidth]{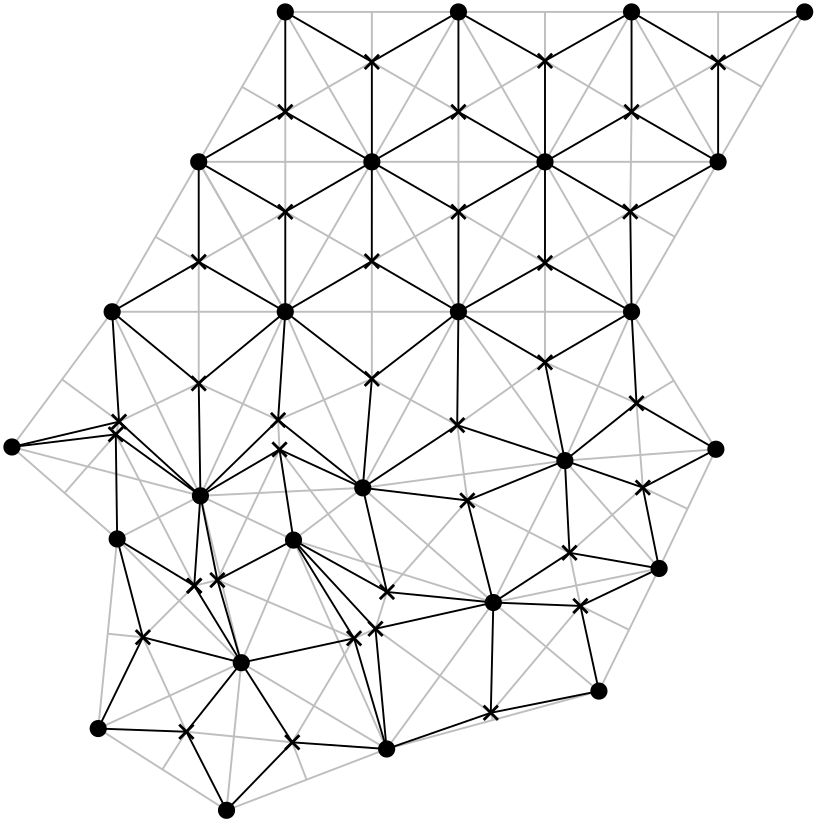}

    Connect each vertex of $\Gamma^\dagger$ to the three vertices of $\Gamma$ nearest to it. Call these new segments ``radius edges''
  \end{minipage}
  \hskip 0.03\textwidth\begin{minipage}{.3\textwidth}
    \includegraphics[clip, trim=0cm 0cm 0cm 0cm, width = 0.96\textwidth]{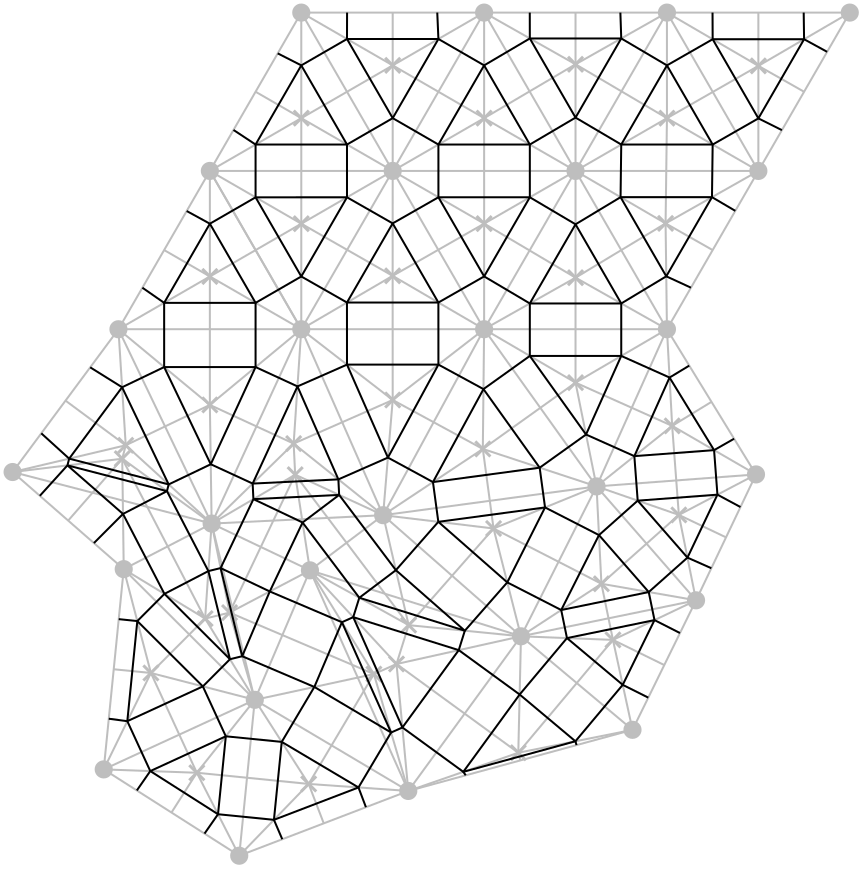}

    Midpoints of radius edges are vertices of $G^\ast$. To draw the edges, connect two vertices if the radius edges are incident
  \end{minipage}
  \hskip 0.03\textwidth\begin{minipage}{.3\textwidth}
    \includegraphics[clip, trim=0cm 0cm 0cm 0cm, width = 0.91\textwidth]{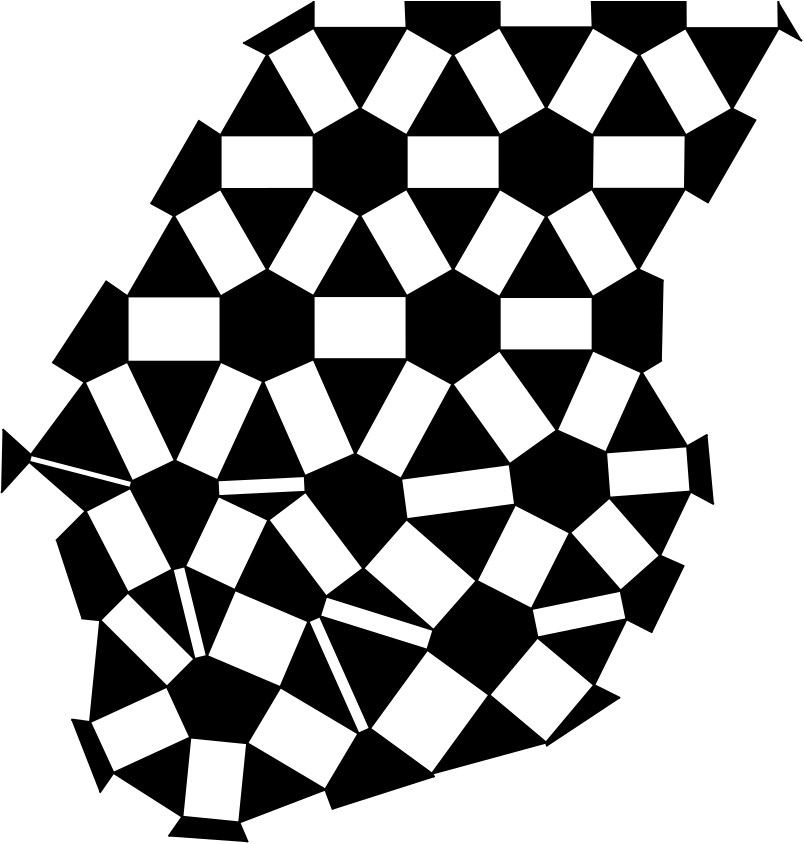}

    Black faces of $G^\ast$ filled. Note that white faces can be very thin, and black faces correspond to faces of $\Gamma$ and $\Gamma^\dagger$
  \end{minipage}
  \caption{The process of extending the t-embedding of $G^\ast$ outside conical singularities.}
  \label{fig:construction_of_t-embedding}
\end{figure}

Now, having defined $G$ and the embedding of its dual, we should verify that $G$ is $(\lambda,\delta)$-adapted. We assume that $\delta$ is small enough, but otherwise arbitrary, and $\lambda$ is small enough, depending on $(\Sigma,ds^2)$ only. We go through the assumptions from Section~\ref{subsec:intro_graphs_on_Sigma0} consequently.

Assumptions~\ref{item:intro_metric_assumptions} and~\ref{item:intro_G_well_embedded} is satisfied by the definition of $\lambda$.

Assumption~\ref{item:intro_Gast_t-embedding}. We first need to show that the embedding of $G^\ast$ is proper. Clearly, it is a local statement, and we need to verify it only outside the $\lambda$-neighborhood of conical singularities. To this end, notice that outside conical singularities each black face $b$ of $G^\ast$ corresponding to a vertex of $\Gamma$ (resp. $\Gamma^\dagger$) is the rescaling by 2 and translation of the corresponding dual face of $\Gamma^\dagger$ (resp. the corresponding triangular face of $\Gamma$); in the same time each white face is a rectangle. Note that the orientation of each face inherited from the embedding of $G^\ast$ coincides with the orientation inherited from the embedding of $G$ (and the fact that $G^\ast$ is dual to $G$). Finally, note that the sum of oriented angles at each vertex of $G^\ast$ except those which are conical singularities is $2\pi$, and there are exactly two white angles among those. This shows that $G^\ast$ is properly t-embedded. The other part of the assumption follows directly.

Assumption~\ref{item:intro_reg_part}. Let $p\in \Sigma$ be as in the assumption, identify $B_\Sigma(p,\lambda)$ with a subset of the Euclidean plane. Extend $\Gamma\cap B_\Sigma(p,\lambda)$ to a $\lambda^{-1}\delta$-net $\Gamma_p$ of the plane, assume that $\Gamma_p$ is also $\lambda\delta$-separated and each circle containing four points from $\Gamma_p$ contains another point inside. Repeat the construction above to get $G_p^\ast$.

Assumptions~\ref{item:intro_conical_sing},~\ref{item:intro_boundary_part} and~\ref{item:intro_dbar_adapted} follows directly from the construction.

\subsubsection{Temperleyan structure and the form \texorpdfstring{$\alpha_G$}{alphaG}}
\label{intro_example:alpha_and_temperley}

In Examples~\ref{intro_example:pillow_surface} and~\ref{intro_example:triangular_graphs} the $(0,1)$-form $\alpha_G$ appears to be related with the holonomy of the metric $ds^2$ on $\Sigma$. We claim that in this special case the $(0,1)$-forms $\alpha_G$ and $\alpha_0$ can be chosen such that
\begin{equation}
  \label{eq:alpha_0_vs_alpha_G_Temperley}
  \alpha_G = -\frac{1}{2}\alpha_0
\end{equation}
We explain this choice for Example~\ref{intro_example:triangular_graphs}, the other example can be analyzed similarly. Recall that white vertices of $G$ are midpoints of edges of $\Gamma^\dagger$ (this holds true even at the conical singularities). Given such a vertex $w$, let us choose an arbitrary non-zero tangent vector $v_w\in T_w\Sigma$ tangent to the edge of $\Gamma^\dagger$ passing through $w$. Put
\[
  \eta_w = \frac{\overline{\omega_0(v_w)}}{|\omega_0(v_w)|},\qquad \eta_b = 1,\ b\in \Gamma,\qquad \eta_b = i,\ b\in \Gamma^\dagger.
\]
Then it is easy to see that if $\alpha_G$ was defined by the equation~\eqref{eq:alpha_0_vs_alpha_G_Temperley} and $K$ was defined by~\eqref{eq:def_of_K}, then $\eta_wK(w,b)\eta_b\in \RR$ for each $w,b$, hence $K$ is gauge equivalent to a real operator.

\subsection{Dimer height function as a 1-form on a Riemann surface}
\label{subsec:intro_height_function}

Traditionally, the dimer model on a bipartite planar graph $G_0$ has the associated (random) height function~\cite[Section~2.2]{kenyon2009lectures}. To define it, we need to choose a reference flow which can be done in several ways, for example, we can choose the reference flow associated with a fixed dimer cover $D_\rf$. Given any other dimer cover $D_0$ we superimpose it with $D_\rf$ to get a collection of loops and double edges. We orient each loop in such a way that all the dimers in $D_0$ are oriented from the white vertex to the black one; then the height function $h$ is defined to be constant on faces of $G_0$ and to jump to $+1$ each time we cross a loop from right to left, and to have no other jumps. This defines $h^{D_\rf}_{D_0}$ up to an additive constant; if $G_0$ has a distinguished face, then we can normalize $h^{D_\rf}_{D_0}$ to be zero on it. For example, if $G_0$ is a discrete domain on the plane with a boundary, then we can require $h$ to be zero on the outer face of $G_0$ (seen as embedded into the sphere).

If we try to apply the same procedure to a graph $G_0$ embedded into a general Riemann surface $\Sigma_0$, we immediately see that the function $h^{D_\rf}_{D_0}$ is not defined consistently: indeed, if the homology class represented by the union of oriented loops in the superposition of $D$ with $D_\rf$ is non-zero, then $h^{D_\rf}_{D_0}$ inherits an additive integer monodromy along any loop having a non-trivial algebraic intersection with this homology class. Note that if we take $D_0$ randomly, then the monodromy of $h^{D_\rf}_{D_0}$ will be random.

Another complication appears when $\partial \Sigma_0\neq \varnothing$. Similarly to the simply-connected case, we can normalize $h^{D_\rf}_{D_0}$ to be zero on one of boundary components. However, this normalization does not fix it on the rest of the boundary of $\Sigma_0$. Indeed, given a path between two components of $\partial \Sigma_0$, the sum of increments of $h_{D_0}^{D_\rf}$ along this path depends on its intersection number with the homology class represented by loops in $D_0\cup D_\rf$. We call this number the `height jump' along this path and observe that height jumps are random similarly to the monodromy.

Recall that given $\Sigma_0$ we introduce the surface $\Sigma$ which is the double of $\Sigma_0$ if $\partial \Sigma_0\neq \varnothing$ and coincides with $\Sigma_0$ else. Both monodromy and height jumps components of the height function on $\Sigma_0$ are encoded in the monodromy of its extension to $\Sigma$. Indeed, according to our convention in Section~\ref{subsec:intro_dimer_model}, each dimer cover $D_0$ of $G_0$ induces a symmetric dimer cover $D$ on the graph $G$ on $\Sigma$ with coincides with $G_0$ on $\Sigma_0$ and is symmetric under the involution if $\partial \Sigma_0\neq \varnothing$. Extending the reference cover $D_\rf$ to $G$ in the same way we can extend $h_{D_0}^{D_\rf}$ on the whole $\Sigma$. Denote the extension by $h_D^{D_\rf}$ by abusing the language slightly. Note that the height jump of $h_{D_0}^{D_\rf}$ along an oriented path $\overrightarrow{\gamma}$ on $\Sigma_0$ is equal to the one half of the monodromy of $h_D^{D_\rf}$ along $\overrightarrow{\gamma}\cup \sigma(\overleftarrow{\gamma})$. In particular, the monodromy of $h_D^{D_\rf}$ along any such loop must be even.

Let us now discuss the definition of the height function more formally. There are many ways to make a strict sense of the concept of a multi-valued function. The most straightforward is probably to consider it on the universal cover. However, for our purpose it is more convenient to replace the height function $h$ with its exterior derivative which is single-valued since the monodromy is constant. A slight issue comes from the fact that $h$ is not differentiable, which does not prevent us from considering $dh$ as a 1-form whose coefficients are distributions, or as a functional on smooth 1-forms, which we will have to do anyway as $h$ is not going to be a function in the scaling limit. We now discuss this more accurately.  

Define a flow $f$ to be any anti-symmetric function on oriented edges of $G$. Assume that the edges of $G$ are piece-wise smooth curves on $\Sigma$. Then with each flow $f$ we can associate a generalized 1-form $\M_f$ defined such that for any 1-form $u$ we have
\[
  \int_\Sigma u\wedge \M_f = \sum_{w\sim b}f(wb)\int_w^b u.
\]
With such a definition $\M_f$ becomes a 1-form with the coefficients from the space of generalized functions. This 1-form is closed outside vertices of our graph meaning that it vanishes when paired with any $u = d\vphi$ where $\vphi$ is a smooth function vanishing at the vertices. This observation allows to apply a variant of the Hodge decomposition to $\M_f$ and write it as
\[
  \M_f = d\Phi + \Psi
\]
where $\Phi$ is a function defined up to an additive constant, bounded and continuous on faces of $G$, and $\Psi$ is a 1-form which is harmonic outside the vertices of $G$ and has simple poles at the vertices. The latter means that $\Psi$ can be written as 
\[
  \Psi = \frac{c}{2\pi}\Im \left( \frac{dz}{z} \right) + \text{harm. term}
\]
locally near a vertex $v$ of $G$, where $z$ is a local coordinate at $v$ and $c$ is equal to the divergence of $f$, i.e. $c = \div f(v) = \sum_{v'\sim v}f(vv')$.

With each dimer cover $D$ we associate the flow
\[
  f_D(wb) = \begin{cases}
    1,\quad wb\in D,\\
    0,\quad \text{else.}
  \end{cases}
\]
Let $\M_D^\fluct$ denote the 1-form corresponding to the flow $f_D^\fluct(wb) = f_D(wb) - \PP[wb\text{ covered}]$. Let
\[
  \M_D^\fluct = d\Phi_D^\fluct + \Psi_D^\fluct
\]
be the Hodge decomposition. Note that $f_D^\fluct$ is divergence-free, hence $\Psi_D^\fluct$ is a harmonic differential on $\Sigma$.

Let for a moment make a step back and consider the multivalued height function $h^{D_\rf}_D$ defined above. The differential $dh^{D_\rf}_D$ is a well-defined generalized 1-form. It is straightforward to verify that
\[
  \M_D^\fluct = dh^{D_\rf}_D - \EE dh^{D_\rf}_D = d(h^{D_\rf}_D - \EE h^{D_\rf}_D) = dh^\fluct_D
\]
where $h^\fluct$ is the dimer height function normalized to have zero mean, i.e. $h^\fluct = h^{D_\rf} - \EE h^{D_\rf}$ (to define $\EE h^{D_\rf}_D$ we actually have to pull $h^{D_\rf}_D$ back to the universal cove of $\Sigma$). In particular, we have the following relation between $h_D^{D_\rf}$ and $(\Phi_D^\fluct, \Psi_D^\fluct)$:

-- the harmonic differential $\Psi_D^\fluct$ has the cohomology class corresponding to the monodromy of $h_D^{D_\rf} - \EE h_D^{D_\rf}$. Note that $\Psi_D^\fluct$ does not necessary have integer cohomologies, but one can always find a deterministic harmonic differential $u$ such that the cohomology class of $\Psi_D^\fluct - u$ is integer.

-- the function $\Phi_D^\fluct$ is responsible for local fluctuations of $h_D^{D_\rf}$ around its mean.

Assume that $\partial \Sigma_0\neq \varnothing$. In this case we have
\[
  \sigma^*\M_D^\fluct = -\M_D^\fluct,
\]
and the same is true for both components $\Psi_D^\fluct$ and $\Phi_D^\fluct$ (given that we chose the additive constant in the definition of $\Phi_D^\fluct$ properly).

Recall that height jumps of the height function on $G_0$ are encoded in the monodromy of its extension to $G$ along symmetric loops intersecting boundary components. Let $l$ be a path on $\Sigma_0$ connecting two boundary points, then $l\cup \sigma(l)$ is a loop on $\Sigma$, which we can orient such that the orientation of $l$ is kept. The height jump between the endpoints of $l$ (note that it depends on the relative homology class of $l$ with respect to the boundary components containing its endpoints) is then $\frac{1}{2}\int_{l\cup \sigma(l)}\Psi_D^\fluct$. In particular, $\int_{C}(\Psi_{D_1}^\fluct - \Psi_{D_2}^\fluct)$ should be an \emph{even} integer for any two dimer covers $D_1,D_2$, whenever $C$ is a symmetric loop crossing the boundary.

\subsection{Compactified free field}
\label{subsec:intro_freefield}

The expected continuous counterpart for the dimer height function on a Riemann surface is a \emph{compactified free field} which we will briefly discuss now. We address the reader to Sections~\ref{subsec:bosonization},~\ref{subsec:cff_formal_def} for more detailed discussion, and to~\cite{DubedatFamiliesOfCR},~\cite[6.3.5, 10.4.1]{francesco2012conformal} for the discussion of the physical meaning behind this object.

We begin with a suitable notation for the Dirichlet energy on $\Sigma_0$: given a 1-form $u$ we set
\[
  \Ss_0(u) = \frac{\pi}{2}\int_{\Sigma_0} u\wedge \ast u,
\]
where $\ast$ is the \emph{Hodge star} normalized such that if $u = d\vphi$ for some smooth $\vphi$, then the integral in the right-hand side above is the Dirichlet energy of $\vphi$ (see also~\eqref{eq:def_of_Hodge_star}). Very informally speaking, the compactified free field is a random multivalued (generalized) function on $\Sigma_0$ with integer monodromy, distributed according to the Gaussian probability measure proportional to ``$e^{-\Ss_0(dh)}\, \Dd(dh)$'', where $\Dd(dh)$ is a `Lebesgue measure' on the space of 1-forms $dh$. Of coarse, this definition does not make any mathematical sense (both $\Dd(dh)$ and $\Ss_0(dh)$ are not well-defined), but let us keep it for a while to make some heuristics.

Following our approach to the height function, we consider the (closed, but not necessary exact) 1-form $dh$ and apply the Hodge decomposition to it:
\begin{equation}
  \label{eq:intro1}
  dh = d\phi + \psi,
\end{equation}
where $\phi$ is a (generalized) function and $\psi$ is a harmonic differential. The components $\vphi$ and $\psi$ are usually called \emph{scalar} and \emph{instanton} components respectively. Note that
\[
  \Ss_0(d\phi + \psi) = \Ss_0(d\phi) + \Ss_0(\psi),
\]
thus $\phi$ and $\psi$ are independent if treated as random variables --- at least according to our informal definition of $h$. 

To define the compactified free field accurately, we describe the components $\phi$ and $\psi$ separately and declare them to be independent; the compactified free field is then defined by the right-hand side of~\eqref{eq:intro1}.

If $\partial \Sigma_0\neq \varnothing$, then we define $\phi$ to be the Gaussian free field on $\Sigma_0$ with zero boundary conditions normalized in such a way that $\int_{\Sigma_0}d\phi\wedge \ast d\vphi\sim \Nn(0, \Ss_0(d\vphi))$ for any test function $\vphi$. We extend $\phi$ to the double $\Sigma$ in such a way that $\sigma^*\phi = -\phi$.

If $\partial \Sigma_0 = \varnothing$, then we again declare $\phi$ to be the Gaussian free field on $\Sigma_0$ normalized as above. Though $\phi$ is defined only up to an additive constant in this way, the 1-form $d\phi$ is defined properly.

Let us now define the instanton component $\psi$. As we want the compactified free field to be the limit of the averaged height function, we need to impose the homology class of $\psi$ to be integer shifted by a deterministic real cohomology class (not purely integer as usual). Assume first that $\partial \Sigma_0 = \varnothing$. Let $\alpha$ be a deterministic anti-holomorphic $(0,1)$-form on $\Sigma$. Then we declare $\psi^\alpha$ to be a random harmonic differential having the following properties
\begin{enumerate}
  \item The cohomology class of $\psi^\alpha - \pi^{-1}\Im\alpha$ is almost surely integer,
  \item For any harmonic differential $u$ such that $u-\pi^{-1}\Im\alpha$ has integer cohomologies we have
    \[
      \PP[\psi^\alpha = u] \sim e^{-\Ss_0(u)}.
    \]
\end{enumerate}
If $\partial \Sigma_0 \neq \varnothing$, we additionally assume that $\sigma^*\alpha = \bar{\alpha}$, and sample $\psi$ as above but conditioned on the event that 
\begin{enumerate}[resume]
  \item for any loop $C$ on $\Sigma$ symmetric under $\sigma$ the integral $\int_C (\psi^\alpha - \pi^{-1}\Im\alpha)$ is equal to an even integer, and $\sigma^*\psi^\alpha = -\psi^\alpha$.
\end{enumerate}
The compactified free field twisted by $\alpha$ is now by definition the sum
\[
  \m^\alpha = d\phi + \psi^\alpha
\]
where $\phi$ and $\psi^\alpha$ are taken to be independent. Note that $\m^\alpha$ is not centered if $\alpha$ is chosen generically:
\[
  \EE\m^\alpha = \EE \psi^\alpha
\]
can be a non-zero harmonic differential. Note also that $\m^\alpha - \EE\m^\alpha$ is not necessary equal to $\m^0$.

\subsection{Formulation of main results}
\label{subsec:intro_main_results}

Assume now that $\lambda\in (0,1)$ and an integer $g\geq 0$ are given, $\{\delta_k\}_{k>0}$ is a sequence of positive numbers tending to zero, and a sequence $(\Sigma^k, p_1^k,\dots, p_{2g-2}^k, G^k)$ is given, where
\begin{enumerate}
  \item Each $\Sigma^k$ is a closed Riemann surface of genus $g$ which is either a double of a Riemann surface $\Sigma_0^k$ of genus $g_0$ with boundary (then $\sigma_k$ is the corresponding involution), or coincides with a closed Riemann surface $\Sigma_0^k$ of genus $g_0$.
  \item $p_1^k,\dots, p_{2g-2}^k$ is a collection of distinct points on $\Sigma^k$ symmetric under $\sigma_k$ if the involution $\sigma_k$ is present.
  \item The sequence of marked surfaces $\Sigma^k, p_1^k,\dots, p_{2g-2}^k$ converges to a marked Riemann surface $\Sigma,p_1,\dots, p_{2g-2}$ in the topology of the moduli space of Riemann surfaces with $2g-2$ marked points; let $\sigma$ denote the involution on $\Sigma$ given by the limit of $\sigma_k$.
  \item For each $k$ let $ds_k^2$ be the locally flat metric on $\Sigma^k$ with conical singularities at $p_i^k$'s with cone angles $4\pi$ normalized such that the area of $\Sigma^k$ is 1. Then $G^k$ is a bipartite graph embedded into $\Sigma^k$ which is $(\lambda,\delta_k)$-adapted, i.e. satisfies all the assumptions from Section~\ref{subsec:intro_graphs_on_Sigma0}.
\end{enumerate}
For each $k$ we consider the dimer model on $\Sigma_0^k$ and denote by $\M^{\fluct,k}_D = d\Phi^{\fluct,k}_D + \Psi^{\fluct,k}_D$ the derivative of the corresponding height function on $\Sigma^k$ as defined in Section~\ref{subsec:intro_height_function}, where $D$ is a random dimer cover of $G^k$ sampled as described in Section~\ref{subsec:intro_height_function}.

For each $k$ we fix an orientation preserving diffeomorphism $\xi_k: \Sigma_k\to \Sigma$ such that the following is satisfied:

\begin{enumerate}
  \item $\xi_k$ tends to identity in $\mC^2$ topology as $k\to \infty$ (see Section~\ref{subsec:teichmuller_space} for the precise meaning of it).
  \item If (starting from some $k$) $\partial \Sigma_0^k\neq \varnothing$, then $\sigma\circ \xi_k = \xi_k\circ \sigma_k$.
  \item For each $k$ and $j = 1,\dots, 2g-2$ we have $\xi_k(p_j^k) = p_j$.
\end{enumerate}
We identify the space of smooth 1-forms on $\Sigma$ with the space of smooth 1-forms on $\Sigma^k$ as follows. Given a smooth 1-form $u$ on $\Sigma$ decompose it as
\[
  u = d\vphi + \ast d\vphi_1 + u_h
\]
where $\vphi,\vphi_1\in \mC^\infty(\Sigma)$, the 1-form $u_h$ is harmonic and $\ast$ is the Hodge star. Let
\begin{equation}
  \label{eq:def_of_uk}
  u^k = d\vphi\circ \xi_k + \ast d\vphi\circ \xi_k + u_h^k
\end{equation}
where $u_h^k$ has the cohomology class equal to the pullback of the cohomology class of $u$ under $\xi_k$. For each $k$ let $\alpha_{G_k}$ be the anti-holomorphic $(0,1)$-form defined as in Section~\ref{subsec:intro_Kasteleyn_operator}. We assume that the sequence of cohomology classes of $\alpha_{G_k}^k$ pulled back to $\Sigma$ along $\xi_k^{-1}$ converges to a cohomology class represented by an anti-holomorphic $(0,1)$-form $\alpha_1$.

\subsubsection{Identification of the limit height field with the compactified free field}
\label{subsubsec:intro_identification_theorem}

Our first main result concerns the reconstruction of the limit of the sequence $\M^{\fluct, k}$ provided the sequence is tight. Recall the definition of $\m^\alpha$ from Section~\ref{subsec:intro_freefield}.

\begin{thma}
  \label{thma:main1}
  Let $\Uu$ be a finite dimensional subspace of the space of all smooth 1-forms on $\Sigma$. Given a generalized 1-form $\M$ on $\Sigma^k$ denote by $\M\vert_\Uu$ the linear functional on $\Uu$ given by $u\mapsto \int_{\Sigma^k} u^k\wedge \M$. Assume that $\Uu$ contains all harmonic 1-forms on $\Sigma$ and the sequence of distributions of $\M^{\fluct, k}\vert_\Uu$ in $\Uu^\ast$ is tight in the weak topology. Then all the subsequential weak limits of $\M^{\fluct, k}\vert_\Uu$ are of the form $(\m^{2\alpha_1} - \m_\Uu)\vert_\Uu$, where $\m_\Uu\in \Uu^\ast$ is a deterministic linear functional, depending on the subsequence. If $\EE\left|\int_{\Sigma^k} u^k\wedge \M^{\fluct, k}\right|$ is uniformly bounded in $k$ for each $u\in \Uu$, then $\M^{\fluct, k}\vert_\Uu$ weakly converges to $\m^{2\alpha_1} - \EE\m^{2\alpha_1}$ restricted to $\Uu$.
\end{thma}

\begin{rem}
  \label{rem:main1}
  Let us emphasize that $\alpha_1 = -\frac{1}{2}\alpha_0$ in Example~\ref{intro_example:triangular_graphs}, where $\alpha_0$ is the anti-holomorphic $(0,1)$-form associated with the metric $ds^2$ (see Section~\ref{subsec:intro_flat_metric}). In particular, the instanton component $\psi^{2\alpha_1}$ that appears in the limit is shifted by the 1-form $2\pi^{-1}\Im\alpha_1 = -\pi^{-1}\Im\alpha_0$ which is the (minus) holonomy one form of the metric $ds^2$.

  Moreover, one can take $\alpha_1 = 0$ in the case of an approximation by pillow surfaces (Example~\ref{intro_example:pillow_surface}), therefore the compactified free field that appears in the limit has integer monodromies in this case.
\end{rem}

\subsubsection{Tightness of the height field}
\label{subsubsec:intro_tightness_theorem}

Our next result concerns the tightness of the sequence $\M^{\fluct, k}$. The methods we use to prove Theorem~\ref{thma:main1} allow us to establish the tightness of $\M^{\fluct, k}$ when the sequence $(\Sigma^k, p_1^k,\dots, p_{2g-2}^k, G^k)$ is in \emph{general position}. We describe the precise genericity assumptions in Theorem~\ref{thmas:second_moment_estimate} given in Section~\ref{subsec:second_thm_tightness}. Here, we only present some examples in which these assumptions are satisfied:

\begin{enumerate}
  \item\label{item:intro1} The form $2\pi^{-1}\Im\alpha_1$ has integer cohomologies and all theta constants of $\Sigma$ corresponding to even theta characteristics are non-zero. Note that the last condition is satisfied if $\Sigma$ was chosen generically enough (outside of an analytic subvariety in the moduli space).
  \item\label{item:intro2} Identify the Jacobian of $\Sigma$ with the quotient $\frac{H^1(\Sigma, \RR)}{H^1(\Sigma, \ZZ)}$ in the standard way (see Section~\ref{subsec:Jacobian}). Then the point in the Jacobian corresponding to the cohomology class of $\pi^{-1}\Im \alpha_1$ is outside of the union of all half-integer shifts of the theta divisor. Note that theta divisor is an analytic subvariety, hence this condition is satisfied provided $\alpha_1$ is generic.
  \item\label{item:intro3} The surface $\Sigma_0$ has the topology of a multiply connected domain.
\end{enumerate}

We now formulate our second result.

\begin{thma}
  \label{thma:main2}
  Assume that at least one of the assumptions~\ref{item:intro1}--\ref{item:intro3} above is satisfied. Then for any 1-form $u$ on $\Sigma$ with $\mC^1$ coefficients the sequence $\EE\left|\int_{\Sigma^k} u^k\wedge \M^{\fluct, k}\right|^2$ is bounded uniformly in the $\mC^1$ norm of the coefficients on $u$.
\end{thma}

Theorem~\ref{thma:main2} implies in particular that if we consider $\M^{\fluct,k}$ as a functional on the space of 1-forms with $\mC^{1+\eps}$ coefficients (where $\mC^{1+\eps}$ is the corresponding Sobolev space), then the corresponding sequence of probability measures on the dual Sobolev space will be tight. Theorem~\ref{thma:main1} then implies that it converges to $\m^{2\alpha_1} - \EE\m^{2\alpha_1}$.

It appears possible to improve the smoothness in Theorem~\ref{thma:main2} to $\mC^\eps$ for an arbitrary $\eps>0$. However, we prefer to keep it $1+\eps$ to shorten the exposition.

\subsection{Relation to the work of Berestycki, Laslier, and Ray}
\label{subsec:intro_relation_to_BLR}

Let us briefly recall the setup from the works~\cite{BerestyckiLaslierRayI,BerestyckiLaslierRayII} that we have already discussed in the introduction. Let a Riemann surface $\Sigma_0$ of genus $g_0$ with $n\geq 0$ boundary components and $2g_0-2+n$ marked points be given, assume that $2g_0 -2+n\geq 0$. Fix a \emph{smooth} Riemannian metric on $\Sigma_0$ continuous up to the boundary and representing the conformal class of the surface. Let $\Gamma^k_0$ be a sequence of graphs embedded into $\Sigma_0$ such that for each $k$ the boundary of $\Sigma_0$ is composed of edges of $\Gamma^k_0$. Let $\Gamma^{k,\dagger}_0$ denote the dual graph to $\Gamma^k_0$. It is also assumed that oriented edges of each $\Gamma^k_0$ are endowed with non-negative weights.

Put $\Sigma_0' = \Sigma_0\smm(\{ p_1,\dots, p_{2g_0-2+n} \}\cup \partial \Sigma_0)$ and let $\widetilde{\Sigma}_0'$ denote the universal cover of $\Sigma_0'$ identified with an open subset of $\CC$ containing the origin. Let $\widetilde{\Gamma}_0^k$ denote the lift of $\Gamma_0^k$ to $\widetilde{\Sigma}_0'$. The following conditions are imposed on $\Gamma^k_0$, see~\cite[Section~2.3]{BerestyckiLaslierRayI}:

\begin{enumerate}
  \item (\textbf{Bounded density}) There exists a constant $C>0$ and a sequence $\delta_k\to 0$ such that for each $k$ and each $x\in \Sigma_0$ the number of vertices of $\Gamma^k_0$ in the ball $\{ z\in \Sigma_0\ \mid\ d_{\Sigma_0}(x,z)\leq \delta_k \}$ is smaller than $C$.
  \item (\textbf{Good embedding}) The edges of $\Gamma_0^k$ and $\Gamma_0^{k,\dagger}$ are embedded as smooth curves and for every compact $K\subset \widetilde{\Sigma}_0'$ the intrinsic winding (see~\cite[eq.~(2.4)]{BerestyckiLaslierRayI}) of every edge of $\widetilde{\Gamma}_0^k$ intersecting $K$ is bounded by a constant depending on $K$ only.
  \item (\textbf{Invariance principle}) The continuous time random walk on $\widetilde{\Gamma}^k_0$ defined by the edge weights of $\Gamma^k_0$ and started at the closest vertex to the origin converges to the standard Brownian motion killed on the boundary up to (possibly random) continuous time change in law in Skorokhod topology.
  \item (\textbf{Uniform crossing estimate}) The continuous time random walk on $\Gamma_0^k$ must satisfy the uniform crossing estimate in any compact subset of $\Sigma_0\smm\partial \Sigma_0$ up to the scale $\delta_k$ with the constants depending on the compact only (see~\cite[Section~2.3]{BerestyckiLaslierRayI} for details).
\end{enumerate}

For each $k$ denote by $G^k_0$ the Temperleyan graph obtained by superimposing $\Gamma^k_0$ and $\Gamma^{k,\dagger}_0$ and removing the boundary vertices of $\Gamma^k_0$. The weight of an edge $e$ of $G^k_0$ is defined to be the weight of the oriented edge of $\Gamma^k_0$ if $e$ is the corresponding tail half edge, or 1 if $e$ is a half edge of an edge of $\Gamma^{k,\dagger}_0$. For each $k$ and $j = 1,\dots, 2g_0-2+n$ remove a white vertex on a distance $o(1)$ from $p_j$ from $G^k_0$. Denote by $(G^k_0)'$ the corresponding graph and by $(\widetilde{G}^k_0)'$ the lift of $(G^k_0)'$ to $\widetilde{\Sigma}_0'$. Let $h^k$ denote the dimer height function on $(\widetilde{G}^k_0)'$ defined with respect to the lift of any reference flow on $(G^k_0)'$. The main result of~\cite{BerestyckiLaslierRayI,BerestyckiLaslierRayII} (see~\cite[Theorem~6.1]{BerestyckiLaslierRayI}) about $h^k$ can be stated as follows:
\begin{thmas}[Berestycki, Laslier, Ray]
  \label{thmas:BLR}
  Let $\mu$ denote the lift of the volume form of $\Sigma_0'$ to $\widetilde{\Sigma}_0'$. Then for any smooth test function $\vphi$ on $\widetilde{\Sigma}_0'$ the random variable $\int (h^k - \EE h^k)\vphi\,d\mu$ converges to a conformally invariant limit which depend only on $\Sigma_0$ and the points $p_1,\dots, p_{2g_0-2+n}$, but not on the sequence $(G^k_0)'$. The convergence is in the sense of all moments.
\end{thmas}

As we already mentioned in the introduction, the scaling limit in Theorem~\ref{thmas:BLR} and our Theorem~\ref{thma:main1} is designed to reconstruct it. To this end we consider the graphs introduced in Example~\ref{intro_example:triangular_graphs}. Let $(\Sigma, ds^2)$ be as in Section~\ref{subsec:intro_flat_metric}
Fix small enough $\lambda>0$ and a sequence $\delta_k\to 0+$, for each $k$ let $(\Gamma^k)', (\Gamma^{k,\dagger})', (G^k)', (G^{k,\ast})'$ be the graphs $(\Gamma, \Gamma^\dagger, G, G^\ast)$ constructed as in Example~\ref{intro_example:triangular_graphs} with the given $\lambda$ and $\delta = \delta_k$. We can also redraw the edges of $(\Gamma^k)'$ and $(G^k)'$ as is shown on Figure~\ref{fig:changing_edges_of_Gamma} to make them smooth (piece-wise smoothness is also enough for the arguments in~\cite{BerestyckiLaslierRayI,BerestyckiLaslierRayII}, though this is not stated explicitly). Put 
\[
  (\Gamma_0^k)' = (\Gamma^k)'\cap \Sigma_0,\qquad (\Gamma^{k,\dagger}_0)' = (\Gamma^{k,\dagger})'\cap\Sigma_0, \qquad (G^k_0)' = (G^k)'\cap \Sigma_0\smm\partial \Sigma_0.
\]
As we already noticed while analysing the construction from Example~\ref{intro_example:triangular_graphs}, the graphs $(\Gamma^k)'$ and $(\Gamma^{k,\dagger})'$ are not dual to each other at conical singularities. In fact, to make them dual everywhere it is enough to add a pair of edges at each conical singularity. Indeed, fix a $j = 1,\dots, 2g_0-2$ and consider the face of $(G^k)'$ which contains $p_j$. Let $b_1,b_1^\dagger, b_2,b_2^\dagger$ be the black vertices of $G$ incident to this face and listed counterclockwise. Note that $b_i\in (\Gamma^k)'$ and $b_i^\dagger\in (\Gamma^{k,\dagger})'$. Draw the additional edge $b_1b_2$ for $(\Gamma^k_0)'$ and $b_1^\dagger b_2^\dagger$ for $(\Gamma^{k,\dagger}_0)'$ in such a way that they intersect at $p_j$ and do not intersect any other edge of $(\Gamma^k)'$ or $(\Gamma^{k,\dagger})'$ (say, by connecting all these 4 vertices to $p_j$ be straight segments), see Figure~\ref{fig:adding_edges}. After we do this for all $j$'s, we obtain graphs $\Gamma^k_0,\Gamma^{k,\dagger}_0$ which are dual to each other everywhere. The corresponding Temperleyan graph $G^k_0$ differs from $(G^k_0)'$ only by $2g_0-2+n$ white vertices located at $p_1,\dots, p_{2g_0-2+n}$. In this way we come back to the combinatorial setup of~\cite{BerestyckiLaslierRayI}.

We equip the oriented edges of $\Gamma_0^k$ by weights as follows. Let $b$ be a vertex of $\Gamma_0^k$ and $b_1,\dots, b_d$ be the neighbors of $b$ in $(\Gamma_0^k)'$. For each $j$ let $b^\dagger_{j,+}b^\dagger_{j,-}$ be the edge of $(\Gamma^{k,\dagger}_0)'$ dual to $bb_j$. We define the weight of $bb_i$ by
\begin{equation}
  \label{eq:def_of_weight_on_Gammak}
  \mathrm{weight}(bb_i) = \frac{m_i}{\dist(b,b_i)^2},\qquad m_i = \frac{\dist(b^\dagger_{i,+},b^\dagger_{i,-})\cdot \dist(b,b_i)}{\sum_{j = 1}^d \dist(b^\dagger_{j,+},b^\dagger_{j,-})\cdot \dist(b,b_j)}.
\end{equation}
We also set $\mathrm{weight}(b_1b_2) = 0$ if $b_1b_2$ is not an edge of $(\Gamma_0^k)'$. It is easy to see that the dimer weights induced via the Temperley bijections by the weights of $\Gamma_0^k$ are gauge equivalent to the dimer weights defined in Section~\ref{subsec:intro_dimer_model}. In particular, if $h^k$ is the dimer height function on $(\widetilde{G}_0^k)'$ defined as above, then $h^k - \EE h^k$ is given by an appropriately chosen primitive of the pullback of $\M^{\fluct,k}$ to the universal cover, where $\M^{\fluct}$ denotes the 1-form associated with the height fluctuations on $G^k_0$ sampled according to our dimer weights, see Section~\ref{subsec:intro_height_function}.

\begin{figure}
  \centering
  \begin{minipage}{.3\textwidth}
    \centerline{\includegraphics[clip, trim=0cm 0cm 0cm 0cm, width = 0.66\textwidth]{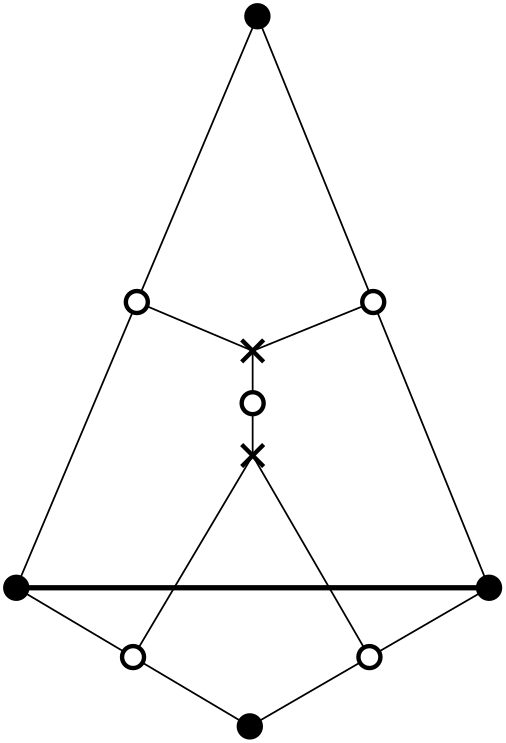}}

    Edges of $(\Gamma^k)'$ do not necessary intersect the corresponding dual edges of $(\Gamma^{k,\dagger})'$
  \end{minipage}
  \hskip 0.03\textwidth \begin{minipage}{.3\textwidth}
    \centerline{\includegraphics[clip, trim=0cm 0cm 0cm 0cm, width = 0.66\textwidth]{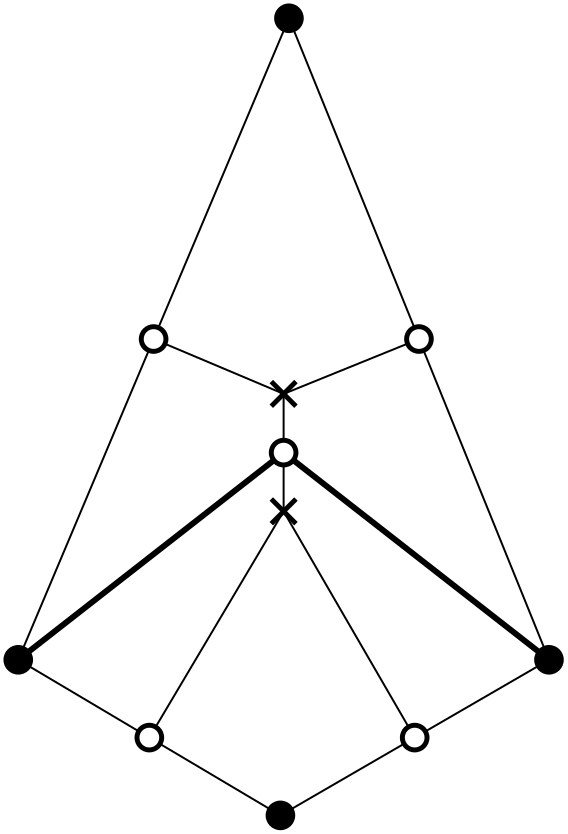}}

    Replace such edges of $(\Gamma^k)'$ with the corresponding pairs of edges of $(G^k)'$ ...
  \end{minipage}
  \hskip 0.03\textwidth\begin{minipage}{.3\textwidth}
    \centerline{\includegraphics[clip, trim=0cm 0cm 0cm 0cm, width = 0.66\textwidth]{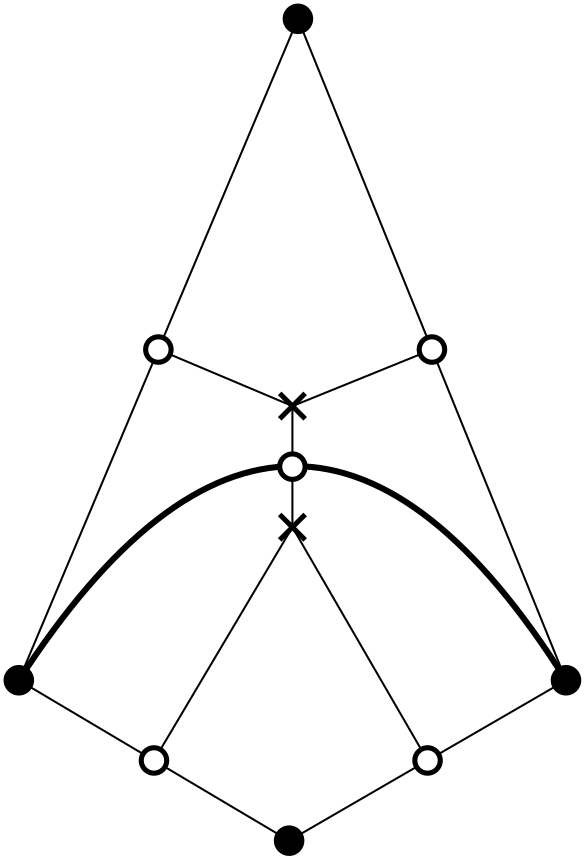}}

    ... and then smooth them (changing edges of $(G^k)'$ accordingly)
   \end{minipage}
  \caption{Changing edges of $(\Gamma_k)'$.}
  \label{fig:changing_edges_of_Gamma}
\end{figure}

\begin{figure}
  \centering
  \centerline{\includegraphics[clip, trim=0cm 0cm 0cm 0cm, width = 0.46\textwidth]{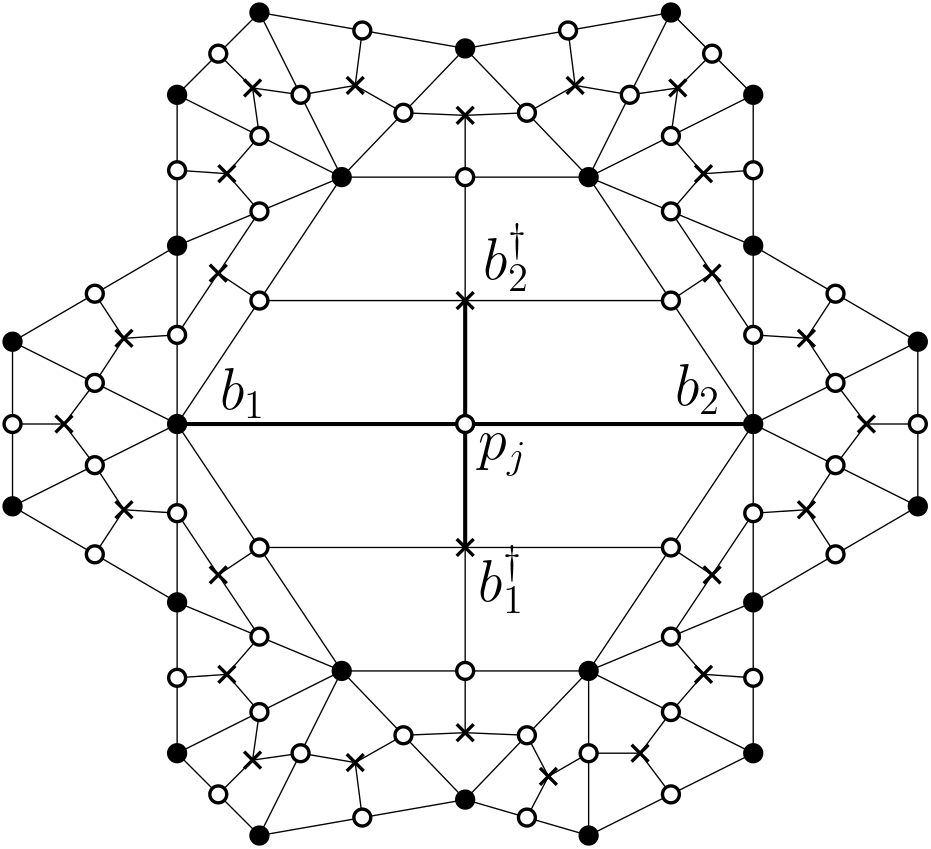}}
  \caption{Adding edges to $(\Gamma^k)'$ and $(\Gamma^{k,\dagger})'$ at conical singularities. The new edges are drawn in bold (cf. Figure~\ref{fig:graph_near_conical}).}
  \label{fig:adding_edges}
\end{figure}

Recall the $(0,1)$-form $\alpha_0$ associated with the metric $ds^2$ in Section~\ref{subsec:intro_flat_metric}. Recall that $\m^{-\alpha_0}$ denotes the derivative of the compactified free field whose instanton part is given by an integer cohomology class minus the class of $\pi^{-1}\Im \alpha_0$, see Section~\ref{subsec:intro_freefield} (recall that $\m^{-\alpha_0}$ is not $\m^0$ shifted by $\pi^{-1}\Im \alpha_0$). Below is a specified version of Theorem~\ref{thmaa:informal_temp}. %

\begin{thma}
  \label{thma:for_BLR}
  Let $(\Sigma_0,p_1,\dots, p_{2g_0-2+n})$ be given and the sequence $\Gamma_0^k, \Gamma_0^{k,\dagger}, G_0^k$ be as above. Then there exists a subsequence of $\Gamma_0^k, \Gamma_0^{k,\dagger}, G_0^k$ satisfying all the assumptions of~\cite{BerestyckiLaslierRayI} with respect to the singular metric $ds^2$. If the limit of the height field $h^k - \EE h^k$ exists in the sense of Theorem~\ref{thmas:BLR}, then it is equal to an appropriately chosen primitive of $\m^{-\alpha_0} - \EE \m^{-\alpha_0}$ restricted to $\Sigma_0$ and pulled back to the universal cover $\widetilde{\Sigma}'_0$.
\end{thma}

\subsection{Ratio of signed dimer partition functions}
\label{subsec:intro_ratio_of_partition_functions}

Let $\Sigma$ be a closed topological surface of genus $g$ and $G$ be a weighted bipartite graph embedded into it. Recall the definition of a Kasteleyn operator given in Section~\ref{subsec:intro_Kasteleyn_operator}. Recall that two Kasteleyn operators $K_1,K_2$ are said to be \emph{gauge equivalent} if there exist functions $U_1,V_1$ with values in $\TT$ such that $K_1(w,b) = U_1(w)K_2(w,b)V_1(b)$ for each $w,b$. The set $\Kk$ of all gauge equivalence classes of Kasteleyn operators form an affine space over $H^1(\Sigma, \TT)$. To define the action of $H^1(\Sigma, \TT)$ on $K\in\Kk$ we represent a given cohomology class by a cocycle $u: \vec{E}(G)\to \TT$ (where $\vec{E}(G)$ is the set of oriented edges) and then replace $K(w,b)$ by $u(wb)K(w,b)$. Recall that any cohomology class can be represented by a cocycle of the form $u(wb) = \exp(2i\Im \int_w^b\alpha)$ where $\alpha$ is an anti-holomorphic $(0,1)$-form and $\int_w^b$ denotes the integral along any smooth path connecting $w$ and $b$ and homotopic to the edge $wb$. Introduce the notation
\[
  K_\alpha(w,b) = \exp(2i\Im \int_w^b\alpha)K(w,b),
\]
so that $K_\alpha$ is the result of applying the cocycle $u$ to $K$.

When $g>0$, the determinant of any Kasteleyn operator computes a `signed' partition function of the dimer model on $G$, where the sign is determined by a certain topological information depending on the cover and on the choice of the gauge class of the Kasteleyn operator. Let $h^\fluct_D$ be the dimer height function normalized to have zero average, where $D$ is a random dimer cover of $G$. Consider its Hodge decomposition
\[
  dh^\fluct_D = d\Phi^\fluct_D + \Psi^\fluct_D,
\]
see Section~\ref{subsec:intro_height_function} for details. Recall that the cohomology class of the harmonic differential $\Psi^\fluct_D$ is integer up to a deterministic shift. Let us denote by $u$ a harmonic differential such that for any dimer cover $D$ we have $\Psi^\fluct_D - u$ to represent an integer cohomology class. Finally, let $K\in \Kk$ be any Kasteleyn operator that is gauge equivalent to a real valued one (e.g. any Kasteleyn operator associated with a Kasteleyn orientation on $G$~\cite{Cimasoni}). Then it can be proven~\cite{Cimasoni} (see also Section~\ref{subsec:Kasteleyn_thm}) that there exists a constant $\epsilon\in \TT$ such that
\begin{equation}
  \label{eq:det_K_abstract}
  \epsilon\det K = \sum_{D\text{ - dimer cover}} \exp(\pi i q(\Psi^\fluct_D - u))\,\prod_{wb\in D}|K(w,b)|,
\end{equation}
where $q: H_1(\Sigma, \ZZ/2\ZZ)\to \ZZ/2\ZZ$ is a quadratic form (see Section~\ref{subsec:quadratic_forms} for the definition of a quadratic form over $\ZZ/2\ZZ$) and $\Psi^\fluct_D - u$ is associated with an integer homology class from $H_1(\Sigma, \ZZ/2\ZZ)$ via the Poincar\'e duality. Note that both $q$ and $\epsilon$ depend on $u$: if we replace $u$ by $u+v$, then $q(\cdot)$ is replaced by $q(\cdot + v) - q(v)$ and $\epsilon$ is replaced by $(-1)^{q(v)}\epsilon$.

Starting from the aforementioned expression for $\det K$ it is also possible to derive an expression for $\det K_\alpha$ for any anti-holomorphic $\alpha$. To this end, let us define
\[
  Q(\alpha) = 2i\Im \sum_{w\sim b}\PP[wb\text{ is covered by a dimer cover}] \cdot \int_w^b\alpha.
\]
It can be shown (see the proof of Lemma~\ref{lemma:detKalpha}) that 
\begin{multline}
  \label{eq:def_of_Zzalpha}
  \Zz_\alpha := \epsilon e^{-Q(\alpha)}\det K_\alpha = \\
  = \sum_{D\text{ - dimer cover}} \exp\left[\pi i q(\Psi^\fluct_D - u) + 2i\int_\Sigma \Im \alpha\wedge \Psi_D^\fluct\right]\,\prod_{wb\in D}|K(w,b)|.
\end{multline}
Given $l\in H^1(\Sigma, \ZZ/2\ZZ)$ denote by $\alpha_l$ any anti-holomorphic $(0,1)$-form such that the cohomology class of $2\pi^{-1}\Im \alpha_l$ is equal to $l$ modulo 2. We have
\[
  e^{-2i\int_\Sigma \Im\alpha_l\wedge u}\cdot \Zz_{\alpha_l} = \sum_{D\text{ - dimer cover}} \exp(\pi i (q+l)(\Psi^\fluct_D - u))\,\prod_{wb\in D}|K(w,b)|
\]
where $q+l$ denotes the quadratic form obtained by adding $l$ considered as a linear functional on $H_1(\Sigma, \ZZ/2\ZZ)$. In particular, we have the formula originally pointed by Kasteleyn (see~\cite{Cimasoni} and references therein):
\begin{equation}
  \label{eq:Kasteleyn_formula_Riemann_surface}
  \Zz_{\mathrm{dimer}} = 2^{-g}\sum_{l\in H^1(\Sigma, \ZZ/2\ZZ)} (-1)^{\Arf(q+l)} e^{-2i\int_\Sigma \Im\alpha_l\wedge u}\cdot \Zz_{\alpha_l}.
\end{equation}
Note that the signs of both $\Zz_{\alpha_l}$ and $(-1)^{\Arf(q+l)}$ depend on the choice of $u$, but the sign of the product $(-1)^{\Arf(q+l)} \Zz_{\alpha_l}$ does not.

A decomposition analogous to~\eqref{eq:Kasteleyn_formula_Riemann_surface} appears when one considers the partition function of free fermions on a Riemann surface~\cite{alvarez1986theta, Bosonization}. This decomposition suggests that the scaling limit of the ratio of $\frac{\Zz_{\alpha_l}}{\Zz_{\mathrm{dimer}}}$ should coincide with a certain expression written in terms of theta constants on $\Sigma$. In the case of $g=2$ this was verified numerically in~\cite{costa2002dimers} under the assumption that the Riemann surface is associated by `lattices on a surface'. The technique that we develope in our paper allows to analyze scaling limits of the aforementioned ratios in the setup of Theorem~\ref{thma:main1} and show that they coincide with those predicted by the aforementioned results. In particular, we show that when $\Sigma$ is chosen generically and is approximated by pillow surfaces, then the scaling limit coincides of the one appearing in the work~\cite{costa2002dimers}.

From now on, let us assume that we are in the setup of Theorem~\ref{thma:main1}, that is, we have a Riemann surface $\Sigma$ with a locally flat metric with conical singularities $ds^2$, and a sequence of Riemann surfaces $\Sigma^k$ with graphs $G^k$ that approximates $\Sigma$ and satisfies all the assumptions from Section~\ref{subsec:intro_main_results}. Recall that for each $k$ we have a diffeomorphism $\xi_k: \Sigma^k\to \Sigma$ fixed. We assume that $\partial \Sigma_0 = \varnothing$ for simplicity.

We choose a symplectic basis $A_1,\dots, A_g,B_1,\dots, B_g\in H_1(\Sigma, \ZZ)$ and use it to define the theta function with characteristics $a,b\in \RR^g$ via the formula
\[
  \theta\chr{a}{b}(z,\Omega) = \sum_{m\in \ZZ^g}\exp\Bigl( \pi i (m+a)^t\cdot \Omega (m+a) + 2\pi i (z-b)^t(m+a) \Bigr),
\]
see Section~\ref{subsec:Szego_kernel}. Note a slightly unusual choice of signs in the definition, cf. Remark~\ref{rem:miunus_in_def_of_theta}. Given a $(0,1)$-form $\alpha$ we put
\[
  a_j(\alpha) = \pi^{-1}\int_{A_j}\Im \alpha,\qquad b_j(\alpha) = \pi^{-1}\int_{B_j}\Im\alpha,\qquad j = 1,\dots, g.
\]
Note that for any $l\in H^1(\Sigma, \ZZ/2\ZZ)$ we have $a(\alpha_l),b(\alpha_l)\in \frac{1}{2}\ZZ^g$ to be the theta characteristics corresponding to $l$. Recall that $\alpha_1$ is the $(0,1)$-form which is the limit of $\alpha_{G^k}$ (see Section~\ref{subsec:intro_main_results}).

Recall that for each graph $G^k$ we have a distinguished Kasteleyn operator $K_k$ that we defined in Section~\ref{subsec:intro_Kasteleyn_operator}. By the construction it is gauge equivalent to a real operator. We choose a harmonic differential $u_k$ and the quadratic form $q_k$ such that~\eqref{eq:det_K_abstract} holds for $K_k$. Given an arbitrary anti-holomorphic $(0,1)$-form $\alpha$ on $\Sigma$ we denote by $\alpha^k$ the $(0,1)$-form on $\Sigma^k$ such that the cohomology class in $H^1(\Sigma^k, \RR)$ represented by $\Im \alpha_k$ is equal to the pullback along the diffeomorphism $\xi_k$ of the cohomology class in $H^1(\Sigma,\RR)$ represented by $\Im\alpha$. If for all $k$ we have $\Sigma^k = \Sigma$, then $\alpha^k = \alpha$. Let $\Zz_{\alpha^k}^k$ be the signed dimer partition function~\eqref{eq:def_of_Zzalpha} and $\Zz_{\mathrm{dimer}}^k$ be the dimer partition function on $G^k$.

\begin{thma}
  \label{thma:ratio_of_partition_functions}
  Assume that we are in the setup of Theorem~\ref{thma:main1} and assume moreover the tightness and first moment hypothesis from this theorem are satisfied. Then there exists a universal non-zero constant $\Zz_1$ and a theta characteristics $a^0, b^0\in \{ 0,1/2 \}^g$ associated with a quadratic form $q_0$ such that for any anti-holomorphic $(0,1)$-form $\alpha$ we have
  \begin{multline*}
    \lim\limits_{k\to +\infty}\frac{(-1)^{\Arf(q_k)}\Zz_{\alpha^k}^k}{\Zz_{\mathrm{dimer}}^k} = \\
    = \Zz_1\theta\chr{a^0 + a(\alpha_1) + a(\alpha)}{b^0 + b(\alpha_1) + b(\alpha)}(0,\Omega)\cdot \overline{\theta\chr{a^0 +a(\alpha_1) - a(\alpha)}{b^0 + b(\alpha_1) - b(\alpha)}(0,\Omega)} \cdot\\
    \cdot\exp[2\pi i a(\alpha)\cdot (b(\alpha_1) - 2b^0) - 2i\int_\Sigma \Im\alpha\wedge \EE\psi^{2\alpha_1} ].
  \end{multline*}
  The constant $\Zz_1$ is determined by the fact that $\sum_{l\in H^1(\Sigma,\ZZ/2\ZZ)}(-1)^{\Arf(q_k)}\Zz_{\alpha^k_l}^k = \Zz_{\mathrm{dimer}}^k$. The spin structure associated with the theta characteristics $a^0, b^0$ is determined by the metric $ds^2$ on $\Sigma$ and the choice of cuts in the definition of $K_k$; the underlying quadratic form $q_0$ is defined by~\eqref{eq:def_of_q0}.
\end{thma}

\begin{proof}
  It is established in Section~\ref{subsec:Kasteleyn_thm} and in particular in the first item of Lemma~\ref{lemma:Kasteleyn_thm} that there exists a flow $f^{A,k}$ such that
  \begin{equation}
    \label{eq:tropf1}
    \det K_k = \epsilon\cdot \sum_{D\text{ - dimer cover of }G_k} \exp[\pi i q_0(\Psi_D^{A,k})]\,\prod_{wb\in D}|K_k(w,b)|
  \end{equation}
  where $q_0$ is the quadratic form defined by~\eqref{eq:def_of_q0} for $\Sigma^k$ (and we use the same notation $q_0$ for all $k$ by abusing the language).
  Declaring
  \[
    u_k = -\Psi^{\fluct, A, k} = -\Psi^{A,k}_D + \Psi_D^{\fluct, k}
  \]
  (note that the right-hand side does not depend on $D$), we can rewrite~\eqref{eq:tropf1} as
  \[
    \det K_k = \epsilon\cdot \sum_{D\text{ - dimer cover of }G_k} \exp[\pi i q_0(\Psi_D^{\fluct,k} - u_k)]\,\prod_{wb\in D}|K_k(w,b)|
  \]
  which now resembles~\eqref{eq:det_K_abstract}. Now by the definition of $\Zz_{\alpha^k}^k$ and by the formula of $u_k$ we have
  \begin{equation}
    \label{eq:tropf2}
    \frac{(-1)^{\Arf(q_0)}\Zz_{\alpha^k}^k}{\Zz_{\mathrm{dimer}}^k} = (-1)^{\Arf(q_0)} \,\EE\exp\left[ \pi i q_0(\Psi_D^{\fluct,k} + \Psi^{\fluct, A, k}) + 2i\int_\Sigma \Im\alpha\wedge \Psi_D^{\fluct,k} \right].
  \end{equation}
  By Theorem~\ref{thma:main1}, using that tightness and first moment hypothesis are satisfied, we know that $\Psi^{\fluct, k}_D$ converges to $\psi^{2\alpha_1} - \EE\psi^{2\alpha_1}$ in distribution. Moreover, by Corollary~\ref{cor:angle_flow} we have that $\Psi^{\fluct, A, k}$ converges to $\EE\psi^{2\alpha_1} - 2\pi^{-1}\Im\alpha_1$ modulo $H^1(\Sigma, 2\ZZ)$. Using these observations and taking the limit of~\eqref{eq:tropf2} we get
  \begin{multline}
    \label{eq:tropf3}
    \lim\limits_{k\to +\infty} \frac{(-1)^{\Arf(q_0)}\Zz_{\alpha^k}^k}{\Zz_{\mathrm{dimer}}^k} =\\
    = (-1)^{\Arf(q_0)} \,\EE\exp\left[ \pi i q_0(\psi^{2\alpha_1} - 2\pi^{-1}\Im \alpha_1) + 2i\int_\Sigma \Im\alpha\wedge (\psi^{2\alpha_1} - \EE \psi^{2\alpha_1}) \right].
  \end{multline}
  Finally, we can evaluate the right-hand side of~\eqref{eq:tropf3} via theta functions using Lemma~\ref{lemma:bosonization_identity}.
\end{proof}

Let us now formulate the particular case of our theorem in the setup when the Riemann surface $\Sigma$ is approximated by the ``square lattice'', that is, by pillow surfaces; see Example~\ref{intro_example:pillow_surface} for the definition of a pillow surface and an image of a square lattice on it.

\begin{cor}
  \label{cor:ratio_of_partition_functions}
  Assume that $(\Sigma, ds^2)$ is a closed Riemann surface of genus $g$ with a locally flat metric with conical singularities with all cone angles equal to $4\pi$ and trivial holonomy. Assume that $\Sigma^k$ is a sequence of pillow surfaces approximating $(\Sigma,ds^2)$ in the moduli space of Riemann surfaces equipped with locally flat metrics with conical singularities, and assume that for each $k$ the square lattice graph $G^k$ on $\Sigma^k$ and the number of vertices of $G^k$ tends to $+\infty$ as $k\to +\infty$. Let $\Zz_{\alpha^k}^k, q_k$ and $a^0, b^0$ be as in Theorem~\ref{thma:ratio_of_partition_functions}. Assume moreover that for any even theta characteristics $a,b\in \{ 0,1/2 \}^g$ we have $\theta\chr{a}{b}(0,\Omega)\neq 0$. Then we have
  \begin{multline*}
    \lim\limits_{k\to +\infty}\frac{(-1)^{\Arf(q_k)}\Zz_{\alpha^k}^k}{\Zz_{\mathrm{dimer}}^k} = (-1)^{\Arf(q_0)}\Zz_1\left|\theta\chr{a^0 + a(\alpha)}{b^0 + b(\alpha)}(0,\Omega)\right|^2 = \\
    = (-1)^{\Arf(q_0)} \Zz_1e^{-\pi (a+a^0)\cdot \Im \Omega (a+a^0)}\left|\theta(-b-b^0 + \Omega(a + a^0),\Omega)\right|^2
  \end{multline*}
  where $\Zz_1$ is defined in the same way as in Theorem~\ref{thma:ratio_of_partition_functions} and $\theta(z,\Omega) = \theta\chr{0}{0}(z,\Omega)$.
\end{cor}

\begin{proof}
  By Theorem~\ref{thma:main2} we know that the tightness and the first moment hypothesis of Theorem~\ref{thma:main1} are satisfied, thus we are in the position to use Theorem~\ref{thma:ratio_of_partition_functions}. In the particular case of pillow surfaces and $G^k$ being embedded square grid on them we may take $\alpha_{G_k} = 0$ for all $k$, hence $\alpha_1 = 0$ and $\EE \psi^{2\alpha_1} = 0$. A straightforward computation shows that
  \begin{equation}
    \label{eq:ropf1}
    \theta\chr{a^0 - a}{b^0 - b}(0,\Omega) = \theta\chr{-a^0 + a}{-b^0 + b}(0,\Omega)\cdot \exp(4\pi i b^0\cdot (a^0 + a)).
  \end{equation}
  Note that $(-1)^{\Arf(q_0)} = (-1)^{4\pi i a^0\cdot b^0}$. Indeed, it follows from the formula~\eqref{eq:def_of_Arf} for the Arf's invariant and the relation between the theta characteristics $a^0,b^0$ and $q_0$ (see~\eqref{eq:a_b_for_quadratic_form}). Substituting~\eqref{eq:ropf1} into the formula in Theorem~\ref{thma:ratio_of_partition_functions} we get the first equality in the corollary. The second follows from Proposition~\ref{prop:of_theta}. 
\end{proof}

\subsection{Organization of the paper}
\label{subsec:organization}

Let us comment on how the rest of the paper is organized. As we already mentioned in the introduction, the proof of our main results is based on generalization of the technique developed by Dub\'edat in~\cite{DubedatFamiliesOfCR}. This requires a certain amount of preparatory lemmas. Section~\ref{sec:full plane case} is devoted to building the necessary discrete complex analysis tools: in Section~\ref{subsec:t-embedding_def} we briefly recall some key lemmas and constructions from~\cite{CLR1}; nothing new appears there, but we still include it for the sake of completeness and to fix the notation. In Sections~\ref{subsec:t-holom-regularity-lemmas},~\ref{subsec:local_inverting_on_t-emb} and~\ref{subsec:circle_pattern} we construct and estimate the discrete full-plane Cauchy kernel. This kernel will be used later as a building block for the discrete Cauchy kernel on a Riemann surface. In Section~\ref{subsec:multivalued} and~\ref{subsec:graph_on_a_cone} we model the situation near a conical singularity. Here multivalued functions on isoradial graphs have to be analyzed. Luckily, a great deal of tools for this was already developed by Dub\'edat in~\cite{DubedatFamiliesOfCR}, which eases our work.

We then jump from a local (full-plane) analysis to the global (on a surface) analysis. In Section~\ref{sec:The Riemann surface: continuous setting} we introduce all the necessary notations and constructions associated with the continuous Riemann surface setup. In particular, we prove the existence of a locally flat metric with prescribed conical singularities. In Section~\ref{subsec:Szego_kernel} we study the continuous Cauchy kernel, which the discrete one approximates.

At this point everything is ready to construct the perturbed discrete Cauchy kernel on a Riemann surface. The kernel is constructed and analyzed in Section~\ref{sec:discrete_Riemann_surface}.

After we construct and estimate the inverting kernel, we study in details the relation between the perturbed Kasteleyn operator and the dimer height function. This is done in Section~\ref{sec:determinant}. Following~\cite{DubedatFamiliesOfCR}, we establish a combinatorial relation between the determinant of the perturbed Kasteleyn operator and the characteristics function of the height field in Sections~\ref{subsec:Kasteleyn_thm} and~\ref{subsec:det_Kalpha}. We analyze the logarithmic variation of the determinant in Section~\ref{subsec:variation_of_det}.

A somewhat continuous analog of these results appears in Section~\ref{sec:free_field}. After a short technical discussion of the construction of the compactified free field we generalize the bosonization identity from the work~\cite{Bosonization}, adapting it for our needs. This identity express certain partition functions of the compactified free field via theta constants, which allows to link it with the determinantal identities for the corresponding partition function of the dimer model via the variation identities obtained in Section~\ref{subsec:variation_of_det}.

Finally, everything is prepared for the proof of the our main results. Section~\ref{sec:proof_of_main_results} is devoted to this: we prove Theorem~\ref{thma:main1} in Section~\ref{subsec:reconstruction_theorem}, Theorem~\ref{thma:main2} in Section~\ref{subsec:second_thm_tightness} and Theorem~\ref{thma:for_BLR} in Section~\ref{subsec:third_thm_BLR_setup}.

In the course of our work we use many different aspects of the theory of Riemann surfaces. We finalize our work with an appendix which contains a short guide through these aspects of the theory.

\section{T-embeddings on the full plane and local kernels}
\label{sec:full plane case}

This section is devoted to the study of ``local'' properties of the inverse Kasteleyn operator. This is done by considering infinite t-embeddings covering the Euclidean plane $\CC$, and constructing and estimating the inverse Kasteleyn operators associated with them.

We begin by reviewing the notion of t-embeddings of planar graphs and related concepts introduced in~\cite{CLR1}. We also recall the necessary facts from the regularity theory for t-holomorphic functions developed in that paper. Then, we use these tools to construct and estimate a full-plane inverse kernel for the Cauchy--Riemann operator on a t-embedding. Next, we proceed to the more special case of isoradial graphs, which we need to analyze discrete holomorphic functions near conical singularities. To prepare for it we study multivalued discrete holomorphic functions on full-plane isoradial graphs following the approach from~\cite[Section~7.1]{DubedatFamiliesOfCR}. We finish the section by analyzing the corresponding inverse Kasteleyn operator in an infinite cone.

Until the rest of the section $G$ denotes an infinite bipartite planar graph. The vertices of $G$ are split into two bipartite classes $B\sqcup W$, we call the vertices from $B$ \emph{black} and the vertices from $W$ \emph{white}. The graph $G^*$ denotes the corresponding dual graph. We will use the same notation for faces of $G^\ast$ and vertices of $G$.

\subsection{T-embeddings, t-holomorphic functions, and T-graphs}
\label{subsec:t-embedding_def}

We begin with the definitions and combinatorial properties of t-embeddings and t-holomorphic functions following~\cite{CLR1}. Consider an embedding $\Tt$ of $G^*$ into the plane which maps all faces from $B\cup W$ to bounded convex polygons, and such that the union of the faces covers the plane.
\begin{defin}
  \label{defin:t-emb}
  An embedding $\Tt$ of $G^*$ into $\CC$ is called an \emph{t-embedding} if for any vertex $v$ of $G^*$ the sum of black angles incident to $\Tt(v)$ is equal to $\pi$.
\end{defin}

There is a natural way to define Kasteleyn weights on $G$ given a t-embedding $\Tt$: if $b\sim w$ and $v_1v_2$ is the dual edge of $G^*$ oriented such that $b$ is on the right, then we set
\begin{equation}
  \label{eq:def_of_KTt}
  K_\Tt(w,b) = \Tt(v_2) - \Tt(v_1).
\end{equation}
It is straightforward to check~\cite[Section~2]{CLR1} that the matrix $K_\Tt$ satisfies the Kasteleyn condition provided $\Tt$ is a t-embedding.
 Another important object in the theory is the \emph{origami square root} function $\eta: B\cup W\to \TT = \{ z\in \CC\ \mid\ |z| = 1\}$. By the definition, $\eta$ is any function such that for any $b\in B,w\in W$
\begin{equation}
  \label{eq:def_of_eta}
  \bar{\eta}^2_b\bar{\eta}^2_w = \left( \frac{K_\Tt(w,b)}{|K_\Tt(w,b)|} \right)^2.
\end{equation}
This definition of $\eta$ is consistent around each vertex of $G^*$ due to the definition of t-embedding, therefore such an $\eta$ always exists. From now on we fix a choice of $\eta$. Following~\cite[Definition~5.1]{CLR1} we define t-holomorphic functions as follows:

\begin{defin}
  \label{defin:discrete_holom}
  Let $w\in W$ and $f$ be a function defined on all vertices $b$ such that $b\sim w$. Then $f$ is called \emph{t-white-holomorphic} at $w$ if for any $b\sim w$ we have $f(b)\in \eta_b \RR$, and the following relation holds:
  \begin{equation}
    \label{eq:Cauchy--Riemann_identity}
    \sum_{b\sim w} K_\Tt(w,b)f(b) = 0.
  \end{equation}
  Given a subset of white vertices and a function $f$ defined on the union of their neighbors, we call $f$ t-white-holomorphic if $f$ is t-white-holomorphic at each white vertex from the subset.
\end{defin}

The notion of a \emph{white splitting} of $\Tt$ is defined as follows. For any $w\in W$, draw an arbitrary maximal collection of non-intersecting diagonals of $w$, thus dividing it into triangles. Let $\Bws$ denote the union of $B$ and the set of diagonals and $\Wws$ denote the set of triangles. We write $b\sim w$ if $b\in \Bws$ is adjacent to $w\in \Wws$. If $b\sim w$, then we define $K_\Tt(w,b)$ by~\eqref{eq:def_of_KTt} using the vertices $v_1,v_2$ of $G^\ast$ adjacent to $b$ and $w$ and enumerated such that $w$ is on the left from $v_1v_2$. Note that the origami square root function $\eta$ can be extended to $\Bws\cup \Wws$ to still satisfy the condition~\eqref{eq:def_of_eta}.

We define a black splitting in the same way and use the notation $\Gbs,\Bbs, \Wbs$ in this case. The following lemma is a rephrasing of~\cite[Proposition~5.4]{CLR1}:
\begin{lemma}
  \label{lemma:function_on_splitting}
  Let $w\in W$ and the function $\eta$ be fixed. Let also a white splitting $\Ttws$ be given and let $w_1,\dots, w_m$ be the triangles corresponding to $w$. A function $f$ is t-white-holomorphic at $w$ if and only if $f$ can be extended to $\{ w_1,\dots,w_m \}$ and all the diagonals separating $w_i$'s in such a way that for any $k = 1,\dots, m$ and $b\in \Bws$ such that $b\sim w_k$ we have
  \[
    f(b) = \Pr(f(w_k),\eta_b\RR),
  \]
  where $\Pr(A, \eta \RR) = \frac{1}{2}\left( A + \eta^2\bar{A} \right)$ is the projection of the complex number $A$ onto the line $\eta\RR$.
\end{lemma}
The values $f(w_k)$ are sometimes referred as ``true complex values'' of the function $f$ in the literature. Notice that the function $f$ considered, for example, as a piece-wise constant function on the black faces of $\Tt$ cannot be regular since the arguments of the values of $f$ on these faces jump from one face to another. But a t-white-holomorphic function $f$ restricted to white faces will have certain regularity properties as we will see below.

We now define the origami map $\Oo$ associated with a t-embedding (cf.~\cite[Definition~2.7]{CLR1}):

\begin{defin}
  \label{defin:of_origami}
  The origami map $\Oo$ is primitive of the 1-form
  \begin{equation}
    \label{eq:def_of_origami}
    d\Oo(z)= \begin{cases}
      \eta_w^2\,dz,\qquad z\in \cup_{w\in W}\Tt(w),\\
      \bar{\eta}_b^2\,d\bar{z},\qquad z\in \cup_{b\in B}\Tt(b).
    \end{cases}
  \end{equation}
  By abusing the language we will also be using the notation $\Oo$ for the map $\Oo\circ\Tt$ from the vertices of $G^\ast$ to $\CC$.
\end{defin}

Let $\alpha\in \TT$ be a unit complex number. Consider the image of $G^*$ under the map $\Tt + \alpha^2\Oo$. From the definition of $\Oo$ it is clear that for any $b\in B$ the image $(\Tt + \alpha^2\Oo)(b)$ is a translation of $2\Pr(\Tt(b), \alpha\bar{\eta}_b\RR)$ and for any $w\in W$ the image $(\Tt + \alpha^2\Oo)(w)$ is a translation of $(1 + (\alpha\eta_w)^2)\Tt(w)$. We endow the image of $G^\ast$ with a graph structure as follows: the vertices are the images of vertices of $G^\ast$, and two vertices are connected by an edge if there is an edge $e$ of $G^\ast$ such that $(\Tt + \alpha^2 \Oo)(e)$ is a segment connecting these two vertices and not containing any other vertex inside. Denote this graph by $\Tt + \alpha^2\Oo$ by abusing the notation. It can be shown~\cite[Proposition~4.3]{CLR1} that the mapping $\Tt + \alpha^2\Oo$ does not have overlaps, hence, we can treat $\Tt + \alpha^2\Oo$ as a planar graph with faces given by those $(\Tt + \alpha^2\Oo)(w)$ for which $1 + (\alpha\eta_2)^2\neq 0$. Those $w$ for which $1 + (\alpha\eta_w)^2=0$ are called \emph{degenerate faces} of $\Tt + \alpha^2\Oo$. 

  Graphs $\Tt + \alpha^2\Oo$ have a special structure making them \emph{T-graphs} with possibly degenerated faces, we address the reader to~\cite[Definitions~4.1,4.2]{CLR1} for the precise definition of the notion of T-graphs. We can also consider $\Tt - \bar{\alpha}^2\bar{\Oo}$ which is again a T-graph with faces corresponding to $b\in B$; in what follows we will need both families of T-graphs.

Consider now a \emph{black} splitting of $\Tt$, so that $\Bbs$ is the set of triangles and $\Wbs$ is the set of diagonals and white faces of $G^*$. Note that for any diagonal $w$ we still have that $(\Tt + \alpha^2\Oo)(w)$ is a translation of $(1 + (\alpha\eta_w)^2)\Tt(w)$, where $\eta_w$ is the extension of $\eta$ to the triangles and the diagonals. We define \emph{transition rates} $q(v\to v')$ between the vertices of the T-graph $\Tt + \alpha^2\Oo$ associated with the given black splitting as follows.

\begin{itemize}
  \item Let $b\in \Bbs$ be a triangle and assume that the length of each side of this triangle remains non-zero in $(\Tt + \alpha^2\Oo)(b)$. Let $v_1,v_2,v_3$ be the vertices of $\Tt + \alpha^2\Oo$ which are the images of the vertices of $b$. Then we chose the $i$ such that $v_i$ lies between $v_{i-1}$ and $v_{i+1}$ and set
    \begin{equation*}
      q(v_i \to v_{i\pm 1}) = \frac1{|v_i - v_{i\pm 1}|\cdot |v_{i+1} - v_{i-1}|}.
    \end{equation*}
  \item Assume that $w\in \Wbs$ is such that $1 + (\alpha\eta_w)^2 = 0$ and let $b_1,b_2,\dots, b_d\in\Bbs$ be the triangles adjacent to $w$. Let $v_1,v_2,\dots,v_d$ be all the vertices of $\Tt + \alpha^2\Oo$ which are images of those vertices of $b_1,\dots,b_d$ that are not adjacent to $w$. Finally, let $v$ denote the vertex of $\Tt + \alpha^2\Oo$ corresponding to the image of $w$. For each $i = 1,\dots,d$ we set
    \[
      q(v\to v_i) = \frac{m_i}{|v-v_i|^2},\qquad m_i = \frac{S_{b_i}}{\sum_{j = 1}^d S_{b_j}} = \frac{|K_\Tt(w,b_i)|\cdot |v - v_i|}{\sum_{j = 1}^d|K_\Tt(w,b_j)|\cdot |v - v_j| }
    \]
    where $S_{b_j}$ is the area of the triangle $b_j$, and $K_\Tt(w,b_j)$ denotes the aforementioned extension of $K_\Tt$ to $\Wbs\times \Bbs$.
  \item If $v,v'$ are any two vertices in $\Tt + \alpha^2\Oo$ such that $q(v\to v')$ was not defined on the previous procedure, we declare $q(v\to v') = 0$.
\end{itemize}
Define $X^v_t$ to be the continuous time random walk on vertices of $\Tt + \alpha^2\Oo$ started at $v$ and jumping with rates specified above.

\begin{rem}
  \label{rem:variation_of_Xt}
  Note that the rates defining $X^v_t$ are set up such that $\Tr\Var(X_t^v) = t$, cf.~\cite[Remark~4.5]{CLR1}. Note that the quantities $S_b$ give an invariant measure for the random walk $X_t$ (cf. Lemma~\ref{lemma:invariant_measure}), which explains the meaning of the coefficients $m_i$ in the definition of intensities at a degenerate vertex, see~\cite[Remark~4.7]{CLR1} for details.
\end{rem}

A complex valued function $H$ on a subset of vertices of $\Tt + \alpha^2\Oo$ is called harmonic if it is a martingale with respect to $X_t$ stopped on the boundary of this subset. In~\cite{CLR1} a correspondence between primitives of t-holomorphic functions on t-embeddings and harmonic functions on T-graphs is constructed. We now describe this correspondence in details.

Let $\alpha\in \TT$ be arbitrary and black and white splittings be fixed. Let $U$ be a subset of the plane and $U_\alpha$ be its image on the T-graph $\Tt + \alpha^2 \Oo$. Let $H$ be a harmonic function on $U_\alpha$. Then it is easy to check (see~\cite[Section~4.2]{CLR1}) that for each $b\in \Bws$ such that $\Tt(b)\subset U$ the function $H$ extends to the segment $(\Tt + \alpha^2\Oo)(b)$ linearly. Define the function $D[H](b)$ by
\begin{equation}
  \label{eq:def_of_DH}
  dH = D[H](b)\,dz\qquad \text{along the segment} (\Tt + \alpha^2\Oo)(b).
\end{equation}

\begin{lemma}
  \label{lemma:Dharm=holom}
  Assume that we are in the setting above. Assume moreover that $H$ takes its values in $\alpha\RR$. Then for each $b\in \Bws$ contained in $U$ we have $D[H](b)\in \eta_b\RR$ and the function $D[H]$ is t-white-holomorphic at each $w\in W$ such that $w$ and all its neighbouring black faces are contained in $U$.
\end{lemma}
\begin{proof}
  See~\cite[Section~4.2]{CLR1}, in particular Definition~4.12 and the remark after Remark~4.13, and~\cite[Section~5]{CLR1} for the treatment of t-embeddings with non-triangular faces.
\end{proof}

Vice versa, given a t-holomorphic function one can integrate it and obtain a harmonic function on a T-graph. Assume that $f$ is a function defined on faces and diagonals from $\Bws\cup \Wws$ whose images are contained in $U$, and assume that for each $b\sim w$ we have $f(b) = \Pr[f(w), \eta_b\RR]$. Following~\cite{CLR1} we introduce the following piece-wise constant 1-form $\omega_f$ on $U$:
\begin{equation}
  \label{eq:def_of_omega_f}
  \omega_f(z) = \begin{cases}
    2f(b)\,dz,\qquad z\text{ belongs to }\Tt(b),\ b\in \Bws,\\
    f(w)\,dz + \overline{f(w)}\,d\Oo(z),\qquad z\text{ belongs to }\Tt(w),\ w\in \Wws.
  \end{cases}
\end{equation}
We have the following 
\begin{lemma}
  \label{lemma:primitive_of_f}
  The form $\omega_f$ is closed. If the primitive $\mathrm{I}_{\alpha\RR}[f]$ of $\Pr(\omega_f, \alpha\RR)$ is well-defined on $U$, then $\mathrm{I}_{\alpha\RR}[f]$ restricted to vertices of $G^*$ descends to a harmonic function on the T-graph $(\Tt + \alpha^2\Oo)(U)$ such that $D[H](b) = f(b)$ for any $b$ contained in $U$.
\end{lemma}
\begin{proof}
  Follows from~\cite[Proposition~3.7]{CLR1} and~\cite[Proposition~4.15]{CLR1}.
\end{proof}

Let us look closely at the case when $U$ is obtained from a simply-connected domain by removing $\Tt(w_0)$ for a fixed $w_0\in W$. Fix such an $U$ and a function $f$ as above. Let $\gamma$ be a simple loop in $U$ encircling $w_0$ and oriented counterclockwise. Note that the face of the T-graph $\Tt - \bar{\eta}_{w_0}^2\Oo$ corresponding to $w_0$ is degenerate. Let $b_1,\dots, b_d$ be the neighbors of $w_0$ and for each $i$ let $v_i$ be the unique vertex of $\Tt - \bar{\eta}_{w_0}^2\Oo$ lying on $(\Tt - \bar{\eta}_{w_0}^2\Oo)(b_i)$ and such that $q(v_0\to v_i)\neq 0$.

\begin{lemma}
  \label{lemma:monodromy_and_Kf}
  We have
  \begin{equation}
    \label{eq:monodromy_and_Kf1}
    \int_\gamma \omega_f = \sum_{b\sim w_0} K_\Tt(w_0,b)f(b)\in \bar{\eta}_{w_0}\RR.
  \end{equation}
  In particular, $H = \mathrm{I}_{i\bar{\eta}_{w_0}\RR}[f]$ is well-defined on $U$ and we have 
  \begin{equation}
    \label{eq:monodromy_and_Kf2}
    \sum_{k = 1}^d (H(v_k) - H(v_0))\cdot q(v_0 \to v_k) = \frac{i\sum_{k = 1}^d K_\Tt(w_0,b_k)f(b_k)}{\sum_{k = 1}^d |K_\Tt(w_0,b_k)|\cdot |v_0 - v_k|}
  \end{equation}

  Vice versa, if we have a function $H$ on $(\Tt - \bar{\eta}_{w_0}^2\Oo)(U)$ harmonic at all the vertices except $(\Tt - \bar{\eta}_{w_0}^2\Oo)(w_0)$, having its values in $i\bar{\eta}_{w_0}\RR$ and satisfying 
  \begin{equation}
    \label{eq:monodromy_and_Kf2.5}
    \sum_{k = 1}^d (H(v_k) - H(v_0))\cdot q(v_0 \to v_k) = \frac{i\alpha}{\sum_{k = 1}^d |K_\Tt(w_0,b_k)|\cdot |v_0 - v_k|}
  \end{equation}
  then $f = D[H]$ as a function on the black faces whose images are contained in $U$ is t-white-holomorphic at all the white faces $w$ such that the images of $w$ and all its neighbors are contained in $U$. Moreover, we have
  \[
    \sum_{k = 1}^d K_\Tt(w_0,b_k)f(b_k) = \alpha.
  \]
\end{lemma}
\begin{proof}
  Since $\omega_f$ is closed on $U$ outside the white face $w_0$, we can assume that $\gamma$ is the boundary of $\Tt(w_0)$. The relation~\eqref{eq:monodromy_and_Kf1} now follows from the definition of $K_\Tt$ given in~\eqref{eq:def_of_KTt}.

  To establish~\eqref{eq:monodromy_and_Kf2}, let us note that by the definition of $D[H]$ we have 
  \[
    H(v_k) - H(v_0) = f(b_k)\cdot (v_k - v_0)
  \]
  which together with the definition of the transition rate $q$ in the case of a degenerate face give
  \begin{equation}
    \label{eq:monodromy_and_Kf3}
    \sum_{k = 1}^d (H(v_k) - H(v_0))\cdot q(v_0 \to v_k) = \frac{\sum_{k = 1}^d |K_\Tt(w_0,b_k)|\frac{(v_k - v_0)}{|v_k - v_0|}f(b_k)}{\sum_{k = 1}^d |K_\Tt(w_0,b_k)|\cdot |v_0 - v_k|}.
  \end{equation}
  Recall that by the definition of $\eta$ we have $K_\Tt(w_0,b_k) \in \bar{\eta}_{w_0}\bar{\eta}_{b_k}\RR$, and in the same time $v_k - v_0\in i\bar{\eta}_{w_0}\bar{\eta}_{b_k}\RR$. It follows that
  \begin{equation}
    \label{eq:monodromy_and_Kf4}
    |K_\Tt(w_0,b_k)|\frac{(v_k - v_0)}{|v_k - v_0|} = \pm iK(w_0,b_k).
  \end{equation}
  The fact that the sign in~\eqref{eq:monodromy_and_Kf4} is ``$+$'' follows easily from orientation arguments. Thus~\eqref{eq:monodromy_and_Kf2} follows from~\eqref{eq:monodromy_and_Kf3} and~\eqref{eq:monodromy_and_Kf4}.

  The reverse statement is clear from the calculations above.
\end{proof}

Another important combinatorial result of~\cite{CLR1} shows that t-holomorphic functions are martingales with respect to the \emph{time-reversed} random walks on T-graphs. Let $\alpha\in \TT$ be an arbitrary unit complex number and consider the T-graph $\Tt - \bar{\alpha}^2\bar{\Oo}$; recall that black faces of $\Tt$ correspond to faces of this T-graph, while white faces are projected to segments. Let us fix an arbitrary \emph{white} splitting $\Ttws$ and define transition rates for the random walk on $\Tt - \bar{\alpha}^2\bar{\Oo}$ similarly as above, but with black and white faces playing opposite roles. For each white triangle $w\in \Wws$ of $\Ttws$ define $S_w$ to be the area of $\Tt(w)$. Let $v$ be a vertex of $\Tt - \bar{\alpha}^2\bar{\Oo}$. Define $\nu(v)$ as follows:

\begin{itemize}
  \item If there exists a white triangle $w(v)\in \Wws$ such that $v$ belongs to the interior of the segment $(\Tt - \bar{\alpha}^2\bar{\Oo})(w(v))$, then set $\nu(v) = S_{w(v)}$.
  \item Otherwise there exists a black face or a black diagonal $b\in \Bws$ such that $v = (\Tt - \bar{\alpha}^2\bar{\Oo})(b)$. Let $w_1,\dots,w_d$ be the white triangles adjacent to $b$. Set $\nu(v) = S_{w_1} + \dots + S_{w_d}$.
\end{itemize}

\begin{lemma}
  \label{lemma:invariant_measure}
  The weights $\nu$ define an invariant measure for the continuous time random walk on $\Tt - \bar{\alpha}^2\bar{\Oo}$.
\end{lemma}
\begin{proof}
  See~\cite[Proposition~4.11(vi)]{CLR1} and~\cite[Section~5]{CLR1}.
\end{proof}

Let $\alpha\in \TT$ and $U\subset \Wws \cup \Bws$ be given. Let us say that a vertex $v$ of the T-graph $\Tt - \bar{\alpha}^2\bar{\Oo}$ is covered if one of the following holds:

\begin{itemize}
  \item either there is a white triangle $w\in U$ such that $v$ is in the interior of the interval $(\Tt - \bar{\alpha}^2\bar{\Oo})(w)$; in this case set $w(v)= w$,
  \item or $v$ coincides with a degenerate face $(\Tt - \bar{\alpha}^2\bar{\Oo})(b)$ and all the triangles adjacent to $b$ belong to $U$; in this case set $w(v)$ to be any of these triangles.
\end{itemize}
Let $U^\alpha$ be the set of vertices of $\Tt - \bar{\alpha}^2\bar{\Oo}$ covered by $(\Tt - \bar{\alpha}^2\bar{\Oo})(U)$ and $\partial U^\alpha$ be the subset of $U^\alpha$ of those $v$ for which there exists $v'\notin U^\alpha$ such that $q(v'\to v)\neq 0$. Denote by $Y_t$ the reversed time random walk on $\Tt - \bar{\alpha}^2\bar{\Oo}$ defined with respect to the invariant measure $\nu$.

Let now $f$ be a function defined on $U$ and having the property that whenever $w,b\in U$ and $w\sim b$ and $f$ we have $\Pr[f(w),\eta_b\RR] = f(b)$.

\begin{lemma}
  \label{lemma:t-holom_are_matringales}
  In the setting above, the function $v\mapsto \Pr(f(w(v)), \alpha\RR)$ defined on $U^\alpha$ is a martingale with respect to the random walk $Y_t$ stopped on $\partial U^\alpha$.
\end{lemma}
\begin{proof}
  See~\cite[Proposition~4.17]{CLR1}.
\end{proof}

\subsection{Regularity lemmas for t-holomorphic functions}
\label{subsec:t-holom-regularity-lemmas}

Let now the parameters $0<\lambda<1$ and (a small) $\delta>0$ be fixed. From now on we assume that all the t-embeddings are weakly uniform as defined in Section~\ref{subsec:intro_graphs_on_Sigma0}. We do not impose the small origami assumption yet, but rather assume a weaker $\Lip = \Lip(1-\lambda, \lambda\delta)$ assumptions from~\cite{CLR1}. Namely, we assume that
\begin{equation}
  \label{eq:Lip_assumption}
  |\Tt(z_1)-\Tt(z_1)|\geq \lambda^{-1}\delta \quad \Rightarrow \quad |\Oo(z_1) - \Oo(z_2)|\leq (1-\lambda)|\Tt(z_1)-\Tt(z_1)|.
\end{equation}

\noindent Note that the argument principle together with the lipschitzness of $\Oo$ imply that whenever $\Ff$ is one of the mappings $\Tt + \alpha^2 \Oo$ we have
\begin{equation}
  \label{eq:balls_on_temb_vs_on_Tgraphs}
  B(\Ff(z), \lambda r) \subset \Ff(B(z,r)) \subset B(\Ff(z), (2-\lambda)r),
\end{equation}
provided $r\geq C\delta$ for some $C$ depending on $\lambda$ only, see~\cite[eq.~(6.1)]{CLR1} and the discussion after it. Given a function $f$ on a set $A$ define
\begin{equation}
  \label{eq:def_of_osc}
  \osc_A f = \sup_{a,b\in A} |f(a) - f(b)|
\end{equation}

\begin{lemma}
  \label{lemma:Harnak_on_T-graph}
  Assume that $\Tt$ is a weakly uniform t-embedding satisfying $\Lip$ assumption, let $\alpha\in \TT$ be given, $v$ be a vertex of the T-graph $\Tt + \alpha^2 \Oo$, and black and white splittings be fixed. Let $H$ be a harmonic function defined on vertices of $\Tt + \alpha^2 \Oo$ lying inside the ball $B(v,r)$. Then we have
  \[
    \max_{B(v, r/2)}|D[H]| \leq \frac{C}{r}\cdot \osc_{B(v,r)}H
  \]
  provided $r\geq C\delta$, where $C$ depends on $\lambda$ only.
\end{lemma}
\begin{proof}
  See~\cite[Theorem~6.17]{CLR1}; note that the second alternative in this theorem cannot occur because black faces are $\lambda\delta$-fat.
\end{proof}

\begin{lemma}
  \label{lemma:Holderness_of_thol_fcts}
  Assume that $\Tt$ is a weakly uniform t-embedding satisfying $\Lip$ assumption $z\in \CC$ is a point and $R>r>0$. Assume that a white splitting is given and $F$ is a function defined on those faces and diagonals from $\Wws\cup \Bws$ which are contained in $B(z,R)$, and assume that $f$ satisfies $\Pr[F(w), \eta_b\RR] = F(b)$ whenever $b\sim w$ are from the domain where $F$ is defined. Let $F^\circ$ denote the restriction of $F$ to white triangles. Then we have
  \[
    \osc_{B(z,r)} F^\circ \leq C(r/R)^\alpha \osc_{B(z,R)} F^\circ
  \]
  provided $r\geq C\delta$, where $C>0$ and $\alpha>0$ depend on $\lambda$ only.
\end{lemma}
\begin{proof}
  See~\cite[Proposition~6.13]{CLR1}.
\end{proof}

\begin{lemma}
  \label{lemma:Boundedness_of_white_values}
  Assume that $\Tt$ is a weakly uniform t-embedding satisfying $\Lip$ assumption and white splitting $\Wws$ is fixed. Let $U\subset \Wws\cup \Bws$ be given. Say that $w\in \Wws\cap U$ lies in the interior of $U$ if for all $b\sim w$ we have $b\in U$. Assume that $f:U\to \CC$ is t-white-holomorphic at each $w$ from the interior of $U$. Assume also that for each black $b\in U$ we have $|f(b)|\leq 1$. Then all the values of $f$ at $w$ from the interior of $U$ are bounded by a constant depending on $\lambda$ only.
\end{lemma}
\begin{proof}
  Weak uniformity of $\Tt$ implies in particular that each white triangle $\Tt(w)$, $w\in \Wws$, has at least one angle from $[\lambda, \pi-\lambda]$. Let $b_1,b_2\in \Bws$ be the two black faces/diagonals that are incident to $w$ and to this angle. If $f$ is t-white-holomorphic at $w$, then $\Pr[f(w); \eta_{b_i}\RR] = f(b_i)$ by the definition, which gives the desired bound on $f(w)$.
\end{proof}

\subsection{Local inverse operator for \texorpdfstring{$K_\Tt$}{KT} on a t-embedding}
\label{subsec:local_inverting_on_t-emb}

Recall that each t-embedding has a natural Kasteleyn operator $K_\Tt$ associated with it, see~\eqref{eq:def_of_KTt}. The goal of this subsection is to construct a good inverse $K_\Tt^{-1}$. We keep assuming that all t-embeddings are weakly uniform and satisfy $\Lip$ with some fixed $\lambda$ and small $\delta$, which is thought of as a mesh size of the embedding. We will prove the following

\begin{prop}
  \label{prop:full_plane_kernel}
  Let $\Tt$ be a full-plane weakly uniform t-embedding satisfying $\Lip$. Then there exists a unique inverting kernel $K_\Tt^{-1}(b,w), (b,w)\in B\times W$, such that
  \begin{enumerate}
    \item $K_\Tt^{-1}$ is both left and right inverse for $K_\Tt$.
    \item For any $b\in B$ and $w\in W$ we have $K_\Tt^{-1}(b,w)\in \eta_b\eta_w\RR$ and 
      \[
        |K_\Tt^{-1}(b,w)|\leq \frac{C}{\dist(\Tt(b),\Tt(w)) + \delta}
      \]
      where $C$ depends only on $\lambda$.
  \end{enumerate}
\end{prop}

To prove Proposition~\ref{prop:full_plane_kernel} we consider a sequence of Green functions on finite T-graphs exhausting the plane. Inverting kernel will be defined to be the limit of derivatives of these functions.

Let $\Tt$ be a full-plane weakly uniform t-embedding of $G^\ast$ satisfying $\Lip$, white and black splittings be fixed, and $w_0\in W$ be a white face of $G^\ast$. Consider the T-graph $\Tt -\bar{\eta}_{w_0}^2\Oo$, recall that the face $(\Tt -\bar{\eta}_{w_0}^2\Oo)(w_0)$ is degenerate, let $v_0$ be the corresponding vertex of this T-graph. Given $N>0$ and a vertex $v\in B(v_0, N\delta)$ of $\Tt -\bar{\eta}_{w_0}^2\Oo$ we define $X_t^v$ to be the random walk on $\Tt -\bar{\eta}_{w_0}^2\Oo$ associated with the given black splitting (see Section~\ref{subsec:t-embedding_def}) started at $v$ and stopped at the first time it left $B(v_0, N\delta)$. Define the Green's function on $B(v_0, N\delta)$ by
\begin{equation}
  \label{eq:def_of_H}
  H_{B(v_0,N\delta)}^\Tt (v) = \frac{\EE[\text{time $X_t^v$ spent at $v_0$}]}{\sum_{k = 1}^d |K_\Tt(w_0,b_k)|\cdot |v_0 - v_k|}
\end{equation}
where $v_1,\dots, v_d$ are as in Lemma~\ref{lemma:monodromy_and_Kf}.
By the definition, $H_{B(v_0,N\delta)}^\Tt$ is harmonic on $B(v_0, \delta N)\smm \{ v_0 \}$ and vanishes outside $B(v_0,\delta N)$. Note that Lemma~\ref{lemma:monodromy_and_Kf} provides another characterization of $H_{B(v_0, N\delta)}^\Tt$: this is the unique function on the vertices of $\Tt - \bar{\eta}_{w_0}^2\Oo$ which is zero outside $B(v_0, N\delta)$, harmonic on $B(v_0, N\delta)\smm\{ v_0 \}$ and satisfying
\begin{equation}
  \label{eq:dbar_of_Gf}
  \sum_{b\sim w_0}K_\Tt(w,b)D[i\bar{\eta}_{w_0}H_{B(v_0, N\delta)}^\Tt](b) = \bar{\eta}_{w_0}.
\end{equation}
Note that $H_{B(v_0,N\delta)}^\Tt$ is scale invariant, i.e. 
\[
  H_{B(kv_0,kN\delta)}^{k\Tt}(kv) = H_{B(v_0,N\delta)}^\Tt(v).
\]
Define the annulus
\[
  A(v_0, r, R) = B(v_0, R)\smm B(v_0, r)
\]

\begin{prop}
  \label{prop:osc_of_Green_function}
  There exists a constant $C>0$ depending only on $\lambda$ such that whenever we are in the setup above we have
  \[
    \osc_{B(v_0, \delta)}H_{B(v_0,N\delta)}^\Tt\leq C,\qquad \osc_{A(v_0, 2^{k-1}\delta, 2^k\delta)}H_{B(v_0,N\delta)}^\Tt\leq C,\quad k = 1,\dots, \log_2N.
  \]
\end{prop}
\begin{proof}
  The proof will use compactness arguments and is very similar to the proof of~\cite[Theorem~3.1]{CLR2}. Assume by contradiction that there is a sequence $G_n^\ast, \Tt_n, w_0^n, N_n$ and $U_n$ such that $\osc_{U_n}H_n\to +\infty$; here each $U_n$ is either an annulus or a ball from the proposition. Let $\delta_n$ be the scale associated with $\Tt_n$; set $H_n = H^{\Tt_n}_{B(v_0^n, N_n\delta_n)}$ for simplicity. We proceed with the following steps.

  \emph{Step 1.} Prove that there is a sequence $k_n\to \infty$ such that 
  \begin{equation}
    \label{eq:osc_of_Green_function1}
    \osc_{A(v_0^n, 2^{k_n-1}\delta_n, 2^{k_n}\delta_n)}H_n\to +\infty.
  \end{equation}
  Assume that there is a $k>0$ and a subsequence of $n$'s such that $\osc_{A(v_0, 2^{l-1}\delta_n, 2^l\delta_n)} H_n$ are bounded simultaneously when $n$ runs along this subsequence and $l = k+1,\dots, \log_2N_n$. Note that $H_{B(v_0, 2^k\delta_n)}^{\Tt_n} = H_n - G_n$ where $G_n$ is the harmonic extension of $H_n$ inside $B(v_0, 2^k\delta_n)$. It follows from our assumption that $\osc_{B(v_0, 2^k\delta_n)} G_n$ is bounded along the subsequence, hence 
  \begin{equation}
    \label{eq:ooGf1}
    \osc_{B(v_0, 2^k\delta_n)}H_{B(v_0, 2^k\delta_n)}^{\Tt_n}\to +\infty\qquad n\text{ from the subsequence.}
  \end{equation}
  On the other hand, one can easily show that $H_{B(v_0, 2^k\delta_n)}^{\Tt_n}$ is bounded uniformly in $n$ by estimating the numerator and the denumerator in~\eqref{eq:def_of_H}. Recall that $H$ is scale invariant, hence we can assume that $\delta_n = 1$ for all $n$ so that the disc $B(v_0, 2^k)$ is fixed and only graphs vary. Then notice that the expected time the random walk $X_t^{v_0}$ has spent in $B(v_0, 2^k)$ before exiting it is bounded by some constant depending on $k$ only (cf. Remark~\ref{rem:variation_of_Xt}), thus the numerator in~\eqref{eq:def_of_H} is bounded. Assumption~\ref{item:full-plane_unif_black_face} from the set of weak uniformity assumptions implies that the denumerator in~\eqref{eq:def_of_H} is bounded from below by a constant depending on $\lambda$ only. It follows that $H_{B(v_0, 2^k\delta_n)}^{\Tt_n}$ is bounded uniformly in $n$ and we get a contradiction.

  \emph{Step 2.} Construct a sequence of auxiliary functions $\tilde{H}_n$. Let us fix a sequence $k_n\to \infty$ such that~\eqref{eq:osc_of_Green_function1} holds. After a proper rescaling and translation we can assume that $2^{k_n}\delta_n = 1$ for each $n$ and both $\Tt_n(w_0^n)$ and $v_0^n$ are in $\cst\cdot \delta_n$-neighborhood of zero where $\cst>0$ is a constant depending on $\lambda$ only. Let $m_n = \min_{v\in B(v_0^n, 1)}H_n(v)$ and $C_n = \osc_{A(v_0^n, 2^{-1}, 1)} H_n$. Consider the new function
  \begin{equation}
    \label{eq:osc_of_Green_function2}
    \tilde{H}_n = C_n^{-1}(H_n - m_n).
  \end{equation}
  Note that $\tilde{H}_n\geq 0$ on $B(v_0^n,1)$ by the construction. Let $v^n_-\in B(v_0^n, 1)$ and $v^n_+\in A(v_0^n, 2^{-1}, 1)$ be such that
  \[
    H_n(v^n_-) = m_n,\qquad H_n(v^n_+) = \max_{v\in A(v_0^n, 2^{-1}, 1)} H_n(v),
  \]
  we have $\tilde{H}_n(v^-_n) = 0$ by the definition. By the maximal principle $v^n_-\in A(v_0^n, 2^{-1}, 1)$, hence $\tilde{H}_n(v^n_+) = 1$ and $\tilde{H}_n(v)\leq 1$ if $v\in A(v_0^n, 2^{-1}, 1)$.

  Since $\tilde{H}_n$ is a harmonic function on $B(v_0^n, N_n\delta_n)\smm \{ v_0^n \}$ and attains its maximum at $v_0^n$, there exists a path $\gamma_-^n$ on the T-graph $\Tt_n - \bar{\eta}^2_{w_0^n}\Oo_n$ which goes from $v^n_-$ to the boundary of $B(v_0^n, N_n\delta_n)$ and such that $\tilde{H}_n$ is non-positive along $\gamma_-^n$. Using the uniform crossing property of the random walk on the T-graph $\Tt - \bar{\eta}^2_{w_0}\Oo$ (see~\cite[Lemma~6.8]{CLR1}), the fact that $\tilde{H}_n$ is non-positive along $\gamma_n^-$ and is bounded by 1 on $A(v_0^n, 1/2,1)$ we conclude that there is an $\eps>0$ such that for all $n$ and $v\in A(v_0^n, 3/4, 1)$ we have $\tilde{H}_n(v)\leq 1 -\eps$. Note in particular that $v_+^n\notin A(v_0^n, 3/4, 1)$. 

  Applying the Harnack estimate~\cite[Proposition~6.9]{CLR1} we conclude that for each compact $K\subset B(0, 1)\smm \{ 0 \}$ the family of functions $\tilde{H}_n$ is bounded on vertices from $K$ uniformly in $n$. Lemma~\ref{lemma:Harnak_on_T-graph} and Arzel\'a--Ascoli lemma then ensure the existence of a continuous function $h$ on $B(0,1)\smm\{ 0 \}$ such that a subsequence of $\tilde{H}_n$ converges to $h$ uniformly on compacts and such that $h$ has the following properties:
  \begin{itemize}
    \item $h$ is non-negative;
    \item for any $z\in A(0,3/4,1)$ we have $h(z)\leq 1 -\eps$;
    \item there is a point $v_+\in B(0,1)\smm\{ 0 \}$ such that $h(v_+) = 1$.
  \end{itemize}

  \emph{Step 3.} Study the derivative of $h$. Let $U_n$ be the set of all faces of $G_n^\ast$ which are mapped to $B(0, 1)$ under the map $\Tt_n - \bar{\eta}^2_{w_0}\Oo_n$. Consider the function $f_n = D[i\bar{\eta}\tilde{H}_n]$ defined on black faces from $U_n$. Due to Lemma~\ref{lemma:Dharm=holom}, functions $f_n$ are t-white-holomorphic at each $w\neq w_0$ which belongs to $U_n$ with all its neighbors. By Lemmas~\ref{lemma:Harnak_on_T-graph} and~\ref{lemma:Boundedness_of_white_values}, for any compact $K\subset B(0,1)\smm\{ 0 \}$ the maximum of values of $f_n$ on black and white faces which are mapped into $K$ is bounded uniformly in $n$.

  Recall that the functions $\Oo_n\circ\Tt_n^{-1}$ are all $(1-\lambda)$-Lipshitz. After passing to a subsequence we may assume that $\Oo_n\circ\Tt_n^{-1}$ converges uniformly on $V = \cap_{m\geq 1}\cup_{n\geq m} V_n$ to a $(1-\lambda)$-Lipshitz function $\vartheta$, where $V_n = (\Tt_n - \bar{\eta}^2_{w_0}\Oo_n)^{-1}(B(0,1))$. We have $B(0,1) = (\mathrm{id} - \bar{\eta}_0^2\vartheta)(V)$ for some $\eta_0\in \TT$. For each $n$, define the function $F_n:V\to \CC$ by 
  \[
    F_n(z) = f_n(w),\qquad \text{$w\in U_n$ is such that $\Tt_n(w)$ is the nearest to $z$ where $f_n$ is defined}.
  \]
  Applying Lemma~\ref{lemma:Holderness_of_thol_fcts} and Arzel\'a--Ascoli lemma and passing to a subsequence we can assume that $F_n$ converge to a continuous function $f$ on $V\smm \{ 0 \}$ uniformly on compacts from the interior of $V$, and we have for $H = h\circ (\mathrm{id} + \vartheta)^{-1}$ and any $z_1,z_2\in V\smm\{ 0 \}$
  \[
    H(z_1) - H(z_2) = \Im\left( \eta_0\int_{z_1}^{z_2} (f\,dz + \overline{f}\,d\vartheta)\right).
  \]
  Recall that the constants $C_n$ from~\ref{eq:osc_of_Green_function2} tend to $+\infty$. Thus, by Lemma~\ref{lemma:monodromy_and_Kf} and the normalization in the definition of Green's function~\eqref{eq:def_of_H}, the form $f\,dz + \overline{f}\,d\vartheta$ is exact in $V\smm\{ 0 \}$. Let $z_+\in V$ be such that $z_+ + \vartheta(z_+) = v_+$. Then we can define
  \begin{equation}
    \label{eq:osc_of_Green_function3}
    F(z) = \int_{z_+}^z (f\,dz + \overline{f}\,d\vartheta)
  \end{equation}
  and we have
  \begin{equation}
    \label{eq:osc_of_Green_function4}
    H(z) = 1 + \Im(\eta_0 F(z)).
  \end{equation}

  \emph{Step 4.} Obtain a contradiction. From the definition of $F$ and the fact that $\vartheta$ is $(1-\lambda)$-Lipshitz it follows that the distortion of $F$ at any point of $V\smm\{ 0 \}$ is bounded by $\frac{2-\lambda}{\lambda}$ from above (see~\cite[Chapter~2.4]{AstalaBook} for the definition of the distortion of a quasiconformal map). Note also that since $f$ is continuous and $\theta$ is Lipshits, we have $F\in W^{1,2}_{\mathrm{loc}}(V\smm \{ 0 \})$. Thus, by Stoilow factorization theorem (see~\cite[Theorem~5.5.1]{AstalaBook}) there exists a homeomorphism $G: V\to V'$ and a holomorphic function $\vphi$ on $V'\smm\{ G(0) \}$ such that 
  \[
    F(z) = \vphi(G(z)).
  \]
  From~\eqref{eq:osc_of_Green_function4} and the properties of $h$ we have that $1 + \Im \eta_0 \vphi \geq 0$, hence $\vphi$ extends holomorphically to $V'$. But on the other hand $\Im \eta_0 \vphi(w) \leq -\eps$, if $w$ is close enough to the boundary of $V'$, and $\Im \eta_0 \vphi(G(z_+)) = 0$, which is a contradiction.
\end{proof}

Now we can prove Proposition~\ref{prop:full_plane_kernel}.

\begin{proof}[Proof of Proposition~\ref{prop:full_plane_kernel}]
  Fix a white face $w_0\in W$ and for any $N>0$ consider
  \[
    K_{\Tt,N}^{-1}(b,w_0) = \eta_{w_0}D[i\bar{\eta}_{w_0}H^\Tt_{B(v_0, N\delta)}]
  \]
  where $H^\Tt_{B(v_0, N\delta)}$ is the Green's function defined by~\eqref{eq:def_of_H}. Using Proposition~\ref{prop:osc_of_Green_function} and Lemma~\ref{lemma:Harnak_on_T-graph} we obtain
  \[
    K_{\Tt,N}^{-1}(b,w_0) = O\left(\frac{1}{\dist(\Tt(b),\Tt(w_0))+ \delta}\right)
  \]
  Uniformly in $b$ and $N$ given that $N>|b|$. Let $K_\Tt^{-1}$ be any subsequential limit of $K_{\Tt,N}^{-1}$. We have for this limit
  \begin{equation}
    \label{eq:infas1}
    K_\Tt^{-1}(b,w_0) = O\left(\frac{1}{\dist(\Tt(b),\Tt(w_0))+ \delta}\right)
  \end{equation}
  and
  \[
    \sum_{b\sim w}K_\Tt(w,b)K_\Tt^{-1}(b,w_0) = \delta_{w_0}(w), \qquad K_\Tt^{-1}(b,w_0)\in \eta_b\eta_{w_0}\RR
  \]
  by Lemma~\ref{lemma:Dharm=holom} and~\eqref{eq:dbar_of_Gf}, so that $K_\Tt^{-1}$ is a right inverse for $K_\Tt$. It remains to show that $K_\Tt^{-1}$ is also a left inverse. To this end fix a $b_0\in B$ and consider the function
  \[
    f(b) = \eta_{b_0}\left( \sum_{w\sim b_0} K_\Tt^{-1}(b,w) K_\Tt(w,b_0) - \delta_{b_0}(b)\right).
  \]
  It is straightforward to see that $f$ is t-white-holomorphic and vanishes at infinity, hence $f\equiv 0$ due to Lemma~\ref{lemma:Holderness_of_thol_fcts}.
\end{proof}

\subsection{Asymptotic of \texorpdfstring{$K_\Tt^{-1}$}{KT-1} under the \texorpdfstring{$O(\delta)$}{O(delta)}-small origami assumption}
\label{subsec:circle_pattern}

In this section we additionally impose the \emph{$O(\delta)$-small origami} assumption~\eqref{eq:intro_small_origami_assumption} on t-embeddings we consider. Recall that this assumption means that the primitive $\Oo$ of the origami 1-form~\eqref{eq:def_of_origami} can be chosen in such a way that
\begin{equation}
  \label{eq:small_origami_assumption}
  |\Oo(z)| \leq \lambda^{-1}\delta
\end{equation}
for all $z\in \CC$. In what follows we fix such a primitive. Recall that any t-embedding satisfying this assumption has a circle pattern with all radii being at most $2\lambda^{-1}\delta$. We fix such a circle pattern and use it to identify vertices of $G$ and points on the plane. Namely, we can fix the positions of the vertices of $G$ such that for every edge $b\sim w$ and the corresponding dual edge $v_1\sim v_2$ of the t-embedding the points $b$ and $w$ are symmetric with respect to the line passing through $v_1$ and $v_2$. The condition~\eqref{eq:small_origami_assumption} implies that the aforementioned positions of the vertices of $G$ can be chosen such that the distance between every the vertices of $G$ and adjacent vertices of the t-embedding does not exceed $2\lambda^{-1}\delta$ (note that all the vertices of $G$ adjacent to a given vertex $v$ of the t-embedding are now placed on a circle centered at $v$).

From now on we will be using the same notation for the vertices of $G$, the faces of $\Tt$ and vertices of the circle pattern corresponding to $\Tt$. Note that for any $b,w$ the distance between the corresponding point and the corresponding face is at most $2\lambda^{-1}\delta$.

\begin{lemma}
  \label{lemma:local_area_relations}
  Let $w$ be a convex $n$-gon on the plane $\CC$ with vertices $v_1,v_2,\dots, v_n$ listed counterclockwise. Let $z_1\in \CC$ be arbitrary and $z_2,z_3,\dots, z_n$ be constructed inductively such that $z_{j+1}$ is obtained from $z_j$ by applying the reflection along $v_{j-1}v_j$ and then the reflection along $v_jv_{j+1}$. Then
  \[
    \sum_{j = 1}^n \bar{z}_j(v_j - v_{j-1}) = 4i\Area(w).
  \]
\end{lemma}
\begin{proof}
  Notice that the formula in the lemma is real analytic in $z_1$, thus it is enough to prove it when $z_1$ belongs to an arbitrary non-empty open set of the plane. To this end, choose a point $z_0$ inside $W$ and set $z_1$ to be its image under the reflection with respect to $v_nv_1$. Then, each $z_j$ becomes the image of $z_0$ under the reflection with respect to $v_{j-1}v_j$ and we have
  \[
    \sum_{j = 1}^n \bar{z}_j(v_j - v_{j-1}) = \sum_{j = 1}^n (\bar{z}_j - \bar{z}_0)(v_j - v_{j-1}) = 4i\Area(w).
  \]
\end{proof}

This computation implies the following:

\begin{lemma}
  \label{lemma:approx_of_dbar}
  In the setting above, let $w\in W$ be an arbitrary white vertex of $G$ and let $\vphi$ be a $C^2$ function defined in a convex neighborhood of $\Tt(w)$ containing all images of vertices $b$ such that $b\sim w$. Keeping the same notation for a vertex and for its image we have
  \[
    \sum_{b\sim w}K_\Tt(w,b)\vphi(b) = 4i\mu_w \frac{\partial}{\partial\bar z} \vphi(w) + \partial \vphi(w)\sum_{b\sim w}K(w,b)(b - w) + O(\|\vphi''\|_\infty\delta^3)
  \]
  where $\mu_w$ is the area of the face $\Tt(w)$ of t-embedding and the constant in $O(\ldots)$ depends only on $\lambda$.
\end{lemma}
\begin{proof}
  Let $v_1,v_2,\dots,v_n$ be the vertices of the polygon corresponding to the face $\Tt(w)$, and let $b_i w$ be the edge of $G$ dual to the edge $v_{i-1}v_i$ of $\Tt$. Note that $K_\Tt(w,b_i) = v_i-v_{i-1}$ by the definition~\eqref{eq:def_of_KTt}, therefore we have
  \[
    \sum_{b\sim w}K_\Tt(w,b)\vphi(b) = \sum_{j = 1}^n \vphi(b_j) (v_j-v_{j-1}).
  \]
  The lemma now follows from the Taylor expansion of $\vphi$ and Lemma~\ref{lemma:local_area_relations}.
\end{proof}

The small origami assumption means in particular that t-holomorphic functions approximate the functions holomorphic in the complex structure of the plane (cf. the discussion in Section~\ref{subsec:intro_graphs_on_Sigma0}). Thus, we expect the inverting kernel from Proposition~\ref{prop:full_plane_kernel} to approximate the Cauchy kernel. In the next theorem we show that this is happening with a polynomial error. 

\begin{thmas}
  \label{thmas:parametrix}
  Let $G$ be a bipartite graph and $\Tt$ be a weakly uniform t-embedding of $G^\ast$ with $O(\delta)$-small origami, where $\lambda,\delta$ are the parameters in the weak uniformity assumption. Assume that the vertices of $G$ are identified with points on the plane such that each vertex is at the distance at most $\lambda^{-1}\delta$ from the corresponding face of $\Tt$. Then there exists a constant $\beta$ depending only on $\lambda$, and a unique inverse kernel $K^{-1}_\Tt(b,w)$ such that
  \[
    \begin{split}
      &K_\Tt^{-1}(b,w) =  \Pr\left[ \frac{1}{\pi i (b - w)}, \eta_b\eta_w\RR \right] + O\left( \frac{\delta^\beta}{|b-w|^{1+\beta}} \right),\\
      &K_\Tt^{-1}(b,w) = O\left( \frac{1}{|b-w| + \delta} \right).
    \end{split}
  \]
  for each $b,w$, where the constants in $O(\ldots)$ depend on $\lambda$ only.
  This kernel satisfies $K_\Tt^{-1}(b,w) \in \eta_b\eta_w\RR$ for each $b,w$.
\end{thmas}

The goal of this subsection is to prove this theorem. Without loss of generality we can assume that the vertices of $G$ are embedded via the corresponding circle pattern. We keep using the same notation for a vertex of $G$ and its image on the plane. Given $b\in B$ or $w\in W$ denote by $\mu_b$ (resp. $\mu_w$) the area of the face $\Tt(b)$ (resp. $\Tt(w)$) of the t-embedding. We begin with some preliminary observations.

\begin{lemma}
  \label{lemma:area_estimate}
  Let $D\subset \CC$ be a convex set of diameter $r$. Then
  \[
    \Area(D) = 2\sum_{b\in D}\mu_b + O(\delta r)
  \]
  provided $r\geq \cst \delta$ where $\cst$ and the constant in $O(\ldots)$ depend only on $\lambda$.
\end{lemma}
\begin{proof}
  Indeed, we have
  \begin{equation*}
    \frac{i}{2}\int_D d\Oo\wedge d\bar{\Oo} = \sum_{w\in D}\mu_w - \sum_{b\in D}\mu_b + O(\delta r) = \frac{i}{2}\int_{\partial D} \Oo\wedge d\bar{\Oo} + O(\delta r) = O(\delta r)
  \end{equation*}
  where the last estimate follows from~\eqref{eq:small_origami_assumption}.
\end{proof}

Let $\vphi$ be a smooth approximation of the identity in $\CC$, that is $C^\infty(\CC)\ni\vphi\geq 0$, $\supp \vphi\subset B(0,1)$ and $\int_\CC\vphi = 1$. Set $\vphi_t(z) = t^{-2}\vphi(zt^{-1})$. Given a function $f$ defined on some portion of black vertices and $t>0$ define the smooth function
\begin{equation}
  \label{eq:def_of_ft}
  f_t(z) = 4\sum_{b\in B} \vphi_t(z+b)f(b)\cdot \mu_b
\end{equation}
for all $z$ for which the expression on the right-hand side makes sense. 
\begin{lemma}
  \label{lemma:f_t_vs_f}
  Assume that $f$ is t-white holomorphic and let $\eps\in (0,1/2)$ be arbitrary. There are constants $C,\beta_1>0$ depending only on $\lambda$ and $\eps$ such that if $z$ is a point where $f_{\delta^\eps}$ is defined, and $\Tt(w_s),w_s\in \Wws,$ is the closest white triangle to $z$, then
  \[
    |f_{\delta^\eps}(z) - f(w_s)| \leq C\cdot \max_{b\in B\colon |b - z| < \delta^{\eps/2}}|f(b)|\cdot \delta^{\beta_1}.
  \]
\end{lemma}
\begin{proof}
  Let us fix $z\in \CC$ and $w_s$ as in the statement of the lemma. From basic properties of t-white-holomorphic functions (see Lemma~\ref{lemma:function_on_splitting}) we have that for any black vertex $b$ and white triangle $w\in \Wws$ adjacent to it we have
  \begin{equation}
    \label{eq:f_t_vs_f1}
    2f(b) = f(w) + \eta_b^2\overline{f(w)}.
  \end{equation} 
  Assume that $b\in B(z,\delta^\eps)$. Then by Lemma~\ref{lemma:Holderness_of_thol_fcts} and Lemma~\ref{lemma:Boundedness_of_white_values} we can replace $f(w)$ with $f(w_s)$ in~\eqref{eq:f_t_vs_f1} by introducing an error $O(\delta^{\alpha\eps/2}\max_{b\in B\colon |b - z| < \delta^{\eps/2}}|f(b)|)$:
  \begin{equation}
    \label{eq:f_t_vs_f15}
    2f(b) = f(w_s) + \eta_b^2\overline{f(w_s)} + O\left(\delta^{\alpha\eps/2}\max_{b\in B\colon |b - z| < \delta^{\eps/2}}|f(b)|\right).
  \end{equation}
  By~\eqref{eq:f_t_vs_f15} and the definition~\eqref{eq:def_of_ft} of $f_{\delta^\eps}$ we can write
  \begin{multline}
    \label{eq:f_t_vs_f2}
    f_{\delta^\eps}(z) =\\
    = 2f(w_s)\sum_{b\in B}\vphi_{\delta^\eps}(z+b)\mu_b + 2\overline{f(w_s)}\sum_{b\in B}\eta_b^2\vphi_{\delta^\eps}(z+b)\mu_b +\\
    + O\left(\delta^{\alpha\eps/2}\max_{b\in B\colon |b - z| < \delta^{\eps/2}}|f(b)|\right).
  \end{multline}
  Let us estimate both sums in the right-hand side of~\eqref{eq:f_t_vs_f2}. Using that the diameter of $\Tt(b)$ is of order $\delta$ (since we assume $\Tt$ to be weakly uniform, see Section~\ref{subsubsec:intro_t-embedding_prelim}) we can write (recall that $\mu_b$ is the area of $\Tt(b)$)
  \[
    2\phi_{\delta^\eps}(z+b)\mu_b = i\int_{\Tt(b)}\vphi_{\delta^\eps}(z+u)\,du\wedge d\bar{u} + \indic[z+b\in \supp(\phi_{\delta^\eps})]\cdot\mu_b\cdot O(\delta^{1-3\eps}).
  \]
  Since the area of the support of $\vphi_{\delta^\eps}$ is of order $\delta^{2\eps}$ and $\eps < 1/2$ the previous display implies
  \begin{equation}
    \label{eq:f_t_vs_f3}
     2\sum_{b\in B}\vphi_{\delta^\eps}(z+b)\mu_b = i\sum_{b\in B} \int_{\Tt(b)} \vphi_{\delta^\eps}(z+u)\,du\wedge d\bar{u} + O(\delta^{1/2}).
  \end{equation}
  Note that for any square $Q$ on the plane with the center $u_0$ and the side length $\delta^{2\eps}$ we have, by Lemma~\ref{lemma:area_estimate},
  \[
    i\int_{Q\cap \bigcup_{b\in B} \Tt(b)} \vphi_{\delta^\eps}(z+u)\,du\wedge d\bar{u} = \vphi_{\delta^\eps}(z + u_0) \Area(Q) + O(\delta^{-\eps}\Area(Q)).
  \]
  Dividing the plane into squares of side length $\delta^{2\eps}$, integrating over those squares that intersect the support of $\vphi_{\delta^\eps}(z+\cdot)$ and using the relation above we can estimate
  \begin{equation}
    \label{eq:f_t_vs_f4}
    2i\sum_{b\in B} \int_{\Tt(b)} \vphi_{\delta^\eps}(z+u) = i\sum_{b\in B} \int_{\Tt(b)} \vphi_{\delta^\eps}(z+u)\,du\wedge d\bar{u} + O(\delta^{1/2}) = 1 + O(\delta^\eps).
  \end{equation}
  To estimate the second sum in~\eqref{eq:f_t_vs_f2}, note that by the definition of the origami~\eqref{eq:def_of_origami} we have
  \[
    2\sum_{b\in B}\eta_b^2\vphi_{\delta^\eps}(z+b)\mu_b = i\sum_{b\in B} \int_{\Tt(b)} \vphi_{\delta^\eps}(z+u)\,d\bar{\Oo}\wedge d\bar{u} + O(\delta^{1/2}),
  \]
  and for any square $Q$ on the plane with the center $u_0$ and the side length $\delta^{2\eps}$
  \begin{multline*}
    \int_{Q\cap \bigcup_{b\in B} \Tt(b)} \vphi_{\delta^\eps}(z+u)\,d\bar{\Oo}\wedge d\bar{u} =\\
    = \vphi_{\delta^\eps}(z+u_0)\int_{Q\cap \bigcup_{b\in B} \Tt(b)} d\bar{\Oo}\wedge d\bar{u} + O(\delta^{-\eps}\Area(Q)).
  \end{multline*}
  We then notice that $d\bar{\Oo}\wedge d\bar{u}$ vanishes on any white face, therefore
  \[
    \int_{Q\cap \bigcup_{b\in B} \Tt(b)} d\bar{\Oo}\wedge d\bar{u} = \int_Qd\bar{\Oo}\wedge d\bar{u} = O(\delta^{1+2\eps})
  \]
  due to the small origami assumption. We conclude that
  \[
    \int_{Q\cap \bigcup_{b\in B} \Tt(b)} \vphi_{\delta^\eps}(z+u)\,d\bar{\Oo}\wedge d\bar{u} = O(\delta^{-\eps}\Area(Q))
  \]
  and thus, by dividing the plane into squares of side length $\delta^{2\eps}$ we get
  \begin{equation}
    \label{eq:f_t_vs_f5}
    2\sum_{b\in B}\eta_b^2\vphi_{\delta^\eps}(z+b)\mu_b = i\sum_{b\in B} \int_{\Tt(b)} \vphi_{\delta^\eps}(z+u)\,d\bar{\Oo}\wedge d\bar{u} + O(\delta^{1/2}) = O(\delta^\eps).
  \end{equation}
  Plugging~\eqref{eq:f_t_vs_f4} and~\eqref{eq:f_t_vs_f5} to~\eqref{eq:f_t_vs_f2} we get the result.
\end{proof}

\begin{lemma}
  \label{lemma:derivative_of_approx_identity}
  For any square $Q$ with the side of length $r$ we have
  \[
    \sum_{b\colon b\in Q}\sum_{w\sim b}K_\Tt(w,b)(w-b) = O(\delta r),\qquad \sum_{b\colon b\in Q}\sum_{w\sim b}\bar{\eta}_w^2\overline{K_\Tt(w,b)}(w-b) = O(\delta r)
  \]
  provided $r\geq \cst \delta$ where $\cst$ and the constants in $O(\ldots)$ depend only on $\lambda$.
\end{lemma}
\begin{proof}
  Using that $\sum_{b\colon b\sim w_0}K_\Tt(w_0,b) = \sum_{w\colon w\sim b_0} K_\Tt(w,b_0) = 0$ for any $b_0,w_0$ we can easily get that
  \[
    \left|\sum_{b\colon b\in Q}\sum_{w\sim b}K_\Tt(w,b)(w-b) \right| \leq \sum_{\substack{b\sim w,\\ \dist(b,\partial Q) = O(\delta)}} |K_\Tt(w,b)|\cdot |w-b| = O(\delta r),
  \]
  where the last equality follows from Assumption~\ref{item:full-plane_bdd_density} from Section~\ref{subsec:t-holom-regularity-lemmas}. The second sum in the lemma can be treated similarly (notice that $\bar{\eta}_w^2\overline{K_\Tt(w,b)} = K_\Tt(w,b)\eta_b^2$).
\end{proof}

\begin{lemma}
  \label{lemma:dbar_of_f_t}
  Assume that $\eps< \frac{\alpha}{\alpha+2}$ where $\alpha$ is the H\"older exponent from Lemma~\ref{lemma:Holderness_of_thol_fcts}. Let $f$ be a t-white holomorphic function. Then there exist $C,\beta_2>0$ depending only on $\lambda$ and $\eps$ such that for any $z\in \CC$ for which $f_{\delta^\eps}$ is defined we have
  \begin{equation}
    \label{eq:d_of_f_t}
    |\partial f_{\delta^{\eps}}(z)|\leq C\cdot \max_{b\in B\colon |b - z| < \delta^\eps}|f(b)|\cdot \delta^{-\eps}
  \end{equation}
  \begin{equation}
    \label{eq:dbar_of_f_t}
    |\dbar f_{\delta^\eps}(z)|\leq C \cdot \max_{b\in B\colon |b - z| < \delta^{\eps/2}}|f(b)|\cdot \delta^{\beta_2}.
  \end{equation}
\end{lemma}
\begin{proof}
  The first inequality is straightforward, let us prove the second one.
  Using that $f$ is t-white holomorphic and Lemma~\ref{lemma:approx_of_dbar} we get the following 
  \begin{multline*}
    0 = \sum_{w\in W}\vphi_{\delta^\eps}(z + w)\sum_{b\sim w}K_\Tt(w,b)f(b) = \sum_{b\in B}f(b)\sum_{w\sim b} K_\Tt(w,b)\vphi_{\delta^\eps}(z + w) = \\
    = i\dbar f_{\delta^\eps}(z) + \sum_{b\in B}f(b)\sum_{w\sim b} K_\Tt(w,b)\partial\vphi_{\delta^\eps}(z + b)(w - b) + O(\delta^{-4\eps + 3}).
  \end{multline*}
  We need to show that
  \[
    \sum_{b\in B}f(b)\sum_{w\sim b} K_\Tt(w,b)\partial\vphi_{\delta^\eps}(z + b)(w - b) = O(\max_{b\in B\colon |b - z| < \delta^\eps}|f(b)|\cdot \delta^{\beta_2})
  \]
  for some $\beta_2>0$ depending on $\eps$ and $\lambda$ only. Let us fix a parameter $\nu>0$ specified later, and for any $n\in \delta^\nu\ZZ^2$ let us denote by $Q_n$ the cube with the center $n$ and the side length $\delta^\nu$, and by $w_n\in \Wws$ an arbitrary white triangle contained in $Q_n$. Given $b\in Q_n\cap B(z,\delta^\eps)$ we can write
  \begin{equation}
    \label{eq:dft1}
    2f(b) = f(w_n) + \eta_b^2\overline{f(w_n)} + O(\delta^{\alpha(\nu - \eps/2)}\max_{b\in B\colon |b - z| < \delta^{\eps/2}}|f(b)|),
  \end{equation}
  where we use Lemma~\ref{lemma:Holderness_of_thol_fcts}, Lemma~\ref{lemma:Boundedness_of_white_values} and properties of t-holomorphic functions (cf.~\eqref{eq:f_t_vs_f15}). We also have
  \begin{equation}
    \label{eq:dft2}
    \partial\vphi_{\delta^\eps}(z + b) = \partial\vphi_{\delta^\eps}(z + w_n) + O(\delta^{\nu-4\eps}).
  \end{equation}
  We conclude that for a given $n$
  \begin{multline}
    \label{eq:dft3}
    \sum_{b\in B\cap Q_n} f(b)\sum_{w\sim b} K_\Tt(w,b)\partial\vphi_{\delta^\eps}(z + b)(w - b) = \\
    = \partial\vphi_{\delta^\eps}(z + w_n)\cdot\Big(f(w_n)\sum_{b\colon b\in Q_k}\sum_{w\sim b}K_\Tt(w,b)(w-b) + \\
    +\overline{f(w_n)}\sum_{b\colon b\in Q_k}\sum_{w\sim b}\eta_b^2K_\Tt(w,b)(w-b) \Big) + \\
    + O((\delta^{\alpha(\nu-\eps/2) -3\eps} + \delta^{\nu-4\eps}) \max_{b\in B\colon |b - z| < \delta^{\eps/2}}|f(b)|\cdot \Area(Q_n)).
  \end{multline}
  Using Lemma~\ref{lemma:derivative_of_approx_identity} to estimates two sums on the right-hand side of~\eqref{eq:dft3} we get
  \begin{multline}
    \label{eq:dft4}
    \sum_{b\in B\cap Q_n} f(b)\sum_{w\sim b} K_\Tt(w,b)\partial\vphi_{\delta^\eps}(z + b)(w - b) =\\
    = O(( \delta^{1-\nu-3\eps} + \delta^{\alpha(\nu-\eps/2) -3\eps} + \delta^{\nu-4\eps}) \max_{b\in B\colon |b - z| < \delta^{\eps/2}}|f(b)|\cdot \Area(Q_n)).
  \end{multline}
  Summing over all squares intersecting the support of $\vphi_{\delta^\eps}(z +\cdot)$ we finally get
  \begin{multline}
    \label{eq:dft5}
    \sum_{b\in B}f(b)\sum_{w\sim b} K_\Tt(w,b)\partial\vphi_{\delta^\eps}(z + b)(w - b) = \\
    = O(( \delta^{1-\nu-\eps} + \delta^{\alpha(\nu-\eps/2) -\eps} + \delta^{\nu-2\eps}) \max_{b\in B\colon |b - z| < \delta^{\eps/2}}|f(b)|).
  \end{multline}
  Plugging $\nu = 2\alpha^{-1}\eps$ we get the result.

\end{proof}

\begin{proof}[Proof of Theorem~\ref{thmas:parametrix}]
  Let $K_\Tt^{-1}$ be as in Proposition~\ref{prop:full_plane_kernel}, let $\lambda$ be fixed. The second estimate in the theorem follows from Proposition~\ref{prop:full_plane_kernel}. We first analyze the asymptotic of $K_\Tt^{-1}(b,w)$ assuming that $\delta>0$ is small enough depending on $\lambda$. Fix $w_0$ and consider the function $\Phi(b) = \bar{\eta}_{w_0}K_\Tt^{-1}(b,w_0)$. We fix a white splitting $\Wws$ and extend $\Phi$ to white triangles. Choose a small parameter $\eps>0$ and consider the function $\Phi_{\delta^\eps}$ defined as in~\eqref{eq:def_of_ft}. Choose another small parameter $0<\nu<\eps/9$ and introduce two circles
  \[
    \gamma_\inn = \partial B(w_0, \delta^{3\nu}), \qquad \gamma_\out = \partial B(w_0, \delta^{-1})
  \]
  both oriented counterclockwise.

  Let us choose $\eps$ small enough so that Lemma~\ref{lemma:f_t_vs_f} and Lemma~\ref{lemma:dbar_of_f_t} hold. These lemmas together with h\"olderness of t-white holomorphic functions (Lemma~\ref{lemma:Holderness_of_thol_fcts}) and the bound on $K_\Tt^{-1}$ imply the following. Let $z\in \CC$ such that $|z-w_0|\geq \delta^{3\nu}$, and let $w_z\in \Wws$ be a triangle on $O(\delta)$ distance from $z$. Let $\alpha$ be the H\"older exponent from Lemma~\ref{lemma:Holderness_of_thol_fcts}. Then there exist parameters $0<\beta_1,\beta_2<\alpha$ both depending on $\lambda$ and $\eps$ solely such that
  \begin{equation}
    \label{eq:Phi_delta_values}
    \Phi_{\delta^\eps}(z) = \Phi(w_z) + O\left(\frac{\delta^{\beta_1}}{|z - w_0|}\right),\qquad |\Phi_{\delta^\eps}(z)| = O\left( \frac{1}{|z-w_0|} \right),
  \end{equation}
  \begin{equation}
    \label{eq:Phi_delta_dbar}
    \dbar \Phi_{\delta^\eps}(z) = O\left(\frac{\delta^{\beta_2}}{|z-w_0|}\right),\qquad \partial \Phi_{\delta^\eps}(z) = O\left( \frac{\delta^{-\eps}}{|z-w_0|} \right).
  \end{equation}

  Let now $z_1\in B$ be such that $\delta^\nu < |w_0 - z_1| < \delta^{-1}$. The upper bound~\eqref{eq:Phi_delta_dbar} implies
  \begin{equation}
    \label{eq:two_integrals}
    \int_{\gamma_\out}\frac{\Phi_{\delta^\eps}(z)\,dz}{2\pi i (z - z_1)} - \int_{\gamma_\inn}\frac{\Phi_{\delta^\eps}(z)\,dz}{2\pi i (z - z_1)} = \Phi_{\delta^\eps}(z_1) + O(\delta^{\beta_2}\log\delta).
  \end{equation}
  Let us estimate each integral in the left-hand side of this equality. For the integral along $\gamma_\out$ we can use a crude bound
  \[
    \left|\int_{\gamma_\out}\frac{\Phi_{\delta^\eps}(z)\,dz}{2\pi i (z - z_1)}\right| = O(\delta)
  \]
  that follows from~\eqref{eq:Phi_delta_values}. To estimate the second integral we need a bit more delicate arguments. Using~\eqref{eq:Phi_delta_values} again we obtain
  \begin{equation}
    \label{eq:int_along_gamma_inn}
    \int_{\gamma_\inn}\frac{\Phi_{\delta^\eps}(z)\,dz}{2\pi i (z - z_1)} = \frac{1}{2\pi i(w_0 - z_1)}\int_{\gamma_\inn}\Phi_{\delta^\eps}(z)\,dz + O(\delta^\nu).
  \end{equation}
  Using the fact that $\Oo = O(\delta)$ and Lemma~\ref{lemma:dbar_of_f_t} we get
  \begin{equation}
    \label{eq:int_along_gamma_inn2}
    \int_{\gamma_\inn}\Phi_{\delta^\eps}(z)\,dz = \int_{\gamma_\inn}\left(\Phi_{\delta^\eps}(z)\,dz + \overline{\Phi_{\delta^\eps}(z)}\,d\Oo\right) + O(\delta^{1-\eps}).
  \end{equation}
  Recall the 1-form $\omega_{\Phi}$ defined in~\eqref{eq:def_of_omega_f}. Using~\eqref{eq:Phi_delta_values} and~\eqref{eq:int_along_gamma_inn2} we can write
  \[
    \int_{\gamma_\inn}\Phi_{\delta^\eps}(z)\,dz = \int_{\gamma_\inn}\omega_{\Phi} + O(\delta^{\beta_1} + \delta^{1-\eps}) = 2\bar{\eta}_{w_0} + O(\delta^{\beta_1} + \delta^{1-\eps})
  \]
  where the last equality follows from~\eqref{eq:monodromy_and_Kf1} and the definition of $\Phi$. Substituting this expression for $\int_{\gamma_\inn}\Phi_{\delta^\eps}(z)\,dz$ into~\eqref{eq:int_along_gamma_inn} we obtain
  \[
    \int_{\gamma_\inn}\frac{\Phi_{\delta^\eps}(z)\,dz}{2\pi i (z - z_1)} = \frac{\bar{\eta}_{w_0}}{\pi i(w_0 - z_1)} + O(\delta^\nu + \delta^{\beta_1 - \nu} + \delta^{1 - \eps-\nu}).
  \]
  Combining this with~\eqref{eq:two_integrals} we finally get
  \begin{equation}
    \label{eq:Phi_t_asymp}
    \Phi_{\delta^\eps}(z_1) = \frac{\bar{\eta}_{w_0}}{\pi i (z_1 - w_0)} + O(\delta^\nu + \delta^{\beta_1 - \nu} + \delta^{1 - \eps-\nu} + \delta^{\beta_2}\log \delta).
  \end{equation}
  Let now $b\in B$ be such that $\delta^\nu \leq |b-w_0|\leq \delta^{-1}$ (recall that we assume black vertices to be embedded into $\CC$) and let $w\in \Wws$ denote a white triangle incident to $b$. Then, by~\eqref{eq:Phi_delta_values},~\eqref{eq:Phi_t_asymp} and the definition of $\Phi$ we get
  \begin{multline*}
    \bar{\eta}_{w_0}K_\Tt^{-1}(b,w_0) = \Pr(\Phi(w), \eta_b\RR) =\\
    = \Pr\left( \frac{\bar{\eta}_{w_0}}{\pi i (b - w_0)} \right) + O(\delta^\nu + \delta^{\beta_1 - \nu} + \delta^{1 - \eps-\nu} + \delta^{\beta_2}\log \delta).
  \end{multline*}
  When $|b-w_0|> \delta^{-1}$, the same estimates follows from the upper bound on $K_\Tt^{-1}$. Choosing $\nu$ to be sufficiently small we conclude that there exists $\beta>0$ such that for all $\delta$ sufficiently small depending on $\lambda$ we have
  \begin{equation}
    \label{eq:non-homogen_asymp}
    K_\Tt^{-1}(b,w) = \Pr\left[\frac{1}{\pi i (b - w)}, \eta_b\eta_w\RR\right] + O(\delta^\beta),\qquad \text{given that}\quad |b-w|\geq \delta^\beta.
  \end{equation}

  Currently we have proven that~\eqref{eq:non-homogen_asymp} holds for all $\delta\leq \delta_0$ uniformly in $\delta$ where $\delta_0$ is some constant depending on $\lambda$ only. When $\delta>\delta_0$, the same estimate (with constants depending on $\delta_0$, which depends $\lambda$ only) follows from the upper bound on $K_\Tt^{-1}$ which we have regardless the value of $\delta$. Taking this into account, we can rewrite~\eqref{eq:non-homogen_asymp} in the following homogeneous form:
  \begin{equation}
    \label{eq:homogen_asymp}
    K_\Tt^{-1}(b,w) = \Pr\left[\frac{1}{\pi i (b - w)}, \eta_b\eta_w\RR\right] + O\left(\frac{\delta^\beta}{|b-w|^{1+\beta}}\right),\qquad |b-w| = 1
  \end{equation}
  uniformly in $\delta>0$ with constants depending on $\lambda$ only. Since both sides of~\eqref{eq:homogen_asymp} are multiplied by the same scalar when we rescale $\Tt$, we conclude that~\eqref{eq:homogen_asymp} holds without the restriction $|w-b| = 1$.
\end{proof}

\subsection{Multivalued holomorphic functions on isoradial graphs}
\label{subsec:multivalued}

In this section we assume that the graph $G$ is embedded into the plane isoradially~\cite{KenyonCriticalPlanarGraphs}. Recall that, by definition, all faces of $G$ are inscribed polygons, each face contains its circumcenter inside and all radii of circumscribed circles of faces are equal to each other. Denote the common radius by $\delta>0$. Let $\Tt$ denote the embedding of $G^\ast$ by circumcenters of faces of $G$. It is easy to see that $\Tt$ is a t-embedding.

By connecting vertices of $G$ with adjacent vertices of $G^\ast$ we get a tiling of the plane by rhombi. We assume that each rhombus has both angles bigger than some fixed constant $\lambda>0$. With this assumption $\Tt$ becomes weakly uniform and satisfies $O(\delta)$-small origami assumption (with some constant $\lambda'$ depending on $\lambda$). Let the origami square root function $\eta$ be chosen arbitrary.

Our goal for this subsection is to study discrete holomorphic functions having a prescribed multiplicative monodromy around a face of $G$. We will be using intensively the results of~\cite[Section~7.1]{DubedatFamiliesOfCR}. For the rest of the section we assume that $0$ is a vertex of $\Tt$, that is, a circumcenter of some face of $G$. Let $\gamma_0$ be a simple infinite path composed of edges of $\Tt$, starting from $0$ and oriented towards infinity. Let $s\in \RR$ be given. Having this data, we modify the Kasteleyn operator $K_\Tt$ as follows:
\begin{equation}
  \label{eq:def_of_K_Gs}
  K_s(w,b) = \begin{cases}
    K_\Tt(w,b),\quad bw\text{ does not cross }\gamma_0,\\
    e^{2\pi i s}K_\Tt(w,b),\quad bw\text{ crosses }\gamma_0\text{ and $b$ is on the left,}\\
    e^{-2\pi i s}K_\Tt(w,b),\quad \text{else.}
  \end{cases}
\end{equation}
Operator $K_s$ may be thought of as operator $K_\Tt$ acting on a space of multivalued functions with monodromy $e^{2\pi i s}$. Indeed, let $\tilde{\CC}\to \CC\smm\{ 0 \}$ be the universal cover and $\chi: \tilde{\CC}\to \tilde{\CC}$ denote the deck transformation corresponding to a single counterclockwise turn around the origin. Let $\tilde G\to G$ be the pullback of $G$ on $\tilde{\CC}$. Consider the set of functions
\[
  \Fun_s(B) = \{ f: B(\tilde G)\to \CC\ \mid\ f(\chi(b)) = e^{2\pi i s}f(b) \};
\]
define the set $\Fun_s(W)$ similarly.
If $\Fun(B)$ denotes the space of all functions from $B$ to $\CC$, then we have a (non-canonical) isomorphism $\Fun_s(B)\cong \Fun(B)$ corresponding to any fundamental domain in $\tilde{\CC}$. Then the left action of $K_\Tt$ on $\Fun_s(B)$ is intertwined with the left action of $K_s$ on $\Fun(B)$ by this isomorphism. Similarly, the right action of $K_\Tt$ on $\Fun_{-s}(W)$ is intertwined with the right action of $K_s$ on $\Fun(W)$. In what follows we will be often jumping between these two formalisms without mentioning it (e.g. considering the function $z^s$ as a function on $G$ by mean of the corresponding identification of $\Fun_s$ with $\Fun$).

For the next lemma we need an additional notation. Let $b_0\in B(G)$ be any black vertex incident to the face containing $0$ and $\eta_0$ be any unit complex number satisfying
\begin{equation}
  \label{eq:def_of_etav0}
  \eta_0^2 = \frac{b_0}{|b_0|}\eta_{b_0}^2.
\end{equation}
From the definition of $\eta_b$ we see that $\eta_0^2$ does not depend on the choice of $b_0$.

\begin{lemma}
  \label{lemma:fs}
  Assume that $G$ is a full-plane isoradial graph as above. For any $s\in [-1/2,+\infty)$ there exist functions $[z^s], [\bar{z}^s]$ on black vertices of $G$ such that $K_s [z^s] = K_{G,-s}[\bar{z}^s] = 0$ and the following asymptotic holds:
  \[
    \begin{split}
      &[\bar{b}^s] = \eta_b^2 \bar{b}^s + \eta_0^2\frac{\Gamma(s+1)}{\Gamma(-s)} \left( \delta^{2s+1} b^{-s-1} + O(\delta b^{s-1})\right),\\
      &[b^s] = b^s + \eta_b^2\bar{\eta}_0^2\frac{\Gamma(s+1)}{\Gamma(-s)}\left( \delta^{2s+1}\bar{b}^{-s-1} + O(\delta b^{s-1})\right),
    \end{split}
  \]
  where the constants in the asymptotic depend only on $\lambda$. Moreover, if $b_0$ is incident to the face containing $0$, then we have
  \[
    \begin{split}
      &[\bar{b}_0^s] = \Gamma(1+s) \eta_b^2\bar{b}_0^s,\\
      &[b_0^s] = \Gamma(1+s)b_0^s.
    \end{split}
  \]
\end{lemma}

\begin{proof}
  The proof is based on analyzing suitable discrete exponents as it was done in~\cite[Lemma~10]{DubedatFamiliesOfCR}. We sketch the arguments for the sake of completeness. 

  We first assume that $\delta = 1$ and $s\in \RR\smm\ZZ$ is arbitrary. Following the conventions introduced in~\cite[Appendix~A.1,~Proof of Theorem~2.5]{ChelkakSmirnov}, for a generic $z\in \CC$ we define the discrete complex exponent $e_v(z)$ inductively by declaring $e_0(z) = 1$ and then for any vertex $v$ of $G^*$ and $b,w\sim v$
  \[
    e_b(z) = (1 + z(b-v))^{-1}e_v(z),\qquad e_w(z) = (1- z(w-v)) e_v(z).
  \]
  It is well-known that $e_b(z)$ is consistently defined rational function of $z$ with poles lying on the unit circle, and that $e_b(z)$ is discrete holomorphic as a function of $b$. (See~\cite[Section~3.2]{KenyonCriticalPlanarGraphs}, where the statement is proven for a collection of functions $f_u(z), u\in $ defined by $f_b(z) = z^{-1}e_b(z^{-1})$). We define
  \[
    f_s(b) := \int_{\gamma_b} z^{-s-1}e_b(z)\,dz
  \]
  where $\gamma_b$ is the contour that starts at the point $(2b)^{-1}$ on the circle $\{|z| = (2|b|)^{-1}\}$, follows this circle in the counterclockwise direction, then the straights segment connecting $(2b)^{-1}$ with $2\bar b$, then follows the circle $\{ |z| = 2|b| \}$ in the clockwise direction, and finally closes itself along the straight segment connecting $2\bar b$ with $(2b)^{-1}$. One can verify easily that $f_s(b)$ is discrete holomorphic as a function of $b$. Choosing a simple path between $0$ and $b$ on $G\cup G^*$ such that its orthogonal projection onto $b\RR$ is monotone we obtain the following asymptotic relations:
  \[
    \begin{split}
      &e_b(z) = \bar{\eta}_0^2\eta_b^2z^{-1}\exp(-\bar{b}z^{-1} + O(bz^{-2})),\qquad |z|\geq |b|^{1/2},\\
      &e_b(z) \leq C\cdot e^{-c\cdot |b|^{1/2}},\qquad |b|^{-1/2}\leq |z|\leq |b|^{1/2}\quad \text{and}\quad bz\in \RR_+,\\
      &e_b(z) = \exp(-bz + O(bz^2)),\qquad |z|\leq |b|^{-1/2}
    \end{split}
  \]
  where $C,c>0$ and all other constants depend on $\lambda$ only. Substituting this into the definition of $f_s$ we get the following asymptotic relation:
  \begin{multline}
    \label{eq:fs1}
    f_s(b) = (1 - e^{2\pi i s})\Gamma(s+1)\bar{\eta}_0^2\eta_b^2 \bar{b}^{-s-1} \left(1 + O((s+1)(s+2)b^{-1})  \right) +\\
    + (1 - e^{2\pi i s})\Gamma(-s)b^s \left(1 + O(s(1-s)b^{-1})  \right).
  \end{multline}
  When $s\in [-1/2,+\infty)$ we set
  \[
    \begin{split}
      &[b^s] = ((1 - e^{2\pi i s})\Gamma(-s))^{-1}\cdot f_s(b),\\
      &[\bar{b}^s] = ((1 - e^{-2\pi i s})\Gamma(-s)\bar{\eta}_0^{2})^{-1}\cdot f_{-s-1}(b);
    \end{split}
  \]
  when $s\in \ZZ$ we just take an appropriate limit.
  The asymptotic expansions of $[b^s]$ and $[\bar{b}^s]$ required in the lemma follow immediately from the asymptotic expansion of $f_s$. Moreover, observing that $e_{b_0}(z) = \frac{1}{1 + zb_0}$ and $|b_0| = 1$ (as we assumed $\delta = 1$) we get the exact relation
  \[
    f_s(b_0) = \int_{\gamma_b}\frac{z^{-s-1}}{(1 + zb_0)} = -2\pi i e^{\pi i s} \bar b_0^{-s} = \int_{\gamma_b}\frac{z^{-s-1}}{(1 + zb_0)} = -2\pi i e^{\pi i s} b_0^s
  \]
  which implies the expression for the values of $[b^s]$ and $[\bar b^s]$ at $b_0$.
\end{proof}

From now on we will be assuming that $G$ is a \emph{Temperleyan isoradial} graph, that is, a superposition of an isoradial graph $\Gamma$ and its dual $\Gamma^\dagger$ embedded by circumcenters. Recall that black vertices of $G$ correspond to vertices of $\Gamma$ and $\Gamma^\dagger$, edges are the half-edges of $\Gamma$ and $\Gamma^\dagger$ and white vertices lie at intersections of edges of $\Gamma$ and $\Gamma^\dagger$. It is clear that such a $G$ is always isoradial, with each face given by a union of two right triangles with a common hypotenuse. We keep the notation $\Tt$ for the t-embedding of the dual graph given by circumcenters of faces (that is, midpoints of the aforementioned hypotenuses). Temperleyan isoradial graphs are usually considered in the framework of classical discrete complex analysis~\cite{ChelkakSmirnov}, in particular, this framework was chosen by Dub\'edat in~\cite{DubedatFamiliesOfCR}. Thus, restricting our focus to this framework will allow us to use results developed in~\cite[Section~7.1]{DubedatFamiliesOfCR} with minimal changes.

Our goal is to analyze the inverting kernel $K_s^{-1}$. This kernel was constructed and estimated by Dub\'edat~\cite[Lemmas~13,14]{DubedatFamiliesOfCR}, however, for our purposes we need to improve this asymptotics, filling some gaps in the aforementioned lemmas. This will be achieved by following the same line of arguments as proposed in the proof of~\cite[Lemmas~13,14]{DubedatFamiliesOfCR}, but with higher precision when it comes to estimates.

We begin by reminding the following results of~\cite[Lemmas~13,14]{DubedatFamiliesOfCR}:

\begin{lemma}
  \label{lemma:Ksinv_a_priori_bound}
  Let $s\in (0,1/2)$ be fixed. Then there is a unique function $K^{-1}_s: B\times W\to \CC$ such that
  \begin{equation}
    \label{eq:estimate_on_K_s}
    K_s^{-1}(b,w) = O\left( \frac{1}{|b-w|} \left( \left|\frac{b}{w}\right|^s + \left|\frac{w}{b}\right|^s \right) \right)
  \end{equation}
  and $K_s^{-1}$ is the (both left and right) inverse operator for $K_s$. Moreover, if $w_0$ is incident to the face of $G$ containing zero, then
  \begin{multline}
    \label{eq:Ksinv_at_zero}
    K_s^{-1}(b,w_0) 
    = \frac{\Gamma(1-s)}{2}\frac{1}{\pi i (b-w_0)}\left( \frac{b}{w_0} \right)^s - \frac{\Gamma(1+s)}{2} \frac{(\eta_b\eta_{w_0})^2 }{\pi i \overline{b-w_0}}\left( \frac{\bar{w_0}}{\bar{b}} \right)^{s} +\\+ O(\delta^{2-s} b^{s-3}).
  \end{multline}
\end{lemma}
\begin{proof}
  By~\cite[Lemma~13]{DubedatFamiliesOfCR}, there exists a right inverse $K_s^{-1}$ of $K_s$ satisfying~\eqref{eq:estimate_on_K_s} when $|b|\geq 2|w|$ and $|b|\leq \frac{|w|}{2}$. The case when $\frac{|w|}{2} \leq |b|\leq 2|w|$ is not indicated in this lemma, but follows directly from the proof. Indeed, it follows from the construction of $K_s^{-1}$ given by Dub\'edat in the proof of~\cite[Lemma~13]{DubedatFamiliesOfCR} that $K_s(b,w) = O(|w|^{-1})$ when $|w|/4\leq |w-b|\leq |w|/2$. Let $K_\Tt^{-1}$ be the kernel constructed Proposition~\ref{prop:full_plane_kernel}. The function $K_s^{-1}(b,w) - K_\Tt^{-1}(b,w)$ considered as a function of $b$ is a discrete holomorphic function in the disc $|b-w|\leq |w|/2$ and bounded by $|w|^{-1}$ on the boundary of this disc. It follows that $K_s^{-1}(b,w) = K_\Tt^{-1}(b,w) + O(|w|^{-1})$ which implies the desired inequality.

  The relation~\eqref{eq:Ksinv_at_zero} is proven in~\cite[Lemma~14]{DubedatFamiliesOfCR}. The uniqueness of $K_s^{-1}(b,w)$ follows from~\cite[Lemma~11]{DubedatFamiliesOfCR}, and the fact that $K_s^{-1}(b,w)$ is the full inverse (not only the right one) follows from the same uniqueness argument exactly as it was done in the proof of Proposition~\ref{prop:full_plane_kernel}.
\end{proof}

Let us now describe the asymptotics of $K_s^{-1}$. Similarly to Theorem~\ref{thmas:parametrix} describing the asymptotics of $K_\Tt^{-1}$ we may expect that $K_s^{-1}$ is asymptotically equal to the continuous kernel
\begin{equation}
  \label{eq:def_of_Cs}
  \Cc_s(b,w) = \frac{1}{2}\left[ \frac{1}{\pi i (b-w)}\left( \frac{b}{w} \right)^s - \frac{(\eta_b\eta_w)^2 }{\pi i\overline{(b-w)}} \left( \frac{\bar{w}}{\bar{b}} \right)^s \right]
\end{equation}
with the error term being uniformly bounded away from the singularities. More precisely, the following assertion holds:

\begin{lemma}
  \label{lemma:existence_of_Ksinv}
  For each $s\in (0,1/2)$ there exist $\beta>0$ such that the kernel $K^{-1}_s: B\times W\to \CC$ from Lemma~\ref{lemma:Ksinv_a_priori_bound} satisfies the following estimates
    \begin{multline}
      \label{eq:asympt_of_K_s_b>2w}
      K_s^{-1}(b,w) = \frac{1}{2}\left[ \frac{1}{\pi i (b-w)}\left( \frac{b}{w} \right)^s - \frac{(\eta_b\eta_w)^2 }{\pi i\overline{(b-w)}} \left( \frac{\bar{w}}{\bar{b}} \right)^s \right] + \\
      + O\left( \left(\frac{|b|^s}{|w|^s} + \frac{|w|^s}{|b|^s}\right)\left( \frac{1}{|b-w|} + \frac{1}{|b|} \right)\frac{\delta^{2\beta}}{|w|^{2\beta}} \right),\qquad |b-w| > |w|^{1-\beta}\delta^\beta,
    \end{multline}
    \begin{multline}
      \label{eq:asympt_of_K_s_b-w<w34}
      K_s^{-1}(b,w) = e^{2\pi i s A(b,w)}\left( K_\Tt^{-1}(b,w) + \frac{s}{2} \left[ \frac{1}{\pi i w} + \frac{(\eta_b\eta_w)^2}{\pi i \bar{w}} \right]\right) + O\left(\frac{\delta^\beta}{|w|^{1+\beta}} \right),\\\qquad |b-w| \leq |w|^{1-\beta}\delta^\beta,
    \end{multline}
    where $A(b,w) = \indic[b\text{ on the left from }\gamma_0] - \indic[w\text{ on the left from }\gamma_0]$.
\end{lemma}
\begin{proof}

  It is enough to prove all the estimates when $\delta = 1$, the case of a general $\delta$ follows from rescaling arguments. Let us now show how to develop a more precise asymptotics of the kernel $K_s^{-1}$ following the general construction suggested by Dub\'edat in the proof of~\cite[Lemma~14]{DubedatFamiliesOfCR}. The idea is to modify the continuous kernel $\Cc_s(b,w)$ near its singularities using discrete holomorphic substitutes such as $[b^s]$ or $K_\Tt^{-1}(b,w)$. If $S$ is the modified kernel, then we can formally write
  \begin{equation}
    \label{eq:Kinv_via_S_one_puncture}
    K_s^{-1} = S - K_s^{-1}T
  \end{equation}
  where
  \begin{equation}
    \label{eq:def_of_T}
    T = K_sS - \Id.
  \end{equation}
  As soon as $T$ is small enough, we can use the estimate~\eqref{eq:estimate_on_K_s} to bound the term $K_s^{-1}T$ on the right-hand side of~\eqref{eq:Kinv_via_S_one_puncture} and to conclude that $K_s^{-1}$ is approximately $S$.

  To define the parametrix $S(b,w)$ we fix two exponents $x,y\in (0,1)$ specified later and for each $w\in W$ we define
  \[
    \Uu_0 = \{ z\in \CC\ \mid\ |z|\leq |w|^x \},\qquad \Uu_w = \{ z\in \CC\ \mid\ |z-w|\leq |w|^y \}.
  \]
  We now define
  \[
    S(b,w) = \begin{cases}
      \Cc_s(b,w), \quad b\notin \Uu_0\cup \Uu_w,\\
      e^{2\pi i s A(b,w)}\left( K^{-1}_\Tt(b,w) + \frac{s}{2}\left[\frac{1}{\pi i w} + \frac{(\eta_b\eta_w)^2}{\pi i\bar{w}}\right]\right),\quad b\in \Uu_w,\\
      -\frac{[b^s]}{2\pi w^{1+s}} + \frac{\eta_w^2[\bar{b}^{-s}]}{2\pi \bar{w}^{1-s}},\quad b\in \Uu_0.
    \end{cases}
  \]
  where $[\bar{b}^s],[b^s]$ are as in Lemma~\ref{lemma:fs} and $A$ is as~\eqref{eq:asympt_of_K_s_b-w<w34}. Let us now estimate $T(u,w) = (K_sS)(u,w) - \Id(u,w)$. By our construction $T(u,w)$ vanishes identically when $u$ is inside $\Uu_0\cup \Uu_w$ and is at the distance more than 1 from their boundaries, and can be estimated by the second derivative of $\Cc_s$ when $u$ is outside of $\Uu_0\cup \Uu_w$ and at the distance more than 1 from their boundaries. When $u$ is at the distance at most 1 from $\partial\Uu_0\cup \partial\Uu_w$, then the estimate on $T(u,w)$ depends on the mismatch between $\Cc_s$ and the discrete holomorphic substitutes. These mismatches can be determined by Lemma~\ref{lemma:fs} and Theorem~\ref{thmas:parametrix}; moreover, as we are currently working with isoradial graphs the exponent $\beta$ in Theorem~\ref{thmas:parametrix} can be taken to be equal to $1$ (see~\cite[Theorem~1]{DubedatFamiliesOfCR} which is closer to our notation or~\cite[Theorem~4.3]{KenyonCriticalPlanarGraphs} where the estimate was originally obtained). Overall, this leads to the following estimates (below $X\lesssim Y$ we mean that $|X|\leq C|Y|$ for an absolute constant $C>0$):
  \begin{equation}
    \label{eq:estimate_on_T_one_puncture}
    T(u,w) \lesssim \begin{cases}
      \left( \frac{1}{|u|^3} + \frac{1}{|u-w|^3} \right)\frac{1}{|u-w|} \left( \left|\frac{u}{w}\right|^s + \left|\frac{w}{u}\right|^s\right),\quad u\notin \Uu_0^{+1}\cup \Uu_w^{+1},\\
      0,\quad u\in \Uu_0^{-1}\cup \Uu_w^{-1},\\
      |w|^{-2y} + |w|^{y-2},\quad u \in \partial \Uu_w+\overline B(0,1),\\
      |w|^{-1-(1-x)(1-s)} + |w|^{-(1-s)(1+x)},\quad u\in \partial \Uu_0+\overline B(0,1),\\
    \end{cases}
  \end{equation}
  where where $\Uu^{+1} = \Uu + \overline B(0,1)$ and $\Uu^{-1} = \Uu\smm (\partial\Uu + \overline B(0,1))$. Let us now estimate $K_s^{-1}T$. We first write by definition
  \[
    |(K_s^{-1}T)(b,w)| \leq \sum_{u\in W} |K_s^{-1}(b,u) T(u,w)|.
  \]
  We can split this sum into three terms depending on the cases on the right-hand side of~\eqref{eq:estimate_on_T_one_puncture}:
  \begin{equation}
    \label{eq:KT_via_I}
    |(K_s^{-1}T)(b,w)| \leq I_1 + I_w + I_0
  \end{equation}
  where 
  \[
    \begin{split}
      &I_1 = \sum_{u\notin \Uu_0^{+1}\cup\,\Uu_w^{+1}}|K_s^{-1}(b,u) T(u,w)|,\\
      &I_w = \sum_{u \in \partial \Uu_w+\overline B(0,1)}|K_s^{-1}(b,u) T(u,w)|,\\
      &I_0 = \sum_{u \in \partial \Uu_0+\overline B(0,1)}|K_s^{-1}(b,u) T(u,w)|.
    \end{split}
  \]
  Let us now estimate each sum separately.

  \textbf{Estimate of $I_1$.} Using the estimate~\eqref{eq:estimate_on_K_s} to bound $K_s^{-1}(b,u)$ (note that we have justified this estimate above already) and the estimate~\eqref{eq:estimate_on_T_one_puncture} to bound $T(u,w)$ we can write
  \begin{multline*}
    I_1\lesssim \int\limits_{\CC\smm (\Uu_0\cup\, \Uu_w)} \frac{1}{|u-b|}\left( \left|\frac{b}{u}\right|^s + \left|\frac{u}{b}\right|^s\right) \cdot \left( \frac{1}{|u|^3} + \frac{1}{|u-w|^3} \right)\frac{1}{|u-w|} \left( \left|\frac{u}{w}\right|^s + \left|\frac{w}{u}\right|^s\right)\,d^2u\\
    \lesssim \int\limits_{\CC\smm(\Uu_0\cup\,\Uu_w)} \left(\frac{|b|^s}{|w|^s} + \frac{|w|^s}{|b|^s} + \frac{|u|^{2s}}{|bw|^s} + \frac{|bw|^s}{|u|^{2s}}\right)\cdot \frac{d^2u}{|(u-b)(u-w)||u|^3} + \\
    +\int\limits_{\CC\smm(\Uu_0\cup\,\Uu_w)} \left(\frac{|b|^s}{|w|^s} + \frac{|w|^s}{|b|^s} + \frac{|u|^{2s}}{|bw|^s} + \frac{|bw|^s}{|u|^{2s}}\right)\cdot\frac{d^2u}{|u-b||u-w|^4}.
  \end{multline*}
  Estimating the integral on the right-hand side is tedious but absolutely straightforward. A direct analysis shows that the main impact comes from integrating over the annuli $\Aa_0 = \{|w|^x\leq |u| \leq 2|w|^x\}$ and $\Aa_w = \{|w|^y\leq |u-w|\leq 2|w|^y\}$. We have the following estimates:
  \[
    \int_{\Aa_0}(\cdots) \lesssim \left(\frac{|w|^s}{|b|^s} + \frac{|bw|^s}{|w|^{2sx}}\right)\cdot \frac{1}{|w|^{x+1}\max(|b|,|w|^x)},
  \]
  \[
    \int_{\Aa_w}(\cdots) \lesssim \left(\frac{|b|^s}{|w|^s} + \frac{|w|^s}{|b|^s}\right)\cdot \frac{1}{|w|^{2y}\max(|b-w|,|w|^y)}
  \]
  which finally implies 
  \begin{multline}
    \label{eq:I1_estimate}
    I_1\lesssim \left(\frac{|w|^s}{|b|^s} + \frac{|bw|^s}{|w|^{2sx}}\right)\cdot \frac{1}{|w|^{x+1}\max(|b|,|w|^x)} +\\+\left(\frac{|b|^s}{|w|^s} + \frac{|w|^s}{|b|^s}\right)\cdot \frac{1}{|w|^{2y}\max(|b-w|,|w|^y)}.
  \end{multline}

  \textbf{Estimate of $I_w$.} Substituting the estimate~\eqref{eq:estimate_on_K_s} for $K_s^{-1}(b,u)$ and the appropriate estimate from~\eqref{eq:estimate_on_T_one_puncture} for $T(u,w)$ we obtain
  \[
    I_w \lesssim \sum_{u\in \partial \Uu_w + \overline B(0,1)}\frac{1}{|u-b|}\left( \left|\frac{b}{u}\right|^s + \left|\frac{u}{b}\right|^s\right)\cdot (|w|^{-2y} + |w|^{y-2}).
  \]
  Note that $\sum_{u\in \partial \Uu_w + \overline B(0,1)}\frac{1}{|u-b|}\asymp \log\left(\frac{|b-w|+|w|^y}{||b-w|-|w|^y|+1}\right)\lesssim \frac{|w|^y\log|w|}{\max(|b-w|,|w|^y)}$ and when $u$ is close to $\Uu_w$ the quantities $|b/u|$ and $|u/b|$ are comparable to $|b/w|$ and $|w/b|$ respectively, thus
  \begin{equation}
    \label{eq:Iw_estimate}
    I_w \lesssim \left(\frac{|b|^s}{|w|^s} + \frac{|w|^s}{|b|^s}\right)\cdot\left( \frac{1}{|w|^{2y}} + \frac{1}{|w|^{2-y}} \right)\frac{|w|^y\log|w|}{\max(|b-w|,|w|^y)}.
  \end{equation}

  \textbf{Estimate of $I_0$.} Substituting the estimate~\eqref{eq:estimate_on_K_s} for $K_s^{-1}(b,u)$ and the appropriate estimate from~\eqref{eq:estimate_on_T_one_puncture} for $T(u,w)$ we obtain
  \[
    I_0\lesssim \sum_{u \in \partial \Uu_0+\overline B(0,1)}\frac{1}{|u-b|}\left( \left|\frac{b}{u}\right|^s + \left|\frac{u}{b}\right|^s\right)\cdot \left( |w|^{-1-(1-x)(1-s)} + |w|^{-(1-s)(1+x)} \right)
  \]
  Similarly as in the previous case we can observe that $\sum_{u \in \partial \Uu_0+\overline B(0,1)}\frac{1}{|u-b|}\lesssim \frac{|w|^x\log|w|}{\max(|b|, |w|^x)}$, so
  \begin{equation}
    \label{eq:I0_estimate}
    I_0\lesssim \left( \frac{|b|^s}{|w|^{sx}} + \frac{|w|^{sx}}{|b|^s}\right)\cdot \left( |w|^{-1-(1-x)(1-s)} + |w|^{-(1-s)(1+x)} \right)\frac{|w|^x\log|w|}{\max(|b|, |w|^x)}.
  \end{equation}

  To conclude the estimates stated in the lemma we now need to plug~\eqref{eq:I1_estimate},~\eqref{eq:Iw_estimate} and~\eqref{eq:I0_estimate} into~\eqref{eq:KT_via_I} and choose the parameters $x,y$. The latter must be chosen so that all the error terms scale as $\lambda^{-1-\beta}$ for some $\beta>0$ when we rescale the graph by $\lambda$. This can be achieved by choosing $x \geq 2s$ and $y>\tfrac12$. Note that the resulting exponent $\beta$ will tend to 0 as $s$ tends to $\tfrac12$ because of the last estimate~\eqref{eq:I0_estimate}. This reflects the fact that the error term in the expansion of $[\bar b^{-s}]$ becomes of the same order as the leading term (see Lemma~\ref{lemma:fs}). We will not investigate it further as we will only use the result for the particular value $s = \tfrac14$ in our what follows. Finally, to obtain the proof for a general $\delta$ we note that to change the scale we just need to replace $K_s^{-1}(b,w)$ with $\delta^{-1}K_s(\delta^{-1}b,\delta^{-1}w)$.

\end{proof}

Note that the error term in the asymptotics~\eqref{eq:asympt_of_K_s_b>2w} blows up when $b$ or $w$ approaches the origin. Indeed, in this regime $b^s$ and $w^s$ in the main term in the asymptotics must be replaced by their discrete analogs defined in Lemma~\ref{lemma:fs} to match the left-hand side. As we will demonstrate in the following corollary, this indeed improves the error term. We will prove the corollary only when $w$ is small, using the relation~\eqref{eq:Ksinv_at_zero} that we borrow from~\cite{DubedatFamiliesOfCR}. The same relation can be proven when $b$ is incident to the origin using the same technique; we will not do it here for the sake of shortness.

\begin{cor}
  \label{cor:Ksinv_when_w_at_sing}
  Let $s\in (0,1/2)$ be fixed and $K_s^{-1}$ be as in Lemma~\ref{lemma:Ksinv_a_priori_bound}. There exists a $\beta>0$ such that 
  \begin{equation}
    \label{eq:big_b_asymptotics_revisited}
    K_s^{-1}(b,w) = \frac{b^s[w^{-s}]}{2\pi i (b-w)} - \frac{\eta_b^2\bar{b}^{-s}[\bar{w}^s]}{2\pi i (\bar{b} - \bar{w})} + O\left(\frac{\delta^\beta}{|b|^{1+\beta}}\right),\qquad |b|\geq 2|w|,
  \end{equation}
  where $[w^{-s}], [\bar{w}^s]$ are the functions defined in Lemma~\ref{lemma:fs} (applied to discrete holomorphic functions on white vertices).
\end{cor}
\begin{proof}
  Choosing $\gamma>0$ small enough we can ensure that~\eqref{eq:big_b_asymptotics_revisited} holds for $\delta^\gamma|b|^{1-\gamma} < |w|\leq |b|/2$ for some $\beta>0$ small enough depending on $s$ and $\gamma$. When $|w|\leq \delta^\gamma|b|^{1-\gamma}$ consider the function
  \[
    F(w) = K_s^{-1}(b,w) - \frac{[w^{-s}]}{2\pi i b^{1-s}} + \frac{[\bar{w}^s]}{2\pi i \bar{b}^{1+s}}.
  \]
  Using~\eqref{eq:asympt_of_K_s_b>2w} to estimate $K_s^{-1}$ and Lemma~\ref{lemma:fs} to estimate the other terms we conclude that
  \begin{equation}
    \label{eq:existence_of_Ksinv2}
    |F(w)| = O\left(\frac{\delta^\beta}{|b|^{1+\beta}}\right)
  \end{equation}
  when $|w| \asymp \delta^\gamma|b|^{1-\gamma}$ for some $\beta>0$ small enough, again provided that $\gamma>0$ is small enough. Moreover, the same estimate~\eqref{eq:existence_of_Ksinv2} can be derived when $w$ is incident to the origin by using~\eqref{eq:Ksinv_at_zero} to estimate $K_s^{-1}(b,w)$ and the exact values of $[w^s]$ and $[\bar w^{-s}]$ given in Lemma~\ref{lemma:fs}.

  Finally, notice that $F$ is a discrete holomorphic multivalued function in the disc $|w|\leq \delta^\gamma|b|^{1-\gamma}$ by the construction. We claim that this implies a maximum principle asserting that the maximal value of $|F|$ in the disc can be bounded via its values on the boundary of the disc or at a vertex incident to the branching point (that is, to the origin). To this end, notice that the function $F$ can be treated as a single-valued discrete holomorphic function on every simply connected subset of the disc punctured $|w|\leq \delta^\gamma|b|^{1-\gamma}$ at the origin. Its two projections $F(w) + \eta_w^2 \overline{F(w)}$ and $i(F(w) - \eta_w^2 \overline{F(w)})$ are also discrete holomorphic and, moreover, are t-holomorphic as defined in Definition~\ref{defin:discrete_holom} (here we switch the roles of black and white vertices which does not create any significant difference). For these functions we can define their `true complex values' (say, $F^\bullet$ and $F^{\bullet,\ast}$) as in Lemma~\ref{lemma:function_on_splitting} and conclude the maximum principle for those values from the martingale property as in Lemma~\ref{lemma:t-holom_are_matringales}. Using that the simply connected set was chosen arbitrary we can conclude that the values of $|F|$ in the disc are bounded by the values of $F^\bullet$ and $F^{\bullet,\ast}$ on black faces of $\Tt$ (that is, black vertices of $G$; remember that the initial $F$ is defined on white vertices of $G$ aka white faces of $\Tt$, so $F^\bullet,F^{\bullet,\ast}$ are defined on black faces of $\Tt$) incident to the origin and to the boundary of $|w|\leq \delta^\gamma|b|^{1-\gamma}$. But these values, in their turn, can be reconstructed from the values of $F$ on the white vertices incident to the origin and to the boundary of the disc, which implies the desired bound. We finally conclude that~\eqref{eq:existence_of_Ksinv2} holds everywhere in the disc $|w|\leq \delta^\gamma|b|^{1-\gamma}$.

\end{proof}

\subsection{Temperleyan isoradial graph on an infinite cone}
\label{subsec:graph_on_a_cone}

We now construct the discrete Cauchy kernel in an infinite cone. We need to adapt our notation. Let $G$ be a Temperleyan isoradial full-plane graph, $G^\ast$ be embedded by circumcenters of faces of $G$. Recall that the latter embedding is a t-embedding of $G^\ast$, let an arbitrary origami square root function $\eta$ on this t-embedding be fixed. Assume that $0\in \CC$ is a vertex of $G^\ast$. Denote by $\Cc$ the plane $\CC$ equipped with the metric $|d(z^2)|^2$, that is, $\Cc$ is an infinite cone with cone angle $4\pi$. Let $\Tt: \Cc\to \CC$ be the mapping given by 
\[
  \Tt(z) = z^2.
\]
Put $G_\Cc = \Tt^{-1}(G),\ G_\Cc^\ast = \Tt^{-1}(G^\ast)$ endowed with the natural graph structure. Note tat $G_\Cc$ and $G_\Cc^\ast$ are dual to each other because $0$ is a vertex of $G^\ast$. 

Recall that the pullback of the Kasteleyn weights from $G$ do not give Kasteleyn weights for $G_\Cc$ because the Kasteleyn condition would not hold around the face containing conical singularity in this case.
Let $\gamma_0^\Cc$ denote an arbitrary simple path connecting $0\in \Cc$ with infinity and such that $\Tt(\gamma_0^\Cc) = \gamma_0$ is the path chosen in the previous section. Let $K_G(w,b)$ denote the usual isoradial Kasteleyn weight of an edge $bw$ of $G$ defined as in~\eqref{eq:def_of_KTt}. Define Kasteleyn weights of $G_\Cc$ by
\begin{equation}
  \label{eq:def_of_K_Tts}
  K_{\Tt,1/2}(w,b) = \begin{cases}
    K_G(\Tt(w),\Tt(b)),\quad bw\text{ does not cross }\gamma_0^\Cc,\\
    -K_G(\Tt(w),\Tt(b)),\quad \text{else}.
  \end{cases}
\end{equation}
Let $\dist(b,w)$ denote the distance between $b,w\in \Cc$ measured in the inner metric of $\Cc$. For example, $\dist(b, 0) = |b|^2$.

\begin{lemma}
  \label{lemma:kernel_of_G4pi}
  There is a unique function $K_{\Tt,1/2}^{-1}(b,w)$ defined on $B(G_\Cc)\times W(G_\Cc)$ such that
  \begin{enumerate}
    \item $K_{\Tt,1/2}^{-1}$ is left and right inverse of $K_{\Tt,1/2}$;

    \item We have
      \begin{equation}
        \label{eq:estimate_on_K_in_cone}
        K_{\Tt,1/2}^{-1}(b,w) = O\left( \frac{1}{|b-w|(|b| + |w|)}\left( \left|\frac{b}{w}\right|^{1/2} + \left|\frac{w}{b}\right|^{1/2} \right) \right)
      \end{equation}
    \item There exists an absolute constant $\beta>0$ such that
      \begin{multline}
        \label{eq:asympt_of_K_Tt_b>2w}
        K_{\Tt,1/2}^{-1}(b,w) = \frac{1}{4}\left[ \frac{1}{\pi i (b-w)\sqrt{bw}} - \frac{(\eta_{\Tt(b)}\eta_{\Tt(w)})^2 }{\pi i\overline{(b-w)}\sqrt{\bar{b}\bar{w}}} \right] + \\
        + O\left( \frac{1}{\sqrt{|bw|}}\left( \frac{1}{|b-w|} + \frac{|w|}{|b|} + 1 \right) \frac{\delta^{2\beta}}{|w|^{4\beta}} \right) ,\qquad \dist(b,w)> |w|^{2-2\beta}\delta^\beta,
      \end{multline}
      \begin{multline}
        \label{eq:asympt_of_K_Tt_b-w<w34}
        K_{\Tt,1/2}^{-1}(b,w) = e^{\pi i A(b,w)} K_G^{-1}(\Tt(b),\Tt(w)) +\\
        + O\left( \frac{\delta^\beta}{|w|^{2+2\beta}} \right),\qquad \dist(b,w) \leq |w|^{2-2\beta}\delta^\beta
      \end{multline}
      where $A(b,w) = \indic[b\text{ on the left from }\gamma_0^\Cc] - \indic[w\text{ on the left from }\gamma_0^\Cc]$.
    \item Moreover, there exists an absolute constant $\beta>0$ such that
      \begin{multline}
        \label{eq:asymp_w_near_sing}
        K_{\Tt,1/2}^{-1}(b,w) = \\
        = \frac{1}{4}\left[ \frac{b^{1/2}[\Tt(w)^{-1/4}] + b^{-1/2}[\Tt(w)^{1/4}]}{\pi i (b-w)(b+w)} - \frac{\eta_{\Tt(b)}^2(\bar{b}^{-1/2}[\Tt(\bar{w})^{1/4}] + \bar{b}^{1/2}[\Tt(\bar{w})^{-1/4}])}{\pi i \overline{(b-w)(b+w)}} \right] +\\
        + O\left( \frac{\delta^\beta}{|b|^{2+2\beta}} \right),\qquad |w|^2\leq |b|^{2-2\beta}\delta^\beta
      \end{multline}
      where $[\bar{z}^s],[z^s]$ are as in Lemma~\ref{lemma:fs}.
  \end{enumerate}
\end{lemma}
\begin{proof}
  Let $K_{1/4}^{-1}$ be defined by Lemma~\ref{lemma:existence_of_Ksinv} applied to $G$, interpreted as a multivalued function in both variables. Fix a white vertex $w$ and define
  \[
    K_{\Tt,1/2}^{-1}(b,w) = \frac{1}{2}\left(K^{-1}_{1/4}(\Tt(b),\Tt(w)) + (\eta_{\Tt(b)}\eta_{\Tt(w)})^2\overline{K^{-1}_{1/4}(\Tt(b),\Tt(w))} \right)
  \]
  For the black vertices $b$ neighborhing $w$. Note that if neither of edges incident to $w$ cross $\gamma_0^\Cc$ we have $(K_{\Tt,1/2}K_{\Tt,1/2}^{-1})(w,w) = 1$. Moreover, extending $K_{\Tt,1/2}^{-1}(b,w)$ we obtain a function multivalued with monodromy $-1$ in both variables and discrete holomorphic everywhere aside from the diagonal where it has the discrete residue 1. Taking a branch of this function in $\CC\smm\gamma_0^\Cc$ we obtain the desired inverse operator $K_{\Tt,1/2}^{-1}$.

  The asymptotic relations~\eqref{eq:estimate_on_K_in_cone}--\eqref{eq:asympt_of_K_Tt_b-w<w34} follow from Lemma~\ref{lemma:existence_of_Ksinv}, and the relation~\eqref{eq:asymp_w_near_sing} follows from Corollary~\ref{cor:Ksinv_when_w_at_sing}. For the uniqueness, note that whenever we have $K_{\pi^{-1}(G),1/2}^{-1}$ as in the lemma, we get
  \[
    K_{1/4}^{-1}(\Tt(b),\Tt(w)) = K_{\Tt,1/2}^{-1}(b,w) - i K_{\Tt,1/2}^{-1}(-b,w)
  \]
  to be an inverse for $K_{1/4}$ satisfying the conditions from Lemma~\ref{lemma:existence_of_Ksinv} which is unique.
\end{proof}

\section{Perturbed Szeg\"o kernel on a Riemann surface}
\label{sec:The Riemann surface: continuous setting}

In this section we introduce the Szeg\"o kernel on $\Sigma$ and a family of its perturbations $\Dd_\alpha$ that will be used to describe the limit of perturbed Kasteleyn operators later. We begin with a short informal discussion. Assume for a moment that the holonomy of $ds^2$ is trivial. Assume that we have a sequence of adapted graphs on $\Sigma$ with the mesh size tending to zero, and assume that the gauge form $\alpha_G$ is zero for each of these graphs (like in the case of a pillow surface, see Example~\ref{intro_example:pillow_surface}). Then `true complex values' (cf. Definition~\ref{defin:discrete_holom}) of bounded discrete holomorphic functions (i.e. those from the kernel of the Kasteleyn operator defined in Section~\ref{subsec:intro_Kasteleyn_operator}) will be approximating values of holomorphic functions on $\Sigma$ having square root singularities $\frac{1}{\sqrt{z}}(a + O(\sqrt{z}))$ at conical singularities of the metric. To define a space of such functions we must specify along which loops they pick up the multiplicative monodromy $-1$. This is equivalent to choosing a cohomology class in $H^1(\Sigma\smm\{ p_1,\dots,p_{2g-2} ,\ZZ/2\ZZ\})$ which has value $-1$ on any small circle around any $p_i$. When $\Sigma$ has a non-trivial topology, this can be done in many different ways. This ambiguity is fixed by choosing simple paths $\gamma_1,\dots, \gamma_{g-1}$ are used to define the Kasteleyn operator Section~\ref{subsec:intro_Kasteleyn_operator}. The cohomology class is then the Poincar\'e dual to the (relative to $p_1,\dots,p_{2g-2}$) homology class of the union of these paths. In other words, having these path we can require our functions to have single-valued branches in $\Sigma\smm\cup_{j = 1}^{g-1}\gamma_j$.

The locally flat metric $ds^2$ and the aforementioned cohomology class patched together determine a spin structure on $\Sigma$. Multivalued functions that we considered appear to be in natural correspondence with smooth (also at conical singularities) sections of the corresponding spin line bundle. We address the reader to Section~\ref{subsec:spin_line_bundles} for more detailed discussion on topological aspects of this correspondence. The continuous analog of the Kasteleyn operator is the Cauchy--Riemann (Dirac, if being more accurate) operator acting on these sections intertwined with the isomorphism between sections and functions with singularities. Presence of a non-trivial holonomy of $ds^2$ or a non-trivial gauge form $\alpha_G$ results in a change of the complex structure in the spin line bundle, which corresponds to a perturbation of the corresponding operator.

These heuristics lead us to expect that the limit of the inverse Kasteleyn matrix should converge to the Szeg\"o kernel with the spin structure we introduced. Classically~\cite{Fay}, a Szeg\"o kernel can be expressed via theta functions and the prime form. Below we use this expression ad hoc, and list some basic asymptotic properties following from it. The proofs of all the statements from this section are given in Appendix.

\subsection{The spin structure associated with \texorpdfstring{$\omega_0$}{omega0} and \texorpdfstring{$\gamma_1,\dots,\gamma_{g-1}$}{gamma1...gammag-1} and the corresponding Szeg\"o kernel}
\label{subsec:Szego_kernel}

Assume that we are in the setup introduces in Section~\ref{subsec:intro_flat_metric}. Let $\gamma_1,\dots, \gamma_{g-1}$ be simple non-intersecting paths connecting $p_1,\dots, p_{2g-2}$ pairwise, and assume that $\sigma(\gamma_i) = \gamma_i$ for each $i = 1,\dots,g-1,$ if the involution $\sigma$ is present. 
Recall that we have $ds^2 = |\omega_0|^2$, where $\omega_0$ is the smooth $(1,0)$-form defined in Proposition~\ref{prop:existence_of_metric_on_Sigma}. Recall that given a smooth path $\gamma:[0,1]\to \Sigma\smm\{ p_1,\dots, p_{2g-2} \}$ we define its winding with respect to $\omega_0$ by
\[
  \wind(\gamma, \omega_0) = \Im\int_0^1 \frac{d}{dt} \log \omega_0(\gamma'(t))\,dt.
\]
Obviously, the winding does not depend on the parametrization. For each smooth oriented simple loop $\gamma$ on $\Sigma\smm\{ p_1,\dots,p_{2g-2} \}$ we define
\begin{equation}
  \label{eq:def_of_q0}
  q_0(\gamma) = (2\pi)^{-1}\wind(\gamma,\omega_0) + \gamma\cdot (\gamma_1+\ldots+\gamma_{g-1}) + 1\mod 2,
\end{equation}
where $\cdot$ denotes the algebraic intersection number. It is easy to show that $q_0$ is correctly defined and depends only on the homology class in $H_1(\Sigma, \ZZ/2\ZZ)$ which $\gamma$ represents.
\begin{lemma}
  \label{lemma:q0_is_quadratic_form}
  The function $q_0: H_1(\Sigma, \ZZ/2\ZZ)\to \ZZ/2\ZZ$ is a quadratic form with respect to the bilinear form given by the intersection product.
\end{lemma}
\begin{proof}
  See Section~\ref{subsec:asymptotics_Ddalpha}.
\end{proof}

Recall that $B_0,\dots, B_{n-1}$ denote the boundary components of $\Sigma_0$. We can complete the homology classes of $B_1,\dots, B_{n-1}$ (oriented according to the orientation of $\Sigma_0$) to a simplicial basis $A_1,\dots, A_g, B_1,\dots, B_g$ in the homology group $H_1(\Sigma, \ZZ)$ in such a way that $\sigma(B_i) = B_i, \sigma(A_i) = -A_i, i = 1,\dots, g$. Let $\omega_1,\dots, \omega_g$ be the normalized Abelian differentials of the first kind, and $\Omega$ be the matrix of B-periods of $\Sigma$, see Section~\ref{subsec:simplicial_basis} for details.

\begin{cor}
  \label{cor:wind_condition}
  There exists a choice of $\alpha_0$ and $\omega_0$ such that~\eqref{eq:wind_condition} is satisfied.
\end{cor}
\begin{proof}
  The normalization~\eqref{eq:wind_condition} is equivalent to $q_0([\gamma]) = 0$ for each simple loop $\gamma$ such that $\sigma(\gamma) = \gamma$, where $[\gamma]\in H_1(\Sigma, \ZZ/2\ZZ)$ is the homology class represented by $\gamma$. This in fact is equivalent to $q_0(A_j) = 0$ for each $j = 1,\dots, n-1$. Indeed, if $\gamma$ is such as above, then $[\gamma]$ can be expressed as a linear combination of $A_1,\dots, A_{n-1}$ and a homology class of the form $[X] + \sigma([X])$ where $X$ is a collection of loops in $\Sigma_0$ disjoint from $A_1,\dots, A_{n-1}$. The relation $q_0(A_j) = 0$ ensures that $q_0$ vanishes on such linear combination.

  It remains to show that the relation $q_0(A_j) = 0$ can be satisfied for each $j = 1,\dots, n-1$ if $\alpha_0$ is chosen properly. Recall that $\alpha_0$ can be replaced with an arbitrary anti-holomorphic $\alpha_1$ such that $\int_\gamma\Im(\alpha_1-\alpha_0)\in \pi\ZZ$ for each loop $\gamma$, that is, the cohomology class represented by $\pi^{-1}\Im(\alpha_1 - \alpha_0)$ is integer. Since we want to keep the symmetry $\sigma^\ast\alpha_0 = \bar\alpha_0$ (that is, $\sigma^\ast\Im\alpha_0 = -\Im\alpha_0$), we must also require this cohomology class to be anti-symmetric with respect to $\sigma^\ast$. In fact, the cohomology class represented by $\pi^{-1}\Im(\alpha_1 - \alpha_0)$ can be an arbitrary class satisfying these restrictions because $\alpha\mapsto\Im\alpha$ is a bijection between anti-holomorphic $(0,1)$-forms and real-valued harmonic differentials on $\Sigma$.

  Once we replace $\alpha_0$ with $\alpha_1$, the differential $\omega_0$ must be replaced by the differential $\omega_1(p) = \exp(2i\int^p \Im (\alpha_1 - \alpha_0))\omega_0(p)$, and so $q_0$ is replaced by $q_0$ plus the cohomology class represented by $\pi^{-1}\Im(\alpha_1 - \alpha_0)$ (modulo~$2\ZZ$). It remains to choose the latter cohomology class so that $q_0(A_j) = 0$ for each $j = 1,\dots, n-1$ which is possible due to the discussion above.
\end{proof}

Recall that for each $p,q\in \Sigma$ the Abel map applied to the divisor $p-q$ is, by definition,
\[
  \Aa(p-q) = (\int_q^p\omega_1,\dots, \int_q^p\omega_g)\ \mod \ZZ^g + \ZZ^g\Omega,
\]
see~\eqref{eq:def_of_Abel_map}. When $p$ and $q$ are close, we will be choosing the path of the integration to be the geodesic between $p$ and $q$, to specify the representative in the equivalence class above. Given $a,b\in \RR^g$, let $\theta\chr{a}{b}(z,\Omega)$ be the theta function with characteristics $\chr{a}{b}$ defined as
\[
  \theta\chr{a}{b}(z,\Omega) = \sum_{m\in \ZZ^g}\exp\Bigl( \pi i (m+a)^t\cdot \Omega (m+a) + 2\pi i (z-b)^t(m+a) \Bigr)
\]
(cf.~\eqref{eq:def_of_theta} and Remark~\ref{rem:miunus_in_def_of_theta} for the unusual choice of the signs of $a$ and $b$ in the definition). Recall that the prime form $\pf(p,q)$ is defined as
\[
  \pf(p,q) = \frac{\theta\chr{a^-}{b^-}(\Aa(p-q),\Omega)}{\sqrt{\omega_-}(p)\sqrt{\omega_-}(q)},
\]
where $[a^-,b^-]$ is some odd theta characteristics and $\omega_-$ is the holomorphic $(1,0)$-form given by the square of a holomorphic section of the corresponding spin line bundle normalized properly, see Section~\ref{subsec:the_prime_form} for details. Write
\begin{equation}
  \label{eq:def_of_varsigma}
  \omega_-(p) = \varsigma(p)\cdot \omega_0(p),
\end{equation}
where $\varsigma$ is a smooth function on $\Sigma\smm\{ p_1,\dots, p_{2g-2} \}$, and introduce the notation
\[
  E(p,q)\sqrt{\omega_0(p)}\sqrt{\omega_0(q)}= \frac{\theta\chr{a^-}{b^-}(\Aa(p-q),\Omega)}{\sqrt{\varsigma(p)}\sqrt{\varsigma(q)}}.
\]

Assume that an anti-holomorphic $(0,1)$-form $\alpha_h$ is given. We associate the vectors $a(\alpha_h),b(\alpha_h)\in \RR^g$ with $\alpha_h$ as follows
\begin{equation}
  \label{eq:def_of_a_b}
  a(\alpha_h)_j = \pi^{-1}\int_{A_j}\Im \alpha_h,\quad b(\alpha_h)_j = \pi^{-1}\int_{B_j}\Im\alpha_h,\qquad j = 1,\dots,g,
\end{equation}
where $A_j,B_j$'s represent the simplicial basis chosen in Section~\ref{subsec:simplicial_basis}. Let also $[a_0, b_0]$ be the characteristics of $q_0$, i.e. $a_0,b_0\in \{ 0,1/2 \}^g$ and we have
\[
  q_0(A_i) = 2a_i,\qquad q_0(B_i) = 2b_i, \qquad i = 1,\dots, g.
\]
We now set
\begin{equation}
  \label{eq:def_of_theta_alpha}
  \theta[\alpha_h](z) = \theta\chr{a(\alpha_h)+a_0}{b(\alpha_h)+b_0}(z,\Omega).
\end{equation}

\begin{prop}
  \label{prop:def_of_S}
  Let $\alpha$ be a $(0,1)$-form on $\Sigma$ with $\mC^2$ coefficients, and let $\alpha = \dbar \vphi + \alpha_h$ be its Dolbeault decomposition. Assume that $\theta[\alpha_h](0)\neq 0$. Let $U\subset \Sigma\smm\{ p_1,\dots,p_{2g-2} \}$ be a non-empty simply-connected open subset. Let $\sqrt{\varsigma}$ be a branch of the square root in $U$. Consider the function $\Dd_\alpha^{-1}(p,q)$ on $U\times U\smm\mathrm{Diagonal}$ defined by
  \begin{equation}
    \label{eq:def_of_S}
      \Dd_\alpha^{-1}(p,q) = \frac{\theta[\alpha_h](\Aa(p-q))}{\pi i \theta[\alpha_h](0)\cdot \pf(p,q) \sqrt{\omega_0(p)}\sqrt{\omega_0(q)}} \cdot \exp\left(\vphi(q) - \vphi(p)-2i\int\limits_q^p\Im \alpha_h\right). 
  \end{equation}
  where the integration path between $p$ and $q$ is taken to lie inside $U$. Then $\Dd_\alpha^{-1}(p,q)$ satisfies the following equations when $p\neq q$:
    \begin{equation}
      \label{eq:equation_on_p_of_Salpha}
      (\dbar_p + \frac{\alpha_0(p)}{2} + \alpha(p))\, \Dd_\alpha^{-1}(p,q) = 0,
    \end{equation}
    \begin{equation}
      \label{eq:equation_on_q_of_Salpha}
      (\dbar_q + \frac{\alpha_0(q)}{2} - \alpha(q))\, \Dd_\alpha^{-1}(p,q) = 0.
    \end{equation}
    Moreover, the function $\Dd^{-1}_\alpha(p,q)$ admits a unique multivalued extension to the space 
    \[
      \Bigl((\Sigma\smm\{ p_1,\dots,p_{2g-2} \})\times (\Sigma\smm\{ p_1,\dots,p_{2g-2} \})\Bigr)\smm\mathrm{Diagonal}
    \]
    satisfying the equations~\eqref{eq:equation_on_q_of_Salpha},~\eqref{eq:equation_on_p_of_Salpha} everywhere on this space and having the multiplicative monodromy $(-1)^{\gamma\cdot (\gamma_1+\ldots + \gamma_{g-1})}$ along each loop $\gamma$ on $\Sigma \smm\{ p_1,\dots, p_{2g-2} \}$.
\end{prop}
\begin{proof}
  See Section~\ref{subsec:asymptotics_Ddalpha}.
\end{proof}

\begin{rem}
  \label{rem:prop_of_S}
  Let $\Ff_{q_0}$ denote the spin line bundle on $\Sigma$ corresponding to the quadratic form $q_0$ (see Section~\ref{subsec:spin_line_bundles} for more details). It can be shown using the definition of $q_0$ that the differential $\omega_0$ gives rise to a smooth section of $\Ff_0$ whose square is equal to $\omega_0$ after we identify $\Ff_0\otimes\Ff_0$ with the cotangent bundle. Denote this section by $\sqrt{\omega_0}$ by abusing the notation. The object $\Dd_\alpha^{-1}(p,q)\sqrt{\omega_0}(p)\sqrt{\omega_0}(q)$ can be viewed as a section of $\Ff_0\boxtimes \Ff_0$ (that is, a section of $\Ff_0$ in $p$ and a section of $\Ff_0$ in $q$) which is smooth outside the diagonal and have a simple pole along it. This section multiplied by $\tfrac12$ is equal to the inverting kernel of the operator $\dbar + \alpha$ acting on smooth sections of $\Ff_0$, that is, for every such smooth section $f$ we have
  \[
    \int_{q\in \Sigma} \Dd_\alpha^{-1}(p,q)\sqrt{\omega_0}(p)\sqrt{\omega_0}(q) (\dbar + \alpha)f(q) = 2f(p)
  \]
  (note that the integrand can be interpreted as a $(\tfrac12,0) + (\tfrac12,1) = (1,1)$-form, so it can be integrated over $\Sigma$). The existence of such a section for a generic $\alpha$ can be shown using Riemann--Roch theorem. Indeed, if one denotes by $\Ll_\alpha$ the holomorphic line bundle whose local holomorphic sections are functions annihilated by the operator $\dbar + \alpha$ (see~\eqref{eq:def_of_Lalpha}), then, for a fixed $q$, the object $\Dd_\alpha^{-1}(p,q)\sqrt{\omega_0}(p)\sqrt{\omega_0}(q)$ becomes a meromorphic section of $\Ff_0\otimes \Ll_\alpha$ with a simple pole at $q$. If $\Ff_0\otimes \Ll_\alpha$ does not admit non-zero holomorphic sections, then the existence of such meromorphic section is guaranteed by the Riemann--Roch formula.
\end{rem}

Below we list some properties of the kernel $\Dd_\alpha^{-1}$. The proof of Proposition~\ref{prop:def_of_S} and the following lemmas will be given in Section~\ref{subsec:asymptotics_Ddalpha}. Recall the multivalued function $\Tt$ on $\Sigma$ we defined in Section~\ref{subsec:intro_graphs_on_Sigma0}, see~\eqref{eq:def_of_Tt}.

\begin{lemma}
  \label{lemma:diagonal_expansion_of_Salpha}
  Given $\alpha = \dbar\vphi + \alpha_h$ with $\vphi\in \mC^1$ define the function $r_\alpha$ by the formula
  \begin{equation}
    \label{eq:formula_for_r}
    r_\alpha(q) = \frac{1}{\pi i} \frac{d_p}{\omega_0(p)}\log \theta[\alpha_h](\Aa(p-q))\vert_{p=q} - \frac{2}{\pi i}\cdot \frac{\partial\Re \vphi}{\omega_0}(q).
  \end{equation}
  Then we have
  \begin{multline}
    \label{eq:def_of_r}
    \Dd_\alpha^{-1}(p,q)\cdot \exp(i\Im \int\limits_q^p (2\alpha+\alpha_0)) =\\ = \frac{1}{\pi i}\exp\left(-2i\Im \int\limits_{p_0}^q\alpha_0 \right)\cdot \left(  \Tt(p) - \Tt(q) \right)^{-1} +  r_\alpha(q) + O\left( \frac{\dist(p,q)}{\dist(q, \{ p_1,\dots, p_{2g-2} \})^{3/2}} \right) 
  \end{multline}
  as $p\to q$ uniformly in $q$ staying away from $p_1,\dots, p_{2g-2}$; the integral $\int\limits_q^p \alpha$ is taken along the geodesic between $p$ and $q$, the value of the product $\exp\left(-2i\Im \int\limits_{p_0}^q\alpha_0 \right)\cdot \left(  \Tt(p) - \Tt(q) \right)^{-1}$ is fixed by requiring that the path of integration between $p_0$ and $q$ in the exponential factor and in the definition of $\omega$ (see~\eqref{eq:def_of_omega}) are the same (the product does not depend on the choice of this path).
\end{lemma}
\begin{proof}
  Follows from direct computations.
\end{proof}

Let $j = 1,\dots, 2g-2$ be given. Note that $\sqrt{\Tt(p) - \Tt(p_i)}$ has a single-valued branch near $p_i$. Denote any such branch by $z_j(p)$. Define the local kernel
\begin{equation}
  \label{eq:def_of_SCc}
  S_j(p,q) = \frac{1}{2\pi i (z_j(p) - z_j(q))\sqrt{z_j(p)}\sqrt{z_j(q)}}
\end{equation}
for $p,q$ are close to $p_j$. Note that replacing $\Tt$ with $\lambda\Tt$ amounts in replacing $S_j(p,q)$ with $\lambda^{-1}S_j(p,q)$.

\begin{lemma}
  \label{lemma:bw_near_singularity_expansion_of_Salpha}
  Let $j = 1,\dots, 2g-2$ be fixed and $\alpha$ be a $(0,1)$-form with $\mC^1$ coefficients.
  \begin{multline}
    \label{eq:bw_near_singularity_expansion_of_Salpha}
    \Dd_\alpha^{-1}(p,q)\cdot \exp(i\Im \int\limits_q^p (2\alpha+\alpha_0)) =\\
    = \exp(-2i\Im\int_{p_0}^q\alpha_0) \cdot  S_j(p,q) + \frac{r_\alpha(q)\sqrt{z_j(q)}}{\sqrt{z_j(p)}} + O\left( \frac{|z_j(p) - z_j(q)|}{\sqrt{z_j(p)}\sqrt{z_j(q)}} \right) 
  \end{multline}
  uniformly in $p,q\in B(p_j,\lambda)$, where $r_\alpha$ is as in Lemma~\ref{lemma:diagonal_expansion_of_Salpha} and $\lambda$ is the constant from Assumption~\ref{item:intro_metric_assumptions} on the graph and the metric.
\end{lemma}
\begin{proof}
  Follows from direct computations.
\end{proof}

Let $\Delta$ be the Laplace operator associated with the metric $ds^2$ defined on $\mC^2$ functions compactly supported in the interior of $\Sigma_0$, i.e.
\[
  -4\dbar\partial f = \Delta f \bar{\omega}_0\wedge \omega_0.
\]

\begin{lemma}
  \label{lemma:derivative_of_theta}
  Assume that $\partial \Sigma_0\neq \varnothing$ and we have an antiholomorphic $(0,1)$-form $\alpha_G$ on $\Sigma$ satisfying $\sigma^*\alpha_G = \bar{\alpha}_G$. Assume that $\alpha_t = \dbar\vphi_t + \alpha_{h,t}$ is a smooth family of $(0,1)$-forms on $\Sigma$ such that for all $t$ we have $\sigma^*\alpha_t = -\bar{\alpha}_t$, $\theta[\alpha_{h,t} + \alpha_G](0)\neq 0$ and $\vphi_t,\dot{\vphi}_t\in \mC^2(\Sigma)$. 
  Let $r_{\pm\alpha_t + \alpha_G }$ be as in Lemma~\ref{lemma:diagonal_expansion_of_Salpha}. Then we have
  \begin{multline}
    \frac{d}{dt}\left(\log \theta[\alpha_{h,t} + \alpha_G](0) + 2\pi i a(\alpha_{h,t})\cdot b(\alpha_G) - \frac{1}{2\pi}\int_{\Sigma_0}\Re \vphi_t \Delta \Re\vphi_t ds^2\right) =\\
    =-\frac{1}{4}\int_\Sigma \left( r_{\alpha_t+\alpha_G}\omega_0\wedge \dot{\alpha}_t - \overline{r_{-\alpha_t + \alpha_G}\omega_0\wedge \dot{\alpha}_t} \right).
  \end{multline}
\end{lemma}
\begin{proof}
  See Section~\ref{subsec:asymptotics_Ddalpha}.
\end{proof}

\section{Inverse Kasteleyn operator: construction and estimates} 
\label{sec:discrete_Riemann_surface}

We will now assume that we are in the setup described in Section~\ref{subsec:intro_flat_metric}, that is, $\Sigma$, the locally flat metric $ds^2$ with conical singularities at $p_1,\dots, p_{2g-2}$ and the associated 1-forms $\omega_0,\alpha_0$ are fixed. Let $\lambda\in (0,1), \delta>0$ be given, let $G$ be a $(\lambda,\delta)$ adapted graph on $\Sigma$ as defined in Section~\ref{subsec:intro_graphs_on_Sigma0}. We think of $\delta>0$ as of a parameter that tends to zero, in particular, we will always be assuming that it is much smaller than any macroscopic parameters such as distances between $p_j$-s or lengths of non-contractible loops on $\Sigma$. Let $\alpha_G$ and $K_\delta = K$ be as defined in Section~\ref{subsec:intro_Kasteleyn_operator}, recall that a collection of paths $\gamma_1,\dots,\gamma_{g-1}$ connecting $p_1,\dots, p_{2g-2}$ pairwise was fixed to define $K$. Note that we defined $K$ to be gauge equivalent to a real-valued operator. Let $\eta$ be the origami square root function as defined in the end of Section~\ref{subsec:intro_Kasteleyn_operator}.
 
Let $\alpha$ be a $(0,1)$-form with $\mC^1$ coefficients. If the involution $\sigma$ is present, then we assume that $\sigma^*\alpha = -\bar{\alpha}$. We define the \emph{perturbed} Kasteleyn operator $K_\alpha$ by
\begin{equation}
  \label{eq:def_of_Kalpha}
  K_\alpha(w,b) = \exp(2i\Im \int_w^b\alpha)\cdot K_\delta(w,b)
\end{equation}
where the integration is taken along the edge $wb$ of $G$. For the next lemma we need some notation. Let us fix an arbitrary~\emph{smooth} metric on $\Sigma$. Then, given an open set $U\subset \Sigma$ we denote by $\mC^n(U)$ the space of $n$ times continuously differentiable functions on $U$ with the usual norm associated with this metric. 
Finally, given two $(0,1)$-forms $\alpha_1,\alpha_2$ we denote by $\alpha_1/\alpha_2$ the function for which $\alpha_1 = (\alpha_1/\alpha_2)\cdot \alpha_2$.

\begin{lemma}
  \label{lemma:Kalpha_and_dbar}
  Let $w\in W$ be a white vertex of $G$ and $U\subset\Sigma$ be an open subset containing $w$ and all its neighbors and isometric to a convex Euclidean polygon in the metric $ds^2$. Assume that $f$ is a $\mC^2$ multivalued function on $U$ that picks the multiplicative monodromy $-1$ when its argument crosses any of the paths $\gamma_1,\dots,\gamma_{g-1}$. Let $\mu_w$ denote the area of the face $w$ of the t-embedding associated with $G$. Then we have
  \begin{multline}
    \sum_{b\sim w}K_\alpha(w,b)f(b) = 4i\mu_w\cdot ( (\dbar + \frac{\alpha_0}{2} + \alpha_G + \alpha) f)(w)\cdot \bar{\omega}_0(w)^{-1} + \\
    + \frac{\|f\|_{\mC^0(B(w,\delta))}\cdot \|(\alpha+\alpha_0+\alpha_G)/\bar{\omega}_0\|_{\mC^1(B(w,\delta))}}{\dist(w, \{ p_1,\dots, p_{2g-2} \})^{1/2}}\cdot O(\delta^3) + \\
    + \frac{\|f\|_{\mC^2(B(w,\delta))}\cdot \|\alpha/\bar{\omega}_0\|_{\mC^0(B(w,\delta))}}{\dist(w, \{ p_1,\dots, p_{2g-2} \})}\cdot O(\delta^3)
  \end{multline}
  as $\delta\to 0$, the constant in $O(\ldots)$ depends only on $\lambda$.
\end{lemma}
\begin{proof}
  Recall that, by the definition~\eqref{eq:def_of_K} of $K$ and the definition~\eqref{eq:def_of_Kalpha}, outside of the cuts $\gamma_1,\dots, \gamma_{g-1}$, the operator $K_\alpha$ can be written as
  \begin{equation}
    \label{eq:K_def_again}
    K_\alpha(w,b) = \exp\left( i\Im\int_w^b(\alpha+2\alpha_G) + i\Im(\int_{p_0}^b\alpha_0 +\int_{p_0}^w\alpha_0) \right)(\Tt(v_2) - \Tt(v_1)).
  \end{equation}
  Recall also that $\omega_0(p) = \exp\left( 2i\int_{p_0}^p\Im\alpha_0 \right)\omega(p)$ where $\omega = d\Tt$ (see~\eqref{eq:def_of_omega} and~\eqref{eq:def_of_Tt}), that is,
  \begin{equation}
    \label{eq:omega_is_dTt}
    \omega_0(p) = \exp\left( 2i\int_{p_0}^p\Im\alpha_0 \right)\,d\Tt.
  \end{equation}
  Finally, recall that $G^\ast$ is t-embedded into $\Sigma$ and, locally, coincides with a t-embedding with $O(\delta)$-small origami.

  We can now compute $\sum_{b\sim w}K_\alpha(w,b)f(b)$ following the steps below:
  \begin{itemize}
    \item Replace $K_\alpha$ with the right-hand side of~\eqref{eq:K_def_again}. This puts us in the setup of Lemma~\ref{lemma:approx_of_dbar} with $f(b)$ being replaced by $\exp\left( i\Im\int_w^b(\alpha+2\alpha_G) + i\Im(\int_{p_0}^b\alpha_0 +\int_{p_0}^w\alpha_0) \right)f(b)$.
    \item Use $\Tt$ as a local coordinate and apply Lemma~\ref{lemma:approx_of_dbar}. Recall that we have to compute the derivative of $\exp\left( i\Im\int_w^b(\alpha+2\alpha_G) + i\Im(\int_{p_0}^b\alpha_0 +\int_{p_0}^w\alpha_0) \right)f(b)$, not just $f(b)$ itself. This gives rise to the appearance of the operator $\dbar + \frac{\alpha_0}{2} + \alpha_G + \alpha$. We divide by $\bar\omega_0$ because the result of applying this operator is not a function but a $(0,1)$. Note that the exponential factor in~\eqref{eq:omega_is_dTt} matches with $\exp\left( i\Im\int_w^b(\alpha+2\alpha_G) + i\Im(\int_{p_0}^b\alpha_0 +\int_{p_0}^w\alpha_0) \right)$ evaluated at $b=w=p$.
    \item Finally, we need to carefully estimate the error terms. The estimates given by Lemma~\ref{lemma:approx_of_dbar} are made for the derivatives takes with respect to the coordinate $\Tt$ which is not smooth at conical singularities. Rewriting this in terms of derivatives taken with respect to a smooth local coordinate we get the error terms as in the lemma.
  \end{itemize}
\end{proof}

The goal of this section is to construct and estimate $K_\alpha^{-1}$. As we can see from Lemma~\ref{lemma:Kalpha_and_dbar}, the matrix $K_\alpha$ approximates the operator $\Dd_\alpha = \dbar + \frac{\alpha_0}{2} + \alpha_G + \alpha$ on smooth multivalued functions with the monodromy $(-1)^{\gamma\cdot (\gamma_1+\ldots+\gamma_{g-1})}$ along a loop $\gamma$. Following an analogy with Theorem~\ref{thmas:parametrix}, we are aiming to prove that $K^{-1}_\alpha(b,w)$ approximates $\Dd_{\alpha + \alpha_G}^{-1}(b,w) + (\eta_b\eta_w)^2\overline{\Dd_{-\alpha + \alpha_G}^{-1}(b,w)}$, where $\Dd_\alpha^{-1}$ is the kernel defined in Proposition~\ref{prop:def_of_S}.

\subsection{Auxiliary notation}
\label{subsec:aux_for_Kalphainv}

We need some auxiliary notation. Recall the multivalued function $\Tt$ on $\Sigma$ defined in Section~\ref{subsec:intro_graphs_on_Sigma0}. We define the multivalued function $K_\Tt(w,b)$, the local parametrix $K_\Tt^{-1}(b,w)$ and functions $[\Tt(b)^s], [\Tt(w)^s]$ as follows:
\begin{enumerate}
  \item Let $bw$ be an edge of $G$ and $v_1v_2$ be the dual edge of $G^\ast$ oriented such that the black face is on the right. Define $K_\Tt(w,b) = \Tt(v_2) - \Tt(v_1)$.
  \item Assume that $w$ is a white vertex and put $r = \dist(w,\{ p_1,\dots, p_{2g-2} \})$ Then, according to Assumption~\ref{item:intro_reg_part} and Assumption~\ref{item:intro_conical_sing} on $G$ from Section~\ref{subsec:intro_graphs_on_Sigma0} the mapping $\Tt$ provides an isometry between $G^\ast\cap B_\Sigma(w,\min(r,\lambda))$ and a subset of a full-plane t-embedding. For each $b\in B_\Sigma(w,\lambda/2)$ we put $K_\Tt^{-1}(b,w)$ to be the unique inverting kernel defined by applying Theorem~\ref{thmas:parametrix} to this t-embedding. Note that $K_\Tt^{-1}(b,w)$ is discrete holomorphic in $b$, but not in $w$, as for different choices of $w$ the full-plane t-embeddings might vary.
  \item Assume that $\dist(w,p_j)\leq \lambda$ for some $j = 1,\dots, 2g-2$. By Assumption~\ref{item:intro_conical_sing} on $G$ from Section~\ref{subsec:intro_graphs_on_Sigma0}, the graph $G\cap B_\Sigma(p_j, 2\lambda)$ is isometric to a subgraph of a Tempreley isoradial graph on an infinite cone as defined in Section~\ref{subsec:graph_on_a_cone}. For each $b\in B_\Sigma(w, \lambda/2)$ we identify $K_\Tt^{-1}$ with $K_{\Tt,1/2}^{-1}$, where the latter is the unique inverting kernel defined by applying Lemma~\ref{lemma:kernel_of_G4pi} to this graph.
  \item Pick a $j = 1,\dots, 2g-2$. Recall that, by Assumption~\ref{item:intro_conical_sing} on $G$ from Section~\ref{subsec:intro_graphs_on_Sigma0}, the mapping $\Tt$ defines a double cover from $G\cap B(p_j,2\lambda)$ onto a subgraph of a full-plane isoradial graph. Given $b,w\in B_\Sigma(p_j, 2\lambda)$ and $s\geq -\frac{1}{2}$ we define $[\Tt(b)^s]$, $[\Tt(w)^s]$ to be the corresponding function constructing by applying Lemma~\ref{lemma:fs} to the black and white vertices of this graph respectively. Recall that $[\Tt(b)^s]$ is a discrete version of $(b-p_j)^{2s}$.
\end{enumerate}
The parametrix $K_\Tt^{-1}$ is intended to describe the leading term in the asymptotic of $K_\alpha^{-1}(b,w)$ at the diagonal. Note that the definition of $K_\Tt^{-1}, K_{\Tt,1/2}^{-1}(b,w)$ and $K_\Tt$ depends on the choice of the branch of $\Tt$; in all forthcoming expressions it is assumed that all the objects whose definition depend on $\Tt$ are defined with respect to the same branch. In fact, all the expressions we are going to build using $\Tt$ will not depend on the choice of this branch: for example, we have a well-defined edge weight $K_\Tt(w,b)K_\Tt^{-1}(b,w)$, which is nothing but the probability of the edge $bw$ to be covered by a dimer computed with respect to the Gibbs measure on the full-plane t-embedding corresponding to $K_\Tt^{-1}$. Another example is an expression of the form $\exp\left( -i\Im(\int_{p_0}^b\alpha_0 + \int_{p_0}^w\alpha_0) \right)K_\Tt^{-1}(b,w)$. Note that by specifying a path between $p_0$ and $w$ we also specify a branch of the multivalued form $\omega$ (see~\eqref{eq:def_of_omega}), which in turn determines a branch of $\Tt$. If we assume that $\dist(b,w)<\lambda$, then we can specify the path between $p_0$ and $b$ as well by concatenating the path between $p_0$ and $w$ with the unique geodesic connecting $w$ and $b$. Under these conventions it is easy to verify that $\exp\left( -i\Im(\int_{p_0}^b\alpha_0 + \int_{p_0}^w\alpha_0) \right)K_\Tt^{-1}(b,w)$ does not depend on the choice of the path between $p_0$ and $w$.

Let $\Mm_g^{t,(0,1)}$ denote the moduli space of Torelli marked curves of genus $g$ with a fixed anti-holomorphic $(0,1)$-form, see Section~\ref{subsec:teichmuller_space} for details. In what follows we will often be using the notation $d$ for $\dist$ and $\underline{p}$ for $\{ p_1,\dots, p_{2g-2} \}$ to shorten the formulae.

\subsection{Construction of the inverse kernel \texorpdfstring{$K_\alpha^{-1}$}{Kalpha-1}}
\label{subsec:Kalphainverse}

Our strategy of constructing $K_\alpha^{-1}$ is essentially the same as in the proof of Lemma~\ref{eq:Kinv_via_S_one_puncture} (see also~\cite[Section~5]{DubedatFamiliesOfCR} where $K_\alpha^{-1}$ is constructed on a torus). It consists of the following steps:
\begin{enumerate}
  \item Construct an approximation $S_\alpha$ of the inverting kernel by modifying the function $\Dd_{\alpha + \alpha_G}^{-1}(b,w) + (\eta_b\eta_w)^2\overline{\Dd_{-\alpha + \alpha_G}^{-1}(b,w)}$ locally near its singularities using discrete holomorphic objects such as $K_\Tt^{-1}$ (near the diagonal) and $[\Tt(b)^s]$ (near conical singularities).
  \item Estimate $T= K_\alpha S_\alpha - \Id$ and $S_\alpha T$.
  \item Define $K_\alpha^{-1} = S_\alpha(\Id + T)^{-1}$ and use the estimates on $T$ to estimate the difference between $K_\alpha^{-1}$ and $S_\alpha$.
\end{enumerate}

We now describe the kernel $S_\alpha(b,w)$. Let $\beta>0$ the minimum of the exponents from Theorem~\ref{thmas:parametrix} and Lemma~\ref{lemma:kernel_of_G4pi}. We fix the mesoscopic scales 
\begin{equation}
  \label{eq:def_of_nu}
  \delta^\beta\ll \nu_1 \ll \nu_2 \ll \nu_3 \ll 1
\end{equation}
specified later. Given $b,w$ such that $d(b,w)\leq \min( d(b,\underline{p}), d(w,\underline{p}) )$ let us set
\begin{equation}
  \label{eq:def_of_sign}
  (-1)^{bw\cdot \gamma} = (-1)^{l(bw)\cdot (\gamma_1+\ldots+\gamma_{g-1})}
\end{equation}
where $l(bw)$ is the geodesic between $b$ and $w$ and $\cdot$ denotes the algebraic intersection number.
We consider the following cases.

\emph{Case 1: Definition of $S_\alpha(b,w)$ when $d(w,\{ p_1,\dots, p_{2g-2} \})\geq \nu_2$.}

\begin{enumerate}
  \item Assume that $d(b,\{ w,p_1,\dots, p_{2g-2} \})\geq \nu_1$. Then we define
    \begin{equation*}
      S_\alpha(b,w) = \frac{1}{2}\left[  \Dd_{\alpha+\alpha_G}^{-1}(b,w) + (\eta_b\eta_w)^2\overline{\Dd_{-\alpha + \alpha_G}^{-1}(b,w)}\right].
    \end{equation*}
  \item Assume that $d(b, w)\leq \nu_1$. Choose a holomorphic coordinate $z$ defined in $B(w,\lambda)$ and write $\alpha = a(z)\,d\bar z$ and $\alpha_c = a(z(w))\,d\bar z$. Note that $\alpha_c$ is closed and $\alpha_c - \alpha$ is small in the $\nu_1$-neighborhood of $w$. We define
    \begin{multline}
      \label{eq:Salpha_case12}
      S_\alpha(b,w) 
      =(-1)^{bw\cdot \gamma}\exp\left[ -i\Im (2\int_w^b (\alpha_c + \alpha_G) + \int_{p_0}^b \alpha_0 + \int_{p_0}^w\alpha_0) \right] K^{-1}_\Tt(b, w) + \\
      + \frac{(-1)^{bw\cdot \gamma}}2 \exp\left[ -i\Im \int_w^b (2\alpha_c + 2\alpha_G + \alpha_0)  \right]\cdot (r_{\alpha+\alpha_G}(w) + (\eta_b\eta_w)^2\overline{r_{-\alpha + \alpha_G}(w)}).
    \end{multline}
    where the integration in $\int_w^b$ is taken along the geodesic between $b$ and $w$ (or any other homotopic path).
  \item Assume that $d(b,p_i)\leq \nu_1$ for some $i\in \{ 1,\dots, 2g-2 \}$. Choose a holomorphic coordinate $z$ defined in $B(p_i,\lambda)$ and write $\alpha = a(z)\,d\bar z$ and $\alpha_c = a(z(w))\,d\bar z$. Define 
    \[
      \widetilde{\Dd}^{-1}_\alpha(b,w) = [\Tt(b)^{-1/4}]\cdot \lim\limits_{p\to p_i} \Dd_\alpha^{-1}(p,w)(\Tt(p)- \Tt(p_i))^{1/4}
    \]
    and
    \begin{multline}
      \label{eq:b_near_pi1}
      S_\alpha(b,w) =\\= \frac{1}{2}\exp\left( -i\Im\int_{p_i}^b(2\alpha_c + 2\alpha_G + \alpha_0) \right) \left[ \widetilde{\Dd}^{-1}_{\alpha + \alpha_G}(b,w) + (\eta_b\eta_w)^2\overline{\widetilde{\Dd}^{-1}_{-\alpha + \alpha_G}(b,w)} \right]
    \end{multline}
    where the integration is taken along the geodesic.
\end{enumerate}

\emph{Case 2: Definition of $S_\alpha(b,w)$ when $d(w,p_i)\leq \nu_2$ for some $i\in \{ 1,\dots, 2g-2 \}$.} Define
\[
  \Dd_{i,\alpha}^{-1}(p,w) = \Dd_{\alpha}^{-1}(p,w)\cdot[(\Tt(w))^{-1/4}](\Tt(w) - \Tt(p_i))^{1/4}.
\]

\begin{enumerate}
  \item Assume that $d(b,\{p_1,\dots, p_{2g-2} \})\geq \nu_3$. Then we define
    \begin{equation}
      \label{eq:Salpha_case21}
      S_\alpha(b,w) = \frac{1}{2}\left[  \Dd_{i,\alpha+\alpha_G}^{-1}(b,w) + (\eta_b\eta_w)^2\overline{\Dd_{i,-\alpha + \alpha_G}^{-1}(b,w)}\right].
    \end{equation}
  \item Assume that $d(b, p_i)\leq \nu_3$. Choose a holomorphic coordinate $z$ defined in $B(p_i,\lambda)$ and write $\alpha = a(z)\,d\bar z$ and $\alpha_c = a(z(p_i))\,d\bar z$. Define
    \begin{multline}
      \label{eq:Salpha_case22}
      S_\alpha(b,w) =\exp\left(-i\Im (2\int_w^b (\alpha_c + \alpha_G) + \int_{p_0}^b\alpha_0 + \int_{p_0}^w\alpha_0)\right)\cdot K_{\Tt,1/2}^{-1}(b,w) + \\
      + \frac{1}{2}\exp\left[ -i\Im \int_w^b (2\alpha_c + 2\alpha_G + \alpha_0)  \right]\cdot\\\cdot \Big(r_{\alpha+\alpha_G}(w)(\Tt(w) - \Tt(p_i))^{1/2}[\Tt(w)^{-1/4}][\Tt(b)^{-1/4}] + \\
      + (\eta_b\eta_w)^2\overline{r_{-\alpha + \alpha_G}(w)(\Tt(w) - \Tt(p_i))^{1/2}[\Tt(w)^{-1/4}][\Tt(b)^{-1/4}]}\Big)
    \end{multline}
    where the integration in $\int_w^b$ is taken along the geodesic.
  \item Assume that $d(b,p_j)\leq \nu_3$ for some $j\neq i$. Choose a holomorphic coordinate $z$ defined in $B(p_j,\lambda)$ and write $\alpha = a(z)\,d\bar z$ and $\alpha_c = a(z(p_j))\,d\bar z$. Define 
    \[
      \widetilde{\Dd}^{-1}_{i, \alpha}(b,w) = [\Tt(b)^{-1/4}]\cdot\lim\limits_{p\to p_j} \Dd_{i,\alpha}^{-1}(p,w)(\Tt(p) - \Tt(p_j))^{1/4}
    \]
    and
    \begin{multline}
      \label{eq:Salpha_case23}
      S_\alpha(b,w) = \\
      = \frac{1}{2}\exp\left( -i\Im\int_{p_j}^b(2\alpha_c + 2\alpha_G + \alpha_0) \right)\cdot\left[ \widetilde{\Dd}^{-1}_{i,\alpha + \alpha_G}(b,w) + (\eta_b\eta_w)^2\overline{\widetilde{\Dd}^{-1}_{i,-\alpha + \alpha_G}(b,w)} \right]
    \end{multline}
    where the integration is taken along the geodesic.
\end{enumerate}

Recall that
\begin{equation}
  \label{eq:def_of_T}
  T = K_\alpha S_\alpha - \Id.
\end{equation}
We will now provide a number of estimates on $T(u,w)$ depending on the relative positions of $u,w$ and the conical singularities. As before, we will use the symbol $A\lesssim B$ when $|A|\leq \cst |B|$ for some constant $\cst>0$. Recall that we use $\underline{p}$ to denote $\{ p_1,\dots,p_{2g-2} \}$. Recall that $\beta>0$ is the minimum of exponents from Theorem~\ref{thmas:parametrix} and Lemma~\ref{lemma:kernel_of_G4pi}. We will now list all the estimates following the same order as in the definition of $S_\alpha$.

\emph{Case 1: Bound of $T(u,w)$ when $\dist(w,\{ p_1,\dots, p_{2g-2} \})\geq \nu_2$.}

\begin{enumerate}
  \item Assume that $\dist(u,\{ w,p_1,\dots, p_{2g-2} \})\geq \nu_1+\delta$. In this case
    \begin{equation}
      \label{eq:bound_on_T1}
      T(u,w) \lesssim \left(\frac{\delta^3}{d(u,w)^2} + \frac{\delta^3}{d(u,\underline{p})^2}\right)\frac{\sqrt{d(u, \underline{p}) + d(w,\underline{p})}}{d(u,w)\sqrt[4]{d(u,\underline{p})d(w,\underline{p})}}
    \end{equation}
  \item Assume that $\dist(u,w) \in [\nu_1-\delta,\nu_1+\delta]$. In this case
    \begin{equation}
      \label{eq:bound_on_T2}
      T(u,w) \lesssim \frac{\delta^{1+\beta}}{\nu_1^{1+\beta}} + \frac{\delta\nu_1}{d(w,\underline p)^{3/2}}.
    \end{equation}
  \item Assume that $d(u,w)\leq \nu_1-\delta$. In this case
    \begin{equation}
      \label{eq:bound_on_T3}
      T(u,w) \lesssim \frac{\delta^2}{d(w,\underline p)}.
    \end{equation}
  \item Assume that $\dist(u,p_i)\in (\nu_1-\delta,\nu_1+\delta)$ for some $i\in \{ 1,\dots, 2g-2 \}$. In this case
    \begin{equation}
      \label{eq:bound_on_T4}
      T(u,w) \lesssim \frac{\delta^{3/2}}{\nu_1^{3/4} d(w,\underline p)^{3/4}} + \frac{\delta\nu_1^{1/4}}{d(w,\underline p)^{3/4}}.
    \end{equation}
  \item Assume that $\dist(u,p_i)\leq \nu_1-\delta$ for some $i\in \{ 1,\dots, 2g-2 \}$. In this case
    \begin{equation}
      \label{eq:bound_on_T5}
      T(u,w)\lesssim \frac{\delta^2}{d(u,p_i)^{1/4}d(w,\underline p)^{3/4}}.
    \end{equation}
\end{enumerate}

\emph{Case 2: Bound on $T(u,w)$ when $\dist(w,p_i)\leq \nu_2$ for some $i\in \{ 1,\dots, 2g-2 \}$.} 
\begin{enumerate}
  \item Assume that $\dist(u,\{p_1,\dots, p_{2g-2} \})\geq \nu_3+\delta$. In this case
    \begin{equation}
      \label{eq:bound_on_T6}
      T(u,w) \lesssim \frac{\delta^3}{d(u,\underline{p})^2\sqrt{d(u,p_i)}\sqrt[4]{d(u,\underline{p})}\sqrt[4]{d(w,p_i)}}
    \end{equation}
  \item Assume that $\dist(u, p_i)\in [\nu_3-\delta, \nu_3+\delta]$. In this case
    \begin{equation}
      \label{eq:bound_on_T7}
      T(u,w) \lesssim \frac{\delta\nu_3^{1/4}}{d(w, p_i)^{1/4}} + \frac{\delta^{1+\beta} }{\nu_3^{1+\beta}} + \frac{\delta^{3/2}}{d(w,p_i)^{1/4}\nu_3^{3/4}}.
    \end{equation}
  \item Assume that $\dist(u, p_i)\leq \nu_3-\delta$. In this case
    \begin{equation}
      \label{eq:bound_on_T8}
      T(u,w) \lesssim \frac{\delta^2\sqrt{d(u,p_i) + d(w,p_i)}}{(d(u,w)+\delta)\sqrt[4]{d(w,p_i)}\sqrt[4]{d(u,p_i)}}
    \end{equation}
  \item Assume that $\dist(u,p_j)\in (\nu_3-\delta,\nu_3+\delta)$ for some $j\neq i$. In this case
    \begin{equation}
      \label{eq:bound_on_T9}
      T(u,w) \lesssim \frac{\delta\nu_3^{1/4}}{d(w,p_i)^{1/4}} + \frac{\delta^{3/2}}{\nu_3^{3/4}d(w,p_i)^{1/4}}.
    \end{equation}
  \item Assume that $\dist(u,p_j)\leq \nu_3-\delta$ for some $j\neq i$. In this case
    \begin{equation}
      \label{eq:bound_on_T10}
      T(u,w) \lesssim \frac{\delta^2 }{d(u,p_j)^{1/4}d(w,p_i)^{1/4}}.
    \end{equation}
\end{enumerate}

In the course of proving these estimates we will often use the following bound which is straightforward:

\begin{lemma}
  \label{lemma:bound_on_Ddalpha}
  Let $\Kk\subset \Mm_g^{t,(0,1)}$ be a compact subset such that for any point $[\Sigma, A,B,\alpha]\in \Kk$ we have $\theta[\alpha_G\pm\alpha_h](0)\neq 0$. Let $R>0$ be given. Then there exists a constant $C>0$ depending only on $\Kk$ and $R$ such that whenever $\alpha = \dbar \vphi + \alpha_h$ is such that 
  \[
    (\Sigma, A_1,\dots, A_g,B_1,\dots,B_g,\alpha_h)\in \Kk, \qquad \|\vphi\|_{\mC^2}\leq R
  \]
  and $b,w\in \Sigma$ are such that $\dist(b,w) \geq \delta$ we have
    \begin{equation}
      \label{eq:bound_on_Ddinv}
      |\Dd_{\alpha_G+\alpha}^{-1}(b,w)| + |\Dd_{\alpha_G-\alpha}^{-1}(b,w)| + |S_\alpha(b,w)| \leq C^2\frac{\sqrt{d(b,\underline{p}) + d(w,\underline{p})}}{d(b,w)\sqrt[4]{d(b,\underline{p})d(w,\underline{p})}}
    \end{equation}
\end{lemma}
\begin{proof}
  The estimate on $|\Dd^{-1}_{\alpha_G\pm\alpha}(b,w)|$ follows immediately from the definition of it given in Proposition~\ref{prop:def_of_S}. Indeed, recall that locally we have
  \begin{equation}
    \label{eq:bDal1}
      \Dd_\alpha^{-1}(p,q) = \frac{\theta[\alpha_h](\Aa(p-q))}{\pi i \theta[\alpha_h](0)\cdot \pf(p,q) \sqrt{\omega_0(p)}\sqrt{\omega_0(q)}} \cdot \exp\left(\vphi(q) - \vphi(p)-2i\int\limits_q^p\Im \alpha_h\right). 
  \end{equation}
  To estimate the right-hand side of~\eqref{eq:bDal1} we can bound $\theta[\alpha_h]$ in the nominator and denominator by a constant, and do the same with the exponential factor. To bound $\pf(p,q) \sqrt{\omega_0(p)}\sqrt{\omega_0(q)}$ we may use a local coordinate $z$ chosen in a neighborhood of $\{ p,q \}$ and write
  \[
    \pf(p,q) \sqrt{\omega_0(p)}\sqrt{\omega_0(q)} = \left(\pf(p,q)\sqrt{dz(p)}\sqrt{dz(q)}\right) \cdot \left(\sqrt{\frac{\omega_0}{dz}(p)}\sqrt{\frac{\omega_0}{dz}(q)}\right).
  \]
  Each pair of brackets on the right-hand side now contains a function. To estimate these functions we notice that, if a conical singularity $p_j$ is close to $p$, then
  \[
    |\tfrac{\omega_0}{dz}(p)|\asymp\sqrt{\dist(p,p_j)},
  \]
  and we also have
  \[
    |\pf(p,q)\sqrt{dz(p)}\sqrt{dz(q)}|\asymp \frac{d(b,w)}{\sqrt{d(b,\underline p) + d(w,\underline p)}}.
  \]
  since $\pf$ has a simple zero along the diagonal. The two estimates above prove the desired estimate for $\Dd^{-1}_{\alpha_G\pm\alpha}$. Note furthermore that our arguments imply that the right-hand side of~\eqref{eq:bound_on_Ddinv} is comparable with $\Dd^{-1}_{\alpha_G\pm\alpha}$ on any compact outside of the zero locus of $\Dd^{-1}_{\alpha_G\pm\alpha}$ (given by the zero locus of $\theta[\alpha_h](\Aa(p-q))$).

  Let us now estimate $|S_\alpha(b,w)|$. It is enough to estimate it when $b$ or $w$ approach conical singularities and/or each other. These are exactly the regimes when $S_\alpha(p,q)$ is different from $S_\alpha(b,w) = \frac{1}{2}\left[  \Dd_{\alpha+\alpha_G}^{-1}(b,w) + (\eta_b\eta_w)^2\overline{\Dd_{-\alpha + \alpha_G}^{-1}(b,w)}\right]$. Recall that to construct $S_\alpha$ in each particular regime we consider the main term in the asymptotics of $\Dd^{-1}_{\alpha_G\pm\alpha}$ and replace it with the corresponding discrete holomorphic function constructed in Sections~\ref{subsec:multivalued},~\ref{subsec:graph_on_a_cone}. The estimates on these discrete holomorphic functions proven in Lemma~\ref{lemma:fs}, Lemma~\ref{lemma:Ksinv_a_priori_bound} and Lemma~\ref{lemma:kernel_of_G4pi} imply that they are comparable with their continuous analogs. This shows in particular that $|S_\alpha(b,w)|$ is bounded by a constant times $|\Dd^{-1}_{\alpha_G+\alpha}(b,w)| + |\Dd^{-1}_{\alpha_G-\alpha}(b,w)|$ if only $\Aa(b-w)$ stays on a definite distance from the zero locus of $\theta[\alpha_G\pm\alpha_h]$.

  Finally, let us fix $w$ and assume that $b$ is close to a zero $p$ of, say, $\Dd_{\alpha_G + \alpha}^{-1}(b,w)$. Let $d$ be the order of this zero. Then we the same arguments as above allows to bound $|S_\alpha(b,w)|$ by $|(z(b)-z(p))^{-d}\Dd_{\alpha_G +\alpha}^{-1}(b,w)|$ which is still bounded by the right-hand side of~\eqref{eq:bound_on_Ddinv} for the same reason as in the beginning of the proof: the only change required is to notice that $(z(b) - z(p))^{-p}\theta[\alpha_G+\alpha_h](\Aa(b-w))$ is bounded from above.
\end{proof}

\begin{lemma}
  \label{lemma:bounds_on_T}
  Let $\Kk\subset \Mm_g^{t,(0,1)}$ be a compact subset such that for any point $[\Sigma, A,B,\alpha]\in \Kk$ we have $\theta[\alpha](0)\neq 0$. Assume that $\alpha = \dbar \vphi + \alpha_h$ is such that
  \[
    (\Sigma, A_1,\dots, A_g, B_1,\dots, B_g,\pm\alpha_h+\alpha_G)\in \Kk.
  \]
  Then the inequalities~\eqref{eq:bound_on_T1}--\eqref{eq:bound_on_T10} hold with some constants depending on $\lambda, \|\vphi\|_{\mC^2(\Sigma)}\leq R$ and $\Kk$ only.
\end{lemma}
\begin{proof}
  We will discuss the inequalities in the same order they are presented above.

  \emph{Case 1: $\dist(w,\{ p_1,\dots, p_{2g-2} \})\geq \nu_2$.}

  \begin{enumerate}
    \item The bound~\eqref{eq:bound_on_T1} follows from Lemma~\ref{lemma:Kalpha_and_dbar} and direct estimates of the derivatives of $\Dd^{-1}_{\pm\alpha+\alpha_G}$ similar to Lemma~\ref{lemma:bound_on_Ddalpha}.

    \item\label{T_bound_item:case12} The bound~\eqref{eq:bound_on_T2} is given by $\delta$ times the mismatch between the continuous kernel used to define $S_\alpha$ when $\dist(b,w)\geq \nu_1$ and the patch used when $\dist(b,w)\leq \nu_1$. We start with the latter and transform it to the former in three steps. First, we replace $K_\Tt^{-1}$ with its continuous analog given in Theorem~\ref{thmas:parametrix}; this creates the error $O(\delta^\beta\nu_1^{-1-\beta})$. Then we replace $\alpha_c$ with $\alpha$. Recall that $\alpha = a(z)\,d\bar z$ and $\alpha_c = a(z(w))\,d\bar z$ where $z$ is a holomorphic local coordinate defined in a macroscopic neighborhood of $w$. Comparing the metrics $|dz|^2$ and $ds^2$ one can show that $a(z) - a(z(w)) = O(\nu_1 d(w,\underline p)^{-1/2})$ and the $|dz|^2$-length of the $ds^2$-geodesic between $b$ and $w$ is of order $O(\nu_1 d(w,\underline p)^{-1/2})$. Combining this with the fact that $K_\Tt^{-1}(b,w) = O(\nu_1^{-1})$ by Theorem~\ref{thmas:parametrix} and $r_{\pm\alpha+\alpha_G}(w) = O(d(w,\underline p)^{-1/2})$ we get the error $O(\nu_1d(w,\underline p)^{-1})$. Finally, we can compare the resulting expression (with $K_\Tt^{-1}$ and $\alpha_c$ being replaced) to the near diagonal expansion of the actual kernel given in Lemma~\ref{lemma:diagonal_expansion_of_Salpha}. The mismatch between these two expressions is of order $O(\nu_1d(w,\underline p)^{-3/2})$. Adding all these error terms and multiplying them by $\delta$ we get~\eqref{eq:bound_on_T2}.

    \item\label{T_bound_item:case13} Using the fact that $\alpha_c$ is a closed form it is easy to show that $(K_{\alpha_c}S_\alpha)(u,w) - \Id(u,w) = 0$. Thus, we have the bound
      \begin{equation}
        \label{eq:T_bound_case13}
        |T(u,w)|\leq |((K_\alpha - K_{\alpha_c})S_\alpha)(u,w)| \leq \max_{b\sim u}|K_\alpha(u,b) - K_{\alpha_c}(u,b)|\cdot \max_{b\sim u}|S_\alpha(b,w)|.
      \end{equation}
      Let $b$ be a vertex incident to $u$. Arguing as in the previous item (Item~\ref{T_bound_item:case12}) we can bound $\int_u^b(\int\alpha_c - \alpha) = O(\delta d(u,w)d(w,\underline p)^{-1})$ (here the integration is taken along the edge between $u$ and $b$ as in the definition~\eqref{eq:def_of_Kalpha} of $K_\alpha$) which implies that
      \[
        \max_{b\sim u}|K_\alpha(u,b) - K_{\alpha_c}(u,b)| \lesssim \frac{\delta d(u,w)}{d(w,\underline p)}\max_{b\sim u}|K_\alpha(u,b)| \lesssim \frac{\delta^2 d(u,w)}{d(w,\underline p)}.
      \]
      Estimating $K^{-1}_\Tt(b,w)$ using Theorem~\ref{thmas:parametrix} we also conclude that $\max_{b\sim u}|S_\alpha(b,w)| \lesssim d(u,w)^{-1}$, therefore
      \[
        |T(u,w)| \lesssim \frac{\delta^2}{d(w,\underline p)}.
      \]

    \item\label{T_bound_item:case14} The bound~\eqref{eq:bound_on_T4} is given by $\delta$ times the mismatch between the continuous kernel used to define $S_\alpha$ when $\dist(b,\underline p)\geq \nu_1$ and the patch~\eqref{eq:b_near_pi1} used when $\dist(b,\underline p)\leq \nu_1$. To compare the patch with the continuous kernel we first notice that $\Dd_\alpha^{-1}(p,w)(\Tt(p)- \Tt(p_i))^{1/4}$ is a smooth function of $p$ near $p_i$ with the derivative (taken with respect to any smooth local coordinate at $p$) being of order $d(w,\underline p)^{-3/4}$, thus replacing $\Dd_\alpha^{-1}(p,w)(\Tt(p)- \Tt(p_i))^{1/4}$ with its limit as $p\to p_i$ creates the error $O(\nu_1^{1/2}d(w,\underline p)^{-3/4})$ (recall that the distance to $p_i$ in our singular metric $ds^2$ scale as a square of the distance in a smooth metric) which results in the total cost of $\delta\nu_1^{1/4}d(w,\underline p)^{-3/4}$ after we multiply back by $(\Tt(p)- \Tt(p_i))^{-1/4}$ and also by $\delta$. Arguing in the same way we can see that removing the exponential factor creates the same error. Finally, replacing $(\Tt(b)- \Tt(p_i))^{-1/4}$ with its discrete companion $[\Tt(b)^{-1/4}]$ creates the error $O(\delta^{1/2}\nu_1^{-3/4})$ according to Lemma~\ref{lemma:fs}, which becomes $\delta^{3/2}\nu_1^{-3/4}d(w,\underline p)^{-3/4}$ when we multiply by $\delta$ and $\lim\limits_{p\to p_i}\Dd_\alpha^{-1}(p,w)(\Tt(p)- \Tt(p_i))^{1/4}$.

    \item\label{T_bound_item:case15} Arguing in the same way as in Item~\ref{T_bound_item:case13} above we arrive to the bound~\eqref{eq:T_bound_case13}. Comparing $\alpha$ and $\alpha_c$ as previously we conclude that
      \[
        \max_{b\sim u}|K_\alpha(u,b) - K_{\alpha_c}(u,b)| \lesssim \delta^2.
      \]
      When $b\sim u$ we have $S_\alpha(b,w) \lesssim d(u,p_i)^{-1/4}d(w,\underline p)^{-3/4}$. Plugging these two bounds to~\eqref{T_bound_item:case13} we get~\eqref{eq:bound_on_T5}.

      We argue in the same way as in Item~\ref{T_bound_item:case13} above, namely, replace $\alpha$ with its Taylor expansion at $u$ and calculate the error. In this case we have $\alpha = \left(\tfrac{\alpha}{d\overline\Tt}(u) + O(\delta d(u, p_i)^{-3/2})\right)\,d\overline \Tt$ and $|K_\alpha(u,b)S_\alpha(b,w)| \lesssim \delta d(w,\underline p)^{-3/4}d(u,p_i)^{-1/4}$, therefore, arguing as after~\eqref{eq:T_bound_case13} we get the bound $\delta^3d(w,\underline p)^{-3/4}d(u,p_i)^{-7/4}$

  \end{enumerate}

  \emph{Case 2: $\dist(w,p_i)\leq \nu_2$ for some $i\in \{ 1,\dots, 2g-2 \}$.}

  \begin{enumerate}
    \item The bound~\eqref{eq:bound_on_T6} follows from Lemma~\ref{lemma:Kalpha_and_dbar} and direct estimates of the derivatives of $\Dd^{-1}_{\pm\alpha+\alpha_G}$ similar to Lemma~\ref{lemma:bound_on_Ddalpha}; in fact, one can just replace $w$ with $p_i$ in the right-hand side of~\eqref{eq:bound_on_T1}.

    \item The bound~\eqref{eq:bound_on_T7} is equal to $\delta$ times the mismatch between two patches applied when $d(b,\underline{p})\geq \nu_3$ and $d(b,p_i)\leq \nu_3$ respectively. Let us consider the patch corresponding to $d(b,p_i)\geq \nu_3$, that is,~\eqref{eq:Salpha_case21}, and estimate the cost of transforming it into the patch~\eqref{eq:Salpha_case22}. We first apply Lemma~\ref{lemma:bw_near_singularity_expansion_of_Salpha} to replace $\Dd_{\pm \alpha + \alpha_G}^{-1}(b,w)$ in the definition of $\Dd_{\pm\alpha + \alpha_G, i}^{-1}$ with the sum of the two terms on the right-hand side of~\eqref{eq:bw_near_singularity_expansion_of_Salpha} which results in
      \begin{multline}
        \label{eq:case2_bound2}
        \Dd_{\pm\alpha+\alpha_G,i}^{-1}(b,w) = \exp(\cdots) \cdot  S_i(b,w)\cdot [(\Tt(w))^{-1/4}](\Tt(w) - \Tt(p_i))^{1/4} + \\
        + \exp(\cdots) \cdot \frac{r_\alpha(w)(\Tt(w) - \Tt(p_i))^{1/2}[\Tt(w)^{-1/4}]}{(\Tt(b) - \Tt(p_i))^{1/4}} + O(\nu_3^{1/4}d(w, p_i)^{-1/4})
      \end{multline}
      where $S_i$ is as defined in~\eqref{eq:def_of_SCc} before Lemma~\ref{lemma:bw_near_singularity_expansion_of_Salpha} and we omit the exponential factors for shortness. We pick the main terms on the right-hand side of~\eqref{eq:case2_bound2} and replace $S_i(b,w)$ with the discrete kernel $K_{\Tt,1/2}^{-1}(b,w)$ and $\frac{r_\alpha(w)(\Tt(w) - \Tt(p_i))^{1/2}[\Tt(w)^{-1/4}]}{(\Tt(b) - \Tt(p_i))^{1/4}}$ with the discrete holomorphic (in $b$) $r_\alpha(w)(\Tt(w) - \Tt(p_i))^{1/2}[\Tt(w)^{-1/4}][\Tt(b)^{-1/4}]$. The first substitute creates the error $O(\delta^\beta \nu_3^{-1-\beta})$ by the last item of Lemma~\ref{lemma:kernel_of_G4pi}, and the second substitute creates an error $O(\delta^{1/2}d(w,p_i)^{-1/4}\nu_3^{-3/4})$ by Lemma~\ref{lemma:fs} (recall that $r_{\pm\alpha+\alpha_G}(w) = O(d(w,p_i)^{-1/2})$). Finally, we have to replace $\alpha$ with $\alpha_c$ in the exponential factors. Comparing $\alpha$ with $\alpha_c$ as in Item~\ref{T_bound_item:case12} in Case~1 and estimating $K_{\Tt,1/2}^{-1}$ using Lemma~\ref{lemma:kernel_of_G4pi} we can bound the error by $O(\nu_3^{1/4}d(w,p_i)^{-1/4})$. Multiplying all these errors by $\delta$ and adding them we get~\eqref{eq:bound_on_T7}.

    \item We prove the bound~\eqref{eq:bound_on_T8} following the same arguments as in Item~\ref{T_bound_item:case15} in Case~1, where we proved~\eqref{eq:bound_on_T5}.

    \item The bound~\eqref{eq:bound_on_T9} is equal to $\delta$ times the mismatch between two patches applied when $d(b,\underline{p})\geq \nu_3$ and $d(b,p_j)\leq \nu_3$ respectively. The latter mismatch can be calculated by estimating the error created by removing the exponential factor and then comparing $\Dd_{\pm\alpha +\alpha_G,i}^{-1}(b,w)$ with $\widetilde{\Dd}_{\pm\alpha+\alpha_G,i}^{-1}(b,w)$. To transform the latter into the former we must first replace $\lim\limits_{p\to p_j} \Dd_{i,\alpha}^{-1}(p,w)(\Tt(p) - \Tt(p_j))^{1/4}$ with $\Dd_{i,\alpha}^{-1}(b,w)(\Tt(b) - \Tt(p_j))^{1/4}$ and then $[\Tt(b)^{-1/4}]$ with $(\Tt(b) - \Tt(p_j))^{-1/4}$. Arguing as in Item~\ref{T_bound_item:case14} in Case~1 we conclude that the first replacement creates the error $O(\nu_3^{1/4}d(w,p_i)^{-1/4})$, while the second replacement creates the error $O(\delta^{1/2}\nu_3^{-3/4}d(w,p_i)^{-1/4})$. Finally, removing the exponential factor costs $O(\nu_3^{1/4}d(w,p_i)^{-1/4})$. Multiplying by $\delta$ and adding the errors we get~\eqref{eq:bound_on_T9}.

    \item We prove the bound~\eqref{eq:bound_on_T10} following the same arguments as in Item~\label{T_bound_item:case15} in Case~1, where we proved~\eqref{eq:bound_on_T5}.
  \end{enumerate}
 \end{proof}

We now estimate $S_\alpha T$. Our goal is to show that $(S_\alpha T)(b,w)\sqrt[4]{d(b,\underline{p})d(w,\underline{p})}$ is $o(1)$ when $\delta\to 0$ uniformly on a compact $\Kk$ such as in Lemma~\ref{lemma:bounds_on_T}. Applying Lemma~\ref{lemma:bounds_on_T} we get
\begin{equation}
  \label{eq:ST_leq_sum}
  |(S_\alpha T)(b,w)|\sqrt[4]{d(b,\underline{p})d(w,\underline{p})} \lesssim \sum_{u\in W}  \frac{\sqrt{d(b,\underline{p}) + d(u,\underline{p})}}{d(b,u)\sqrt[4]{d(u,\underline{p})}}|T(u,w)|\sqrt[4]{d(w,\underline{p})}.
\end{equation}
We will now use the bounds~\eqref{eq:bound_on_T1}--\eqref{eq:bound_on_T10} to estimate the right-hand side of~\eqref{eq:ST_leq_sum}. Similarly to how we derived these bounds we consider the cases when $w$ is far or close to conical singularities separately.

\emph{Case 1: Bound on $(S_\alpha T)(b,w)$ when $\dist(w,\{ p_1,\dots, p_{2g-2} \})\geq \nu_2$.}
Define
\[
  \Uu^{(1)}_0 = \{ u\in \Sigma\ \mid\ d(u,\underline p)\leq \nu_1+\delta \},\qquad \Uu^{(1)}_w = \{ u\in \Sigma\ \mid\ d(u,w)\leq \nu_1+\delta \}
\]
and, similarly to the proof of Lemma~\ref{lemma:existence_of_Ksinv} write
\begin{equation}
  \label{eq:ST_three_sums_case1}
  \sum_{u\in W}  \frac{\sqrt{d(b,\underline{p}) + d(u,\underline{p})}}{d(b,u)\sqrt[4]{d(u,\underline{p})}}|T(u,w)|\sqrt[4]{d(w,\underline{p})} \leq I^{(1)}_1 + I^{(1)}_w + I^{(1)}_0
\end{equation}
where
\[
  I^{(1)}_1 = \sum_{u\notin \Uu^{(1)}_0\cup\,\Uu^{(1)}_w}(\cdots),\qquad I^{(1)}_w = \sum_{u\in \Uu^{(1)}_w}(\cdots),\qquad I^{(1)}_0 = \sum_{u\in \Uu^{(1)}_0}(\cdots).
\]
Let us estimate each of $I^{(1)}_1,I^{(1)}_0$ and $I^{(1)}_w$ separately. These estimates will be rather similar to those applied in the proof of Lemma~\ref{lemma:existence_of_Ksinv}, so we will allow ourselves to omit some details.

\emph{Estimate of $I^{(1)}_1$.} Applying~\eqref{eq:bound_on_T1} we get the bound
\[
  I^{(1)}_1 \lesssim \sum_{u\notin \Uu^{(1)}_0\cup\,\Uu^{(1)}_w}\left( \frac{\sqrt{d(b,\underline{p}) + d(u,\underline{p})}}{d(b,u)\sqrt[4]{d(u,\underline{p})}} \right)\cdot\left( \left(\frac{\delta^3}{d(u,w)^2} + \frac{\delta^3}{d(u,\underline{p})^2}\right)\frac{\sqrt{d(u, \underline{p}) + d(w,\underline{p})}}{d(u,w)\sqrt[4]{d(u,\underline{p})}} \right).
\]
Arguing as in the proof of Lemma~\ref{lemma:existence_of_Ksinv} one can show that the main contribution in the sum above comes from the terms corresponding to $\nu_1\leq d(u,\underline p)\leq 2\nu_1$ and $\nu_1\leq d(u,w)\leq 2\nu_1$. In both regimes all the terms apart from $d(b,u)$ can be replaced with their maximal or minimal values which results in the following estimate:
\begin{equation}
  \label{eq:I1_case1}
  I^{(1)}_1\lesssim \frac{\delta}{\nu_1\nu_2^{1/2}} + \frac{\delta}{\nu_1^{3/2}}.
\end{equation}

\emph{Estimate of $I^{(1)}_w$.} In this case the bound on $T$ depend on whether $\nu_1-\delta \leq d(u,w)\leq \nu_1+\delta$ or $d(u,w)<\nu_1-\delta$. In the first case we can apply the estimate~\eqref{eq:bound_on_T2} and in the second case we can apply the estimate~\eqref{eq:bound_on_T3}. This results in
\begin{multline*}
  I^{(1)}_w \lesssim \sum_{\nu_1-\delta \leq d(u,w)\leq \nu_1+\delta} \left( \frac{\sqrt{d(b,\underline{p}) + d(u,\underline{p})}}{d(b,u)\sqrt[4]{d(u,\underline{p})}} \right)\cdot\left( \frac{\delta^{1+\beta}}{\nu_1^{1+\beta}} + \frac{\delta\nu_1}{d(w,\underline p)^{3/2}} \right)\cdot \sqrt[4]{d(w,\underline{p})} + \\
  +\sum_{d(u,w)<\nu_1-\delta} \left( \frac{\sqrt{d(b,\underline{p}) + d(u,\underline{p})}}{d(b,u)\sqrt[4]{d(u,\underline{p})}} \right)\cdot\frac{\delta^2}{d(w,\underline p)}\cdot\sqrt[4]{d(w,\underline{p})}
\end{multline*}
which implies
\begin{equation}
  \label{eq:Iw_case1}
  I^{(1)}_w \lesssim \left( \frac{\delta^\beta}{\nu_1^{1+\beta}} + \frac{\nu_1}{\nu_2} \right)|\log\delta| + \frac{\nu_1}{\nu_2^{3/2}}.
\end{equation}

\emph{Estimate of $I^{(1)}_0$.} In this case the bound on $T$ depend on whether $\nu_1-\delta \leq d(u,\underline p)\leq \nu_1+\delta$ or $d(u,\underline p)<\nu_1-\delta$. In the first case we can apply the estimate~\eqref{eq:bound_on_T4} and in the second case we can apply the estimate~\eqref{eq:bound_on_T5}. This results in
\begin{multline*}
  I^{(1)}_0 \lesssim \sum_{\nu_1-\delta \leq d(u,\underline p)\leq \nu_1+\delta} \left( \frac{\sqrt{d(b,\underline{p}) + d(u,\underline{p})}}{d(b,u)\sqrt[4]{d(u,\underline{p})}} \right)\cdot\left( \frac{\delta^{3/2}}{\nu_1^{3/4} d(w,\underline p)^{1/2}} + \frac{\delta\nu_1^{1/4}}{d(w,\underline p)^{1/2}} \right) + \\
  +\sum_{d(u,\underline p)<\nu_1-\delta} \left( \frac{\sqrt{d(b,\underline{p}) + d(u,\underline{p})}}{d(b,u)\sqrt[4]{d(u,\underline{p})}} \right)\cdot \frac{\delta^2}{d(u,p_i)^{1/4}d(w,\underline p)^{1/2}}
\end{multline*}
which implies
\begin{equation}
  \label{eq:I0_case1}
  I^{(1)}_0 \lesssim \left( \frac{\delta^{1/2}}{\nu_1\nu_2^{1/2}} + \frac{1}{\nu_2^{1/2}} \right)\sqrt{d(b,\underline p) + \nu_1}\log\left( \frac{d(b,\underline p)+\nu_1}{d(b,\underline p)+\delta} \right) 
  + \frac{\nu_1^{1/2}}{\nu_2^{1/2}}.
\end{equation}
Let us now move on to the case when $w$ is close to a conical singularity.

\emph{Case 2: Bound on $(S_\alpha T)(b,w)$ when $\dist(w,p_i)\leq \nu_2$ for some $i\in \{ 1,\dots, 2g-2 \}$.}
Define
\[
  \Uu_j^{(2)} = \{ u\in \Sigma\ \mid\ d(u,p_j)\leq \nu_3+\delta \},\quad j = 1,\dots, 2g-2
\]
and similarly as in the previous case write
\begin{equation}
  \label{eq:ST_three_sums_case2}
  \sum_{u\in W} \frac{\sqrt{d(b,\underline{p}) + d(u,\underline{p})}}{d(b,u)\sqrt[4]{d(u,\underline{p})}}|T(u,w)|\sqrt[4]{d(w,\underline{p})} \leq I_1^{(2)} + I_w^{(2)} + I_0^{(2)}
\end{equation}
where
\[
  I_1^{(2)} = \sum_{u\notin \bigcup_j\Uu_j^{(2)}}(\cdots),\qquad I_w^{(2)} = \sum_{u\in \Uu_i^{(2)}}(\cdots),\qquad I_0^{(2)} = \sum_{u\in \bigcup_{j\neq i}\Uu_j^{(2)}}(\cdots).
\]
Let us estimate each of $I_1^{(2)},I_0^{(2)}$ and $I_w^{(2)}$ separately.

\emph{Estimate of $I^{(2)}_1$.} Applying~\eqref{eq:bound_on_T6} we get the bound
\[
  I^{(2)}_1 \lesssim \sum_{u\notin \bigcup_j\Uu_j^{(2)}}\left( \frac{\sqrt{d(b,\underline{p}) + d(u,\underline{p})}}{d(b,u)\sqrt[4]{d(u,\underline{p})}} \right)\cdot\left( \frac{\delta^3}{d(u,\underline{p})^2\sqrt{d(u,p_i)}\sqrt[4]{d(u,\underline{p})}} \right).
\]
Again one can show that the main contribution in the sum above comes from the terms corresponding to $\nu_3\leq d(u,\underline p)\leq 2\nu_3$. Simplifying the expression we get the following estimate:
\begin{equation}
  \label{eq:I1_case2}
  I^{(2)}_1\lesssim \frac{\delta}{\nu_3^2}.
\end{equation}

\emph{Estimate of $I^{(2)}_w$.} In this case the bound on $T$ depend on whether $\nu_3-\delta \leq d(u,p_i)\leq \nu_3+\delta$ or $d(u,p_i)<\nu_3-\delta$. In the first case we can apply the estimate~\eqref{eq:bound_on_T7} and in the second case we can apply the estimate~\eqref{eq:bound_on_T8}. This results in
\begin{multline*}
  I^{(2)}_w \lesssim \sum_{\nu_3-\delta \leq d(u,p_i)\leq \nu_3+\delta} \left( \frac{\sqrt{d(b,\underline{p}) + d(u,\underline{p})}}{d(b,u)\sqrt[4]{d(u,\underline{p})}} \right)\cdot\left( \delta\nu_3^{1/4} + \frac{\delta^{1+\beta} d(w,\underline{p})^{1/4}}{\nu_3^{1+\beta}} + \frac{\delta^{3/2}}{\nu_3^{3/4}} \right) + \\
  +\sum_{d(u,p_i)<\nu_3-\delta} \left( \frac{\sqrt{d(b,\underline{p}) + d(u,\underline{p})}}{d(b,u)\sqrt[4]{d(u,\underline{p})}} \right)\cdot\left( \frac{\delta^2\sqrt{d(u,p_i) + d(w,p_i)}}{(d(u,w)+\delta)\sqrt[4]{d(u,p_i)}} \right)
\end{multline*}
which implies
\begin{equation}
  \label{eq:Iw_case2}
  I^{(2)}_w \lesssim \left( 1 + \frac{\delta^{\beta}}{\nu_3^{1+\beta}} + \frac{\delta^{1/2}}{\nu_3} \right)\cdot \sqrt{d(b,\underline p)+\nu_3}\log\left( \frac{d(b,\underline p)+\nu_3}{d(b,\underline p)+\delta} \right) + \nu_3^{1/2}.
\end{equation}

\emph{Estimate of $I^{(2)}_0$.} In this case the bound on $T$ depend on whether $\nu_3-\delta \leq d(u,p_j)\leq \nu_3+\delta$ or $d(u,p_j)<\nu_3-\delta$ for some $j\neq i$. In the first case we can apply the estimate~\eqref{eq:bound_on_T9} and in the second case we can apply the estimate~\eqref{eq:bound_on_T10}. This results in
\begin{multline*}
  I^{(2)}_0 \lesssim \sum_{\nu_3-\delta \leq d(u,p_j)\leq \nu_3+\delta} \left( \frac{\sqrt{d(b,\underline{p}) + d(u,\underline{p})}}{d(b,u)\sqrt[4]{d(u,\underline{p})}} \right)\cdot\left( \delta\nu_3^{1/4} + \frac{\delta^{3/2}}{\nu_3^{3/4}} \right) + \\
  +\sum_{d(u,p_j)<\nu_3-\delta} \left( \frac{\sqrt{d(b,\underline{p}) + d(u,\underline{p})}}{d(b,u)\sqrt[4]{d(u,\underline{p})}} \right)\cdot \frac{\delta^2 }{d(u,p_j)^{1/4}}
\end{multline*}
which implies
\begin{equation}
  \label{eq:I0_case2}
  I^{(2)}_0 \lesssim \left( \delta + \frac{\delta^{3/2}}{\nu_3} \right)|\log\delta| + \nu_3^{1/2}.
\end{equation}

Note that all the constants in the above inequalities depend only on $\lambda$ and the compact $\Kk$ which was chosen as in Lemma~\ref{lemma:bounds_on_T}. We now formulate the main proposition of the current subsection.

\begin{prop}
  \label{prop:Kernel_for_Kalpha}
  Let $\Kk\subset \Mm_g^{t,(0,1)}$ be a compact subset such that for any point $[\Sigma, A,B,\alpha]\in \Kk$ we have $\theta[\alpha](0)\neq 0$. Let $\lambda,R>0$ are fixed. Then there exists $\delta_0,\beta_0>0$ such that for any $(\lambda,\delta)$-adapted graph $G$ on $\Sigma$ with $\delta\leq \delta_0$, and for all $\alpha = \dbar\vphi + \alpha_h$ such that 
  \[
    (\Sigma, A_1,\dots, A_g, B_1,\dots, B_g,\pm\alpha_h+\alpha_G)\in \Kk
  \]
  and $\|\vphi\|_{\mC^2(\Sigma)}\leq R$ the operator $K_\alpha$ has an inverse and we have
  \[
    K_\alpha^{-1}(b,w) = S_\alpha(b,w) + O\left(\frac{\delta^{\beta_0}}{\sqrt[4]{\dist(b,\{ p_1,\dots,p_{2g-2} \})\dist(w,\{ p_1,\dots,p_{2g-2} \})}}\right)
  \]
  where the constants in $O(\ldots)$ depend on $\lambda, R$ and $\Kk$ only.
\end{prop}
\begin{proof}
  Let $L^\infty(W, \sqrt[4]{d(\cdot, \underline{p})})$ denote the space of functions on white vertices with the norm
  \[
    \|f\|_{L^\infty(W, \sqrt[4]{d(\cdot, \underline{p})})} = \max_{w\in W}|f(w)\sqrt[4]{d(w, \underline{p})}|.
  \]
  Choosing $\nu_i = \delta^{\beta_i}$ for some small $0<\beta_3<\beta_2<\beta_1$ we can achieve that the right-hand sides of~\eqref{eq:I1_case1},~\eqref{eq:Iw_case1},~\eqref{eq:I0_case1},~\eqref{eq:I1_case2},~\eqref{eq:Iw_case2} and~\eqref{eq:I0_case2} are all of order $\delta^{\beta_0}$ for some $\beta_0>0$ depending only on $\lambda,R$ and $\Kk$. Plugging this into the right-hand side of~\eqref{eq:ST_three_sums_case1} and~\eqref{eq:ST_three_sums_case2} we conclude that
  \begin{equation}
    \label{eq:KfK1}
    \sum_{u\in W}  \frac{\sqrt{d(b,\underline{p}) + d(u,\underline{p})}}{d(b,u)\sqrt[4]{d(u,\underline{p})}}|T(u,w)|\sqrt[4]{d(w,\underline{p})} = O(\delta^{\beta_0})
  \end{equation}
  uniformly in $b,w$. Combining this with~\eqref{eq:ST_leq_sum} we conclude that for any $b\in B$
  \begin{equation}
    \label{eq:bound_on_ST_common}
    \|\sqrt[4]{d(b,\underline{p})}(S_\alpha T)(b,\cdot)\|_{L^\infty(W, \sqrt[4]{d(\cdot, \underline{p})})} = O(\delta^{\beta_0}).
  \end{equation}
  Moreover, for this choice of $\nu_1,\nu_2,\nu_3$ we have that
  \begin{equation}
    \label{eq:bound_on_norm_of_T}
     \|T\|_{L^\infty(W, \sqrt[4]{d(\cdot, \underline{p})})\to L^\infty(W, \sqrt[4]{d(\cdot, \underline{p})})} = O(\delta^{\beta_0})
  \end{equation}
  for the right action of $T$. Indeed it is enough to notice that
  \[
    \frac{1}{\sqrt[4]{d(u,\underline{p})}} \lesssim \frac{\sqrt{d(b,u) + d(b,\underline{p}) + d(u,\underline{p})}}{d(b,u)\sqrt[4]{d(u,\underline{p})}},
  \]
  then~\eqref{eq:bound_on_norm_of_T} follows from~\eqref{eq:KfK1}. Recall that we have
  \[
    K_\alpha^{-1} = S_\alpha (\Id + T)^{-1}
  \]
  by the definition~\eqref{eq:def_of_T} of $T$. By~\eqref{eq:bound_on_ST_common} and~\eqref{eq:bound_on_norm_of_T} the right-hand side is well-defined and the $K_\alpha^{-1}$ satisfies the desired properties provided $\delta$ is small enough.
\end{proof}

Recall that $G_0 = G\cap \Sigma_0\smm\partial \Sigma_0$. Denote by $K_0$ and $K_{0,\alpha}$ restrictions of $K$ and $K_\alpha$ to $G_0$.

\begin{lemma}
  \label{lemma:Kalphainv_Sigma_with_boundary}
  Assume that all the assumptions of Proposition~\ref{prop:Kernel_for_Kalpha} are satisfied and $\partial \Sigma\neq \varnothing$. Assume that $\alpha$ additionally satisfies $\sigma^*\alpha = -\bar{\alpha}$. Then $K_{0,\alpha}$ is invertible and we have the following formula for the inverse:
  \[
    K_{0,\alpha}^{-1}(b,w) = K_\alpha^{-1}(b,w) + \eta_w^2K_\alpha^{-1}(b,\sigma(w)).
  \]
\end{lemma}
\begin{proof}
  Define $K^{-1}_{0,\alpha}$ by the formula above. First, note that
  \[
    (K_\alpha K^{-1}_{0,\alpha})(u,w) = \delta_w(u) + \eta_{w}^2\delta_{\sigma(w)}(u).
  \]
  Second, note that $K_\alpha(\sigma(b),\sigma(w)) = -(\eta_b\eta_w)^2K_\alpha(b,w)$, therefore
  \[
    K_\alpha^{-1}(\sigma(b),\sigma(w)) = - (\bar{\eta}_b\bar{\eta}_w)^2 K_\alpha^{-1}(b,w).
  \]
  This and the fact that $\eta_b\in \RR$ if $b\in \partial \Sigma_0$ implies that $K^{-1}_{0,\alpha}(b,w) = 0$ when $b\in \partial\Sigma_0$. These two observations imply that $K^{-1}_{0,\alpha}$ is the inverse of $K_{0,\alpha}$.
\end{proof}

\subsection{Near-diagonal expansion of \texorpdfstring{$K_\alpha^{-1}$}{Kalpha-1}}
\label{subsec:near-diag_expansion_Kalpha}

\begin{lemma}
  \label{lemma:near-diag-expansion-of-Kalphainv}
  Let all the assumptions of Proposition~\ref{prop:Kernel_for_Kalpha} be satisfied. Let $bw$ be an arbitrary edge of $G$. The following asymptotic relation holds:
    \begin{multline*}
      K_\alpha(w,b)K_\alpha^{-1}(b,w) = K_\Tt(w,b)K_\Tt^{-1}(b,w) +\\
      + \exp\left[ 2i\int_{p_0}^w\Im\alpha_0 \right]\cdot K_\Tt(w,b) \cdot \frac{1}{2}\left[ r_{\alpha + \alpha_G}(w) + (\eta_b\eta_w)^2\overline{r_{-\alpha + \alpha_G}(w)} \right] + \\
      +O\left( \frac{\delta^{\beta_0+1}(\delta^{1/2}\log\delta^{-1} + \sqrt{\dist(w,\{ p_1,\dots,p_{2g-2} \})})}{\dist(w,\{ p_1,\dots,p_{2g-2} \})^{3/2}} \right).
    \end{multline*}
\end{lemma}

\begin{proof}
  Let $\nu_1,\nu_2,\nu_3$ be as in the construction of the parametrix $S_\alpha$. When $d(w, \underline{p})\geq \nu_2$ the asymptotic relation from the first item follows from Proposition~\ref{prop:Kernel_for_Kalpha} and the definition of $S_\alpha$. Assume now that $d(w,\underline{p})<\nu_2$. In this case we can use Proposition~\ref{prop:Kernel_for_Kalpha}, the definition of $S_\alpha$ and the third item of 
  Lemma~\ref{lemma:kernel_of_G4pi}.
\end{proof}

\section{Determinant of \texorpdfstring{$K_\alpha$}{Kalpha} as an observable for the dimer model}
\label{sec:determinant}

Recall that the perturbed Kasteleyn operator $K_\alpha$ is defined by~\eqref{eq:def_of_Kalpha}. It is well-known~\cite{Cimasoni} that the determinant of a Kasteleyn matrix on a Riemann surface enumerates dimer covers signed by a global sign depending on the monodromy of the height function (cf. Section~\ref{subsec:intro_ratio_of_partition_functions}). By perturbing the Kasteleyn matrix and considering the ratio $\frac{\det K_\alpha}{\det K}$ we obtain the characteristic function of the height function with respect to this sing indefinite measure. In this section we explore this combinatorial relation in details, and estimate how $\det K_\alpha$ varies with respect to $\alpha$. We continue using the setup introducing in Sections~\ref{subsec:intro_flat_metric},~\ref{subsec:intro_graphs_on_Sigma0},~\ref{subsec:intro_dimer_model}. We assume that the height function is introduced as in Section~\ref{subsec:intro_height_function}.

Recall the definition of a flow and the corresponding 1-form with generalized coefficients and its Hodge decomposition discussed in Section~\ref{subsec:intro_height_function}. Let $f$ be a flow and $\M = d\Phi + \Psi$ be the corresponding 1-form.
The following lemma is straightforward:
\begin{lemma}
  \label{lemma:primitive_of_flow}
  Let $v_1,v_2$ be two points lying inside two faces of $G$ and $l$ be a smooth path connecting $v_1$ and $v_2$. Assume that $l$ crosses edges $e_1,\dots, e_k$ of $G$ transversally, orient each $e_j$ such that $l$ crosses this edge from the left to the right locally at the intersection point. Then we have
  \[
    \sum_{j = 1}^k f(e_j) = \Phi(v_2) - \Phi(v_1) + \int_l \Psi.
  \]
\end{lemma}

We need the following auxiliary definition
\begin{defin}
  \label{defin:of_barw}
  For each $w\in W$ pick any of its black neighbors and define $\bar{w}$ to be the point on $\Sigma$ symmetric to $b$ with respect to the edge of $G^\ast$ dual to $bw$. Note that if $w$ is close to a conical singularity, that $\bar{w} = w$ by our assumptions on the graph. Outside the singularities we can locally identify $\Sigma$ with a piece of the Euclidean plane and define the reflection with respect to the edge there. Assumption~\ref{item:intro_reg_part} from Section~\ref{subsec:intro_graphs_on_Sigma0} guarantees that $\bar{w}$ does not depend on the choice of $b$.
\end{defin}

\subsection{Kasteleyn theorem for \texorpdfstring{$G$}{G}}
\label{subsec:Kasteleyn_thm}

There is a number of flows which appear naturally on $G$. First, given any subset $D$ of edges we can set
\begin{equation}
  \label{eq:def_of_f_D}
  f_D(wb) = \begin{cases}
    1,\quad wb\in D,\\
    0,\quad wb\notin D.
  \end{cases}
\end{equation}
Note that if $D$ is a dimer cover, then $\div f_D(b) = -1$ and $\div f_D(w) = -1$ for any $b,w$. Second, we can define the \emph{angle flow} $f^A$ as follows. 
Let $wb$ be an edge of $G$ and $v_1v_2$ be the dual edge of $G^\ast$ oriented such that the black face is on the right. Let $\vartheta$ denote the oriented angle $v_2bv_1$. We define
\begin{equation}
  \label{eq:def_of_fA}
  f^A(wb) = \frac{\vartheta}{2\pi},
\end{equation}
see Figure~\ref{fig:KasTh0}. By the construction, for any $b\in B$ we have $\div f^A(b) = -1$. Moreover, we claim that $\div f^A(w) = 1$ for any $w\in W$. Indeed, pick such a $w$, let $b\sim w$ be its arbitrary neighbor and let $\bar{w}$ be defined by Definition~\ref{defin:of_barw}. Then the angle $\vartheta = 2\pi f^A(wb)$ is equal to the oriented angle $v_1\bar{w}v_2$, which implies the claim.

Let $\M_D$ and $\M^A$ denote 1-forms associated with $f_D$ and $f^A$ respectively, and $\M_D^A$ be the 1-form given by
\begin{equation}
  \label{eq:def_of_MAD}
  \M^A_D = \M_D - \M^A - \frac{2}{\pi} \Im \alpha_G.
\end{equation}
Let
\begin{equation}
  \label{eq:Hodge_for_MDA}
  \M_D^A = d\Phi_D^A + \Psi_D^A
\end{equation}
be the Hodge decomposition. Note that if $D$ is a dimer cover, then $f_D - f^A$ is divergence-free and hence $\Psi_D^A$ is harmonic.

Let us say that a harmonic differential $\Psi$ is \emph{Poincar\'e dual} to a homology class $[C]$ if for any homology class $[C_1]$ we have
\begin{equation}
  \label{eq:def_of_dual_class}
  \int_{C_1} \Psi = C\cdot C_1.
\end{equation}
Given two dimer covers $D_1,D_2$ of $G$ we can draw them simultaneously on $G$ and get a set of double edges and loops $C_1,C_2,\dots, C_k$ on $\Sigma$. Orient each $C_j$ such that dimers from $D_1$ are oriented from black vertices to white. Then the harmonic 1-form $\Psi^A_{D_1} - \Psi^A_{D_2}$ is Poincar\'e dual to $[C_1] + \ldots + [C_k]$.

If $\partial \Sigma_0\neq \varnothing$ we also want to be able to associate the restriction $\Psi\vert_{\Sigma_0}$ with a homology class on $\Sigma_0$. Assume that $\Psi$ is a harmonic differential on $\Sigma$ such that $\sigma^*\Psi = -\Psi$. Then there is a unique class $[C]\in H_1(\Sigma_0, \RR)$ characterized by the property that for any relative homology class $[l]\in H_1(\Sigma_0;\partial \Sigma_0,\ZZ)$ we have
\[
  \int_l \Psi = [C]\cdot [l];
\]
note that $\Psi$ is zero along boundary curves, hence the integral on the left-hand side above is well-defined. We say that $\Psi\vert_{\Sigma_0}$ is Poincar\'e dual to the homology class $[C]$. Recall the notation for the simplicial basis in the first homologies of $\Sigma$ introduced in Section~\ref{sec:The Riemann surface: continuous setting}. The following lemma is straightforward
\begin{lemma}
  \label{lemma:Poincare_dual}
  Assume that $\Psi$ is a harmonic differential on $\Sigma$ having integer cohomologies and such that $\sigma^*\Psi = -\Psi$. Then $\Psi\vert_{\Sigma_0}$ is Poincar\'e dual to an integer homology class if and only if for any $j = 1,\dots, n-1$ we have $\int_{A_j}\Psi$ to be even.
\end{lemma}

Recall that we have defined the quadratic form $q_0$ on $H_1(\Sigma, \ZZ/2\ZZ)$ in Section~\ref{sec:The Riemann surface: continuous setting} by~\eqref{eq:def_of_q0}. For any simple oriented loop $C$ representing an element from $H_1(\Sigma, \ZZ/2\ZZ)$ we have
\begin{equation}
  \label{eq:def_of_q_0_again}
  q_0(C) = (2\pi)^{-1}\wind(C,\omega_0) + C\cdot\gamma + 1\mod 2
\end{equation}
where $\gamma = \gamma_1\cup\dots\cup \gamma_{g-1}$ are the cuts introduces in Section~\ref{subsec:intro_Kasteleyn_operator}. Given a harmonic 1-form $\Psi$ Poincar\'e dual to an integer homology class $[C]$ set $q_0(\Psi) = q_0(C)$. Furthermore, if $\Psi\vert_{\Sigma_0}$ is Poincar\'e dual to an integer homology class $[C]$ on $\Sigma_0$, we set 
\begin{equation}
  \label{eq:q0_of_restricted_differential}
  q_0(\Psi\vert_{\Sigma_0})= q_0(C),
\end{equation}
where $C$ is embedded to $\Sigma$ under the natural embedding $\Sigma_0\hookrightarrow \Sigma$.

Assume that $D_0$ is a dimer cover of $G_0$. If $E$ is a dimer cover of the boundary cycles, then $D = D_0\cup E\cup \sigma(D_0)$ is a dimer cover of $G$.
Note that $\sigma^*\M^A_D = - \M^A_D$ since both flows $f^A$ and $f_D$ are symmetric and $\sigma^*\Im\alpha_G = -\Im\alpha_G$. The following lemma is a specification of the general result~\cite{Cimasoni} on the determinant of a Kasteleyn matrix on a Riemann surface to our setup.

\begin{figure}
  \centering
  \includegraphics[clip, trim=0cm 0cm 0cm 0cm, width = 0.96\textwidth]{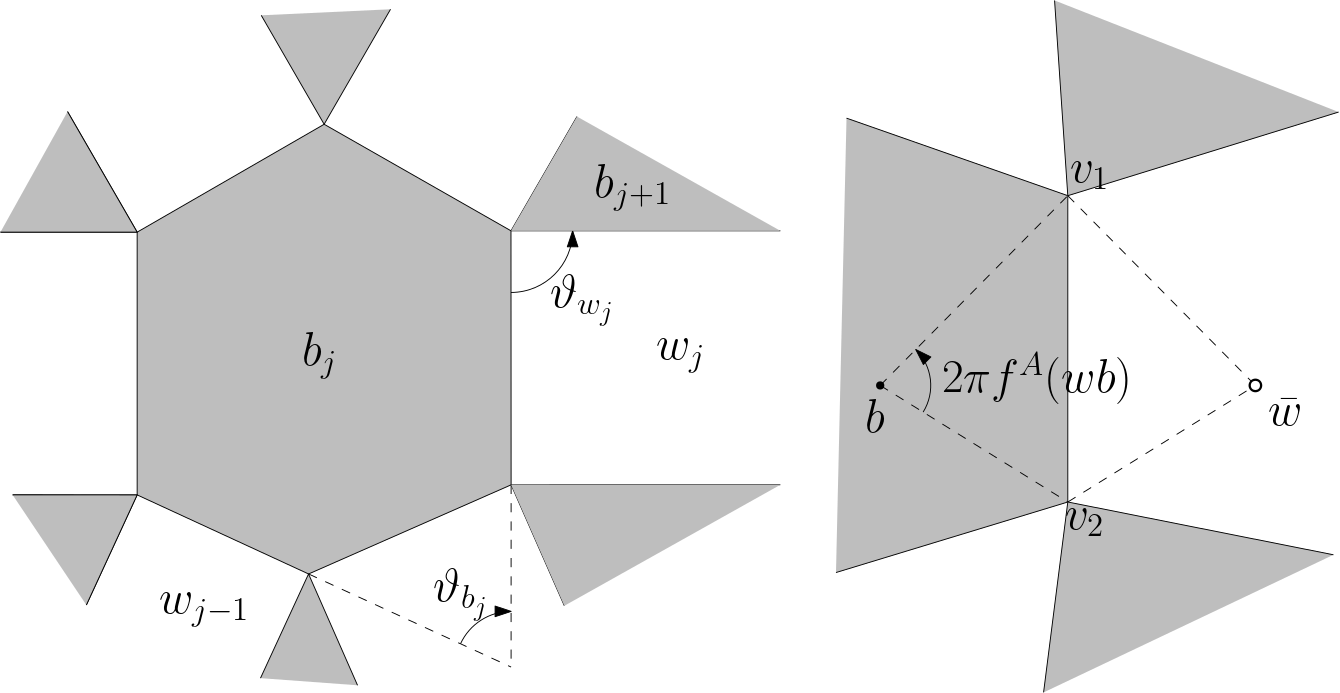}

  \caption{The definition of the angles $\vartheta_b,\vartheta_w$ and the flow $f_\pm^A$.}
  \label{fig:KasTh0}
\end{figure}

\begin{lemma}
  \label{lemma:Kasteleyn_thm}
  For any dimer cover $D$ of $G$ the differential $\Psi^A_D$ has integer cohomologies. Let $\vphi$ be an arbitrary function on edges of $G$. The following formulas hold:
  \begin{enumerate}
    \item Assume that $\partial \Sigma_0 = \varnothing$. Enumerate black and white vertices of $G$ arbitrary. Then there exists a constant $\epsilon\in \TT$ that does not depend on $\vphi$ such that the expansion of the determinant of $K\vphi$ looks as follows:
      \[
        \det (K(w,b)\vphi(wb))_{b\in B,w\in W} = \epsilon \cdot \sum_{D\text{ - dimer cover of }G}\exp[\pi i q_0(\Psi^A_D)] \prod_{wb\in D} |K(w,b)|\vphi(wb)
      \]
    \item Assume that $\partial \Sigma_0 \neq \varnothing$. Enumerate black and white vertices of $G_0$ arbitrary. Then there exist a constant $\epsilon\in \TT$ that does not depend on $\vphi$ and a unique choice of the dimer cover $E$ of boundary cycles such that for any dimer cover $D$ of the form $D_0\cup E\cup \sigma(D_0)$ the 1-form $\Psi^A_D$ is Poincar\'e dual to an integer homology class and the expansion of the determinant of $K\vphi$ restricted to the vertices of $G_0$ looks as follows:
      \[
        \det (K(w,b)\vphi(wb))_{b\in B_0,w\in W_0} = \epsilon \cdot\! \sum_{D_0\text{ - dim. cov. of }G_0}\exp[\pi i q_0(\Psi^A_D\vert_{\Sigma_0})] \prod_{wb\in D_0} |K(w,b)|\vphi(wb),
      \]
      where $D$ denotes a dimer cover of $G$ of the form $D_0\cup E\cup \sigma(D_0)$.
  \end{enumerate}
\end{lemma}
\begin{proof}
  We assume that $\vphi$ is real valued to simplify the notation; the proof for complex valued $\vphi$ can be repeated verbatim. We can also assume that $\vphi$ never vanishes. Let $C = b_1w_1b_2\ldots w_{k-1}b_kw_k$ be an oriented loop in $G$. For each $j = 1,\dots, k$ identify the face $w_j$ of $G^\ast$ with a Euclidean polygon, let $(b_jw_j)^\ast$ be the edge of $G^\ast$ dual to $b_jw_j$ and oriented such that the face $b_j$ is on the left, and let $(w_jb_{j+1})^\ast$ be the edge dual to $w_jb_{j+1}$ and oriented such that $w_j$ is on the left. Let $\vartheta_{w_j}$ be the oriented angle from $(b_jw_j)^\ast$ to $(w_jb_{j+1})^\ast$. The angle $\vartheta_{b_j}$ is defined similarly, see Figure~\ref{fig:KasTh0}. The definition of $K(w,b)$ implies the following relation.
  \begin{multline}
    \label{eq:KasTh1}
    K(C)= -\prod_{j = 1}^k\left( \frac{-K(w_j, b_j)\vphi(w_jb_j)}{K(w_j, b_{j+1})\vphi(w_jb_{j+1})} \right) =\\
    = \exp \left[ -i\Im\int_C(\alpha_0 + 2\alpha_G) + \pi i (C\cdot\gamma+1) - i\sum_{j = 1}^k \vartheta_{w_j} \right]\times\\
    \times\left|\prod_{j = 1}^k\left( \frac{-K(w_j, b_j)\vphi(w_jb_j)}{K(w_j, b_{j+1})\vphi(w_jb_{j+1})} \right)\right|.
  \end{multline}
  Let us now express $q_0(C)$ in terms of angles $\vartheta_b,\vartheta_w$. To this end we note that $C$ is homotopic to a simple loop which traverses through faces $b_1,w_1,b_2,\dots$ of $G^\ast$ and crosses edges $(b_jw_j)^\ast$, $(w_jb_{j+1})^\ast$ perpendicularly. Computing the winding of this curve we obtain
  \begin{equation}
    \label{eq:KasTh2}
    q_0(C) = \frac{1}{2\pi}\sum_{j = 1}^k (\vartheta_{w_j} + \vartheta_{b_j}) + \pi^{-1}\Im \int_C \alpha_0 + C\cdot\gamma + 1\mod 2.
  \end{equation}
  Let $D$ be a dimer cover of $G$. We now want to express $\int_C\Psi^A_D$ via the angles $\vartheta_b,\vartheta_w$. For this purpose we replace the curve $C$ with the curve $\widetilde{C}$ homotopic to it as drawn on Figure~\ref{fig:KasTh}: instead of following the edges of $G$ the curve $\widetilde{C}$ connects midpoints of edges $b_1w_1,w_1b_2,\dots$, goes around $b_1,b_2,\dots$ clockwise and around $w_1,w_2,\dots$ counterclockwise. Given an edge $bw$ of $G$ and the dual edge $v_1v_2$ of $G^\ast$ oriented such that the black face of $G^\ast$ is on the right we define the flows
  \[
    f_+^A(wb) = \frac{1}{2\pi}\angle v_2b\bar{w},\qquad f_-^A(wb) = \frac{1}{2\pi} \angle \bar{w}bv_1
  \]
  where $\angle$ means an oriented angle, see Figure~\ref{fig:KasTh}, and $\bar{w}$ is defined by Definition~\ref{defin:of_barw}. Clearly, we have $f^A(wb) = f_+^A(wb) + f_-^A(wb)$. Using that the segment $b\bar{w}$ is perpendicular to the dual edge $(bw)^\ast$ of $G^\ast$ we also find the following relation. Given $j = 1,\dots, k$ let $w_{j-1}b_j, w^1_jb_j,\dots, w_j^{l_j}b_j, w_jb_j$ be the edges incident to $b_j$ which $\widetilde{C}$ is crossing consequently, see Figure~\ref{fig:KasTh}. Then
  \begin{equation}
    \label{eq:KasTh21}
    f_+^A(w_{j-1}b_j) + \sum_{s = 1}^{l_j} f^A(w_j^sb_j) + f_-^A(w_jb_j) = \frac{1}{2} - \frac{\vartheta_{b_j}}{2\pi}
  \end{equation}
  Similarly, let $w_jb_j, w_jb_j^1,\dots, w_jb_j^{m_j}, w_jb_{j+1}$ be the edges incident to $w_j$ and crossed by $\widetilde{C}$, then
  \begin{equation}
    \label{eq:KasTh22}
    f_+^A(w_jb_j) + \sum_{s = 1}^{m_j} f^A(w_jb_j^s) + f_-^A(w_jb_{j+1}) = \frac{1}{2} + \frac{\vartheta_{w_j}}{2\pi}.
  \end{equation}
  Using these equalities and Lemma~\ref{lemma:primitive_of_flow} we conclude that
  \begin{multline}
    \label{eq:KasTh3}
    \int_C\Psi^A_D = \int_{\widetilde{C}}\Psi^A_D =\\
    = \#[\text{edges of $C$ covered by $D$}] - k + \frac{1}{2\pi}\sum_{j = 1}^k (\vartheta_{b_j} - \vartheta_{w_j}) - \frac{2}{\pi}\Im \int_C\alpha_G.
  \end{multline}
  \begin{figure}
    \centering
    \begin{minipage}{.65\textwidth}
      \includegraphics[clip, trim=0cm 0cm 0cm 0cm, width = 0.96\textwidth]{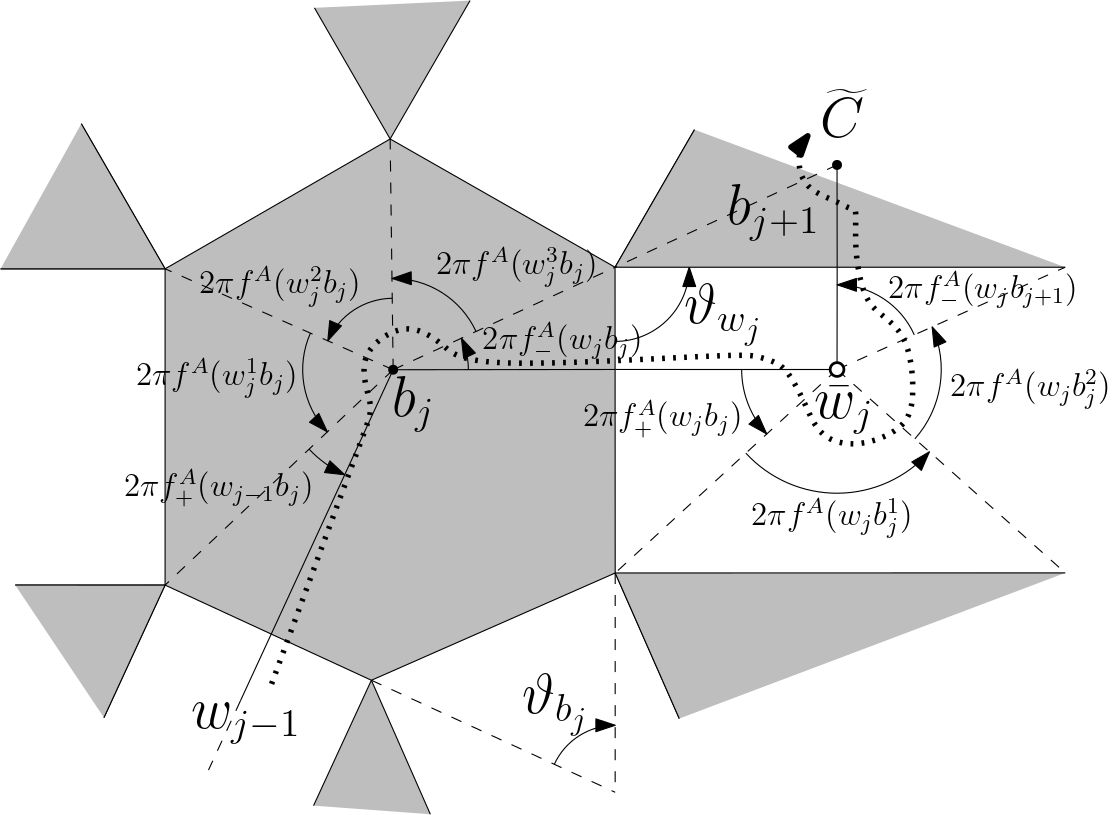}
    \end{minipage}
    \hskip 0.07\textwidth\begin{minipage}{.25\textwidth}
      \includegraphics[clip, trim=0cm 0cm 0cm 0cm, width = 0.96\textwidth]{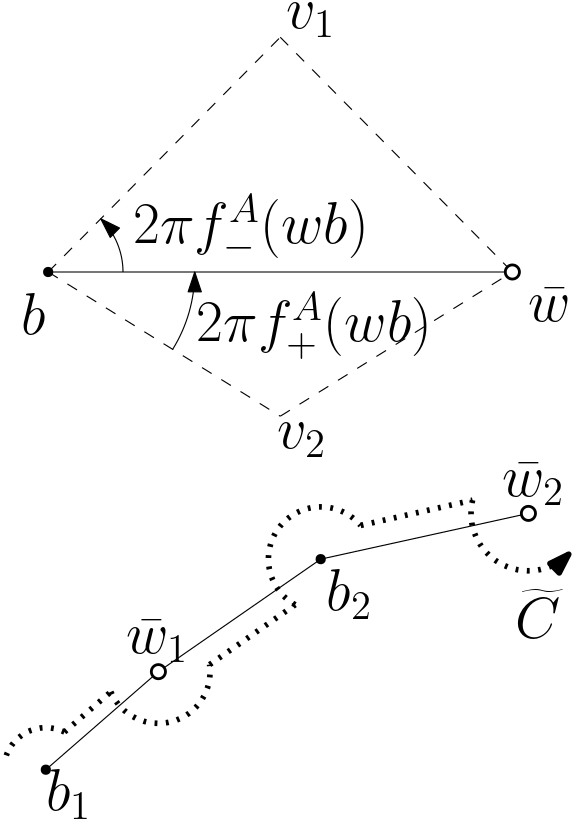}
    \end{minipage}

    \caption{On the left: illustration for the formulas~\eqref{eq:KasTh21} and~\eqref{eq:KasTh22}. On the right: the flow $f_\pm^A$ and the curve $\widetilde{C}$ (the dashed curve on the last picture).}
    \label{fig:KasTh}
  \end{figure}
  Define 
  \[
    \Delta_D(C) = \#[\text{edges of $C$ covered by $D$}] - k.
  \]
  Comparing the exponent on the right-hand side of~\eqref{eq:KasTh1} with the expressions~\eqref{eq:KasTh2},~\eqref{eq:KasTh3} we get
  \begin{equation}
    \label{eq:KasTh3.5}
    K(C) = \exp\left[-\pi iq_0(C) + \pi i \int_C\Psi_D^A - \pi i \Delta_D(C) \right]|K(C)|.
  \end{equation}
  In particular, note that by our choice of weights the operator $K$ is gauge-equivalent to a real-valued operator, hence we have $K(C)\in \RR$ for $K(C)$ defined by~\eqref{eq:KasTh1}. Since $C$ was arbitrary,~\eqref{eq:KasTh3.5} combined with this fact implies that $\Psi^A_D$ has integer cohomologies.

  Assume now that $\partial \Sigma_0 = \varnothing$. Fix an enumeration of black and white vertices and expand
  \[
    \det (K(w,b)\vphi(wb))_{b\in B,w\in W} = \sum_{\tau \text{ - permutation}} (-1)^{\tau}\prod_jK(w_j,b_{\tau(j)})\vphi(w_jb_{\tau(j)}).
  \]
  A given term on the right-hand side is non-zero if and only if $\{ w_jb_{\tau(j)} \}_j$ is a dimer cover. Let $D_1,D_2$ be two dimer covers and $\tau_1,\tau_2$ be the corresponding permutations. Let $C_1,\dots, C_m$ be the loops obtained by superposing of $D_1$ and $D_2$ and orienting the loops in the superposition such that each edge in $D_1$ is oriented from its black to its white end. We have
  \begin{equation}
    \label{eq:KasTh4}
    \frac{(-1)^{\tau_1}\prod_jK(w_j,b_{\tau_1(j)})\vphi(w_jb_{\tau_1(j)})}{(-1)^{\tau_2}\prod_jK(w_j,b_{\tau_2(j)})\vphi(w_jb_{\tau_2(j)})} = \prod_{j = 1}^m K(C_m)
  \end{equation}
  where $K(C)$ is as in~\eqref{eq:KasTh1}. Let us now use~\eqref{eq:KasTh3.5} to evaluate $\prod_{j = 1}^mK(C_m)$. We notice that $\Delta_{D_1}(C_m) = 0$ since exactly half of the edges of $C_m$ is covered by dimers from $D_1$, and we have $q_0(C_1) + \ldots + q_0(C_m) = q_0(C_1 + \ldots + C_m)$ because $C_1,\ldots,C_m$ are disjoint. We conclude that
  \begin{multline}
    \label{eq:KasTh4.5}
    \prod_{j = 1}^m K(C_j) =\\
    = \exp\left[ - \pi i q_0(C_1 + \ldots + C_m) + \pi i (\int_{C_1}\Psi^A_{D_1} + \ldots + \int_{C_m}\Psi^A_{D_1}) \right]\cdot \prod_{j = 1}^m |K(C_j)|
  \end{multline}
  Assume now that $\Psi_{D_i}^A$ is Poincar\'e dual to a class $C_{D_i}$ as defined in~\eqref{eq:def_of_dual_class}, $i = 1,2$. As we observed after~\eqref{eq:def_of_dual_class}, we have $\Psi_{D_1}^A - \Psi_{D_2}^A$ is Poincar\'e dual to $C_1+\ldots + C_m$ which can be expressed as $C_{D_1} - C_{D_2} = C_1+\ldots + C_m$ in $H_1(\Sigma,\ZZ)$. We conclude that
  \begin{multline}
    \label{eq:KasTh4.6}
    q_0(C_1 + \ldots + C_m) = q_0(C_{D_1} - C_{D_2}) =\\
    = q_0(C_{D_1}) - q_0(C_{D_2}) + C_{D_1}\cdot C_{D_2} = q_0(\Psi^A_{D_1}) - q_0(\Psi^A_{D_2}) + (C_{D_2} - C_{D_1})\cdot C_{D_1} =\\
    = q_0(\Psi^A_{D_1}) - q_0(\Psi^A_{D_2}) + (C_1+\ldots + C_m)\cdot C_{D_1} =\\
    = q_0(\Psi^A_{D_1}) - q_0(\Psi^A_{D_2}) + \int_{C_1}\Psi^A_{D_1} + \ldots + \int_{C_m}\Psi^A_{D_1}
  \end{multline}
  where we use the definition of $q_0(\Psi)$ introduced above~\eqref{eq:q0_of_restricted_differential} and~\eqref{eq:def_of_dual_class} and all the equalities are taken modulo 2. Substituting this into~\eqref{eq:KasTh4.5} we obtain
  \begin{equation}
    \label{eq:KasTh4.7}
    \prod_{j = 1}^m K(C_j) = \frac{\exp(\pi i q_0(\Psi^A_{D_1}))}{\exp(\pi iq_0(\Psi^A_{D_2}))}\cdot \prod_{j = 1}^m |K(C_j)|.
  \end{equation}
  Comparing this with~\eqref{eq:KasTh4} we see that there is a unimodal constant $\epsilon$ such that
  \[
    (-1)^{\tau}\prod_jK(w_j,b_{\tau(j)})\vphi(w_jb_{\tau(j)}) = \epsilon\exp(\pi i q_0(\Psi^A_{D}))\cdot \left|\prod_jK(w_j,b_{\tau(j)})\vphi(w_jb_{\tau(j)})\right|.
  \]
  for each permutation $\tau$ which finalizes the proof of the first item of the lemma.

  Let now $\partial \Sigma_0 \neq \varnothing$. Recall that $D = D_0\cup E\cup \sigma(D_0)$ where $E$ is some dimer cover of the boundary cycles. Note that $\sigma^*\Psi^A_D = -\Psi^A_D$. By Lemma~\ref{lemma:Poincare_dual} the differential $\Psi^A_D$ is Poincar\'e dual to an integer cohomology class if and only if $\Psi^A_D$ has even periods along $A_1,\dots, A_{n-1}$. Let us show that there is a unique choice of $E$ for which this condition is satisfied with every $D$.

  Fix $D_0$ and note that changing the cover $E$ along a boundary cycle $B_i$ amounts to adding or subtracting from $\Psi^A_D$ a harmonic differential Poincar\'e dual to $B_i$. Making such adjustments we can achieve that $\int_{A_i}\Psi^A_D$ is even for any $i = 1,\dots n-1$. Replacing $D_0$ with $D_1$ amounts to adding to $\Psi^A_D$ the harmonic differential Poincar\'e dual to $D_0-D_1 + \sigma(D_0-D_1)$ where we denote by $D_0-D_1$ the superposition of $D_0$ and $D_1$. This differential hs even $A_i$-periods for $i = 1,\dots,n-1$. We conclude that the above choice of $E$ guarantees that $\Psi^A_D$ has even $A_1,\dots,A_{n-1}$-periods for each $D_0$, and so $\Psi^A_D\vert_{\Sigma_0}$ is Poincar\'e dual to an integer homology class on $\Sigma_0$.

  Now the formula for $\det (K(w,b)\vphi(wb))_{b\in B_0,w\in W_0}$ follows from the same arguments as in the previous case.
\end{proof}

\subsection{Variation of \texorpdfstring{$\det K_\alpha$}{det Kalpha}}
\label{subsec:variation_of_det}

In this subsection we will be using auxiliary notation defined in Section~\ref{subsec:aux_for_Kalphainv}. We begin with the following technical lemma.

\begin{lemma}
  \label{lemma:KvsD}
  Assume that we are in the setup of Proposition~\ref{prop:Kernel_for_Kalpha} and $\partial \Sigma_0\neq \varnothing$. There exist a constant $\beta_1>0$ a real-valued function $F(b)$ that do not depend on $\alpha$ such that
  \[
    F(b) = O\left( \frac{\delta^{\beta_1}}{(\dist(b, \partial \Sigma_0) + \delta)^{1+\beta_1}} \right)
  \]
  and for any edge $wb$ of $G$ we have
  \begin{multline*}
    K_\alpha^{-1}(b,\sigma(w)) = \frac{1}{2}\left[ \Dd^{-1}_{\alpha+\alpha_G}(b,\sigma(w)) + (\eta_b\bar{\eta}_w)^2\overline{\Dd^{-1}_{-\alpha + \alpha_G}(b,\sigma(w))} \right] + \\
    +F(b) + O\left( \frac{\delta^{\beta_1}}{\dist(b, \{ p_1,\dots, p_{2g-2} \})^{\frac{1}{2}+\beta_1}}  \right)
  \end{multline*}
  where $\beta_1$ and constants in $O(\ldots)$ depend on $\Kk,R$ and $\lambda$ only (see Proposition~\ref{prop:Kernel_for_Kalpha}).
\end{lemma}
\begin{proof}
  By Proposition~\ref{prop:Kernel_for_Kalpha} we have
  \[
    K_\alpha^{-1}(b,\sigma(w)) = S_\alpha(b,\sigma(w)) + O\left( \frac{\delta^\beta}{\dist(b, \{ p_1,\dots, p_{2g-2} \})^{\frac{1}{2}}}  \right)
  \]
  for a small $\beta>0$, thus it is enough to estimate $S_\alpha(b,\sigma(w))$. Recall the construction of $S_\alpha$ given in the beginning of Section~\ref{subsec:Kalphainverse}. By this construction $S_\alpha(b,\sigma(w))$ coincides with $\frac{1}{2}\left[ \Dd^{-1}_{\alpha+\alpha_G}(b,\sigma(w)) + (\eta_b\bar{\eta}_w)^2\overline{\Dd^{-1}_{-\alpha + \alpha_G}(b,\sigma(w))} \right]$ as soon as $\dist(b,\sigma(w)) > \nu_1$ and $\dist(w,\{ p_1,\dots, p_{2g-2} \}) > \nu_2$ (recall that $\eta_{\sigma(w)} = -\bar{\eta}_w$). It remains to compare $S_\alpha(b,\sigma(w))$ with $\frac{1}{2}\left[ \Dd^{-1}_{\alpha+\alpha_G}(b,\sigma(w)) + (\eta_b\bar{\eta}_w)^2\overline{\Dd^{-1}_{-\alpha + \alpha_G}(b,\sigma(w))} \right]$ when one of the aforementioned inequalities break.

  Assume that $\dist(w,\{ p_1,\dots, p_{2g-2} \}) \leq \nu_2$, let $i$ be such that $\dist(\sigma(w),p_i)\leq \nu_2$ and $p_j = \sigma(p_i)$. To evaluate $S_\alpha(b,\sigma(w))$ in this regime we should substitute $b$ and $\sigma(w)$ into the kernel~\eqref{eq:Salpha_case23} with this choice of $i,j$. This kernel is obtained from the bulk kernel by applying the following two operations:
  \begin{itemize}
    \item Replace the kernel $\Dd_\alpha(b,\sigma(w))$ with the new kernel 
      \begin{multline}
        \widetilde{\Dd}^{-1}_{i,\alpha}(b,\sigma(w)) =\\
        = [\Tt(b)^{-1/4}][\Tt(\sigma(w))^{-1/4}]\lim\limits_{p\to p_j}\Dd_\alpha^{-1}(p,\sigma(w))(\Tt(p) - \Tt(p_j))^{1/4}(\Tt(\sigma(w)) - \Tt(p_i))^{1/4}.
      \end{multline}
    \item Multiply by the resulting expression by $\exp\left( -i\Im\int_{p_j}^b(2\alpha_c + 2\alpha_G + \alpha_0) \right)$.
  \end{itemize}
  The error coming from the first replacement can be controlled using Lemma~\ref{lemma:fs}, while the exponential factor on the second step is of order $1 + O(\dist(b,p_j)^{1/2})$. We conclude that the overall error has the required order.

  Let us now assume that $\dist(b,\sigma(w)) \leq \nu_1$. Note that in this case $b$ and $w$ are close to the boundary and thus far from conical singularities. Thus, to evaluate $S_\alpha(b,\sigma(w))$ we should use the kernel given in~\eqref{eq:Salpha_case12}. To compare this kernel to the bulk kernel we can use Theorem~\ref{thmas:parametrix} to replace $K_\Tt^{-1}$ with its continuous counterpart and then Lemma~\ref{lemma:diagonal_expansion_of_Salpha} to compare $\Dd_\alpha^{-1}$ with its near-diagonal asymptotics. A direct check shows that the error term has the required order.

\end{proof}

Proposition~\ref{prop:Kernel_for_Kalpha} permits us to analyze logarithmic variations of $\det K_\alpha$ in a similar fashion as it was done in~\cite{DubedatFamiliesOfCR} by expressing it via the near-diagonal asymptotics of the kernel $K_\alpha^{-1}$. However, an additional technical difficulty appears when $\partial\Sigma_0\neq \varnothing$. In this case we are interested in the variation of $\det K_{0,\alpha}$, and so we must use the kernel $K_{0,\alpha}^{-1}$ of the operator $K_\alpha$ restricted to $\Sigma_0$. This kernel is related with $K_\alpha^{-1}$ via the expression given in Lemma~\ref{lemma:Kalphainv_Sigma_with_boundary}. We see that the additional term in the expression blows up when we substitute $b\sim w$ and $w$ approaches the boundary. This can potentially create an additional term in the logarithmic variation of $\det K_{0,\alpha}$ which is handled by the lemma below.

Assume that $\partial \Sigma_0\neq \varnothing$ and $\alpha$ is a $(0,1)$-form with $\mC^1$ coefficients and such that $
\sigma^*\alpha = -\bar{\alpha}$. Consider the flow
\begin{equation}
  \label{eq:def_of_fS}
  f^S_\alpha(wb) = \eta_w^2K_\alpha(w,b)K^{-1}_\alpha(b,\sigma(w)).
\end{equation}
Using that
\[
  K_\alpha(\sigma(b),\sigma(w)) = -(\eta_b\eta_w)^2K_\alpha(b,w),\qquad K_\alpha^{-1}(\sigma(b),\sigma(w)) = - (\bar{\eta}_b\bar{\eta}_w)^2 K_\alpha^{-1}(b,w)
\]
one can check that
\begin{equation}
  \label{eq:div_of_f^S-alpha}
  \div f^S_\alpha(b) = \begin{cases}
    0,\quad b\notin \partial \Sigma_0,\\
    1,\quad b\in \partial \Sigma_0,
  \end{cases}
  \qquad \div f^S_\alpha(w) = \begin{cases}
    0,\quad w\notin \partial \Sigma_0,\\
    -1,\quad w\in \partial \Sigma_0.
  \end{cases}
\end{equation}
Let $\M^S_\alpha$ denote the 1-form associated with $f^S_\alpha$.

\begin{lemma}
  \label{lemma:reference_flow}
  There exists a flow $f_0^S$ such that the following holds. Let $\Kk\subset \Mm_g^{t,(0,1)}$ be a compact subset such for each $[C,A,B,\alpha]\in \Kk$ we have $\theta[\alpha](0)\neq 0$. Let $R>0$ be fixed. Let $\alpha = \dbar \vphi + \alpha_h$ be such that 
  \[
    (\Sigma, A_1,\dots, A_g,B_1,\dots,B_g,\pm\alpha_h + \alpha_G)\in \Kk
  \]
  and $\|\vphi\|_{\mC^2(\Sigma)}\leq R$, and $\sigma^*\alpha = -\bar{\alpha}$. Let $M^S_0$ be the 1-form associated with $f^S_0$ and 
    \[
      M^S_\alpha - M^S_0 = d\Phi^S_\alpha + \Psi^S_\alpha
    \]
    be the Hodge decomposition. Then $f_0^S$ can be chosen such that
  \begin{enumerate}
    \item $\supp f_0^S\subset \partial \Sigma_0$ and $f_0^S$ is real-valued;
    \item $f^S_\alpha - f^S_0$ is divergence-free;
    \item $\Psi^S_\alpha = o(1)$ as $\delta\to 0$ uniformly in $\alpha$ and $(\lambda, \delta)$-adapted graphs $G$;
    \item $\Phi_\alpha^S$ is bounded uniformly in $\alpha$ and $(\lambda, \delta)$-adapted graphs $G$ provided $\delta$ is small enough;
    \item $\Phi_\alpha^S(p)\to 0$ as $\delta\to 0$ uniformly in $\alpha$ and $(\lambda, \delta)$-adapted graphs $G$, and in $p$ from any compact subset in $\Sigma$ not intersecting $\partial \Sigma_0$.
  \end{enumerate}
\end{lemma}
\begin{proof}
  Let $l$ be an oriented path on $G^\ast$ composed of oriented edges $e_1^\ast,\dots, e_k^\ast$ and let $e_1,\dots, e_k$ be the corresponding edges of $G$ oriented as in Lemma~\ref{lemma:primitive_of_flow}. Assume that $l\subset \Sigma_0$. We want to estimate 
  \begin{equation}
    \label{eq:integral_of_MSalpha}
    \sum_{j = 1}^k f^S_\alpha(e_j).
  \end{equation}
  Recall the multivalued function $\Tt$ associated with the t-embedding of $G^\ast$, see Section~\ref{subsubsec:intro_graph_assumptions}. Given an edge $wb$ of $G$ and the dual edge $(wb)^\ast$ of $G^\ast$ define the 1-form $\o$ on the tangent space to $(wb)^\ast$ by
  \begin{equation}
    \label{eq:def_of_Oo_surface}
    \o = \eta_w^2 \exp\left[ i\Im\left(\int_{p_0}^w \alpha_0 + \int_{p_0}^b\alpha_0+ 2\int_w^b\alpha_G\right) \right]d\Tt\vert_{(wb)^*}.
  \end{equation}
  Note that, by the definition of $K$~\eqref{eq:def_of_K}, we have $\eta_w^2K(w,b) = \pm\int_{(wb)^\ast}\o$ where the sign is $-$ if $wb$ crosses one of the cuts $\gamma_1,\dots, \gamma_{g-1}$. Note that $K_\alpha(w,b) = K(w,b)(1+O(\delta))$. It follows that we can write
  \begin{equation}
    \label{eq:rf1}
    \sum_{j = 1}^k f^S_\alpha(e_j) = \sum_{j = 1}^k\pm K_\alpha^{-1}(b_j, \sigma(w_j))\int_{e_j^\ast}\o\cdot (1 + O(\delta)),
  \end{equation}
  where $e_j = w_jb_j$. We now use Lemma~\ref{lemma:KvsD} to replace $K_\alpha^{-1}(b,\sigma(w))$ with the corresponding expression involving $\Dd_{\pm\alpha+\alpha_G}^{-1}$, and then replace summation with integration. Assuming that $\dist(l, \partial \Sigma_0)\geq \delta$ we obtain
  \begin{multline}
    \label{eq:rf2}
    \sum_{j = 1}^k f^S_\alpha(e_j) = \frac{1}{2}\int_l \left( \Dd^{-1}_{\alpha + \alpha_G}(p,\sigma(p))\,\o(p) + (\eta_b\bar{\eta}_w)^2\overline{ \Dd^{-1}_{-\alpha + \alpha_G}(p,\sigma(p))} \,\o(p) \right) + \\
    + \int_lF_1(b)\, |\omega_0| + \int_l O\left(  \frac{\delta^{\beta_1}}{(\dist(b, \partial \Sigma_0) + \delta)^{1+\beta_1}} \right) |\omega_0|
  \end{multline}
  where the $\pm$ sign is compensated by choosing an appropriate branch of $\Dd^{-1}$ and
  \[
    F_1(b) = O\left( \frac{\delta^{\beta_1}}{\dist(b, \{ p_1,\dots, p_{2g-2} \})^{\frac{1}{2}+\beta_1}}\right),
  \]
  $F_1$ is real valued and does not depend on $\alpha$. Indeed, the error term $F_1$ comes from Lemma~\ref{lemma:KvsD} (this part is controlled by the function $F$ in the lemma), and from replacing summation with integration (in this case the independence on $\alpha$ can be derived by examining $\Dd^{-1}$ locally at the boundary via Lemma~\ref{lemma:diagonal_expansion_of_Salpha}).

  Let us estimate the first integral in the right-hand side of~\eqref{eq:rf2}. For this purpose let us recall that we have a multivalied origami map $\Oo$ associated with the t-embedding of $G^\ast$ in $\Sigma$, see Section~\ref{subsubsec:intro_graph_assumptions}. By Assumption~\ref{item:intro_reg_part} on $G$ (see Section~\ref{subsubsec:intro_graph_assumptions}), there exists a branch of $\Oo$ such that $\sup_{z\in \Sigma}|\Oo(z)|\leq C\delta$ where $C$ is a constant depending on the constant $\lambda$ from the assumptions. Let us put
  \[
    \eta_w^{(1)} = \exp\left( i\Im(\int_{p_0}^w\alpha_0 - 2\int_{p_0}^w\alpha_G) \right)\eta_w,\qquad \eta_b^{(1)} = \exp\left( i\Im(\int_{p_0}^b\alpha_0 + 2\int_{p_0}^b\alpha_G) \right)\eta_b.
  \]
  The relation $(\bar{\eta}_b\bar{\eta}_w)^2 = \frac{K(w,b)^2}{|K(w,b)|^2}$ implies $(\bar{\eta}^{(1)}_b\bar{\eta}^{(1)}_w)^2 = \frac{K_\Tt(w,b)^2}{|K_\Tt(w,b)|^2}$, where $K_\Tt(w,b) = \int_{v_1}^{v_2}d\Tt$. It follows that, after replacing $\Oo$ with $e^{it}\Oo$ for a suitable $t\in \RR$, we can write
  \begin{equation}
    \label{eq:rf3}
    d\Oo(z) = \begin{cases}
      (\eta_w^{(1)})^2\,d\Tt,\qquad z\in w,\\
      (\bar{\eta}^{(1)}_b)^2\,d\bar{\Tt},\qquad z\in b
    \end{cases}
  \end{equation} 
  (cf.~\eqref{eq:def_of_origami}), where $z\in w$ or $z\in b$ means that $z$ belongs to the corresponding face of the t-embedding. This allows us to deduce that for any edge $wb$ of $G$
  \begin{equation}
    \label{eq:rf4}
    \begin{split}
      &\o\vert_{(wb)^\ast} = \exp\left[ i\Im \int_w^b\alpha_0\right] \cdot \exp\left[2i\Im\left(\int_{p_0}^w\alpha_G + \int_{p_0}^b\alpha_G \right) \right]\,d\Oo\vert_{(wb)^\ast},\\
      &(\eta_b\bar{\eta}_w)^2\o\vert_{(wb)^\ast} = \exp\left[- i\Im \int_w^b\alpha_0\right] \cdot \exp\left[-2i\Im\left( \int_{p_0}^w\alpha_G + \int_{p_0}^b\alpha_G \right) \right]\,d\bar{\Oo}\vert_{(wb)^\ast}.
    \end{split}
  \end{equation}
  Replacing $\exp\left[ i\Im \int_w^b\alpha_0\right]$ with $1 + O(\delta)$, substituting these expressions into~\eqref{eq:rf2} and integrating by parts we get
  \begin{multline}
    \label{eq:rf5}
    \frac{1}{2}\int_l \left( \Dd^{-1}_{\alpha + \alpha_G}(p,\sigma(p))\,\o(p) + (\eta_b\bar{\eta}_w)^2\overline{ \Dd^{-1}_{-\alpha + \alpha_G}(p,\sigma(p))} \,\o(p) \right) =\\
    = \int_lF_2(b)\, |\omega_0| + \int_l O\left(  \frac{\delta^{\beta_1}}{(\dist(b, \partial \Sigma_0) + \delta)^{1+\beta_1}} \right) |\omega_0|
  \end{multline}
  where 
  \[
    F_2(b) = O\left( \frac{\delta^{\beta_1}}{\dist(b, \{ p_1,\dots, p_{2g-2} \})^{\frac{1}{2}+\beta_1}}\right)
  \]
  $F_2$ is real valued does not depend on $\alpha$. We conclude that
  \begin{equation}
    \label{eq:rf6}
    \sum_{j = 1}^k f^S_\alpha(e_j) = \int_l(F_1(b) + F_2(b))\, |\omega_0| + \int_l O\left(  \frac{\delta^{\beta_1}}{(\dist(b, \partial \Sigma_0) + \delta)^{1+\beta_1}} \right) |\omega_0|.
  \end{equation}
  Note that this estimate remains valid even when $l$ crosses the boundary of $\Sigma_0$ since $|f_\alpha^S(e)|\leq C$ for some $C$ depending on $\Kk, R$ and $\lambda$ only due to Proposition~\ref{prop:Kernel_for_Kalpha}.

  Let $E$ be the dimer cover of boundary cycles chosen as in the second item of Lemma~\ref{lemma:Kasteleyn_thm}. Recall that the flow $f_E$ is defined by
  \[
    f_E(wb) = \begin{cases}
      1,\quad wb\in E,\\
      0,\quad wb\notin E.
    \end{cases}
  \]
  The flow $f^S_\alpha + f_E$ is divergence-free. Put $\tilde{f}_0^S = -f_E$ and $\tilde{M}^S_0$ to be the corresponding 1-form, let 
  \[
    \M^S_\alpha - \tilde{\M}^S_0 = d\tilde{\Phi}^S_\alpha + \tilde{\Psi}^S_\alpha.
  \]
  be its Hodge decomposition. Note that $\sigma^*(\M^S_\alpha - \tilde{\M}^S_0) = \M^S_0 - \tilde{\M}^S_\alpha$, hence we can choose $\tilde{\Phi}^S_\alpha$ such that $\sigma^*\tilde{\Phi}^S_\alpha = -\tilde{\Phi}^S_\alpha$. From the estimate~\eqref{eq:rf6} and Lemma~\ref{lemma:primitive_of_flow} we obtain the following estimates as $\delta\to 0$:
  \begin{equation}
    \label{eq:estimates_of_tildePsiPhi}
    \begin{split}
      &\int_{B_j} \tilde{\Psi}^S_\alpha = o(1),\quad j = 1,\dots, g,\\
      &\int_{A_j}\tilde{\Psi}^S_\alpha = o(1),\quad j = n, n+1,\dots, g,\\
      &\exists C_j\in \RR\ \colon\ \int_{A_j}\tilde{\Psi}^S_\alpha = C_j + o(1),\quad j = 1,\dots, n-1,\\
      &\exists C\in \RR\ \colon\ \forall p\in \Sigma_0,\ \tilde{\Phi}^S_\alpha(p) = C + O\left( \frac{\delta}{(\dist(p,\partial\Sigma_0)+\delta)} \right)^{\beta_1}
    \end{split}
  \end{equation}
  and $C,C_1,\dots, C_{n-1}$ are uniformly bounded and do not depend on $\alpha$. Recall that loops $B_0,B_1,\dots, B_{n-1}$ are boundary components of $\Sigma_0$ oriented according to its orientation. For each $j$ let $f_j$ be the flow defined by
  \[
    f_j(wb) = \begin{cases}
      1,\quad w,b\in B_j,\ wb\text{ coincides with the orientation of }B_j,\\
      -1,\quad w,b\in B_j,\ wb\text{ contradicts with the orientation of }B_j,\\
      0,\quad \text{else}.
    \end{cases}
  \]
  Define
  \[
    f_0^S = -f_E - Cf_0 - \sum_{j = 1}^{n-1}(C_j+C)f_j. 
  \]
  Estimates~\eqref{eq:estimates_of_tildePsiPhi} imply that $f_0^S$ satisfies all the necessary properties.
\end{proof}

Recall that $K_\Tt(b,w)$ and $K_\Tt^{-1})(b,w)$ denote the Kasteleyn operator and its full-plane inverse respectively associated with the t-embedding obtained by identifying a neighborhood of $w$ in $G^\ast$ with a full-plane t-embedding, see Section~\ref{subsec:aux_for_Kalphainv}. Given a $(0,1)$-form $\alpha$ we define
\begin{equation}
  \label{eq:def_of_Palpha}
  P(\alpha) = 2i \sum_{b\sim w,\ b,w\in G} K_\Tt(w,b) K_\Tt^{-1}(b,w)\int_w^b \Im\alpha 
\end{equation}
If $\partial \Sigma_0\neq \varnothing$, then we also define
\begin{equation}
  \label{eq:def_of_P0alpha}
  P_0(\alpha) = i\sum_{b\sim w,\ b,w\in G} (K_\Tt(w,b) K_\Tt^{-1}(b,w) + f_0^S(wb))\int_w^b \Im\alpha.
\end{equation}
Recall that with any anti-holomorphic $(0,1)$-form $\alpha$ we associate vectors $a(\alpha),b(\alpha)\in \RR^g$ defined by
\[
  a(\alpha)_j = \pi^{-1}\int_{A_j} \Im \alpha,\qquad b(\alpha)_j = \pi^{-1}\int_{B_j} \Im \alpha.
\]
Recall the definition of $\theta[\alpha_h](z)$ given in~\eqref{eq:def_of_theta_alpha}.
Let $\Delta$ be the Laplace operator associated with the metric $ds^2$ defined on $\mC^2$ functions compactly supported in the interior of $\Sigma_0$, i.e.
\[
  -4\dbar\partial f = \Delta f \bar{\omega}_0\wedge \omega_0.
\]

\begin{prop}
  \label{prop:variations_of_detKalpha}
  Let $\Kk\subset \Mm_g^{t,(0,1)}$ be a compact such that for any $[\Sigma, A,B,\alpha]\in \Kk$ we have $\theta[\alpha](0)\neq 0$. Let $\lambda,R>0$ be fixed, and $\delta_0$ be as in Proposition~\ref{prop:Kernel_for_Kalpha}. Assume that the graph $G$ on $\Sigma$ is $(\lambda,\delta)$-adapted, $\delta\leq \delta_0$ and $\alpha_t = \dbar\vphi_t + \alpha_{h,t}$ is a smooth family of $(0,1)$-forms on $\Sigma$ such that for all $t$ 
  \[
    (\Sigma, A_1,\dots, A_g,B_1,\dots,B_g,\pm\alpha_{h,t} + \alpha_G)\in \Kk
  \]
  $\|\vphi_t\|_{\mC^2(\Sigma)}\leq R$, $\|\frac{d}{dt}\vphi_t\|_{\mC^2(\Sigma)}\leq R$ and $\int_\Sigma \alpha_{h,t}\wedge \ast \alpha_{h,t}\leq R$. 
  Then we have the following: 
    \begin{enumerate}
      \item If $\partial \Sigma_0 = \varnothing$, then 
        \begin{multline*}
          \frac{d}{dt}\log\left[ e^{-P(\alpha_t)}\det K_{\alpha_t}\right] = \\
          = \frac{d}{dt}\left[ \log \left( \theta[\alpha_{h,t}+\alpha_G](0) \cdot \overline{\theta[-\alpha_{h,t} + \alpha_G](0)}\right) + 4\pi ia(\alpha_{h,t})\cdot b(\alpha_G)\right] -\\
          - \frac{d}{dt}\,\frac{1}{2\pi}\int_{\Sigma_0}\Re \vphi_t \Delta \Re\vphi_t ds^2  + o(1).         
        \end{multline*}
      \item Assume that $\partial \Sigma_0 \neq \varnothing$ and $\alpha_t$ satisfies $\sigma^*\alpha_t = -\bar{\alpha}_t$. Recall that $K_{0,\alpha_t}$ denotes the restriction of $K_{\alpha_t}$ to the vertices of $G_0$. We have
        \begin{multline*}
          \frac{d}{dt}\log\left[e^{-P_0(\alpha_t)}\det K_{0,\alpha_t}\right] = \\
          = \frac{d}{dt}\left[ \log \theta[\alpha_{h,t} + \alpha_G](0) + 2\pi i a(\alpha_{h,t})\cdot b(\alpha_G)\right] - \\ 
          - \frac{d}{dt}\,\frac{1}{2\pi}\int_{\Sigma_0}\Re \vphi_t \Delta \Re\vphi_t ds^2  + o(1).
        \end{multline*}
    \end{enumerate}
    Here $o(1)$ is in $\delta\to 0$, and the estimate depends only on $\lambda,\Kk$ and $R$.
\end{prop}
\begin{proof}
  Let us prove the second item; the first one follows if one substitutes $\Sigma = \Sigma_0\sqcup \Sigma_0^\op$. Recall the formula for $K_{0,\alpha_t}^{-1}$ given in Lemma~\ref{lemma:Kalphainv_Sigma_with_boundary}. Let us expand
  \begin{multline}
    \label{eq:vofd1}
    \frac{d}{dt}\log\det K_{0,\alpha_t} = \sum_{b\sim w,\ b,w\in G_0}\frac{d}{dt}(K_{\alpha_t}(w,b)) K_{0,\alpha_t}^{-1}(b,w) =\\
    =\sum_{b\sim w,\ b,w\in G} i\int_w^b \Im \dot{\alpha}_t \cdot K_{\alpha_t}(w,b) \left( K_{\alpha_t}^{-1}(b,w) + \eta_w^2 K_{\alpha_t}^{-1}(b,\sigma(w)) \right)
  \end{multline}
  where $\dot{\alpha}_t$ denotes the derivative of $\alpha_t$ with respect to $t$ (not that $K_{\alpha_t}^{-1}(b,w) + \eta_w^2 K_{\alpha_t}^{-1}(b,\sigma(w))$ vanishes when $b$ or $w$ belongs to $\partial \Sigma_0$). We now apply Lemma~\ref{lemma:near-diag-expansion-of-Kalphainv} to replace the expression $K_{\alpha_t}(w,b)K^{-1}_{\alpha_t}(b,w)$ that appears in the sum above with its asymptotic expansion. Using the definition of $\M_{\alpha_t}^S$ and $\M_0^S$ from Lemma~\ref{lemma:reference_flow} we obtain
  \begin{multline}
    \label{eq:vofd2}
    \frac{d}{dt}\log\det K_{0,\alpha_t} = \frac{d}{dt}P_0(\alpha_t) + \\
    + i\sum_{b\sim w,\ b,w\in G}\int_w^b\Im \dot{\alpha}_t\cdot \exp\left[ 2i\int_{p_0}^w\Im\alpha_0 \right]\cdot K_\Tt(w,b) \times\\
    \times\frac{1}{2}\left[ r_{\alpha_t + \alpha_G}(w) + (\eta_b\eta_w)^2\overline{r_{-\alpha_t + \alpha_G}(w)} \right] + \\
    + i\int_\Sigma \Im\dot{\alpha}_t \wedge (\M_{\alpha_t}^S - \M_0^S) + o(1)
  \end{multline}
  as $\delta\to 0$.
  Applying Lemma~\ref{lemma:reference_flow} we get
  \begin{equation}
    \label{eq:vofd3}
    \int_\Sigma \Im\dot{\alpha}_t \wedge (\M_{\alpha_t}^S - \M_0^S) = o(1),\qquad \delta\to 0.
  \end{equation}
  Let us analyze the second summand in~\eqref{eq:vofd2}. Using Assumption~\ref{item:intro_dbar_adapted} on $G$ from Section~\ref{subsec:intro_graphs_on_Sigma0} and Lemma~\ref{lemma:derivative_of_theta} we expand
  \begin{multline}
    \label{eq:vofd4}
     i\sum_{b\sim w,\ b,w\in G}\int_w^b\Im \dot{\alpha}_t\cdot \exp\left[ 2i\int_{p_0}^w\Im\alpha_0 \right]\cdot K_\Tt(w,b) \times\\
     \times \frac{1}{2}\left[ r_{\alpha_t + \alpha_G}(w) + (\eta_b\eta_w)^2\overline{r_{-\alpha_t + \alpha_G}(w)} \right] = \\
     = i\sum_{w\in G}\mu_w \left(\frac{\dot{\alpha}_t}{\omega_0}r_{\alpha_t + \alpha_G}(w) + \overline{\frac{\dot{\alpha}_t}{\omega_0} r_{-\alpha_t + \alpha_G}(w)} \right) + o(1) =\\
     = -\frac{1}{4}\int_\Sigma \left( r_{\alpha_t + \alpha_G}\omega_0\wedge \dot{\alpha}_t - \overline{r_{-\alpha_t + \alpha_G}\omega_0\wedge \dot{\alpha}_t} \right) + o(1) =\\
     = \frac{d}{dt}\left(\log \theta[\alpha_{h,t} + \alpha_G](0) + 2\pi i a(\alpha_{h,t})\cdot b(\alpha_G) - \frac{1}{2\pi}\int_{\Sigma_0}\Re \vphi_t \Delta \Re\vphi_t ds^2\right) + o(1).
  \end{multline}
  Substituting~\eqref{eq:vofd3} and~\eqref{eq:vofd4} into~\eqref{eq:vofd2} we get the result.
\end{proof}

\subsection{Relation between \texorpdfstring{$\det K_\alpha$}{det Kalpha} and the height function}
\label{subsec:det_Kalpha}

We finalize this section by interpreting the determinant of $K_\alpha$ in probabilistic terms. Let us introduce yet another flow $f^K$ on $G$. We put
\begin{equation}
  \label{eq:def_of_fK_nobdy}
  f^K(wb) = K_\Tt(w,b)K_\Tt^{-1}(b,w),\qquad \text{if }\partial \Sigma_0 = \varnothing.
\end{equation}
If $\partial \Sigma_0 \neq \varnothing$, then this definition needs some modifications. Recall the dimer cover $E$ of the boundary cycles defined in the second item of Lemma~\ref{lemma:Kasteleyn_thm}. Recall also the flow $f_0^S$ is defined in Lemma~\ref{lemma:reference_flow}. Recall the notation $K_\Tt,K_\Tt^{-1}$ from Section~\ref{subsec:aux_for_Kalphainv}. We define
\begin{equation}
  \label{eq:def_of_fK_bdypresent}
  f^K(wb) = K_\Tt(w,b)K_\Tt^{-1}(b,w) + f^S_0(wb) + f_E(wb),\qquad \text{if }\partial \Sigma_0 \neq \varnothing.
\end{equation}
Given a dimer cover $D$ of $G$ define $f^K_D = f_D - f^K$ and let
\begin{equation}
  \label{eq:def_of_mKD}
  \M^K_D = \M_D - \M^K,\qquad \M^K_D = d\Phi^K_D + \Psi^K_D
\end{equation}
be the associated 1-form. Note that $\Psi^K_D$ is a priori meromorphic since $f_D - f^K$ is not necessary divergence-free, because $K_\Tt^{-1}(b,w)$ may depend on the point $w$.

\begin{rem}
  \label{rem:why_f0+fE}
  When $\partial\Sigma_0\neq 0$, the role of the corrections $f_0^S$ and $f_E$ in the definition of $f^K$ is to make $M_D^K$ restricted to $\Sigma_0$ to be the 1-form corresponding to the flow
  \begin{equation}
    \label{eq:why_f0+fE}
    f_{D_0}(wb) - K_0(w,b)\cdot \big[\text{main term in the near-diagonal asymptotics of }K_0^{-1}(b,w)\big].
  \end{equation}
  (where $K_0$ is $K$ restricted to $\Sigma_0$). As we already observed before Lemma~\ref{lemma:reference_flow}, the main term in the near-diagonal asymptotics of $K_0^{-1}(b,w)$ differs from the one of $K^{-1}(wb)$ when $w$ is close to the boundary as we have $K_0^{-1}(b,w) = K^{-1}(b,w) + \eta_w^2 K^{-1}(b,\sigma(w))$ (see Lemma~\ref{lemma:Kalphainv_Sigma_with_boundary}) and the second term also blows up. This introduces a correction, which, as shown in Lemma~\ref{lemma:reference_flow}, is approximately (meaning that the terms in the Hodge decomposition of the difference are small) equal to $f_0^S$. Therefore, the expression~\eqref{eq:why_f0+fE} can be replaced with
  \[
    f_{D_0}(wb) - K_\Tt(wb)K^{-1}_\Tt(wb) - f_0^S(wb).
  \]
  Extending this to $\Sigma$ by symmetry and using our usual notation $D = D_0\cup E\cup \sigma(D_0)$ we get
  \[
    f_{D_0\cup\sigma(D_0)}(wb) - K_\Tt(wb)K^{-1}_\Tt(wb) - f_0^S(wb) = f_D(wb) - f^K(wb)
  \]
  which explains the presence of $f_E$.
\end{rem}

\begin{lemma}
  \label{lemma:detKalpha}
  Let $\alpha = \dbar \vphi + \alpha_h$ be an arbitrary $(0,1)$-form and $P(\alpha)$ and $P_0(\alpha)$ be defined by~\eqref{eq:def_of_Palpha} and~\eqref{eq:def_of_P0alpha} respectively. 
  \begin{enumerate}
    \item Assume that $\partial \Sigma_0 = \varnothing$. Let $\Zz_G$ denote the dimer partition function on $G$ and $\epsilon$ be the constant from item 1 of Lemma~\ref{lemma:Kasteleyn_thm}. We have
      \[
        e^{-P(\alpha)}\det K_\alpha = \epsilon \Zz_G \cdot \EE \exp\left[ \pi i q_0(\Psi^A_D) + 2i\int_{\Sigma_0} \Im\alpha\wedge \M^K_D\right].
      \]
    \item Assume that $\partial \Sigma_0\neq\varnothing$ and $\sigma^*\alpha = -\bar{\alpha}$. Let $\Zz_{G_0}$ denote the dimer partition function on $G_0$. Them, in the notation of item 2 of Lemma~\ref{lemma:Kasteleyn_thm}, we have
      \[
        e^{-P_0(\alpha)}\det K_{0,\alpha} = \epsilon \Zz_{G_0} \cdot \EE \exp\left[ \pi i q_0(\Psi^A_D\vert_{\Sigma_0}) + i\int_\Sigma \Im\alpha\wedge \M^K_D\right],
      \]
      where $D$ denotes a dimer cover of $G$ of the form $D_0\cup E\cup \sigma(D_0)$.
  \end{enumerate}
\end{lemma}
\begin{proof}
  Let us prove the second item, the first one can be proven similarly. Fix a dimer cover $D_0$ of $G_0$ and let $D = D_0\cup E\cup \sigma(D_0)$ as usual. Using the definition of $\M_D^K$ and the fact that $\sigma^\ast\Im\alpha = \Im\alpha$ we can write
  \begin{multline}
    \label{eq:detKalpha0}
    i\int_\Sigma\Im\alpha\wedge \M_D^K = i\sum_{wb\in D}\int_w^b\Im\alpha - i\sum_{w\sim b} f^K(wb)\int_w^b\Im\alpha =\\
    = 2i\sum_{wb\in D_0}\int_w^b\Im\alpha - i \sum_{w\sim b}(f^K(wb) - f_E(wb))\int_w^b\Im\alpha.
  \end{multline}
  Comparing the definition~\eqref{eq:def_of_P0alpha} of $P_0$ with the definition~\eqref{eq:def_of_fK_bdypresent} of $f^K$ we conclude that
  \[
    i \sum_{w\sim b}(f^K(wb) - f_E(wb))\int_w^b\Im\alpha = P_0(\alpha),
  \]
  therefore~\eqref{eq:detKalpha0} implies
  \begin{equation}
    \label{eq:detKalpha1}
    \exp\left[i\int_\Sigma \Im\alpha\wedge \M_D^K\right] = e^{-P_0(\alpha)}\prod_{wb\in D_0}\exp\left[ 2i\int_w^b\Im \alpha \right].
  \end{equation}
  The formula for the determinant now follows from~\eqref{eq:detKalpha1} and the second item of Lemma~\ref{lemma:Kasteleyn_thm} if one substitutes $\vphi(wb) = \exp\left[ 2i\int_w^b\Im \alpha \right]$.
\end{proof}

We conclude this section with the following short discussion about the roles of the reference flows introduced above.

\begin{rem}
  \label{rem:flows}
  Recall that our main goal is to prove that $\M^{\fluct, k}$ converges to $\m^{2\alpha_1} - \EE\m^{2\alpha_1}$, where $\alpha_1$ is the limit of $\alpha_{G_k}^k$ (cf. Theorem~\ref{thma:main1}). Here $\M^{\fluct, k} = dh^{\fluct,k}$, where $h^{\fluct,k}$ is the centered height function, that is, the height function defined with respect to the reference flow given by edge densities $\PP[wb\in \text{dimer cover}]$. A natural question is whether there is an intrinsic definition of a reference flow $f$ such that the corresponding height differentials $dh^k$ would converge to $\m^{2\alpha_1}$. We can approach this in two ways.

  First one is ``topological'': we can try to define the reference flow such that $h^k - 2\pi^{-1}\int\Im\alpha_{G_k}$ has integer monodromies to mimic this property of $\int\m^{2\alpha_1}$. This of course does not specify the reference flow uniquely; however, if we seek for one admitting a ``local'' expression, then, perhaps, the angle flow $f^A$ would be the most natural choice (recall that $\Psi_D^A$ corresponds to $f_D(wb) - f^A(wb) - 2\pi^{-1}\int_w^b\Im\alpha_G$, see~\eqref{eq:def_of_MAD},~\eqref{eq:Hodge_for_MDA}, and the fact that $\Psi_D^A$ has integer monodromies proved in Lemma~\ref{lemma:Kasteleyn_thm}).

  The second way is ``analytic'' and uses the observable $\det K_\alpha$. Examining the proof of Lemma~\ref{lemma:detKalpha} one can conclude that in order to replace $\M_D^K$ with $dh$ defined with respect to some flow $f$ different from $f^K$ it is enough to redefine $P_0$ using the flow $f(wb)$ instead of $(K_\Tt(w,b) K_\Tt^{-1}(b,w) + f_0^S(wb))$. What makes the flow $f^K$ special is that this flow makes $e^{-P_0(\alpha)}\det K_\alpha$ to have bounded variations as we see in Proposition~\ref{prop:variations_of_detKalpha}. Moreover, as we will see in the next section, the scaling limit of this variation identity give the variation identity for a similar observable of $\m^{2\alpha_1}$.

  Unfortunately for us, there is no reason to expect that the flows $f^K$ and $f^A$ coincide exactly, unless the graph $G$ is isoradially embedded, in which case it miraculously happens (see e.g.~\cite[Theorem~1]{DubedatFamiliesOfCR}). Yet our results show that these flows are close at least in a certain weak sense, see Corollary~\ref{cor:angle_flow}.
\end{rem}

\section{Compactified free field and bosonization identity}
\label{sec:free_field}

\subsection{Scalar and instanton components of the compactified free field}
\label{subsec:bosonization}

In this section we develop the so-called ``bosonization identity'' (cf.~\cite[Section~5.2]{DubedatFamiliesOfCR}) for the compactified free field $\m^\alpha$ introduced in Section~\ref{subsec:intro_freefield}. A slightly less general form of this identity was proven in~\cite{Bosonization}; our generalization corresponds to the case when the involution $\sigma$ which makes $\m^\alpha$ to possess an additional symmetry.

We normalize the Hodge star $\ast$ on $\Sigma$ such that for any $(0,1)$-form $\alpha$ we have
\begin{equation}
  \label{eq:def_of_Hodge_star}
  \ast \alpha = -i\bar{\alpha}.
\end{equation}
Note that if $\vphi$ is a $\mC^1$ function on $\Sigma$, then
\begin{equation}
  \label{eq:Dirichlet_inner_product}
    \int_\Sigma d\vphi\wedge \ast d\vphi
\end{equation}
is the Dirichlet inner product.
Recall that
\begin{equation}
  \label{eq:def_of_Ss}
  \Ss_0(u) = \frac{\pi}{2}\int_{\Sigma_0} u\wedge \ast u.
\end{equation}
Recall that $\Omega$ is the matrix of $B$-periods on $\Sigma$ (see Section~\ref{subsec:simplicial_basis}). Let us write
\begin{equation}
  \label{eq:def_of_RT}
  \Omega = R + iT
\end{equation}
where $R,T$ are real symmetric and $T$ is positive.
\begin{lemma}
  \label{lemma:inner_product_via}
  Let $u$ be a harmonic differential on $\Sigma$ and $a,b\in \RR^g$ denote the vectors of $A$- and $B$-periods of $u$. Then we have
  \[
    \int_\Sigma u\wedge \ast u = (\Omega a - b)\cdot T^{-1}(\bar{\Omega}a - b).
  \]
\end{lemma}
\begin{proof}
  The lemma is classical (see~\cite{Fay}), but let us repeat the proof of it for the sake of completeness. Write $u = \Re v$ for some holomorphic differential $v$, and let $z$ be the vector of $A$-periods of $v$. Then the vector of $B$-periods of $v$ is given by $\Omega z$ and we have
  \[
    a = \Re z,\qquad b = \Re\Omega z
  \]
  whence 
  \[
    z = iT^{-1}(\bar{\Omega}a - b).
  \]
  Note that $\ast u = \Im v$. By~\eqref{eq:Riemann bilinear relations} we have
  \[
    \int_\Sigma u\wedge \ast u = \frac{i}{2}\int_\Sigma v\wedge \bar{v} = \Im \left[ \bar{z}\cdot \Omega z \right] = (\Omega a - b)\cdot T^{-1}(\bar{\Omega} a - b).
  \]
\end{proof}

Recall that for any anti-holomorphic $(0,1)$-forms $\alpha$ we introduce the vectors $a(\alpha), b(\alpha)\in \RR^g$ defined by
  \[
    a(\alpha)_j = \pi^{-1}\int_{A_j}\Im \alpha, \qquad b(\alpha) = \pi^{-1}\int_{B_j}\Im \alpha
  \]
  Recall that $a^0,b^0\in \{ 0,1/2 \}^g$ are the vectors representing the form $q_0$, see Section~\ref{subsec:Szego_kernel}. Let $\theta[\alpha](z)$ be defined as in~\eqref{eq:def_of_theta_alpha}. Recall also that if $u$ is a harmonic differential on $\Sigma$ such that $\sigma^\ast u = -u$ and $u$ is Poincare\'e dual to an integer homology class on $\Sigma_0$, then we denote by $q_0(u\vert_{\Sigma_0})$ the evaluation of the quadratic form $q_0$ on this homology class on $\Sigma_0$, see~\eqref{eq:q0_of_restricted_differential} and the discussion above it. Let $\m^{\alpha_1} = d\phi + \psi^{\alpha_1}$ denote the compactified free field introduced in Section~\ref{subsec:intro_freefield}.

\begin{lemma}
  \label{lemma:bosonization_identity}
  Let $\alpha,\alpha_1$ be an anti-holomorphic $(0,1)$-forms.
  \begin{enumerate}
    \item Assume that $\partial \Sigma_0 = \varnothing$. There exists a constant $\Zz_1$ which depends only on the surface $\Sigma_0$, on the choice of the simplicial basis in the first homologies and on $\alpha_1$ such that 
      \begin{multline*}
        \EE \exp\left[ \pi i q_0(\psi^{\alpha_1} - \pi^{-1}\Im\alpha_1) + 2i \int_{\Sigma_0} \Im\alpha\wedge \psi^{\alpha_1} \right] = \\
        = \Zz_1\cdot (-1)^{\Arf(q_0)} \theta[\alpha_1/2+\alpha](0)\overline{\theta[\alpha_1/2-\alpha](0)}\cdot e^{2\pi i a(\alpha)\cdot( b(\alpha_1) - 2b^0) )}
      \end{multline*}
    \item Assume that $\partial \Sigma_0 \neq \varnothing$ and $\sigma^*\alpha = -\bar{\alpha}$, $\sigma^*\alpha_1 = \bar{\alpha}_1$. There exists a constant $\Zz_1$ which depends only on the surface $\Sigma_0$, on the choice of the simplicial basis in the first homologies of the double $\Sigma$ and on $\alpha_1$ such that
      \begin{multline*}
        \EE \exp\left[ \pi i q_0((\psi^{\alpha_1} - \pi^{-1}\Im\alpha_1)\vert_{\Sigma_0}) + i \int_\Sigma \Im\alpha\wedge \psi^{\alpha_1} \right] = \\
        = \Zz_1\cdot \theta[\alpha_1/2 + \alpha](0)e^{\pi i (a(\alpha)\cdot b(\alpha_1) + b^0\cdot a(\alpha_1))}
      \end{multline*}
  \end{enumerate}
\end{lemma}

\begin{proof}
  We begin by proving the second item. For this purpose we will engage the same Poisson resummation trick that was used in~\cite[Section~4]{Bosonization}. Define $J$ to be the permutation matrix
\begin{equation*}
  \label{eq:def_of_J_again}
  \begin{split}
    & J_{i,i} = 1,\quad i = 1,\dots, n-1,\\
    & J_{g-2g_0+i,g-g_0 + i} = J_{g-g_0+i, g-2g_0 + i} = 1,\quad i = 1,\dots, g_0.
  \end{split}
\end{equation*}
  We have 
  \[
    J\Omega J = -\bar{\Omega},
  \]
  see eq.~\eqref{eq:Omega_under_involution} in Section~\ref{subsec:simplicial_basis} for details. Recall that $R = \Re \Omega$ and $T = \Im \Omega$. It follows that
  \begin{equation}
    \label{eq:bi1}
    JRJ = -R,\qquad JTJ = T.
  \end{equation}
  Let
  \[
    \RR^g = V_+\oplus V_-
  \]
  be the eigenvalues decomposition with respect to $J$, that is for each $v_\pm \in V_\pm$ we have $Jv_\pm = \pm v$. Note that $V_+$ and $V_-$ are orthogonal with respect to the Euclidean scalar product that are invariant subspaces for $T$, while $RV_\pm \subset V_\mp$.  Given $v\in V_+$ the vector $\bar{v}\in V_-$ is defined by antisymmetrizing $v$:
  \[
    \begin{split}
      &\bar{v}_1 = \ldots = \bar{v}_{n-1} = 0,\\
      &\bar{v}_{n-1+i} = v_{n-1+i},\quad \bar{v}_{n-1 + g_0+i} = - v_{n-1+g_0+i} \qquad i = 1,\dots, g_0.
    \end{split}
  \]
  Let
  \[
    \begin{split}
      &\Lambda_+ = \{ x\in V_+\cap \ZZ^g\ \mid\ x_1,\dots, x_{n-1}\in 2\ZZ \},\\
      &\Lambda_- = V_-\cap \ZZ^g.
    \end{split}
  \]
  Let $N,M$ denote vectors of $A$- and $B$-periods of $\psi^{\alpha_1} - \pi^{-1}\Im\alpha_1$ respectively. Since $\sigma^*\alpha = -\bar{\alpha}$, $\sigma^*\psi^{\alpha_1} = -\psi^{\alpha_1}$ and $\sigma^*\alpha_1 = \bar{\alpha}_1$ we have
  \begin{equation}
    \label{eq:bi2}
    \begin{split}
      &b(\alpha),a(\alpha_1)\in V_+,\qquad a(\alpha), b(\alpha_1)\in V_-,\\
      &\quad N\in \Lambda_+,\qquad\qquad\qquad M\in \Lambda_-.
    \end{split}
  \end{equation}
  Since $q_0$ is the quadratic form corresponding to $\Ff_0$ we have (see~\eqref{eq:a_b_for_quadratic_form} and~\eqref{eq:def_of_q0})
  \begin{equation}
    \label{eq:bi3}
     \begin{split}
      &2a^0_j = q_0(A_j) = \wind(A_j, \omega_0) + \gamma\cdot A_j + 1\mod 2,\\
      &2b^0_j = q_0(B_j) = \wind(B_j, \omega_0) + \gamma\cdot B_j + 1\mod 2.
    \end{split}   
  \end{equation}
  Since $\sigma^*\omega_0 = \bar{\omega}_0$ we have $a^0,b^0\in V_+$, and the normalization~\eqref{eq:wind_condition} implies $a^0_j = 0$ for $j = 1,\dots, n-1$.
  
  Using~\eqref{eq:Riemann bilinear relations} and Lemma~\ref{lemma:inner_product_via} we obtain
  \begin{multline}
    \label{eq:bi4}
    \Zz\EE\exp\Big[ \pi i q_0((\psi^{\alpha_1} -\pi^{-1}\Im\alpha_1) \vert_{\Sigma_0}) + i\int_\Sigma \Im \alpha\wedge \psi^{\alpha_1} \Big] = \\
    = \sum_{\substack{N\in \Lambda_+\\ M\in \Lambda_-}}\exp\Big[\pi i ( \frac{ \bar{N}\cdot  M + 2\bar{a}^0\cdot M + 2b^0\cdot  N }{2} + a(\alpha)\cdot  (M+b(\alpha_1)) - b(\alpha)\cdot  (N+a(\alpha_1))) - \\
    - \frac{\pi}{4}(\Omega (N+a(\alpha_1)) - M-b(\alpha_1))\cdot T^{-1}(\bar{\Omega}(N+a(\alpha_1)) - M - b(\alpha_1)) \Big]
  \end{multline}
  where $\Zz$ is the partition function which depends only on $\Sigma$ and $\alpha_1$. To prove the lemma we need to interpret the sum on the right-hand side of~\eqref{eq:bi4} as a value of a theta function at zero times some exponential corrections. This cannot be done straightforwardly because the lattice over which the summation in~\eqref{eq:bi4} is performed is essentially different from the lattice used to define the theta function. To overcome this problem we will adapt the approach of~\cite{Bosonization} and apply Poisson re-summation in $M$. This changes the sublattice $\Lambda_-$ to the dual one and eventually brings the whole lattice to the correct form.

  Note that by the Poisson summation formula for any $U\in V_-$ we have
  \begin{multline}
    \label{eq:bi5}
    \sum_{M\in \Lambda_-}\exp \left[ -\frac{\pi}{4}M\cdot  T^{-1}M + \frac{\pi}{2} U\cdot T^{-1}M \right] =\\
    = 2^g\sqrt{\det(T\vert_{V_-})}\sum_{M\in V_-} \exp\left[ -\pi M\cdot TM - \pi i U\cdot M + \frac{\pi}{4}U\cdot T^{-1}U \right]
  \end{multline}
  In order to apply this formula to the right-hand side of~\eqref{eq:bi4} we need to substitute
  \[
    U = i(\bar{N} + 2\bar{a}^0 + 2a(\alpha))T + (N + a(\alpha_1))R - b(\alpha_1).
  \]
  We proceed by applying this formula to the right-hand side of~\eqref{eq:bi4} and simplifying the expression. Then, we substitute $T = -\frac{i}{2}(\Omega + J\Omega J)$ and $R = \frac{1}{2}(\Omega - J\Omega J)$ and use that $V_\pm$ are eigenspaces of $J$ to get rid of it. Simplifying the expression again we finally get
  \begin{multline}
    \label{eq:bi6}
    \frac{\Zz}{2^g\sqrt{\det T\vert_{V_-}}}\EE\exp\Big[ \pi i q_0((\psi^{\alpha_1} -\pi^{-1}\Im\alpha_1) \vert_{\Sigma_0}) + i\int_\Sigma \Im \alpha\wedge \psi^{\alpha_1} \Big] = \\
    = \!\!\sum_{\substack{N\in \Lambda_+,\\ M\in \Lambda_-}}\!\! \exp\Big[ \pi i (\frac{\bar{N} + N}{2} -M + \frac{a(\alpha_1)}{2} + a(\alpha) + \bar{a}^0)\cdot \Omega (\frac{\bar{N} + N}{2} -M + \frac{a(\alpha_1)}{2} + a(\alpha) + \bar{a}^0) - \\
    - 2\pi i (\frac{b(\alpha_1)}{2} + b(\alpha) + b^0)\cdot (\frac{\bar{N} + N}{2} -M + \frac{a(\alpha_1)}{2} + a(\alpha) + \bar{a}^0) + \\
  + \pi i (a(\alpha)\cdot b(\alpha_1) + b^0\cdot a(\alpha_1)) \Big].
  \end{multline}
  The relation $m = \frac{\bar{N} + N}{2}-M +\bar{a}^0 - a^0$ provides a bijection between $\Lambda_+\times \Lambda_-$ and $\ZZ^g$ (where $(N, M)\in \Lambda_+\times \Lambda_-$). Substituting this formula for $m$ to the right-hand side of~\eqref{eq:bi6} we obtain an expression defining the theta function. We conclude that
  \begin{multline}
    \label{eq:bi7}
    \frac{\Zz}{2^g\sqrt{\det T\vert_{V_-}}}\EE\exp\Big[ \pi i q_0((\psi^{\alpha_1} -\pi^{-1}\Im\alpha_1) \vert_{\Sigma_0}) + i\int_\Sigma \Im \alpha\wedge \psi^{\alpha_1} \Big] = \\
    = \theta\chr{a(\alpha_1)/2 + a(\alpha) + a^0}{b(\alpha_1)/2 + b(\alpha) + b^0}(0,\Omega)\cdot e^{\pi i(a(\alpha)\cdot b(\alpha_1) + b^0\cdot a(\alpha_1))}
  \end{multline}
  which finishes the proof due to~\eqref{eq:def_of_theta_alpha}.

  It remains to prove the first item. To this end we observe that in the proof above we never used the fact that $\Sigma$ is connected. Thus, we can put $\Sigma = \Sigma_0\sqcup \Sigma_0^\op$ and use the formula developed in the second item. To this end we choose the simplicial basis in $H_1(\Sigma,\ZZ)$ in the same way as before and extend $\alpha,\alpha_1$ to $\Sigma$ so that $\sigma^\ast\alpha = -\bar\alpha$ and $\sigma^\ast\alpha_1 = \bar\alpha_1$. Choosing the basis of normalized differentials as previously we obtain a matrix $\Omega$ of $B$-periods on $\Sigma$ that is block diagonal with blocks $\Omega_0$ and $-\overline\Omega_0$ where $\Omega_0$ is the matrix of $B$-periods on $\Sigma_0$. Similarly, the vectors $a = a(\alpha_1)/2 + a(\alpha) + a^0$ and $b = b(\alpha_1)/2 + b(\alpha) + b^0$ that correspond to the periods of $\Im\alpha,\Im\alpha_1$ on $\Sigma$ split into components: $a = a_0 \oplus a_0^\op$ and $b = b_0 \oplus b_0^\op$ where (here by definition $v_0$ is the vector composed of the first half of the coordinates of $v$)
  \[
    \begin{split}
      &a_0 = a(\alpha_1)_0/2 + a(\alpha)_0 + a^0_0,\qquad a_0^\op = a(\alpha_1)_0/2 - a(\alpha)_0 + a_0^0,\\
      &b_0 = b(\alpha_1)_0/2 + b_0(\alpha)_0 + b^0_0,\qquad b_0^\op = -b(\alpha_1)_0/2 + b(\alpha)_0 + b_0^0.
    \end{split}
  \]
  Then~\eqref{eq:bi7} can be read as follows
    \begin{multline}
      \label{eq:bi8}
      \EE \exp\left[ \pi i q_0(\psi^{\alpha_1} - \pi^{-1}\Im\alpha_1) + 2i \int_{\Sigma_0} \Im\alpha\wedge \psi^{\alpha_1} \right] = \\
      = \Zz_1\cdot \sum_{m_0,m_0^\op\in \ZZ^{g_0}} \exp\Big[ \pi i \big((m_0+a_0)\cdot\Omega(m_0+a_0) - (m_0^\op+a_0^\op)\cdot\overline\Omega_0(m_0^\op+a_0^\op) \big) -\\
      - 2\pi i \big(b_0\cdot(m_0+a_0) + b_0^\op\cdot(m_0^\op + a_0^\op) \big) \Big]\cdot e^{2\pi i a(\alpha)_0\cdot b(\alpha_1)_0}
    \end{multline}
    The right-hand side of~\eqref{eq:bi8} splits into a product of two theta functions corresponding to $\Sigma_0$. We conclude that
    \begin{multline}
      \label{eq:bi9}
      \EE \exp\left[ \pi i q_0(\psi^{\alpha_1} - \pi^{-1}\Im\alpha_1) + 2i \int_{\Sigma_0} \Im\alpha\wedge \psi^{\alpha_1} \right] = \\
      = \Zz_1\cdot \theta[\alpha_1/2+\alpha](0)\overline{\theta[\alpha_1/2-\alpha](0)}\cdot e^{-4\pi ib^0\cdot (\frac{a(\alpha_1)}{2} + a(\alpha) + a^0)} \cdot e^{2\pi i (a(\alpha)\cdot b(\alpha_1) + b^0\cdot a(\alpha_1))}.
    \end{multline}
    To finish the proof of the lemma we can observe that $(-1)^{\Arf(q_0)} = e^{4\pi i a^0\cdot b^0}$ (see~\eqref{eq:def_of_Arf}).
\end{proof}

\subsection{Compactified free field as a random distribution}
\label{subsec:cff_formal_def}

Let $\alpha_1$ be an anti-holomorphic $(0,1)$-form and $\m^{\alpha_1}$ be defined as in Section~\ref{subsec:intro_freefield}. Recall that we view $\m^{\alpha_1}$ as a random 1-form with coefficients from the space of generalized functions, or a random functional on the space of smooth 1-forms. We have the Hodge decomposition $\m^{\alpha_1} = d\phi+\psi^{\alpha_1}$. In this section we make some elementary observations relating this decomposition and its relation with how $\m^{\alpha_1}$ acts on harmonic, exact and coexact forms. This helps us to build some notation needed for the proof of main theorems.

\begin{enumerate}
  \item Let $\mC^\infty_1(\Sigma)$ denote the space of all \emph{real-valued} 1-forms with smooth coefficients. We consider $\mC^\infty_1(\Sigma)$ as a Fr\'echet vector space over $\RR$. Let $\mC'_1(\Sigma)$ denote its Fr\'echet dual.
  \item Let $\mC^\infty_{(0,1)}(\Sigma)$ denote the space of all $(0,1)$-forms with smooth coefficients. Note that $\CC$ acts on $\mC^\infty_{(0,1)}(\Sigma)$, nevertheless we consider it as a vector space over $\RR$ equipped with the Fr\'echet topology.
\end{enumerate}

Hodge decomposition in case of $\mC^\infty_1(\Sigma)$ and Dolbeault decomposition in case of $\mC^\infty_{(0,1)}(\Sigma)$ provide splittings
\begin{equation}
  \label{eq:1forms_splitting}
  \begin{split}
    &\mC^\infty_1(\Sigma) = \mC^\infty_{1,\exact}(\Sigma)\oplus\mC^\infty_{1,\coexact}(\Sigma)\oplus \mC_{1,\harm}(\Sigma),\\
    &\mC^\infty_{(0,1)}(\Sigma) = \mC^\infty_{(0,1),\dbar-\exact} \oplus \mC_{(0,1),\antiholom}(\Sigma).
  \end{split}
\end{equation}
Note that for an arbitrary $\alpha = \dbar \vphi + \alpha_h\in \mC^\infty_{(0,1)}(\Sigma)$ we have
\begin{equation}
  \label{eq:expansion_of_Imalpha}
  \Im \alpha = \frac{1}{2}(d\Im \vphi - \ast d\Re\vphi) + \Im \alpha_h.
\end{equation}
It follows that we have a real linear isomorphism between 
\begin{equation}
  \label{eq:iso_between_01_and_1}
  \mC^\infty_{(0,1)}(\Sigma) \xrightarrow{\Im} \mC^\infty_1(\Sigma),\qquad \alpha\mapsto \Im \alpha.
\end{equation}
This isomorphism respects the splitting: the space $\mC^\infty_{(0,1),\dbar-\exact}$ is mapped onto the space $\mC^\infty_{1,\exact}(\Sigma)\oplus\mC^\infty_{1,\coexact}(\Sigma)$ and the space $\mC_{(0,1),\antiholom}(\Sigma)$ onto the space $\mC_{1,\harm}(\Sigma)$.

Note that $\mC^\infty_1(\Sigma)$ acts on itself via the pairing
\begin{equation}
  \label{eq:wedge_pairing}
  \langle u,v\rangle = \int_\Sigma u\wedge v.
\end{equation}
This induces an embedding $\mC^\infty_1(\Sigma)\hookrightarrow \mC'_1(\Sigma)$. Note that we have
\begin{equation}
  \label{eq:C_1_into_C'_1}
  \mC^\infty_{1,\exact}(\Sigma)\hookrightarrow \mC'_{1,\coexact}(\Sigma),\quad \mC^\infty_{1,\coexact}(\Sigma)\hookrightarrow \mC'_{1,\exact}(\Sigma),\quad \mC_{1,\harm}(\Sigma)\xrightarrow{\cong}\mC'_{1,\harm}(\Sigma)
\end{equation}
where the superscript always denotes the Fr\'echet dual space.

If $\partial \Sigma_0\neq \varnothing$, we have the involution $\sigma$ acting on $\Sigma$. Define $\sigma$-symmetric subspaces by
\begin{equation}
  \label{eq:symmetric_forms}
  \begin{split}
    &\mC^\infty_1(\Sigma)^\sigma = \{ u\in \mC^\infty_1(\Sigma)\ \mid\ \sigma^*u = u \},\\
    &\mC^\infty_{(0,1)}(\Sigma)^\sigma = \{ \alpha\in \mC^\infty \ \mid\ \sigma^*\alpha = -\bar{\alpha} \}.
  \end{split}
\end{equation}
The operation of taking $\sigma$-symmetric subspace commutes with splittings~\eqref{eq:1forms_splitting} and the isomorphism~\eqref{eq:iso_between_01_and_1}. Note that the dual space to $\mC^\infty_1(\Sigma)^\sigma$ consists of \emph{anti-symmetric} generalized 1-forms $u\in \mC'_1(\Sigma)$, i.e. 
\begin{equation}
  \label{eq:anti-symmetric_forms}
  \mC'_1(\Sigma)^\sigma = (\mC^\infty_1(\Sigma)^\sigma)^\ast = \{ u\in \mC'_1(\Sigma)\ \mid\ \sigma^*u = -u \}.
\end{equation}

Now, let us consider $\m^{\alpha_1}$ as a random element of $\mC'_1(\Sigma)$, recall that the action of $\m^{\alpha_1}$ on a 1-form $u$ is given by
\begin{equation}
  \label{eq:action_of_malpha1}
  u\mapsto \int_\Sigma u\wedge (d\phi + \psi^{\alpha_1}).
\end{equation}
Then $\m^{\alpha_1}$ has the following properties:
\begin{itemize}
  \item $\m^{\alpha_1}$ is almost surely closed (i.e. vanishes on exact forms). In particular, if $\alpha = \dbar \vphi + \alpha_h\in \mC^\infty_{(0,1)}(\Sigma)$, then the action of $\m^{\alpha_1}$ on $\Im \alpha$ depends only on $\alpha_h$ and $\Re\vphi$;
  \item if $\partial \Sigma_0\neq \varnothing$, then almost surely $\sigma^*\m^{\alpha_1} = - \m^{\alpha_1}$.
\end{itemize}

The following corollary is immediate after Lemma~\ref{lemma:bosonization_identity}.

\begin{cor}
  \label{cor:char_fct_of_malpha1}
  Let $\alpha_1$ be an anti-holomorphic $(0,1)$-form on $\Sigma$ and $\alpha = \dbar \vphi + \alpha_h$ be a smooth $(0,1)$-form, where $\vphi$ is a smooth function and $\alpha_h$ is anti-holomorphic. Let $\Zz_1$ be as in Lemma~\ref{lemma:bosonization_identity}. We have the following:
    \begin{enumerate}
    \item Assume that $\partial \Sigma_0 = \varnothing$. Then we have 
      \begin{multline*}
        \EE \exp\left[ \pi i q_0(\psi^{\alpha_1} - \pi^{-1}\Im\alpha_1) + 2i \int_{\Sigma_0} \Im\alpha\wedge \m^{\alpha_1} \right] = \\
        = \Zz_1\cdot (-1)^{\Arf(q_0)} \theta[\alpha_1/2+\alpha_h](0)\overline{\theta[\alpha_1/2-\alpha_h](0)}\cdot e^{2\pi i a(\alpha_h)\cdot( b(\alpha_1) - 2b^0) }\times\\
        \times\exp\left[ -\frac{1}{2\pi}\int_{\Sigma_0} d\Re \vphi\wedge \ast d\Re \vphi \right]
      \end{multline*}
    \item Assume that $\partial \Sigma_0 \neq \varnothing$ and $\sigma^*\alpha = -\bar{\alpha},\ \sigma^\ast\alpha_1 = \bar\alpha_1$. Then we have
      \begin{multline*}
        \EE \exp\left[ \pi i q_0((\psi^{\alpha_1} - \pi^{-1}\Im\alpha_1)\vert_{\Sigma_0}) + i \int_\Sigma \Im\alpha\wedge \m^{\alpha_1} \right] = \\
        = \Zz_1\cdot \theta[\alpha_1/2 + \alpha_h](0)\cdot e^{\pi i (a(\alpha_h)\cdot b(\alpha_1) + b^0\cdot a(\alpha_1))}\cdot \exp\left[ -\frac{1}{2\pi}\int_{\Sigma_0} d\Re \vphi\wedge \ast d\Re \vphi \right]
      \end{multline*}
  \end{enumerate}
\end{cor}

\begin{proof}
  Let us prove the second item; the proof of the first is similar. Recall that 
  \[
    \Im \alpha = \frac{1}{2}(d\Im \vphi - \ast d\Re\vphi) + \Im \alpha_h
  \]
  and this decomposition is orthogonal with respect to the Dirichlet inner product~\eqref{eq:Dirichlet_inner_product}. Since $\phi$ and $\psi^{\alpha_1}$ are independent we conclude that
  \begin{multline*}
    \EE \exp\left[ \pi i q_0((\psi^{\alpha_1} - \pi^{-1}\Im\alpha_1)\vert_{\Sigma_0}) + i \int_\Sigma \Im\alpha\wedge \m^{\alpha_1} \right] = \\
    = \EE \exp\left[ \pi i q_0((\psi^{\alpha_1} - \pi^{-1}\Im\alpha_1)\vert_{\Sigma_0}) + i \int_\Sigma \Im\alpha_h\wedge \psi^{\alpha_1} \right]\times\\
    \times \EE\exp \left[-\frac{i}{2}\int \ast d\Re \vphi\wedge d\phi\right].
  \end{multline*}
  The corollary now follows from Lemma~\ref{lemma:bosonization_identity} and our normalization of $\phi$.
\end{proof}

Combining Corollary~\ref{cor:char_fct_of_malpha1} with Lemma~\ref{lemma:derivative_of_theta} we get the following corollary. Recall that $r_\alpha$ denotes the second term in the near-diagonal expansion of $\Dd_\alpha^{-1}$, see Lemma~\ref{lemma:diagonal_expansion_of_Salpha} for details.

\begin{cor}
  \label{cor:variation_of_char_fct}
  Assume that $\partial \Sigma_0 \neq \varnothing$ and $\alpha_1$ be an anti-holomorphic $(0,1)$-form on $\Sigma$ such that $\sigma^\ast\alpha_1 = \bar\alpha_1$. Assume that $\alpha_t = \dbar\vphi_t + \alpha_{h,t}$ is a smooth family of $(0,1)$-forms on $\Sigma$ such that for all $t$ we have $\sigma^*\alpha_t = -\bar{\alpha}_t$, $\theta[\alpha_1/2+\alpha_{h,t}](0)\neq 0$ and $\vphi_t,\dot{\vphi}_t\in \mC^2(\Sigma)$. Let $r_{\alpha_1/2\pm\alpha_t}$ be as in Lemma~\ref{lemma:diagonal_expansion_of_Salpha}. Then we have
  \begin{multline*}
    \frac{d}{dt}\log \EE \exp\left[ \pi i q_0((\psi^{\alpha_1} - \pi^{-1}\Im\alpha_1)\vert_{\Sigma_0}) + i \int_\Sigma \Im\alpha\wedge \m^{\alpha_1} \right] =\\
    = -\frac{1}{4}\int_\Sigma \left( r_{\alpha_1/2+\alpha_t}\omega_0\wedge \dot{\alpha}_t - \overline{r_{\alpha_1/2-\alpha_t}\omega_0\wedge \dot{\alpha}_t} \right) = -\frac{1}{2}\int_\Sigma r_{\alpha_t + \alpha_G}\omega_0\wedge \dot{\alpha}_t.
  \end{multline*}
\end{cor}
\begin{proof}
  The first equality follows from Corollary~\ref{cor:char_fct_of_malpha1} and Lemma~\ref{lemma:derivative_of_theta}. The second equality follows from the proof of Lemma~\ref{lemma:derivative_of_theta} given in Section~\ref{subsec:asymptotics_Ddalpha}, see~\eqref{eq:doth1.5}.
\end{proof}

\section{Proof of main results}
\label{sec:proof_of_main_results}

\subsection{Proof of Theorem~\ref{thma:main1}}
\label{subsec:reconstruction_theorem}

We assume that $\partial \Sigma_0\neq \varnothing$. To treat the other case it is enough to apply the same arguments as below with $\Sigma = \Sigma_0\sqcup \Sigma_0^\op$. We follow the notation developed in Section~\ref{subsec:intro_main_results}.

According to~\eqref{eq:iso_between_01_and_1} there exists a finite dimensional $\RR$-linear subspace $\Vv\subset \mC^\infty_{(0,1)}(\Sigma)$ such that $\Uu = \Im \Vv$ and $\mC_{(0,1),\antiholom}(\Sigma)\subset \Vv$. Moreover, we can assume that $\sigma^*\Vv\subset \Vv$ without loss of generality, since $\int_{\Sigma^k}u^k\wedge \M^{\fluct, k}$ depends only on the symmetric (with respect to $\sigma$) part of $u^k$. Given $\alpha\in \Vv$ we denote by $\alpha^k$ the $(0,1)$-form on $\Sigma_k$ defined such that $(\Im\alpha)^k = \Im \alpha^k$, where $(\Im \alpha)^k$ is the pullback of the 1-form $\Im \alpha$ to $\Sigma^k$ as defined in Section~\ref{subsec:intro_main_results}.

Passing to a subsequence and using the assumption of the theorem we can assume that $\M_D^{\fluct,k}\vert_\Uu$ converges weakly to a random functional $\m^\fluct\in \Uu^\ast$. Note that $\sigma^*\m^\fluct = -\m^\fluct$. Let $\psi^\fluct$ be the restriction of $\m^\fluct$ to $\Im\mC_{(0,1),\antiholom}(\Sigma) = \mC_{1,\harm}(\Sigma)$, then $\Psi_D^{\fluct,k}$ converge weakly to $\psi^\fluct$. Since the pairing~\eqref{eq:wedge_pairing} is non-degenerate on $\mC_{1,\harm}(\Sigma)$, we can interpret $\psi^\fluct$ as a random harmonic differential.

Let $\M_D^{K,k}$ and $\M_D^{A,k}$ denote the 1-forms on $\Sigma_k$ defined by~\eqref{eq:def_of_mKD} and~\eqref{eq:def_of_MAD} respectively. Take an arbitrary dimer cover $D$ of $G$ and put
\begin{equation}
  \label{eq:reference_flows}
  \Psi^{\fluct,A,k} = \Psi_D^{A,k} - \Psi_D^{\fluct,k}, \qquad \M^{\fluct,K,k} = \M_D^{K,k} - \M_D^{\fluct,k}.
\end{equation}
Clearly, $\Psi^{\fluct,A,k}$ and $\M^{\fluct,K,k}$ do not depend on the choice of $D$ and $\Psi^{\fluct,A,k}$ is a harmonic differentials (since the corresponding flow is divergence-free). Consider the sequence of points on the torus $\frac{H^1(\Sigma, \RR)}{H^1(\Sigma, 2\ZZ)}$ represented by cohomology classes of $\Psi^{\fluct,A,k}$. Passing to a subsequence we may assume that this sequence converges to a point represented by a harmonic differential $\psi^{\fluct,A}$. Consider the functions
\begin{equation}
  \label{eq:char_fcts}
  \begin{split}
    &F^\fluct_k(\alpha) = \EE\exp\left[ \pi iq_0((\Psi_D^{A,k})\vert_{\Sigma_0}) + i\int_{\Sigma^k} \Im\alpha^k\wedge \M_D^\fluct  \right],\\
    &F^K_k(\alpha) = \EE\exp\left[ \pi iq_0((\Psi_D^{A,k})\vert_{\Sigma_0}) + i\int_{\Sigma^k} \Im\alpha^k\wedge \M_D^{K,k}  \right]
  \end{split}
\end{equation}
defined on $\Vv$. Note that
\begin{equation}
  \label{eq:FH_vs_FK}
  F^K_k(\alpha) = F^\fluct_k(\alpha)\cdot \exp\left[ i\int_{\Sigma^k} \Im \alpha^k\wedge \M^{\fluct,K,k}\right].
\end{equation}

By our assumptions functions $F^\fluct_k$ converge to the function $F^\fluct_0$ given by
\begin{equation}
  \label{eq:char_fcts2}
  F^\fluct_0(\alpha) = \EE\exp\left[ \pi iq_0((\psi^\fluct + \psi^{\fluct,A})\vert_{\Sigma_0}) + i\int_\Sigma \Im\alpha\wedge \m^\fluct  \right]
\end{equation}
uniformly on compacts in $\Vv$.
Note that $F^\fluct_0$ is non-zero since it is a characteristic function of a non-zero (sign indefinite) measure on $(\Im\Vv)^\ast$. It follows that we can fix an anti-holomorphic $(0,1)$-form $\beta$ such that 
\[
  F^\fluct_0(\beta)\neq 0,\qquad \theta[\pm\alpha_1 + \beta](0)\neq 0.
\]

Now let us analyze $F^K_k$. Note that by Proposition~\ref{prop:variations_of_detKalpha} and Corollary~\ref{cor:char_fct_of_malpha1} the logarithmic variation of $e^{-P(\alpha)}\det K_\alpha$ (resp. $e^{-P_0(\alpha)}\det K_{0,\alpha}$ if $\partial\Sigma_0\neq \varnothing$) with respect to $\alpha$ converges to the logarithmic variation of $\EE \exp\left[ \pi i q_0((\psi^{2\alpha_1} - 2\pi^{-1}\Im\alpha_1)\vert_{\Sigma_0}) + i \int_\Sigma \Im\alpha\wedge \m^{2\alpha_1} \right]$. Combining this with Lemma~\ref{lemma:detKalpha} we find out that 
\begin{equation}
  \label{eq:char_fct2.5}
  \lim\limits_{k\to \infty}\frac{F^K_k(\alpha)}{F^K_k(\beta)} = \frac{\EE \exp\left[ \pi i q_0((\psi^{2\alpha_1} - 2\pi^{-1}\Im\alpha_1)\vert_{\Sigma_0}) + i \int_\Sigma \Im\alpha\wedge \m^{2\alpha_1} \right]}{\EE \exp\left[ \pi i q_0((\psi^{2\alpha_1} - 2\pi^{-1}\Im\alpha_1)\vert_{\Sigma_0}) + i \int_\Sigma \Im\beta\wedge \m^{2\alpha_1} \right]}
\end{equation}
when $\alpha$ belongs to an open dense subset of $\Vv$; moreover, the convergence is uniform on compacts from this subset. Note that
\[
  \exp\left[ i\int_{\Sigma^k} (\Im \alpha^k-\Im\beta^k)\wedge \M^{\fluct,K,k}\right] = \frac{F^K_k(\alpha) }{F^K_k(\beta)}\cdot \frac{F^\fluct_k(\beta)}{F^\fluct_k(\alpha)}
\]
by~\eqref{eq:FH_vs_FK}. By~\eqref{eq:char_fct2.5} and the fact that $F^\fluct_k(\beta)$ converge it follows that 
\begin{equation}
  \label{eq:char_fct29}
  \exp\left[ i\int_{\Sigma^k} (\Im \alpha^k-\Im\beta^k)\wedge \M^{\fluct,K,k}\right]
\end{equation}
converge uniformly in $\alpha$ taken from any compact subset of an open dense subset of $\Vv$. This implies that $\M^{\fluct,K,k}$, considered as functionals on $\Uu$, converge to some $\m_\Uu\in \Uu^\ast$. Thus, combining~\eqref{eq:FH_vs_FK} and~\eqref{eq:char_fct2.5}, we can write:
\begin{multline}
  \label{eq:char_fct3}
  \frac{F^\fluct_0(\alpha)}{F^\fluct_0(\beta)} = \\
  = \frac{\EE \exp\left[ \pi i q_0((\psi^{2\alpha_1} - 2\pi^{-1}\Im\alpha_1)\vert_{\Sigma_0}) + i \int_\Sigma \Im\alpha\wedge \m^{2\alpha_1} \right]\cdot \exp\left[ -i\int_\Sigma \Im \alpha\wedge \m_\Uu\right]}{\EE \exp\left[ \pi i q_0((\psi^{2\alpha_1} - 2\pi^{-1}\Im\alpha_1)\vert_{\Sigma_0}) + i \int_\Sigma \Im\beta\wedge \m^{2\alpha_1} \right]\cdot \exp\left[ -i\int_\Sigma \Im \beta\wedge \m_\Uu\right]}
\end{multline}
when $\alpha$ belongs to an open dense subset of $\Vv$. Since both sides are continuous, we conclude that the equality holds for all $\alpha$. Note that by~\eqref{eq:char_fcts2} both sides of~\eqref{eq:char_fct3} are characteristic functions of certain (sign indefinite) measures on $\Uu^\ast$. More precisely, let $\PP^\fluct$ denote the law of $\m^\fluct\vert_\Uu$ and $\PP^{2\alpha_1}$ denote the law of $\m^{2\alpha_1}\vert_\Uu$. The equality~\eqref{eq:char_fct3} combined with~\eqref{eq:char_fcts2} %
implies that
\begin{multline}
  \label{eq:char_fct4}
  \frac{\exp(\pi iq_0((\psi + \psi^{\fluct,A})\vert_{\Sigma_0}))\cdot \mathrm{d}\PP^\fluct(\m)}{F^\fluct_0(\beta)} = \\
  = \frac{\exp(\pi iq_0((\psi - 2\pi^{-1}\Im\alpha_1)\vert_{\Sigma_0}))\cdot \mathrm{d}\PP^{2\alpha_1}(\m + \m_\Uu)}{\EE \exp\left[ \pi i q_0((\psi^{2\alpha_1} - 2\pi^{-1}\Im\alpha_1)\vert_{\Sigma_0}) + i \int_\Sigma \Im\beta\wedge \m^{2\alpha_1} \right]\cdot \exp\left[ - i\int_\Sigma \Im \beta\wedge \m_\Uu\right]},
\end{multline}
where $\m = d\phi + \psi$. 
Taking absolute values of both sides of~\eqref{eq:char_fct4} and using the fact that both measures are probabilistic we conclude that $\PP^\fluct = \PP^{2\alpha_1}(\cdot + \m_\Uu)$, which proves the first assertion of the theorem.

Assume finally that $\EE\left|\int_{\Sigma^k} u^k\wedge \M^{\fluct, k}\right|$ is uniformly bounded in $k$ for each $u\in \Uu$. By the definition we have $\EE\M^{\fluct, k} = 0$, which now implies $\EE \m^\fluct\vert_\Uu = 0$. But then $\PP^\fluct = \PP^{2\alpha_1}(\cdot + \m_\Uu)$ implies $\m_\Uu = \EE (\m^{2\alpha_1}\vert_\Uu)$. To conclude the second assertion of the theorem note that $\EE \m^{2\alpha_1}$ depends only on the instanton component of $\m^{2\alpha_1}$ and hence $\EE (\m^{2\alpha_1}\vert_\Uu) = (\EE \m^{2\alpha_1})\vert_\Uu$. The proof of Theorem~\ref{thma:main1} is concluded.

The proof of Theorem~\ref{thma:main1} in fact gives us some information about the asymptotic behaviour of $\M^{\fluct, K, k}$ and $\Psi^{\fluct, A,k}$, which we would like to point out in the following

\begin{cor}
  \label{cor:angle_flow}
  Recall that $\M^{\fluct, K, k}$ and $\Psi^{\fluct, A,k}$ are defined by~\eqref{eq:reference_flows}:
  \begin{equation*}
    \Psi^{\fluct,A,k} = \Psi_D^{A,k} - \Psi_D^{\fluct,k}, \qquad \M^{\fluct,K,k} = \M_D^{K,k} - \M_D^{\fluct,k}.
  \end{equation*}
  Assume that tightness and first moment hypothesis of Theorem~\ref{thma:main1} are fulfilled. Then 
  \begin{enumerate}
    \item The sequence $\M^{\fluct, K ,k}$ converges to $\EE\m^{2\alpha_1}$ in the $\ast$-weak topology in $\mC_1'(\Sigma)$.
    \item The sequence in $\frac{H^1(\Sigma, \RR)}{H^1(\Sigma, 2\ZZ)}$ induced by cohomology classes of $\Psi^{\fluct, A,k}$ converges to the cohomology class of $\EE\psi^{2\alpha_1} - 2\pi^{-1}\Im\alpha_1$.
  \end{enumerate}
\end{cor}
\begin{proof}
  The first item follows the remark after~\eqref{eq:char_fct29} asserting that $\M^{\fluct, K, k}$ converges to $\m_\Uu$ and the identification of $\m_\Uu$ with $\EE\m^{2\alpha_1}$ made in the end of the proof.

  To prove the second item let us recall that $\psi^{\fluct,A}$ denotes an arbitrary partial limit of the sequence $\Psi^{\fluct, A,k}$ considered modulo $H^1(\Sigma, 2\ZZ)$. By the definition of $\Psi^{\fluct, A, k}$ and Lemma~\ref{lemma:Kasteleyn_thm} we know that for any dimer cover $D$ the harmonic differential $\Psi_D^{A,k} = \Psi_D^{\fluct,k} + \Psi^{\fluct, A, k}$ has integer cohomologies. By Theorem~\ref{thma:main1}, the support of the distribution of $\Psi_D^{\fluct,k}$ in the limit when $k\to +\infty$ coincides with the support of $\psi^{2\alpha_1} - \EE \psi^{2\alpha_1}$. Combining these two observations we conclude that 
  \begin{equation}
    \label{eq:caf1}
    \psi^{\fluct,A}\equiv \EE\psi^{2\alpha_1} - 2\pi^{-1}\Im\alpha_1\quad \mod H^1(\Sigma, \ZZ)
  \end{equation}

  It remains to show that the equality in~\eqref{eq:caf1} holds modulo $H^1(\Sigma, 2\ZZ)$. To simplify the presentation, we assume that $\partial \Sigma_0 = \varnothing$; the other case can be treated similarly. By the equality~\eqref{eq:char_fct4} and the fact that $d\PP^\fluct(\cdot) = d\PP^{2\alpha_1}(\cdot + \EE\m^{2\alpha_1})$ we conclude that there exists a sign $\epsilon\in \{ \pm1 \}$ such that for any harmonic differential $\psi$ which can occur in the support of the distribution of $\psi^{2\alpha_1} - \EE\psi^{2\alpha_1}$ we have
  \[
    \exp[\pi i q_0(\psi + \psi^{\fluct, A})] = \epsilon\exp[\pi i q_0(\psi + \EE\psi^{2\alpha_1} - 2\pi^{-1}\Im\alpha_1)].
  \]
  Denote by $u$ the cohomology class in $H^1(\Sigma, \ZZ/2\ZZ)$ represented by $\psi^{\fluct, A} - \EE\psi^{2\alpha_1} + 2\pi^{-1}\Im\alpha_1$ and by $v$ the cohomology class represented by $\psi + \EE\psi^{2\alpha_1} - 2\pi^{-1}\Im\alpha_1$. We need to show that $u = 0$. The equality above can be rewritten as
  \[
    \exp[\pi i q_0(v+u)] = \epsilon\exp[\pi i q_0(v)]
  \]
  for a fixed $u\in H^1(\Sigma, \ZZ/2\ZZ)$ and an arbitrary $v\in H^1(\Sigma, \ZZ/2\ZZ)$. Using that
  \[
    \#\{ v\in H^1(\Sigma, \ZZ/2\ZZ) \ \mid\ q_0(v) = 1\}\neq \#\{ v\in H^1(\Sigma, \ZZ/2\ZZ) \ \mid\ q_0(v) = 0\}
  \]
  we conclude that $\epsilon = 1$, thus $q_0(v+u) = q_0(v)$ for all $v$. Using that $q_0(v+u) - q_0(v) = u\cdot v + q_0(u)$ and $v$ is arbitrary we conclude that $u = 0$.
\end{proof}

\subsection{Proof of Theorem~\ref{thma:main2}}
\label{subsec:second_thm_tightness}

We begin with the following lemma. See Section~\ref{subsec:quadratic_forms} for the definition of odd and even quadratic forms over $\ZZ/2\ZZ$ and their Arf invariants.

\begin{lemma}
  \label{lemma:sum_of_even_quadratic_forms}
  Let $T$ be a closed compact oriented surface of genus $g\geq 1$ and $\Qq_+$ denote the space of all even quadratic forms on $H_1(T,\ZZ/2\ZZ)$. Then for any $v\in H_1(T,\ZZ/2\ZZ)$ we have
  \[
    2^{g-1} \leq \sum_{q\in \Qq_+} \exp\left[ \pi i q(v) \right]\leq 2^{g-1}(2^g+1).
  \]
\end{lemma}
\begin{proof}
  The second inequality follows from the fact that $\Qq_+$ has $2^{g-1}(2^g+1)$ elements. We prove the first one by induction on $g$. The inequality can be easily verified in the case when $g = 1$. Assume that $g\geq 2$ and the inequality holds for each surface of genus less than $g$. Let $\Qq_-$ be the space of odd quadratic forms. Set
  \begin{equation}
    \label{eq:qf2}
    X_g(v)= \sum_{q\in \Qq_+} \exp\left[ \pi i q(v) \right],\qquad Y_g(v)=\sum_{q\in \Qq_-} \exp\left[ \pi i q(v) \right].
  \end{equation}
  Choose a symplectic basis $A_1,\dots, A_g,B_1,\dots, B_g\in H_1(T,\ZZ/2\ZZ)$ and let $q_0\in \Qq_+$ be the quadratic form vanishing on this basis. As follows from Lemma~\ref{lemma:shifted_forms}, each quadratic form $q$ can be written as
  \begin{equation}
    \label{eq:qf1}
    q(v) = q_0(v+u) + \Arf(q) = q_0(v+u) + q_0(u)
  \end{equation}
  which provides a bijection between $\Qq_+\cup\Qq_- $ and $H^1(T,\ZZ/2\ZZ)$. Applying~\eqref{eq:qf1} we get
  \begin{multline}
    \label{eq:qf3}
    X_g(v) - Y_g(v) = \sum_{q\in \Qq_+\cup \Qq_-}\exp[\pi i (q(v) + \Arf(q))] = \sum_{u\in H_1(T,\ZZ/2\ZZ)} \exp\left[ \pi i q_0(v+u) \right] = \\
    = \sum_{u\in H_1(T,\ZZ/2\ZZ)} \exp\left[ \pi i q_0(u) \right] = \sum_{q\in \Qq_+\cup \Qq_-}\exp[\pi i \Arf(q)] = |\Qq_+| - |\Qq_-| = 2^g.
  \end{multline}
  Write $u = u_1+u_2, v = v_1 + v_2$, where $u_1,v_1$ are the projections of $u$ and $v$ onto the span of $A_1,B_1,\dots, A_{g-1}, B_{g-1}$. Note that $q_0(u+v) = q_0(u_1+v_1) + q_0(u_2 + v_2)$, therefore
  \begin{multline}
    \label{eq:qf4}
    X_g(v) = \left(\sum_{u_1\colon\ q_0(u_1) = 0}\exp\left[ \pi i q_0(v_1 + u_1) \right]\right)\cdot \left(\sum_{u_2\colon\ q_0(u_2) = 0}\exp\left[ \pi i q_0(v_2 + u_2) \right]\right) +\\
    + \left(\sum_{u_1\colon\ q_0(u_1) = 1}\exp\left[ \pi i q_0(v_1 + u_1) \right]\right)\cdot \left(\sum_{u_2\colon\ q_0(u_2) = 1}\exp\left[ \pi i q_0(v_2 + u_2) \right]\right) =\\
    = X_{g-1}(v_1)X_1(v_2) + Y_{g-1}(v_1)Y_1(v_2) =\\
    = X_{g-1}(v_1)X_1(v_2) + (2^{g-1} - X_{g-1}(v_1))(2 - X_1(v_2)),
  \end{multline}
  where we used~\eqref{eq:qf3} in the last equality. Note that $X_1(v_2)\in \{ 1,3 \}$. In particular, the right-hand side of~\eqref{eq:qf4} is either equal to $2^{g-1}$, or equal to $4X_{g-1}(v_1) - 2^{g-1}$ which is greater of equal to $2^{g-1}$ by the induction hypothesis.
\end{proof}

Introduce the following sets:
\begin{enumerate}
  \item If $\partial \Sigma_0 = \varnothing$, then we set 
    \[
      L_0^\varnothing = \{ l\in H^1(\Sigma, \ZZ/2\ZZ)\ \mid\ q_0 + l\text{ is even} \}.
    \]
  \item Assume that $\partial \Sigma_0 \neq \varnothing$. Let $\bar{\Sigma}_0$ denote the closed surface obtained from $\Sigma_0$ by gluing $n$ discs to the boundary components. Note that if $l\in H^1(\Sigma,\ZZ/2\ZZ)$ is such that $l(A_1) = \ldots = l(A_{n-1}) = 0$ and $l(B_1) = \ldots = l(B_{n-1}) = 1$, then $q_0 + l$ induces a well-defined quadratic form on $H_1(\bar{\Sigma}_0, \ZZ/2\ZZ)$; we denote it by $(q_0 + l)\vert_{\bar{\Sigma}_0}$. Note that if $\sigma^*l = l$, then $l(A_1) = \ldots = l(A_{n-1}) = 0$. We now set
    \[
      L_0^{\partial \Sigma_0} = \{ l\in H^1(\Sigma, \ZZ/2\ZZ)\ \mid\ \sigma^*l = l,\ l(B_1) = \ldots = l(B_{n-1}) = 1,\ (q_0 + l)\vert_{\bar{\Sigma}_0}\text{ is even}\};
    \]
    if the genus of $\bar{\Sigma}_0$ is zero, then we put $L_0^{\partial \Sigma_0}$ to contain the only $l\in H^1(\Sigma, \ZZ/2\ZZ)$ satisfying the first two conditions.
\end{enumerate}
Given an element $l\in H^1(\Sigma, \RR)$, denote by $\alpha_l$ the anti-holomorphic $(0,1)$ form such that for any cycle $C\in H_1(\Sigma, \ZZ)$ we have
\[
  \pi^{-1}\int_C \Im \alpha_l = l(C).
\]
Note that if $\sigma^*l = l$, then $\sigma^*\alpha_l = -\bar{\alpha}_l$. In what follows we identify $l\in H^1(\Sigma, \ZZ/2\ZZ)$ with $l\in H^1(\Sigma, \ZZ)$ such that $l(A_j),l(B_j)\in \{ 0, 1\}$ for any $j = 1,\dots, g$, where $A_1,\dots, A_g,B_1,\dots, B_g$ is the simplicial basis in $H_1(\Sigma, \ZZ)$.

Let $D$ be an arbitrary dimer cover of $G$. Define
\begin{equation}
  \label{eq:def_of_PsiAK}
  \Psi^{A,K} = \Psi_D^A - \Psi_D^K,
\end{equation}
where $\Psi_D^A$ and $\Psi_D^K$ are as in~\eqref{eq:def_of_MAD} and~\eqref{eq:def_of_mKD} respectively. Clearly, $\Psi^{A,K}$ does not depend on $D$.

\begin{lemma}
  \label{lemma:partition_function_estimate}
  The following inequalities hold:
  \begin{enumerate}
    \item Assume that $\partial \Sigma_0 = \varnothing$. Let $\epsilon$ be the constant from the first item of Lemma~\ref{lemma:Kasteleyn_thm} and $\Zz_G$ be the partition function of the dimer model on $G$. Then we have
      \begin{equation*}
        2^{g-1}\Zz_G\leq \epsilon^{-1}\sum_{l\in L_0^\varnothing} \exp\left[ i\int_{\Sigma_0}\Im\alpha_l \wedge\Psi^{A,K} \right] e^{-P(\frac{1}{2}\alpha_l)} \det K_{\frac{1}{2}\alpha_l} \leq 2^{g-1}(2^g+1)\Zz_G.
      \end{equation*}
    \item Assume that $\partial \Sigma_0 \neq \varnothing$. Let $\epsilon$ be the constant from the second item of Lemma~\ref{lemma:Kasteleyn_thm} and $\Zz_{G_0}$ be the partition function of the dimer model on $G_0$. Then we have
      \begin{equation*}
        2^{\bar{g}_0-1}\Zz_{G_0}\leq \epsilon^{-1}\!\!\sum_{l\in L_0^{\partial \Sigma_0}} \exp\left[ \frac{i}{2}\int_\Sigma\Im\alpha_l \wedge\Psi^{A,K} \right] e^{-P_0(\frac{1}{2}\alpha_l)} \det K_{0,\frac{1}{2}\alpha_l} \leq 2^{\bar{g}_0-1}(2^{\bar{g}_0}+1)\Zz_{G_0},
      \end{equation*}
      where $\bar{g}_0$ is the genus of $\bar{\Sigma}_0$.
  \end{enumerate}
\end{lemma}
\begin{proof}
  We will prove only the second item, the first one can be proven similarly. Given a dimer cover $D_0$ of $G_0$ and the corresponding dimer cover $D = D_0\cup E\cup \sigma(D_0)$ constructed as in the second item of Lemma~\ref{lemma:detKalpha} we denote by $u_D\in H_1(\Sigma_0,\ZZ)$ the homology class Poincar\'e dual to $\Psi_D^A\vert_{\Sigma_0}$. If $l\in L_0^{\partial \Sigma_0}$, then $(q_0 + l)(u_D)$ depends only on the projection of $u_D$ to $H_1(\bar{\Sigma}_0,\ZZ)$. Note that 
  \[
    \exp\left[ \frac{i}{2}\int_\Sigma\Im\alpha_l \wedge\Psi^{A,K} \right]\cdot \exp\left[ \frac{i}{2}\int_\Sigma \Im\alpha_l\wedge \M_D^K \right] = \exp\left[ \pi i l(u_D) \right]
  \]
  whenever $D$ is of the form $D = D_0\cup E \cup \sigma(D_0)$.
  Using these observations and the second item of Lemma~\ref{lemma:detKalpha} we can write
  \begin{equation}
    \label{eq:pfe1}
    \epsilon^{-1} \exp\left[ \frac{i}{2}\int_\Sigma\Im\alpha_l \wedge\Psi^{A,K} \right] e^{-P_0(\frac{1}{2}\alpha_l)} \det K_{0,\frac{1}{2}\alpha_l} = \Zz_{G_0}\EE\exp\left[ \pi i \left((q_0 + l)\vert_{\bar{\Sigma}_0}\right)(u_D) \right].
  \end{equation}
  The inequalities now follow from Lemma~\ref{lemma:sum_of_even_quadratic_forms}.
\end{proof}

The following technical result implies Theorem~\ref{thma:main2} immediately.

\begin{thmas}
  \label{thmas:second_moment_estimate}
  Let $\Kk\subset \Mm_g^{t,(0,1)}$ be a compact subset such that for any $[\Sigma, A,B,\alpha]\in \Kk$ we have $\theta[\alpha](0)\neq 0$. Let $\lambda,R>0$ be fixed and $\delta>0$ be small enough depending on $\Kk,\lambda, R$. There exists a constant $C>0$ depending only on $\Kk, R$ and $\lambda$ such that for each $(\lambda,\delta)$-adapted graph $G$ on $\Sigma$ the following holds:
  \begin{enumerate}
    \item Assume that $\partial \Sigma_0 = \varnothing$. Assume that $\alpha_G$ can be chosen such that
      \[
        (\Sigma, A_1,\dots, A_g,B_1,\dots,B_g,\pm\frac{1}{2}\alpha_l + \alpha_G)\in \Kk
      \]
      for any $l\in L_0^\varnothing$. Then for any $\alpha = \dbar \vphi + \alpha_h$ such that $\|\vphi\|_{\mC^2(\Sigma)}\leq R$ and $\int_\Sigma \alpha_h\wedge\ast \alpha_h\leq R$ we have
      \[
        \EE \left(\int_\Sigma\Im \alpha\wedge \M_D^K \right)^2 \leq C.
      \]
    \item Assume that $\partial \Sigma_0 \neq \varnothing$. Assume that $\alpha_G$ can be chosen such that
      \[
        (\Sigma, A_1,\dots, A_g,B_1,\dots,B_g,\pm\frac{1}{2}\alpha_l + \alpha_G)\in \Kk
      \]
      for any $l\in L_0^{\partial \Sigma_0}$. Then for any for any $\alpha = \dbar \vphi + \alpha_h$ such that $\|\vphi\|_{\mC^2(\Sigma)}\leq R$ and $\int_\Sigma \alpha_h\wedge\ast \alpha_h\leq R$ and such that $\sigma^*\alpha = -\bar{\alpha}$ we have
      \[
        \EE \left(\int_{\Sigma_0}\Im \alpha\wedge \M_D^K \right)^2 \leq C.
      \]
  \end{enumerate}
\end{thmas}

\begin{proof}[Proof of Theorem~\ref{thma:main2}]
  Assume that we are in the setting of Section~\ref{subsec:intro_main_results}. To derive Theorem~\ref{thma:main2} from Theorem~\ref{thmas:second_moment_estimate} it is enough to show that  
  \begin{equation}
    \label{eq:poT2}
    \theta[\pm\frac{1}{2}\alpha_l + \alpha_1](0)\neq 0.
  \end{equation}
  for any $l\in L_0^{\partial \Sigma_0}$ provided one of the assumptions~\ref{item:intro1}--\ref{item:intro3} formulated prior Theorem~\ref{thma:main2} is satisfied.

  Assume that the form $2\pi^{-1}\Im\alpha_1$ has integer cohomologies and all theta constants of $\Sigma$ corresponding to even theta characteristics are non-zero. By choosing $\alpha_{G^k}$ properly we can assume that $\alpha_1 = 0$. In this case~\eqref{eq:poT2} boils down to $\theta[\frac{1}{2}\alpha_l](0)\neq 0$ given that $l\in L_0^{\partial \Sigma_0}$. But the definition of $L_0^{\partial \Sigma_0}$ and of $\theta[\frac{1}{2}\alpha_l]$ (see~\eqref{eq:def_of_theta_alpha}) imply that $\theta[\frac{1}{2}\alpha_l](z)$ is a theta function with an even theta characteristics, therefore $\theta[\frac{1}{2}\alpha_l](0)\neq 0$ by our assumption.

  Assume that Assumption~\ref{item:intro2} above Theorem~\ref{thma:main2} is satisfied. In this case we note that the definition~\eqref{eq:def_of_theta_alpha} of $\theta[\pm\frac{1}{2}\alpha_l + \alpha_1](z)$ and basic properties of theta function (see item~3 of Proposition~\ref{prop:of_theta} and~\eqref{eq:Phi_alpha_via_a_b}) imply that $\theta[\pm\frac{1}{2}\alpha_l + \alpha_1](0) = 0$ only if $\pi^{-1}\Im\alpha_1$ belongs to a half-integer shift of the theta divisor, hence~\eqref{eq:poT2} hold given Assumption~\ref{item:intro2}.

  Assume now that $\Sigma_0$ is a multiply-connected domain. Recall that $L_0^{\partial\Sigma_0}$ consists of one $l$ such that $l(A_1) = \ldots = l(A_{n-1}) = 0$ and $l(B_1) = \ldots = l(B_{n-1}) = 1$. We claim that~\eqref{eq:poT2} holds for any $\alpha_1$ such that $\sigma\alpha_1 = \bar{\alpha}_1$. We can prove it using Lemma~\ref{lemma:bosonization_identity}. Indeed, consider $\m^{2\alpha_1}$. By the definition, $(\psi^{2\alpha_1} - 2\pi^{-1}\Im \alpha_1)\vert_{\Sigma_0}$ is Poincare dual to a random integer homology class of the form $k_1B_1 + \ldots + k_{n-1}B_{n-1}$. It follows that
  \[
    q_0((\psi^{2\alpha_1} - 2\pi^{-1}\Im\alpha_1)\vert_{\Sigma_0}) + (2\pi)^{-1}\int_\Sigma \Im\alpha_l\wedge (\psi^{2\alpha_1} - 2\pi^{-1}\Im\alpha_1) = 0\mod 2
  \]
  whence
  \begin{multline}
    \label{eq:poT3}
    \exp\left[ \pi i q_0((\psi^{2\alpha_1} - 2\pi^{-1}\Im\alpha_1)\vert_{\Sigma_0}) + \frac{i}{2} \int_\Sigma \Im\alpha_l\wedge \psi^{2\alpha_1} \right] =\\
    = \exp\left[ i\pi^{-1}\int_\Sigma \Im \alpha_l \wedge \Im \alpha_1\right].
  \end{multline}
  But the right-hand side of~\eqref{eq:poT3} is deterministic, hence, applying the second item of Lemma~\ref{lemma:bosonization_identity} we get
  \[
    \theta[\frac{1}{2}\alpha_l + \alpha_1](0) \neq 0.
  \]
  The same applies to $-\alpha_l$, hence~\eqref{eq:poT2} holds.
\end{proof}

\begin{proof}[Proof of Theorem~\ref{thmas:second_moment_estimate}]
  As usually, we prove only the second item; the first one can be proved similarly. For any $l\in H^1(\Sigma, \RR)$ set
  \begin{equation}
    \label{eq:sme1}
    \alpha_l(t) = \frac{1}{2}\alpha_l + t\alpha.
  \end{equation}
  Arguing as in the proof of Lemma~\ref{lemma:partition_function_estimate} one can show that
  \begin{multline}
    \label{eq:sme2}
    \frac{d^2}{dt^2}\sum_{l\in L_0^{\partial \Sigma_0}} \exp\left[ \frac{i}{2}\int_\Sigma\Im\alpha_l \wedge\Psi^{A,K} \right] e^{-P_0(\alpha_l(t))} \det K_{0,\alpha_l(t)}\vert_{t=0}\asymp\\
    \asymp\Zz_{G_0}\EE \left(\int_{\Sigma_0}\Im \alpha\wedge \M_D^K \right)^2,
  \end{multline}
  where $A\asymp B$ if $c^{-1}|B|\leq |A|\leq c|B|$ for some constant $c>1$  
  (we have $c = 2^{\bar{g}_0 -1}(2^{\bar{g}_0}+1)$ in our case). By the assumptions of the theorem and Proposition~\ref{prop:variations_of_detKalpha} for any two $l_1,l_2\in L_0^{\partial \Sigma_0}$ we have
  \begin{equation}
    \label{eq:sme3}
    \det K_{0,\frac{1}{2}\alpha_{l_1}} \asymp \det K_{0,\frac{1}{2}\alpha_{l_2}}.
  \end{equation}
  Using this observation, Lemma~\ref{lemma:partition_function_estimate}, equation~\eqref{eq:sme2} and Proposition~\ref{prop:variations_of_detKalpha} again we conclude that the inequality for the second moment would follow if we prove that for some $l\in L_0^{\partial \Sigma_0}$ and some constant $C>0$ depending only on $\Kk,R$ and $\lambda$ we have
  \begin{equation}
    \label{eq:sme4}
    \left|\frac{d^2}{dt^2}\log \det K_{0,\alpha_l(t)}\vert_{t=0}\right| \leq C.
  \end{equation}
  From now on we will write $= O(1)$ instead of $\leq C$. Expanding~\eqref{eq:sme4} we obtain
  \begin{multline}
    \label{eq:sme5}
    \frac{d^2}{dt^2}\log \det K_{0,\alpha_l(t)}\vert_{t=0} = -4\sum_{b\sim w,\ b,w\in G_0} \left( \int_w^b\Im \alpha \right)^2\cdot K_{\frac{1}{2}\alpha_l}(w,b)K_{0,\frac{1}{2}\alpha_l}^{-1}(b,w) + \\
    + 4 \sum_{\substack{b_1\sim w_1,\ b_1,w_1\in G_0 \\ b_2\sim w_2,\ b_2,w_2\in G_0}} \left(\int_{w_1}^{b_1}\Im\alpha\right)\cdot \left( \int_{w_2}^{b_2}\Im\alpha \right)\times\\
    \times K_{\frac{1}{2}\alpha_l}(w_1,b_1)K_{\frac{1}{2}\alpha_l}(w_2,b_2)K_{0,\frac{1}{2}\alpha_l}^{-1}(b_1,w_2)K_{0,\frac{1}{2}\alpha_l}^{-1}(b_2,w_1).
  \end{multline}
  We estimate two sums on the right-hand side of~\eqref{eq:sme5} separately. By Lemma~\ref{lemma:near-diag-expansion-of-Kalphainv} and Theorem~\ref{thmas:parametrix} we have 
  \begin{equation}
    \label{eq:sme6}
    K_{\frac{1}{2}\alpha_l}(w,b)K_{0,\frac{1}{2}\alpha_l}^{-1}(b,w) = O(1),
  \end{equation}
  hence
  \begin{equation}
    \label{eq:sme7}
    \sum_{b\sim w,\ b,w\in G_0} \left( \int_w^b\Im \alpha \right)^2\cdot K_{\frac{1}{2}\alpha_l}(w,b)K_{0,\frac{1}{2}\alpha_l}^{-1}(b,w) = O(1).
  \end{equation}
  Let us deal with the second sum. Recall that by Lemma~\ref{lemma:Kalphainv_Sigma_with_boundary} we have 
  \begin{equation}
    \label{eq:sme8}
    K_{0,\frac{1}{2}\alpha_l}^{-1}(b,w) = K_{\frac{1}{2}\alpha_l}^{-1}(b,w) + \eta_w^2 K_{\frac{1}{2}\alpha_l}^{-1}(b,\sigma(w));
  \end{equation}
  using this formula we can extend $K^{-1}_{0,\frac{1}{2}\alpha_l}(b,w)$ to the whole $G$. Fix an edge $w_2b_2$ of $G_0$ and consider the flow
  \begin{equation}
    \label{eq:sme9}
    f_{w_2,b_2}(w_1b_1) = K_{\frac{1}{2}\alpha_l}(w_1,b_1)K_{\frac{1}{2}\alpha_l}(w_2,b_2)K_{0,\frac{1}{2}\alpha_l}^{-1}(b_1,w_2)K_{0,\frac{1}{2}\alpha_l}^{-1}(b_2,w_1).
  \end{equation}
  By its definition the flow $f$ has zero divergence at all vertices except of $b_2,w_2,\sigma(b_2), \sigma(w_2)$, where we have
  \begin{equation}
    \label{eq:sme10}
    \begin{split}
      &\div f_{w_2,b_2}(b_2) = -\div f_{w_2,b_2}(w_2) = K_{\frac{1}{2}\alpha_l}(w_2,b_2)K_{0,\frac{1}{2}\alpha_l}^{-1}(b_2,w_2),\\
      &\div f_{w_2,b_2}(\sigma(b_2)) = -\div f_{w_2,b_2}(\sigma(w_2)) = \eta_w^2K_{\frac{1}{2}\alpha_l}(w_2,b_2)K_{0,\frac{1}{2}\alpha_l}^{-1}(b_2,\sigma(w_2)).
    \end{split}
  \end{equation}
  (recall that we have $K_{\frac{1}{2}\alpha_l}(\sigma(b),\sigma(w)) = -(\eta_b\eta_w)^2K_{\frac{1}{2}\alpha_l}(b,w)$ and $K_{\frac{1}{2}\alpha_l}^{-1}(\sigma(b),\sigma(w)) = - (\bar{\eta}_b\bar{\eta}_w)^2 K_{\frac{1}{2}\alpha_l}^{-1}(b,w)$). Let $\M_{w_2,b_2}$ be the 1-form associated with $f_{w_2,b_2}$ and let
  \begin{equation}
    \label{eq:sme11}
    \M_{w_2,b_2} = d\Phi_{w_2,b_2} + \Psi_{w_2,b_2}
  \end{equation}
  be its Hodge decomposition. Using Lemma~\ref{lemma:primitive_of_flow} we can estimate $\Psi_{w_2,b_2}$ and $\Phi_{w_2,b_2}$. Fix an arbitrary smooth metric on $\Sigma$ and denote by $|p-q|$ the distance between $p$ and $q$ in this metric. For any loop $\gamma$ on $\Sigma$ on a definite distance from $b_2,w_2,\sigma(b_2), \sigma(w_2)$ and from conical singularities we have
  \begin{equation}
    \label{eq:sme12}
    \int_\gamma \Psi_{w_2,b_2} = O(\delta\length(\gamma))
  \end{equation}
  since for any edge $b_1w_1$ intersecting $C$ we have $|f_{w_2,b_2}(w_1b_1)| = O(\delta^2)$. Using this observation and expressing the residues of $\Psi_{w_2,b_2}$ via the divergence of $f_{w_2,b_2}$ we conculude that 
  \begin{equation}
    \label{eq:sme13}
    |\Psi_{w_2,b_2}(p)| = O\left( \delta\cdot \left( \frac{1}{|p - b_2||p - w_2|} + \frac{1}{|p - \sigma(b_2)||p - \sigma(w_2)|} \right) \right)
  \end{equation}
  uniformly in positions of $w_2,b_2$. Using this estimate on $\Psi$, Proposition~\ref{prop:Kernel_for_Kalpha} together with Theorem~\ref{thmas:parametrix} and Lemma~\ref{lemma:kernel_of_G4pi} to estimate $f_{w_2,b_2}(w_1b_1)$, and Lemma~\ref{lemma:primitive_of_flow} again we conclude that $\Phi_{w_2,b_2}$ can be chosen such that
  \begin{multline}
    \label{eq:sme14}
    |\Phi_{w_2,b_2}(p)| = O\left( \delta\cdot \frac{|\log|p-b_2|| + |\log|p-w_2|| + 1}{|p - b_2| + |p - w_2|}\right) + \\
    + O\left(\delta\cdot  \frac{|\log|p-\sigma(b_2)|| + |\log|p-\sigma(w_2)|| + 1}{|p - \sigma(b_2)| + |p - \sigma(w_2)|}\right)
  \end{multline}
  uniformly in positions $w_2,b_2$. We now estimate the second sum in~\eqref{eq:sme5}. Using~\eqref{eq:sme12} and~\eqref{eq:sme13} we can write
  \begin{equation}
    \label{eq:sme15}
    \begin{split}
      &\sum_{\substack{b_1\sim w_1,\ b_1,w_1\in G_0 \\ b_2\sim w_2,\ b_2,w_2\in G_0}} \left(\int_{w_1}^{b_1}\Im\alpha\right)\cdot \left( \int_{w_2}^{b_2}\Im\alpha \right)\times\\
      & \times K_{\frac{1}{2}\alpha_l}(w_1,b_1)K_{\frac{1}{2}\alpha_l}(w_2,b_2)K_{0,\frac{1}{2}\alpha_l}^{-1}(b_1,w_2)K_{0,\frac{1}{2}\alpha_l}^{-1}(b_2,w_1)= \\
      & = \frac{1}{2}\sum_{b_2\sim w_2,\ b_2,w_2\in G_0} \left( \int_{w_2}^{b_2}\Im\alpha \right)\cdot \int_\Sigma \Im \alpha\wedge \M_{w_2,b_2} =\\
      &= \frac{1}{2}\sum_{b_2\sim w_2,\ b_2,w_2\in G_0} \left( \int_{w_2}^{b_2}\Im\alpha \right)\cdot \int_\Sigma \left(\frac{1}{2}\Phi_{w_2,b_2} d\ast d\Re\vphi + \Im \alpha_h\wedge \Psi_{w_2,b_2}\right) = O(1)
    \end{split}
  \end{equation}
  because
  \[
    \int_\Sigma \left(\frac{1}{2}\Phi_{w_2,b_2} d\ast d\Re\vphi + \Im \alpha_h\wedge \Psi_{w_2,b_2}\right) = O(\delta).
  \]
  Indeed, $\int_\Sigma \Phi_{w_2,b_2} d\ast d\Re\vphi = O(\delta)$ due to~\eqref{eq:sme14}, and to estimate the second sum one can observe that, by Riemann bilinear relations (see~\eqref{eq:Riemann bilinear relations}) and residue formula, we have
  \begin{multline*}
    \int_\Sigma\Im \alpha_h\wedge \Psi_{w_2,b_2} = K_{\frac{1}{2}\alpha_l}(w_2,b_2)K_{0,\frac{1}{2}\alpha_l}^{-1}(b_2,w_2)\int_{w_2}^{b_2}\alpha_h + \\
    + \eta_w^2K_{\frac{1}{2}\alpha_l}(w_2,b_2)K_{0,\frac{1}{2}\alpha_l}^{-1}(b_2,\sigma(w_2))\int_{\sigma(w_2)}^{\sigma(b_2)}\Im \alpha_h + O(\delta) = O(\delta).   
  \end{multline*}
  Combining~\eqref{eq:sme15} and~\eqref{eq:sme7} with~\eqref{eq:sme5} we conclude~\eqref{eq:sme4} and finish the proof.
\end{proof}

\subsection{Proof of Theorem~\ref{thma:for_BLR}}
\label{subsec:third_thm_BLR_setup}

In this section we prove Theorem~\ref{thma:for_BLR}. Note, the graphs $(G^k)'$ and $(G^{\ast,k})'$ satisfy the assumptions of Section~\ref{subsec:intro_graphs_on_Sigma0} with $\lambda$ and $\delta_k$ chosen as in Section~\ref{subsec:intro_relation_to_BLR}, and the gauge form $\alpha_{G_k}$ can be chosen to be equal to $-\frac{1}{2}\alpha_0$, see~\eqref{eq:alpha_0_vs_alpha_G_Temperley}. Thus, the proof that the limit of $h^k - \EE h^k$ is the primitive of $\m^{-\alpha_0} - \EE \m^{-\alpha_0}$ restricted to $\Sigma_0$ and pulled back to the universal cover $\widetilde{\Sigma}'_0$ is a straightforward application of Theorem~\ref{thma:main1}.

It remains to prove one can extract a subsequence from $\Gamma_0^k, \Gamma_0^{k,\dagger}, G_0^k$ satisfying the assumptions of~\cite{BerestyckiLaslierRayI} with respect to the singular metric $ds^2$ on $\Sigma_0$. The ``bounded density'' and ``good embedding'' assumptions are clearly satisfied by the whole sequence. The rest two assumptions concern the random walk on the graph $\Gamma^k_0$. In fact, both assumptions impose only local restrictions on the random walk, thus we begin by studying the local structure of the random walk on $\Gamma_0^k$. The ``uniform crossing estimate'' assumption follows from~\cite[Lemma~6.8]{CLR1}; the random walk near the conical singularities can be controlled using that $\Gamma_0^k$ is locally a double cover of a Temperleyan isoradial graph.

Let us now show that ``invariance principle'' assumption holds at least along a subsequence. By the construction (see Section~\ref{subsec:intro_relation_to_BLR}, Example~\ref{intro_example:triangular_graphs}), locally outside the conical singularities the graph $\Gamma_0^k$ is a subgraph of a full plane graph obtained in the following way. Let $\lambda,\delta>0$ be fixed, let $\Gamma$ be the Delaunay triangulation associated with a discrete subset of $\CC$ which is a $\lambda^{-1}\delta$-net and $\lambda\delta$ separated, and whose points are in the general position. Let $\Gamma^\dagger$ be the corresponding Voronoi diagram and let $\Tt$ be the t-embedding constructed by midpoints of segments connecting dual vertices of $\Gamma$ and $\Gamma^\dagger$ as in Example~\ref{intro_example:triangular_graphs}. Recall that $\Tt$ is weakly uniform and has $O(\delta)$-small origami (see Section~\ref{subsec:intro_graphs_on_Sigma0} for the precise definition of these properties). The black faces of $\Tt$ correspond to vertices of $\Gamma$ and $\Gamma^\dagger$ and white faces of $\Tt$ correspond to edges of $\Gamma$. As we noticed in Section~\ref{intro_example:alpha_and_temperley}, there is a natural choice of the origami square root function $\eta$ such that 
\[
  \eta_b = 1,\ b\in \Gamma,\qquad \eta_b = i,\ b\in \Gamma^\dagger.
\]
We fix this choice. We define the weights $\mathrm{weight}(b_1b_2)$ on oriented edges of $\Gamma$ following the formula~\eqref{eq:def_of_weight_on_Gammak} defining the weights for $\Gamma_0^k$. The following lemma follows from direct computations (see~\cite[Section~8.1]{CLR1}):

\begin{lemma}
  \label{lemma:Gamma_is_Tgraph}
  If the additive normalization of the origami map $\Oo$ is chosen properly, then the T-graph $\Tt + \bar{\Oo}$ is equal to the Delaunay triangulation $\Gamma$. Moreover, for any edge $b_1b_2$ of $\Gamma$ we have $\mathrm{weight}(b_1b_2) = q(b_1\to b_2)$, where $q$ is the transition rate for the random walk on $\Tt + \bar{\Oo}$ defined as in Section~\ref{subsec:t-embedding_def} for an arbitrary splitting.
\end{lemma}

From Lemma~\ref{lemma:Gamma_is_Tgraph} we see that the random walk on $\Gamma$ coincides with the random walk on the T-graph $\Tt + \bar{\Oo}$, hence we can use the machinery from~\cite{CLR1} to work with it. We immediately get the following

\begin{cor}
  \label{cor:conv_to_BM_full-plane}
  Let $\lambda>0$ be fixed and $\delta_k\to 0$ be a sequence of positive numbers. Let $\Gamma_k$ be a Delaunay triangulation of the plane constructed as above with the given $\lambda$ and $\delta=\delta_k$. Let $b_0^k$ be a sequence of vertices of $\Gamma_k$ approximating $0\in \CC$. Let $X_t^{b_0^k}$ be the continuous time random walk on $\Gamma_k$ started at $b_0^k$. Then any subsequence of $X_t^{b_0^k}$ converging in the Skorokhod topology converges to $B_{\phi(t)}$, where $B_t$ is the standard Brownian motion on $\CC$ started at the origin and $\phi$ is a random continuous increasing function such that $\EE\phi(t) = t$ for each $t$.
\end{cor}
\begin{proof}
  Recall that by Lemma~\ref{lemma:Gamma_is_Tgraph} $X_t^{b_0^k}$ is the random walk on the T-graph $\Tt + \bar{\Oo}$. From the time normalization of $X_t^{b_0^k}$ (see Remark~\ref{rem:variation_of_Xt}) it is easy to deduce that any subsequential limit of $X_t^{b_0^k}$ has continuous trajectories almost surely. Denote by $X_t$ some subsequential limit; we have
  \begin{equation}
    \label{eq:CBM1}
    \Var\Tr X_t^2 = t.
  \end{equation}
  Given $r>0$ let $\tau_r$ be the first time $X_t$ exits the disc $B(0,r)$. We claim that for each harmonic function $h$ in $B(0,r)$ continuous up to the boundary, the process $h(X_{t\wedge \tau_r})$ is a martingale with respect to the filtration generated by $X_t$. To prove it, fix a $\lambda\in \CC, |\lambda|>r,$ and consider the function $h_\lambda(z) = \log|z-\lambda|$. Applying Lemma~\ref{lemma:primitive_of_f} to the inverting kernel from Theorem~\ref{thmas:parametrix} and sending $k$ to $+\infty$ we find out that $h_\lambda(X_{t\wedge \tau_r})$ is a martingale for any $\lambda$. Given an arbitrary harmonic function $h$ on $B(0,r)$ one can construct a sequence of finite linear combinations of $h_{\lambda}'s$ approximating $h$ in the topology of uniform convergence on compacts of $B(0,r)$. This implies that $h(X_{t\wedge \tau_r})$ is a martingale.

  The discussion above implies the following. Let $\tau$ be an arbitrary stopping time and $\tau_\eps = \inf\{ t>\tau\ \mid\ |X_\tau - X_t|\geq \eps \}$. Then, conditioned on $\tau$, the exist point $X_{\tau_\eps}$ is uniformly distributed on the circle $\partial B(X_\tau, \eps)$. One can apply the same arguments as above to the process $X_{\tau_\eps + t}$ conditioned on $X_{\tau_\eps}$ and prove that its exit point from $B(X_{\tau_\eps},\eps)$ is again uniformly distributed on the circle. Repeating this arguments we get a discrete process $X_0,X_{\tau_\eps}, X_{\tau_\eps + \tau_\eps^{(1)}},\dots$, which converges to the standard Brownian motion up to a random time change $\phi$. Equation~\eqref{eq:CBM1} ensures that $\phi$ is continuous and satisfies $\EE\phi(t) = t$.
\end{proof}

Let us now deduce the ``invariance principle'' assumption from this corollary and by taking a suitable subsequence. Recall that we have a multivalued mapping $\Tt$ on $\Sigma$ which is locally one-to-one outside conical singularities and a branched double cover at conical singularities. Using the time normalization of $X^{b_k}_t$ it is easy to deduce the tightness of $\Tt(X^{b_k}_t)$ in the Skorokhod topology, which in fact implies the tightess on $X^{b_k}_t$ itself. Let $X_t$ be any subsequential limit. Uniform crossing estimates imply that $X_t$ almost surely never visits conical singularities. Now, Corollary~\ref{cor:conv_to_BM_full-plane} implies that $X_t = B_{\phi(t)}$ where $B$ is the Brownian motion on $\Sigma$ and $\phi$ is a random continuous increasing function s.t. $\EE \phi(t) = t$ for any $t$.

\begin{appendix}
\section{}%
\label{appendix:apA}

The main goal of this section is to prove the existence of the locally flat metric $ds^2$ promised in Section~\ref{intro_example:torus}, and to fill in the details missing in Section~\ref{sec:The Riemann surface: continuous setting}. This is a technical task, using quite a lot of machinery from the classical theory of Riemann surfaces. For the sake of completeness we decided to make a brief introduction into the necessary parts of this theory. If the reader does not need such an introduction, then we suggest him to jump to Section~\ref{subsec:locally_flat_metric_existence} skipping previous subsections. All the facts stated before this section are classical and can be found in the standard literature such as~\cite{Johnson},~\cite{GriffitsHarris},~\cite{MumfordTata1},~\cite{Fay},~\cite{ImayoshiTaniguchi}.

\subsection{Sheaves and vector bundles}
\label{subsec:sheaves_and_vectore_bundles}

In what follows it will be convenient for us to use the language of sheaves. Thus, we briefly introduce this notion and related things.

Given an arbitrary category $\Cc$, a \emph{presheaf} of objects of this category on a topological space $\Sigma$ is a contravariant functor from the category of open subsets of $\Sigma$ to $\Cc$. In other words, to specify a presheaf $\Ff$ we have to choose an object $\Gamma(U,\Ff)\in \mathrm{Ob}(\Cc)$ for each open $U\subset \Sigma$ and a morphism $\vphi_{U,V}: \Gamma(U,\Ff)\to \Gamma(V,\Ff)$ for each pair $U\supset V$ such that
\[
  \vphi_{U,U} = \Id,\qquad \vphi_{U,V}\circ\vphi_{V,W} = \vphi_{U,W},\quad \text{if }V\supset W.
\]
The object $\Gamma(U,\Ff)$ is called the space of sections of $\Ff$ over $U$ and the object $\Gamma(\Sigma,\Ff)$ is called the space of global sections. A presheaf $\Ff$ is called a \emph{sheaf} if for any set of indices $I$ and a cover of a set $U$ by sets $\{ U_i \}_{i\in I}$ one has
\begin{itemize}
  \item if $u,v\in \Gamma(U,\Ff)$ and for any $i\in I$ we have $\vphi_{U,U_i}(u) = \vphi_{U,U_i}(v)$, then $u=v$;
  \item if we are given a collection $u_i\in \Gamma(U_i,\Ff)$ such that for any $i,j\in I$ the compatibility relation $\vphi_{U_i,U_i\cap U_j}(u_i) = \vphi_{U_j,U_i\cap U_j}(u_j)$ holds, then there exists an $u\in \Gamma(U,\Ff)$ such that $u_i = \vphi_{U,U_i}(u)$ for any $i\in I$.
\end{itemize}
Note that given a category $\Cc$, the sheaves on $\Sigma$ themself form a category, where the morphisms are defined in a natural way.

If $\Sigma$ is a Riemann surface, then there is a certain amount of natural sheaves on it. First to come is the \emph{structure sheaf} $\Oo_\Sigma$ with spaces of sections given by
\[
  \Gamma(U, \Oo_\Sigma) = \{ f: U\to \CC\ \mid\ f\text{ is holomorphic} \}.
\]
In this case the underlying category may be taken to be the category of commutative rings, or algebras over $\CC$. Next, given a holomorphic vector bundle $E\to \Sigma$ we can introduce the sheaf $\Ee$ of its sections:
\[
  \Gamma(U,\Ee) = \{ f: U\to E\vert_U\ \mid\ f\text{ is a holomorphic section of $E$ over $U$} \}.
\]
Note that the space $\Gamma(U,\Ee)$ is a module over the ring $\Gamma(U, \Oo_\Sigma)$; in this case the sheaf $\Ee$ is called a sheaf of modules over the structure sheaf and the underlying category is the category of modules over commutative rings.

The important property of the sheaf $\Ee$ is that the vector bundle $E$ itself can be reconstructed from it. For, we say that a sheaf $\Ee$ of modules over the structure sheaf is locally free if any $p\in \Sigma$ has a neighborhood $U$ such that $\Ee\vert_U$ is isomorphic to a direct sum of $r$ copies of $\Oo_U$. Given that $\Sigma$ is connected, we see that the number $r$ does not depend on the point $p$. We say that $\Ee$ is a locally free sheaf of rank $r$ in this case. The following proposition is standard

\begin{prop}
  \label{prop:bundles_are_locally_free_sheaves}
  A sheaf $\Ee$ of modules over the structure sheaf is locally free of rank $r$ if and only if it is isomorphic to the sheaf of sections of some holomorphic vector bundle $E$ of rank $r$. Moreover, the correspondence between isomorphism classes of vector bundles and locally free sheaves is one-to-one and functorial in $\Sigma$.
\end{prop}
\begin{proof}
  Let us sketch the proof of this proposition. The fact that the sheaf of sections of any vector bundle is a locally free sheaf of the corresponding rank is straightforward. Conversely, let $\Ee$ be a locally free sheaf of rank $r$. Take a $p\in \Sigma$ and a small neighborhood $U$ of $p$, let $\Ii_{U,p}\subset \Oo_\Sigma(U)$ be the ideal consisting of functions vanishing at $p$. Then it is easy to see that the vector space quotient
  \[
    E_p = \frac{\Gamma(U, \Ee)}{\Ii_{U,p}\cdot \Gamma(U,\Ee)}
  \]
  does not depend on the choice of $U$ (provided $U$ is small enough) and is a $\CC$-vector space of rank $r$. It is straightforward to show that the family of vector spaces $E_p,\ p\in \Sigma,$ form a vector bundle of rank $r$ and $\Ee$ is isomorphic to the sheaf of sections of it.
\end{proof}

From now on we will not make a difference between holomorphic vector bundles and locally free sheaves, abusing the notation slightly. 

Let $D = \sum_{i = 1}^d m_i\cdot p_i$ be an arbitrary divisor on $\Sigma$; here $p_1,\dots, p_d\in \Sigma$ are some points and $m_1,\dots, m_d\in \ZZ$. Let $\Oo_\Sigma(D)$ be the sheaf given by
\begin{equation}
  \label{eq:def_of_Oo(D)}
  \Gamma(U, \Oo_\Sigma(D)) = \{ f\text{ --- meromorphic on $U$ and $\div f \geq - D\cap U$}  \}.
\end{equation}
It is straightforward to see that $\Oo_\Sigma(D)$ is a locally free sheaf of rank 1, that is $\Oo_\Sigma(D)$ is a line bundle on $\Sigma$. One can easily check that for any two divisors $D_1$ and $D_2$ we have $\Oo_\Sigma(D_1)\otimes \Oo_\Sigma(D_2)\cong \Oo_\Sigma(D_1 + D_2)$, where the tensor product means the tensor product of the corresponding line bundles. In particular, there is a natural isomorphism $\Oo_\Sigma(-D)\cong\Oo_\Sigma(D)^\vee$, where $\Oo_\Sigma(D)^\vee$ is the sheaf of sections of the dual (i.e. obtained by taking the dual vector spaces fiber-wise) line bundle. If $D_1,D_2$ are two divisors such that $D_1\leq D_2$, then there is a natural non-zero morphism $\Oo_\Sigma(D_1)\to \Oo_\Sigma(D_2)$. Note that this morphism is injective as the morphism of sheaves (i.e. sends non-zero sections to non-zero sections), but in general is not injective on the level of line bundles (i.e. vanishes on some fibers). In particular, we have a natural morphism $\Oo_\Sigma\to \Oo_\Sigma(D)$ for any positive $D$. 

Another natural sheaf is the canonical sheaf $K_\Sigma$ of $\Sigma$. It is defined by
\[
  \Gamma(U, K_\Sigma) = \{ \omega\text{ --- holomorphic $(1,0)$-form on U} \}.
\]
Again, it is clear that $K_\Sigma$ is a locally free sheaf of rank 1. As a line bundle, it coincides with the holomorphic cotangent bundle of the surface $\Sigma$.

\subsection{Quadratic forms over \texorpdfstring{$\ZZ/2\ZZ$}{Z/2Z}}
\label{subsec:quadratic_forms}

Let $V$ be a vector space over $\ZZ/2\ZZ$ of dimension $2g$ with a non-degenerate skew-symmetric (i.e. $a\cdot a = 0$) bilinear form. A function $q: V\to \ZZ/2\ZZ$ is called a quadratic form if it satisfies
\begin{equation}
  \label{eq:quadratic_relation}
  q(a+b) = q(a) + q(b) + a\cdot b
\end{equation}
for all $a,b\in V$. The space of quadratic forms on $V$ is an affine space over the dual space $V^\vee$: for each $q$ and $l\in V^\vee$ the function $q+l$ is again a quadratic form. In particular, if there are precisely $2^{2g}$ quadratic forms on $V$.

The symplectic group of $(V,\cdot)$ acts on the space of quadratic forms of $V$. It is well-known that two forms belong to the same orbit if and only if they have the same Arf invariant which is defined as follows. Fix a simplicial basis $A_1,\dots, A_g, B_1,\dots, B_g$ (so that $A_i\cdot A_j = B_i\cdot B_j = 0$ and $A_i\cdot B_j = \delta_{ij}$) and put
\begin{equation}
  \label{eq:def_of_Arf}
  \Arf(q) = \sum_{i = 1}^g q(A_i)q(B_i).
\end{equation}
Quadratic form $q$ is called \emph{even} if $\Arf(q) = 0$, otherwise it is called \emph{odd}. We have the following lemma:
\begin{lemma}
  \label{lemma:shifted_forms}
 Let $q_0$ be the quadratic form satisfying $q_0(A_i) = q_0(B_i) = 0$ for $i = 1,\dots, g$. Then for every other quadratic form $q$ there exists a $u\in V$ such that
 \[
   q(v) = q_0(v+u) + \Arf(q) = q_0(v+u) + q_0(u).
 \]
 for each $v\in V$.
\end{lemma}
\begin{proof}
  There exists a $u\in V$ such that $q(v) = q_0(v) + u\cdot v$ for each $v\in V$. Note that this can be rewritten as
\begin{equation}
  \label{eq:shifted_form}
  q(v) = q_0(v) + u\cdot v = q_0(u+v) + q_0(u)
\end{equation}
using~\eqref{eq:quadratic_relation}. Write $u = \sum_{i = 1}^g(a_iA_i + b_iB_i)$. Then we have
\begin{equation}
  \label{eq:Arf(q_u)}
  q_0(u) = \sum_{i = 1}^ga_ib_i = \Arf(q_u),
\end{equation}
where we used~\eqref{eq:quadratic_relation} to write the first equation and the relations $q(A_i) = b_i$ and $q(B_i) = a_i$ to write the second equation. A combination of~\eqref{eq:shifted_form} and~\eqref{eq:Arf(q_u)} proves the lemma.
\end{proof}

\subsection{Spin line bundles and quadratic forms on \texorpdfstring{$H_1(\Sigma, \ZZ/2\ZZ)$}{H1(Sigma, Z/2Z)}}
\label{subsec:spin_line_bundles}
Let $\Sigma$ be a smooth Riemann surface of genus $g$. A spin line bundle is, by definition, a holomorphic line bundle $\Ff\to \Sigma$ and an isomorphism $\beta:\Ff^{\otimes 2}\cong K_\Sigma$. Classically~\cite{MumfordThetaCharacteristics, Atiyah}, there are precisely $2^{2g}$ isomorphism classes of spin bundles on $\Sigma$, classified by quadratic forms in $H_1(\Sigma, \ZZ/2\ZZ)$. In this section we review this correspondence in some details.

We begin with the relation between quadratic forms and spin structures established by Johnson~\cite{Johnson}. In what follows we will consider quadratic forms on $V = H_1(\Sigma, \ZZ/2\ZZ)$ taken with the intersection product (see Section~\ref{subsec:quadratic_forms}). In this case $V^\vee = H^1(\Sigma, \ZZ/2\ZZ)$.

Let $UT\Sigma$ be the \emph{unit tangent bundle}, which is obtained by removing the zero fiber from the total space of the tangent bundle $T\Sigma$. Let $z\in H_1(UT\Sigma, \ZZ/2\ZZ)$ denote the non-zero element in the kernel of the natural map $H_1(UT\Sigma, \ZZ/2\ZZ)\to H_1(\Sigma, \ZZ/2\ZZ)$.
\begin{defin}
  \label{defin:def_of_spin_structure}
  A class $\xi\in H^1(UT\Sigma, \ZZ/2\ZZ)$ is a \emph{spin structure} on if $\xi(z) = 1$.
\end{defin}
Clearly, the set of spin structures is affine over $H^1(\Sigma, \ZZ/2\ZZ)$ (identified with its natural image in $H^1(UT\Sigma, \ZZ/2\ZZ)$). Note that any loop in $UT\Sigma$ corresponds to a loop on $\Sigma$ with a vector field along it. Given a smooth oriented loop $\gamma$ on $\Sigma$ denote by $\tilde\gamma$ the loop in $UT\Sigma$ corresponding to $\gamma$ with the tangent frame on it. In~\cite{Johnson}, the following theorem is proven:

\begin{thmas}[Johnson]
  \label{thmas:Johnson_spin_structures}
  Given a spin structure $\xi\in H^1(UT\Sigma, \ZZ/2\ZZ)$ and a smooth simple loop $\gamma$ on $\Sigma$ define $q_\xi(\gamma) = \xi(\tilde\gamma)+1\mod 2$. Then $q$ depends only on the homology of $\gamma$ and extends to a quadratic form $q_\xi: H_1(\Sigma,\ZZ/2\ZZ)\to \ZZ/2\ZZ$. The correspondence $\xi\mapsto q_\xi$ is an affine isomorphism between the set of spin structures and quadratic forms on $H_1(\Sigma, \ZZ/2\ZZ)$. 
\end{thmas}

The following proposition is folklore.

\begin{prop}
  \label{prop:topological_charact_of_spin_bundles}
  The following sets are in natural bijection:
  \[
    \xymatrix{
      {\begin{array}{c}
        \text{Isomorphism classes}\\
        \text{of spin line bundles }\Ff\to \Sigma 
      \end{array}}
      \ar@{<->}[d]^{F: \Ff^\vee \xrightarrow{\mathrm{diag}} (\Ff^\vee)^{\otimes 2}\xrightarrow{\cong} T\Sigma}\\
      {\begin{array}{c}
        \text{Isomorphism classes of double covers }F:V\to UT\Sigma\\
        \text{s.t. preimages of fibers are connected}
      \end{array}}
      \ar@{<->}[d]^{\xi \text{ is cohomology class of double cover}}\\
      {\text{The set of spin structures $\xi\in H^1(UT\Sigma, \ZZ/2\ZZ)$}}
      \ar@{<->}[d]^{q = \xi + 1}\\
      {\text{The set of quadratic forms $q$ on $H_1(\Sigma, \ZZ/2\ZZ)$}}
    }
  \]
\end{prop}
\begin{proof}
  Let us give descriptions of the bijections. Given a spin line bundle $\Ff\to \Sigma$ we define $F: U\Ff^\vee\to UT\Sigma$ as a composition of the diagonal map $\Ff^\vee\to (\Ff^\vee)^{\otimes 2}$ and the isomorphism between $(\Ff^\vee)^{\otimes 2}$ and the tangent bundle $T\Sigma$. Vise versa, given a double cover $F:V\to UT\Sigma$ such that the preimages of fibers are connected, one can reconstruct the spin line bundle $\Ff$ uniquely up to isomorphism. Given a double cover $F:V\to UT\Sigma$ we take $\xi\in H^1(UT\Sigma, \ZZ/2\ZZ)$ to be its cohomology class; this gives the bijection between spin structures and isomorphism classes of double covers such that preimages of fibers are connected. Finally, the correspondence between spin structures and quadratic forms is made by applying Theorem~\ref{thmas:Johnson_spin_structures}.

\end{proof}

Recall that given a holomorphic vector bundle $\Ll\to \Sigma$ we denote by $\Gamma(\Sigma, \Ll)$ the space of its holomorphic sections. The following theorem was proven by Johnson~\cite{Johnson} based on the results of Atiyah~\cite{Atiyah} and Mumford~\cite{MumfordThetaCharacteristics}:
\begin{thmas}
  \label{thmas:Arf_and_h0}
  Let $q$ be a quadratic form on $H_1(\Sigma, \ZZ/2\ZZ)$ and $\Ff$ be the corresponding spin line bundle on $\Sigma$. Then
  \[
    \Arf(q) = \dim \Gamma(\Sigma, \Ff) \mod 2.
  \]
\end{thmas}
Following this theorem, we call a spin line bundle odd or even if the corresponding quadratic form is odd or even. Theorem~\ref{thmas:Arf_and_h0} tells us in particular that odd spin line bundles always admit non-zero sections. Moreover, it can be shown that if $\Sigma$ is chosen generic and $g\geq 3$, then $\dim \Gamma(\Sigma, \Ff)$ is either 0 or 1 for any spin line bundle $\Ff$.

Let us now explore in details how the spin structure of a spine line bundle is related with ``windings'' of its smooth sections.
Let $\omega$ be a smooth (not necessary holomorphic) $(1,0)$-form on $\Sigma$ and $\gamma$ be a smooth oriented loop on $\Sigma$ such that $\omega$ does not vanish along $\gamma$. Let $r:[0,1]\to \Sigma$ be a smooth parametrization of $\gamma$. Define the winding of $\gamma$ with respect to $\omega$ by
\begin{equation}
  \label{eq:def_of_wind}
  \wind(\gamma, \omega) = \Im \int_0^1 \frac{d}{dt} \log \omega(r'(t))\,dt.
\end{equation}

We have the following

\begin{lemma}
  \label{lemma:wind_lift}
  Let $\xi\in H^1(UT\Sigma, \ZZ/2\ZZ)$ be a spin structure and $\Ff$ a spin line bundle corresponding to $\xi$ under the bijection from Proposition~\ref{prop:topological_charact_of_spin_bundles}. Let $Z\subset \Sigma$ be a finite set and $\omega$ be a smooth (not necessary holomorphic) $(1,0)$-form on $\Sigma\smm Z$ vanishing nowhere. Then $\omega$ is an image of a smooth section of $\Ff\vert_{\Sigma\smm Z}$ if and only if for any simple smooth oriented curve $\gamma$ on $\Sigma\smm Z$ we have
  \begin{equation}
    \label{eq:wind_lift1}
    \xi(\tilde\gamma) = (2\pi)^{-1}\wind(\gamma, \omega)\mod 2.
  \end{equation}
\end{lemma}

\subsection{Basis in \texorpdfstring{$H_1(\Sigma,\ZZ)$}{H1(Sigma, Z)}}
\label{subsec:simplicial_basis}

The space $H_1(\Sigma,\ZZ)$ has a natural non-degenerate simplicial form given by the intersection product. We fix a simplicial basis 
\[
  A_1,\dots,A_g, B_1,\dots, B_g\in H_1(\Sigma, \ZZ)
\]
with respect to this form fixed by the condition
\[
  A_i\cdot B_j = \delta_{ij},\qquad A_i\cdot A_j = B_i\cdot B_j = 0.
\]
With some abuse of the notation we can assume that $A_i$'s and $B_i$'s are also oriented simple closed curves on $\Sigma$ representing the corresponding homology classes.

Assume that $\Sigma$ is a double of a Riemann surface $\Sigma_0$ and $\sigma:\Sigma\to \Sigma$ is the corresponding anti-holomorphic involution as in Section~\ref{subsec:intro_flat_metric}. Assume that $n$ is the number of boundary components of $\Sigma_0$ and $g_0 = g(\Sigma_0)$, then we have $g = 2g_0 + n-1$. In this setting we can choose the basis above in such a way that
\begin{equation}
  \label{eq:symmetries_of_simplicial_basis}
  \sigma_*A_i = -A_i,\quad \sigma_* B_i = B_i,\quad i = 1,\dots,g.
\end{equation}

Having a simplicial basis fixed we can introduce \emph{normalized} holomorphic differentials, that is the basis $\omega_1,\dots,\omega_g\in H^0(\Sigma, K_\Sigma)$ normalized by the condition
\[
  \int_{A_i} \omega_j = \delta_{ij},\qquad i,j = 1,\dots,g.
\]
If the involution $\sigma$ is present, is straightforward to see that
\begin{equation}
  \label{eq:symmetries_of_normalized_differentials}
  \begin{split}
    & \sigma^*\omega_i = -\bar{\omega}_i,\quad i = 1,\dots, n-1,\\
    & \sigma^*\omega_{n-1+i} = -\bar{\omega}_{n-1+g_0 + i},\quad i = 1,\dots, g_0.
  \end{split}
\end{equation}
The matrix of $b$-periods $\Omega = (\Omega_{i,j})_{i,j = 1,\dots,g}$ is defined by
\[
  \Omega_{i,j} = \int_{B_i}\omega_j = \int_{B_j}\omega_i,
\]
where the last equality follows from Riemann bilinear relations (see~\cite[Chapter~II.2]{GriffitsHarris}):
\begin{equation}
  \label{eq:Riemann bilinear relations}
  \int_\Sigma u\wedge v = \sum_{i = 1}^g\left( \int_{A_i}u\cdot \int_{B_i}v - \int_{A_i}v\cdot \int_{B_i}u\right)
\end{equation}
for any harmonic 1-differentials $u,v$. The matrix $\Omega$ is symmetric and has positive imaginary part, which corresponds to the natural Hermitian product on $H^0(\Sigma, K_\Sigma)$:
\[
  \Im\Omega_{i,j} = \frac{1}{2}\int_\Sigma\omega_i\wedge*\omega_j = \frac{i}{2} \int_\Sigma\omega_i\wedge\bar{\omega}_j,
\]
where $*$ is the Hodge star.

Assume that the involution $\sigma$ is present; then it is easy to see that
\begin{equation}
  \label{eq:Omega_under_involution}
  J\Omega J = -\overline{\Omega}
\end{equation}
where $J$ is the permutation matrix given by 
\begin{equation*}
  \label{eq:def_of_J}
  \begin{split}
    & J_{i,i} = 1,\quad i = 1,\dots, n-1,\\
    & J_{g-2g_0+i,g-g_0 + i} = J_{g-g_0+i, g-2g_0 + i} = 1,\quad i = 1,\dots, g_0.
  \end{split}
\end{equation*}

We finalize this section by recalling some basic facts about harmonic differentials on $\Sigma$. Let $H^{1,0}(\Sigma)$ and $H^{0,1}(\Sigma)$ denote the spaces of holomorphic $(1,0)$-forms and anti-holomorphic $(0,1)$-forms on $\Sigma$ respectively. Then $H^{1,0}(\Sigma)\oplus H^{0,1}(\Sigma)$ is the space of harmonic 1-forms on $\Sigma$. The Hodge decomposition (see~\cite{GriffitsHarris}) provides an isomorphism between $H^{1,0}(\Sigma)\oplus H^{0,1}(\Sigma)$ and $H^1(\Sigma, \CC)$, where each $u\in H^{1,0}(\Sigma)\oplus H^{0,1}(\Sigma)$ is sent to its cohomology class. The Hodge star is an involution on $H^{1,0}(\Sigma)\oplus H^{0,1}(\Sigma)$ given by $*u = i\bar{u}$ for $u\in H^{1,0}(\Sigma)$. The skew-symmetric form $(u,v)\mapsto \int_\Sigma u\wedge v$ coincides with the cap product on $H^1(\Sigma, \CC)$, and the bilinear form 
\begin{equation}
  \label{eq:scalar_product_of_harmonic_differentials}
  (u,v)\mapsto \int_\Sigma u\wedge *v
\end{equation}
defines a scalar product on $H^{1,0}(\Sigma)\oplus H^{0,1}(\Sigma)$.

\subsection{Families of Cauchy--Riemann operators and the Jacobian of a Riemann Surface}
\label{subsec:Jacobian}

Recall that a line bundle is a vector bundle of rank 1. Since $H^2(\Sigma, \ZZ)\cong \ZZ$, the first Chern class of any line bundle is just an integer, called the degree of the bundle. We have for example
\[
  \deg \Oo_\Sigma(D) = \deg D,\qquad \det K_\Sigma = 2g-2.
\]
The degree of a line bundle is the unique topological invariant: any two line bundles of the same degree are isomorphic as $\mC^\infty$ line bundles. But when $g\geq 1$, each $\mC^\infty$ complex line bundle has infinitely many complex structures on it, parametrized by a $g$-dimensional complex torus called the Jacobian of $\Sigma$ and denoted by $\Jac(\Sigma)$. To describe the Jacobian, it is enough to describe complex structures on the trivial line bundle.

Let $\Sigma\times \CC$ be the trivial line bundle. Given an open $U\subset \Sigma$, the space of smooth sections of $\Sigma\times \CC$ over $U$ is just the function space $\mC^\infty(U)$. The Cauchy--Riemann operator $\dbar$ acts on $C^\infty(U)$ naturally; its kernel is the space of holomorphic sections of $\Sigma\times \CC$ over $U$. This endows $\Sigma\times \CC$ with a \emph{complex structure} --- for each $U$ we know which smooth sections are holomorphic. Denote the corresponding \emph{holomorphic} line bundle by $\Ll_0$ In the language of sheaves the holomorphc line bundle $\Ll_0$ is described as
\[
  \Gamma(U,\Ll_0) = \{ f\in \mC^\infty(U)\ \mid\ \dbar f = 0 \}.
\]
Now, let $\alpha$ be a smooth $(0,1)$-form on $\Sigma$ and consider the perturbed operator $\dbar + \alpha$. Replacing $\dbar$ with $\dbar + \alpha$ we obtain another holomorphic line bundle $\Ll_\alpha$, with the underlying sheaf described as 
\begin{equation}
  \label{eq:def_of_Lalpha}
  \Gamma(U, \Ll_\alpha) = \{ f\in \mC^\infty(U) \ \mid\ (\dbar + \alpha)f = 0 \}.
\end{equation}
To see that this is a locally free sheaf of rank 1, write $\alpha = \dbar \vphi + \alpha_h$, where $\vphi\in C^\infty(\Sigma)$ and $\alpha_h$ is an anti-holomorphic $(0,1)$-form on $\Sigma$; this is always possible by Daulbeaut decomposition. Assume that $U$ is simply-connected, so that a primitive $\int\alpha_h$ is defined on $U$. Then it is straightforward to see that
\[
  f\in \Gamma(U, \Ll_\alpha)\qquad \Longleftrightarrow \qquad e^{\vphi + \int\alpha_h}\cdot f\text{ is holomorphic.}
\]
From this it is clear that we have defined a locally free sheaf of rank 1. In the next lemma we show that $\Ll_\alpha$ exhaust all the possible complex structures on the trivial bundle.

\begin{lemma}
  \label{lemma:deg0_bundles_via_connections}
  Let $\Sigma$ be a smooth closed Riemann surface. Then for any holomorphic line bundle $\Ll\to \Sigma$ of degree zero there exists an antiholomorphic $(0,1)$-form $\alpha$ such that $\Ll$ is isomorphic to $\Ll_\alpha$ as a holomorphic line bundle.
\end{lemma}
\begin{proof}
  Denote by $C^\infty(U,\Ll)$ the space of smooth sections of $\Ll$ over $U$. Let $U\subset \Sigma$ be simply-connected and open, and let $\phi\in \Gamma(U,\Ll)$ be a non-vanishing holomorphic section. Then any other smooth section over $U$ has the form $f\phi$, where $f\in \mC^\infty(U)$. Define $\dbar(f\phi)= \dbar f\cdot \phi$ (viewed as $\Ll$-valued differential form). Since $\Ll$ is holomorphic, this definition does not depend on the choice of $\phi$ and extends to any open set $U$ via a partition of unity. Note that if $\phi\in \mC^\infty(\Sigma,\Ll)$, then $\dbar\phi$ vanishes over $U$ if and only if $\phi\vert_U$ is a holomorphic section.

  Since $\deg\Ll = 0$, it is trivial as a smooth bundle; in particular, there exists a nowhere vanishing $\phi_0\in C^\infty(\Sigma, \Ll)$. Any other smooth section $\phi$ is of the form $\phi = f\phi_0$ for some $f\in C^\infty(\Sigma)$. Note that
  \[
    \dbar(f\phi_0) - \dbar f\cdot \phi_0
  \]
  is linear in $f$. It follows that there exists a smooth $(0,1)$-form $\alpha$ such that
  \begin{equation}
    \label{eq:deg0_bundles_via_connections1}
    \dbar (f\phi_0) = (\dbar + \alpha)f\cdot \phi_0
  \end{equation}
  Let us write $\alpha = \dbar \vphi + \alpha_h$ where $\alpha_h$ is antiholomorphic. For each open $U\subset \Sigma$ define $\Phi_U: \Gamma(U, \Ll_{\alpha_h})\to \Gamma(U,\Ll)$ (recall that $\Gamma(U, \Ll)$ is the space of holomorphic sections) by
  \[
    \Phi_U(f) = e^{-\vphi} f\phi_0.
  \]
  We have $\dbar(e^{-\vphi} f \phi_0) = e^{-\vphi}(\dbar + \alpha_h)f\cdot \phi_0 = 0$, hence $\Phi_U$ is defined correctly. It is easy to see that $\Phi_U$ descends to an isomorphism between $\Ll_{\alpha_h}$ and $\Ll$.
\end{proof}

It follows from Lemma~\ref{lemma:deg0_bundles_via_connections} that isomorphism classes of line bundles of degree 0 are parametrized by anti-holomorphic $(0,1)$-forms. It is natural to ask when two different forms define the same isomorphism class. 
\begin{lemma}
  \label{lemma:alpha_1_sim_alpha_2}
  Let $\alpha_1,\alpha_2$ be two antiholomorphic $(0,1)$-forms. Then the holomorphic line bundles $\Ll_{\alpha_1}$ and $\Ll_{\alpha_2}$ are isomorphic if and only if all periods of $\pi^{-1}\Im (\alpha_1-\alpha_2)$ are integer.
\end{lemma}
\begin{proof}[Sketch of a proof]
  If all periods of $\pi^{-1}\Im (\alpha_1-\alpha_2)$ are integer, then the isomorphism between $\Ll_{\alpha_1}$ and $\Ll_{\alpha_2}$ is given by the multiplication by the function $\exp(2i\int \Im(\alpha_1-\alpha_2))$. For the converse statement see~\cite[p.314]{GriffitsHarris}.
\end{proof}

Let $H^{0,1}(\Sigma)$ denote the vector space of all anti-holomorphic $(0,1)$-forms on $\Sigma$ and let
\begin{equation}
  \label{eq:def_of_Lambda}
  \Lambda = \{ \alpha\in H^{0,1}(\Sigma)\ \mid\ \pi^{-1}\Im \alpha\text{ has integer periods} \}.
\end{equation}
Let also
\begin{equation}
  \label{eq:def_of_Pic0}
  \Pic^0(\Sigma) = \{ \text{isomorphism classes of $\deg$ 0 holomorphic line bundles on }\Sigma \}
\end{equation}
Lemmas~\ref{lemma:deg0_bundles_via_connections}~and~\ref{lemma:alpha_1_sim_alpha_2} imply that we have a bijection
\begin{equation}
  \label{eq:Pic_via_01}
  \Psi: \Pic^0(\Sigma) \to \frac{H^{0,1}(\Sigma)}{\Lambda}.
\end{equation}

Using the simplicial basis chosen in Section~\ref{subsec:simplicial_basis}, we can introduce explicit coordinates on $\frac{H^{0,1}(\Sigma)}{\Lambda}$. Recall that $\omega_1,\dots,\omega_g$ is the basis of normalized differentials, see Section~\ref{subsec:simplicial_basis}. Consider the mapping
\begin{equation}
  \label{eq:def_of_Phi}
  \Phi: H^{0,1}(\Sigma)\to \CC^g,\qquad \Phi(\alpha) = (2\pi i)^{-1}(\int_\Sigma \omega_1\wedge \alpha,\dots, \int_\Sigma \omega_g\wedge \alpha).
\end{equation}
Using the fact that the bilinear form~\eqref{eq:scalar_product_of_harmonic_differentials} is non-degenerate, it is easy to see that $\Phi$ is an isomorphism. Note that $\omega\wedge \alpha = 2i\omega\wedge \Im\alpha$ for any $(1,0)$-form $\omega$. It follows immediately from~\eqref{eq:Riemann bilinear relations} that
\begin{equation}
  \label{eq:Phi_alpha_via_a_b}
  \Phi(\alpha) = \pi^{-1}(\int_\Sigma \omega_1\wedge \Im\alpha,\dots, \int_\Sigma \omega_g\wedge \Im\alpha) = b-a\Omega,
\end{equation}
where $a,b\in \RR^g$ are the vectors $A$- and $B$-periods of $\pi^{-1}\Im \alpha$ respectively. In particular
\[
   \Phi(\alpha)\in \ZZ^g + \ZZ^g\Omega \quad \Leftrightarrow \quad \alpha\in \Lambda.
\]
We conclude that $\Phi$ induces an isomorphism
\begin{equation}
  \label{eq:iso_between_Pic_and_complex_torus}
  \Phi:\frac{H^{0,1}(\Sigma)}{\Lambda}\to \frac{\CC^g}{\ZZ^g + \ZZ^g \Omega} = \Jac(\Sigma).
\end{equation}
The torus on the right-hand side of~\eqref{eq:iso_between_Pic_and_complex_torus} is called the \emph{Jacobian} of the surface $\Sigma$ and is denoted by $\Jac(\Sigma)$. From~\eqref{eq:Pic_via_01} and~\eqref{eq:iso_between_Pic_and_complex_torus} we deduce that there is a bijection
\begin{equation}
  \label{eq:def_of_tilde_Aa}
  \Phi\circ\Psi: \Pic^0(\Sigma)\to \Jac(\Sigma).
\end{equation}

Another way to describe $\Pic^0(\Sigma)$ is using the group of divisors. Any holomorphic line bundle $\Ll\to \Sigma$ admits a global meromorphic section (see~\cite[Chapter~I.2]{GriffitsHarris}) $\phi$. Let $D = \div \phi$. Any holomorphic section of $\Ll$ over $U$ is of the form $f\phi$, where $f$ is a meromorphic function on $U$ such that $\div f\geq -D\cap U$. By mapping such a function $f$ to $f\phi$ we obtain an isomorphis $\Ll\cong \Oo_{\Sigma}(D)$ (cf.~\eqref{eq:def_of_Oo(D)}). It follows that we have a natural bijection
\begin{equation}
  \label{eq:def_of_Div0}
  \Pic^0(\Sigma) \cong \Div^0(\Sigma) = \{ D\ \mid\ \deg D = 0 \}/_\sim,
\end{equation}
where $D_1\sim D_2$ if $\Oo_{\Sigma}(D_1)\cong \Oo_\Sigma(D_2)$.
The corresponding map between $D^0(\Sigma)$ and $\Jac(\Sigma)$ is called the Abel map. It has the following description. Let
\[
  D = \sum_{i = 1}^m k_ip_i,\qquad \sum_{i = 1}^m k_i = 0
\]
be a divisor of degree 0. Let $2N = \sum_{i = 1}^m |k_i|$ and $\gamma = \gamma_1\cup\ldots\cup \gamma_N$ be the union of some oriented curves on $\Sigma$ such that $\partial \gamma = D$. Define the \emph{Abel map}
\begin{equation}
  \label{eq:def_of_Abel_map}
  \Aa(D) = \left( (\int_\gamma \omega_1,\dots,\int_\gamma\omega_g)\mod \ZZ^g + \ZZ^g\Omega\right)\in \Jac(\Sigma),
\end{equation}
where $\int_\gamma = \int_{\gamma_1} + \ldots + \int_{\gamma_N}$. It is clear that $\Aa(D)$ does not depend on $\gamma$.

\begin{lemma}
  \label{lemma:PsiPhi=Aa}
  Let $D$ be a divisor on $\Sigma$ such that $\deg D = 0$. If $\alpha\in H^{0,1}(\Sigma)$ is such that $\Ll_\alpha \cong \Oo_\Sigma(D)$, then
  \[
    \Phi(\alpha) = \Aa(D).
  \]
\end{lemma}
\begin{proof}
  The fact that $\Ll_\alpha \cong \Oo_\Sigma(D)$ means that there is a meromorphic section $\vphi$ of $\Ll_\alpha$ with $\div \vphi = D$. By the construction of $\Ll_\alpha$, $\vphi$ is a function on $\Sigma\smm D$ satisfying $(\dbar + \alpha)\vphi = 0$ and having prescribed singularities at the support of $D$. Write
  \[
    D = \sum_{i = 1}^m k_ip_i,\qquad \sum_{i = 1}^m k_i = 0
  \]
  and let $\gamma$ be as above.
  Let us construct the function $\vphi$ explicitly. By Riemann-Roch theorem, there exists a meromorphic $(1,0)$-form $\omega_D$ with simple poles at $p_i$'s and $\Res_{p_i}\omega_D = k_i$. Let the basis cycles $A_1,\dots,A_g,B_1,\dots,B_g$ be represented by simple curves not intersecting $\gamma$. By substracting from $\omega_D$ an appropriate linear combination of $\omega_i$'s we can assume that $\int_{A_i}\omega_D = 0$ for any $i = 1,\dots, g$. Let $u$ be the harmonic differential such that
  \begin{equation}
    \label{eq:PsiPhi=Aa1}
    \int_{A_j}u = 0,\quad \int_{B_j} u = \int_{B_j} \omega_D,\qquad j = 1,\dots,g.
  \end{equation}
  Fix a reference point $p_0\in \Sigma$ and consider the function
  \[
    \vphi(p) = \exp(\int\limits_{p_0}^p (\omega_D - u)).
  \]
  Define
  \[
    \alpha = u^{0,1}.
  \]
  Then it is clear that $\vphi$ defines a meromorphic section of $\Ll_\alpha$ with $\div \vphi = D$, hence $\Oo_\Sigma(D)\cong \Ll_\alpha$.

  A straightforward repetition of the proof of Riemann bilinear relations implies that
  \[
    \Aa(D) = (2\pi i)^{-1} (\int_{B_1}\omega_D,\dots,\int_{B_g}\omega_D) 
  \]
  On the other hand, applying~\eqref{eq:Riemann bilinear relations} and~\eqref{eq:PsiPhi=Aa1} we get that
  \[
    \Phi(\alpha) = (2\pi i)^{-1}(\int_\Sigma \omega_1\wedge u,\dots,\int_\Sigma \omega_g\wedge u) = (2\pi i)^{-1} (\int_{B_1}\omega_D,\dots,\int_{B_g}\omega_D).
  \]
  We conclude that $\Phi(\alpha) = \Aa(D)$.
\end{proof}

We can summarize the discussion in the following commutative diagram of bijections:
\begin{equation}
  \label{eq:different_Pic}
  \vcenter{\xymatrix{& \Pic^0(\Sigma) \ar[d]^\cong \ar[dr]^\Psi & \\
  \Div^0(\Sigma)\ar[ur]^{D\mapsto \Oo_\Sigma(D)} \ar[r]^{\Aa} & \Jac(\Sigma) & \frac{H^{0,1}(\Sigma)}{\Lambda} \ar[l]_\Phi}}
\end{equation}

We finish this subsection with the following lemma:
\begin{lemma}
  \label{lemma:difference_of_two_spin_bundles}
  Let $q_1,q_2$ be two quadratic forms on $H_1(\Sigma, \ZZ/2\ZZ)$ and $\Ff_1,\Ff_2$ be the corresponding spin line bundles given by Proposition~\ref{prop:topological_charact_of_spin_bundles}. Let $\alpha\in H^{0,1}(\Sigma)$ be representing $\Psi(\Ff_2\otimes \Ff_1^\vee)$. Then $2\pi^{-1}\Im \alpha$ has an integer cohomology class and we have
  \[
    q_2 - q_1 = 2\pi^{-1}\Im\alpha\mod 2.
  \]
\end{lemma}
\begin{proof}
  The fact that $2\pi^{-1}\Im \alpha$ follows from Lemma~\ref{lemma:alpha_1_sim_alpha_2} immediately. Let $D_i$ be any divisor such that $\Ff_i \cong \Oo_\Sigma(D_i)$, and let $\omega_i$ be a meromorphic differential with divisor $2D_i$. Let $f$ be the meromorphic function defined by $\omega_2 = f\omega_1$. From Lemma~\ref{lemma:wind_lift} and Theorem~\ref{thmas:Johnson_spin_structures} it follows that for any curve $C$
  \begin{equation}
    \label{eq:tsb1}
    q_2(C) - q_1(C) = \frac{1}{2\pi i}\int_C d\log f \mod 2.
  \end{equation}
  Let $D = D_2- D_1$ and let $\omega_D = \frac{1}{2}d\log f$. Let $u$ be a harmonic differential on $\Sigma$ such that for any curve $C$ we have $\frac{1}{2\pi i}\int_C u = \frac{1}{2\pi i}\int_C\omega_D \mod 2$. Repeating the arguments from the proof of Lemma~\ref{lemma:PsiPhi=Aa} we can show that $\alpha$ can be taken to be $u^{0,1}$. Since $u$ is purely imaginary, it is equivalent to set
  \[
    u = 2i\Im \alpha.
  \]
  But in this case~\eqref{eq:tsb1} becomes the desired equality.
\end{proof}

\subsection{Theta function: definition and basic properties}\label{subsec:theta_function}

In this section we recall the definition of the theta function with a characteristics and listen some of its basic properties that we will use later. Given two real vectors $a,b\in \RR^g$ and $z\in \CC^g$ we define the function $\theta\chr{a}{b}(z,\Omega)$ as follows 
\begin{equation}
  \label{eq:def_of_theta}
  \theta\chr{a}{b}(z,\Omega) = \sum_{m\in \ZZ^g}\exp\Bigl( \pi i (m+a)^t\cdot \Omega (m+a) + 2\pi i (z-b)^t(m+a) \Bigr).
\end{equation}
The function $\theta\chr{a}{b}$ is called \emph{theta function with characteristics} $(a,b)$. We will use the traditional notation
\begin{equation}
  \label{eq:def_of_theta_zero_char}
  \theta(z,\Omega) = \theta\chr{0}{0}(z,\Omega).
\end{equation}

\begin{rem}
  \label{rem:miunus_in_def_of_theta}
  Classically~\cite[Chapter~II]{MumfordTata1}, theta function with characteristics $[a,b]$ is defined to be equal to $\theta\chr{a}{-b}$ in our notation. We choose a non-standard normalization to simplify the notation for the periodicity properties of $\theta$, see Proposition~\ref{prop:of_theta}.
\end{rem}

We now follow the notation from Section~\ref{subsec:simplicial_basis}. Let $q$ be a quadratic form on $H_1(\Sigma, \ZZ/2\ZZ)$ (see Section~\ref{subsec:spin_line_bundles}). Then there exist $a_i,b_i \in \{ 0,\frac{1}{2} \}$ such that
\begin{equation}
  \label{eq:a_b_for_quadratic_form}
  q(A_i) = 2a_i,\qquad q(B_i) = 2b_i,\qquad i = 1,\dots,g.
\end{equation}
In this case we say that the quadratic form $q$ has characteristics $[a,b]$. We denote by $q_\zero$ the quadratic form with zero characteristics, and by $\Ff_\zero$ the corresponding spin line bundle.

The following proposition is a straightforward computation (see~\cite[Chapter~II]{MumfordTata1}).

\begin{prop}
  \label{prop:of_theta}
  The function $\theta\chr{a}{b}$ defined as above satisfies the following properties:
  \begin{enumerate}
    \item for any $k\in \ZZ^{g}$ one has
      \[
        \begin{split}
          &\theta\chr{a}{b}(z+k,\Omega) = \exp(2\pi i k^t\cdot a)\theta\chr{a}{b}(z,\Omega),\\
          &\theta\chr{a}{b}(z+\Omega k, \Omega) = \exp(2\pi ik^t\cdot b ) \exp(-\pi i k^t\cdot\Omega k - 2\pi iz^t\cdot k)\theta\chr{a}{b}(z,\Omega).
        \end{split}
      \]

    \item Let $a,b\in \frac{1}{2}\ZZ^g$ the characteristics of a quadratic form $q$, then 
      \[
        \theta\chr{a}{b}(-z,\Omega) = (-1)^{\mathrm{Arf(q)}}\theta\chr{a}{b}(z,\Omega)
      \]

    \item For all $a,b\in \RR^g$ we have
      \[
        \theta\chr{a}{b}(z,\Omega) = \exp(\pi i a^tBa + 2\pi i (z-b)^t\cdot a)\theta(z-b+Ba,\Omega).
      \]
  \end{enumerate}
\end{prop}

The following theorem is usually referred as Riemann theorem on theta divisor (see~\cite[Chapter~II.3]{MumfordTata1} and~\cite[Chapter~II.7]{GriffitsHarris}). Recall the notation from Section~\ref{subsec:Jacobian}. Define
\[
  \Theta = \{ z\in \Jac(\Sigma)\ \mid\ \theta(z,\Omega) = 0 \}
\]
\begin{thmas}
  \label{thmas:theta_divisor}
  Let $D_0$ be a divisor on $\Sigma$ such that $\Ff_\zero = \Oo_\Sigma(D_0)$. Then for any other divisor $D$ we have
  \[
    \ord_{\Aa(D-D_0)}\theta(\cdot,\Omega) = \dim \Gamma(\Sigma, \Oo(D)),
  \]
  where $\Gamma(\Sigma, \Oo(D))$ is the space of global holomorphic sections of $\Oo(D)$. In particular,
  \[
    \Theta = \{ \Aa(D - D_0)\in \Jac(\Sigma)\ \mid\ \deg D = g-1,\ D\geq 0 \}.
  \]
\end{thmas}
\begin{proof}
  This result is classical, but it is usually not mentioned in the literature that the spin line bundle $\Oo(D_0)$ corresponds to $q_\zero$ in the sense of topological Proposition~\ref{prop:topological_charact_of_spin_bundles}. This topological fact is straightforward, but not completely elementary, so let us prove it here for the sake of completeness. Below we assume that $D_0$ is a divisor such that $\Ff_0 = \Oo(D_0)$ is a spin line bundle and all the assertions of the theorem, except $\Ff_0 \cong \Ff_\zero$, hold. Let $q_0$ be the quadratic form corresponding to $\Ff_0$.

  Consider any other spin line bundle $\Ff$ corresponding to a quadratic form $q$. Recall the notation from Section~\ref{subsec:Jacobian}. Let $\alpha\in H^{0,1}(\Sigma)$ represent $\Psi(\Ff\otimes \Ff_0^\vee)$. Let $a,b\in \RR^g$ to be the vectors of $A$- and $B$-periods of $\pi^{-1}\Im \alpha$ respectively, recall that
  \[
    \Phi(\alpha) = b - a\Omega 
  \]
  by~\eqref{eq:Phi_alpha_via_a_b}. Let $D$ be such that $\Ff \cong \Oo_\Sigma(D)$. By Lemma~\ref{lemma:PsiPhi=Aa} and the 3rd item of Proposition~\ref{prop:of_theta} we have
  \[
    \theta(z + \Aa(D-D_0),\Omega) = \theta(z + \Phi(\alpha), \Omega) = \theta(z + b-a\Omega,\Omega) = \exp(\ldots)\cdot \theta\chr{a}{b}(z,\Omega).
  \]
  This and the 2nd item of Proposition~\ref{prop:of_theta} imply that the order of $\theta(\cdot, \Omega)$ at $\Aa(D-D_0)$ is odd if and only if $[a,b]$ is an odd characteristics. We conclude with the properties of $D_0$ that
  \begin{equation}
    \label{eq:td1}
    \Ff \text{ is odd }\quad \Leftrightarrow \quad [a,b]\text{ is an odd characteristics.}
  \end{equation}
  Note that, by Lemma~\ref{lemma:difference_of_two_spin_bundles}, the form $2\pi^{-1}\Im \alpha$ defines a cohomology class $l\in H^1(\Sigma, \ZZ/2\ZZ)$ and $q = q_0 + l$. By Theorem~\ref{thmas:Arf_and_h0}, and the definition of $q_\zero$ the~eq.~\eqref{eq:td1} is equivalent to
  \begin{equation}
    \label{eq:td2}
    q_0 + l \text{ is odd } \quad \Leftrightarrow \quad q_\zero + l\text{ is odd}.
  \end{equation}
  Since $l$ is arbitrary, this immediately implies that $q_0 = q_\zero$.
\end{proof}

\subsection{The prime form}\label{subsec:the_prime_form}

In this section we briefly recall and study the construction of the prime form on the surface $\Sigma$. We refer the reader to~\cite[Chapter~II]{Fay} for more complete exposition. Informally speaking, the prime form $\pf(x,y)$ is a proper analogy of the function $x-y$ when $\CC$ is replaced by $\Sigma$. We begin by recalling how sections of spin line bundles can be constructred in terms of theta functions.

Let $\Ff_-$ be an arbitrary odd spin line bundle (see Section~\ref{subsec:spin_line_bundles}) with the corresponding quadratic form $q_-$, and let $[a_-,b_-]$ be the characteristics of $q_-$ (see Section~\ref{subsec:theta_function}). Let $\omega_1,\dots, \omega_g$ be the basis of normalized differentials, see Section~\ref{subsec:simplicial_basis}. Define
\begin{equation}
  \label{eq:def_of_dif_eta}
  \omega_-(p) = \sum_{j = 1}^g \frac{\partial}{\partial z_j} \theta\chr{a^-}{b^-}(0, \Omega)\omega_j(p).
\end{equation}

\begin{lemma}
  \label{lemma:on_dif_eta}
  The differentia $\omega_-$ is a square of a holomorphic section of $\Ff_-$. The line bundle $\Ff_-$ can be chosen such that $\omega_-$ is non-zero.
\end{lemma}
\begin{proof}
  For the existence of such a $\Ff_-$ that $\omega_-\neq 0$ see~\cite[Remark after Definition 2.1]{Fay}. Assume that $\Ff_-$ is chosen in this way. Let $D_->0$ be an effective divisor such that $\Ff_- \cong \Oo_\Sigma(D_-)$. Then, in the notation of Theorem~\ref{thmas:theta_divisor} we have $\Aa(D_- - D_0) = b_- - a_-\Omega$, as follows immediately from Lemma~\ref{lemma:difference_of_two_spin_bundles},~\eqref{eq:Phi_alpha_via_a_b} and Lemma~\ref{eq:PsiPhi=Aa1}. Due to this fact and~\cite[Corollary~1.4]{Fay} we have
  \[
    \div \omega_- = 2D_-.
  \]
  It follows that $\omega_-$ is a square of a holomorphic section of $\Ff_-$.
\end{proof}

Following~\cite[Chapter~II]{Fay} we define the prime form by
\begin{equation}
  \label{eq:def_of_prime_form}
  \pf(p,q) = \frac{\theta\chr{a^-}{b^-}(\Aa(p-q),\Omega)}{\sqrt{\omega_-}(p)\sqrt{\omega_-}(q)},
\end{equation}
where $\Aa$ is the Abel map, see~\eqref{eq:def_of_Abel_map}, and by $\sqrt{\omega_-}$ we mean a section of $\Ff_-$ whose square is $\omega_-$. This definition obviously has some ambiguities as one can choose different path of integration in the definition of $\Aa$, and it is not entirely clear what does it mean to divide by a section of a line bundle. To make things more precise, we will always assume that if $p$ and $q$ live in a simply connected domain then the path of integration lives inside this domain, and we identify $\sqrt{\omega_-}$ with a function using some trivialization of $\Ff_-$ over this domain. Otherwise we consider $\pf(p,q)$ to be a multiply-defined function with $(-1/2,0)$-covariance in each variable.

We have the following proposition, see~\cite[Chapter~II]{Fay}:
\begin{prop}
  \label{prop:on_pf}
  The prime form has the following properties:
  \begin{itemize}
    \item $\pf(p,q)$ vanishes when $p=q$ and is non-zero otherwise;
    \item $\pf(q,p) = -\pf(p,q)$
    \item Let $z$ be a local coordinate on $\Sigma$. Then
      \[
        \pf(p,q) = \frac{z(p) - z(q)}{\sqrt{dz}(p)\sqrt{dz}(q)} (1 + O(z(p) - z(q))),\quad \text{as }p\to q
      \]
      where $\sqrt{dz}$ is any local section of $\Ff_-$ such that $\sqrt{dz}\otimes \sqrt{dz}$ maps to $dz$ under $\Ff^{\otimes 2}_-\to T^*\Sigma$.
  \end{itemize}
\end{prop}

\subsection{Existence of a locally flat metric with conical singularities}
\label{subsec:locally_flat_metric_existence}

The main goal of this subsection is to prove the following proposition.

\begin{prop}
  \label{prop:existence_of_metric_on_Sigma}
  Let $\Sigma$ and $p_1,\dots,p_{2g-2}$ be as in Section~\ref{subsec:intro_flat_metric}. There exists a unique locally flat metric $ds^2$ on $\Sigma\smm\{ p_1,\dots,p_{2g-2} \}$ with conical singularities at $p_i$'s with cone angles equal to $4\pi$ for each $i = 1,\dots,2g-2$, and such that the area of $\Sigma$ is equal to 1. Moreover, we have the following:
  \begin{enumerate}
    \item\label{item:existence_of_metric_on_Sigma1} The metric $ds^2$ has the form $ds^2 = |\omega_0|^2$ where $\omega_0$ is a smooth $(1,0)$-form on $\Sigma$ satisfying $\sigma^*\omega_0 = \bar{ \omega }_0$ if the involution $\sigma$ is present.
    \item\label{item:existence_of_metric_on_Sigma2} The form $\omega_0$ satisfies the differential equation $(\dbar - \alpha_0)\omega_0 = 0$, where $\alpha_0$ is some antiholomorphic form on $\Sigma$ with the property that $\sigma^*\alpha_0 = \bar{\alpha}_0$ if the involution $\sigma$ is present.
    \item\label{item:existence_of_metric_on_Sigma3} The holonomy map of the metric $ds^2$ along any closed curve $\gamma$ on $\Sigma$ is given by $\exp(2i\int_\gamma\Im\alpha_0)$.
  \end{enumerate}
\end{prop}

Recall that each 2-form on a Riemann surface $\Sigma$ has the type $(1,1)$. A $(1,1)$-form $\Phi$ is called real if for each open set $U\subset \Sigma$ we have $\int_U\Phi\in \RR$. The following lemma is standard (see e.g.~\cite[p.~149]{GriffitsHarris}):
\begin{lemma}
  \label{lemma:potential}
  Assume that $\Phi$ is a smooth $(1,1)$-form on a compact Riemann surface $\Sigma$ such that $\int_\Sigma\Phi = 0$. Then there exists a function $\vphi\in \mC^\infty(\Sigma)$ such that
  \[
    \partial\dbar\vphi = \Phi.
  \]
  The function $\vphi$ is unique up to an additive constant, and if $\Phi$ is real, then $\vphi$ can be taken to have pure imaginary values.
\end{lemma}

\begin{proof}[Proof of Proposition~\ref{prop:existence_of_metric_on_Sigma}]
  Let $ds_0^2$ be some smooth metric on $\Sigma$ lying in the conformal class of $\Sigma$. If the involution $\sigma$ is present, then we assume that $ds_0^2$ is invariant under $\sigma$. Given a holomorphic local coordinate $z$ on $\Sigma$ we can write $ds_0^2 = e^{\vphi_0} |dz|^2$ for some smooth real-valued $\vphi$. Consider the $(1,1)$-form 
  \[
    \Phi_0 = \partial\dbar\vphi_0.
  \]
  It is straightforward to see that $\Phi_0$ does not depend on a local coordinate, therefore it is well-defined as a global $(1,1)$-form on $\Sigma$. In fact, on the level of volume forms we have
  \[
    \Phi_0 = iKds_0^2,
  \]
  where $K$ is the Gaussian curvature of $ds_0^2$, see~\cite[p.~77]{GriffitsHarris}. Hence, we have
  \begin{equation}
    \label{eq:moS1}
    \int_\Sigma \Phi_0 = 4\pi i (2-2g)
  \end{equation}
  by Gauss--Bonnet theorem.

  Given $p\in \Sigma$ denote by $\delta_p$ the $\delta$-measure at $p$, considered as a $(1,1)$-form with generalized coefficients. Applying Lemma~\ref{lemma:potential} to a suitable smooth approximation of $\delta$-measures and taking the limit we can find a real-valued function $\vphi$, smooth on $\Sigma\smm\{ p_1,\dots, p_{2g-2} \}$, having logarithmic singularities at $p_1,\dots, p_{2g-2}$ and satisfying the equation
  \begin{equation}
    \label{eq:moS2}
    \partial\dbar\vphi = -4\pi i\sum_{j = 1}^{2g-2}\delta_{p_j} - \Phi_0.
  \end{equation}
  Moreover, if the involution $\sigma$ is present, then we have $\sigma^*\vphi = \vphi$. Define
  \[
    ds^2 = e^\vphi ds_0^2.
  \]
  A straightforward local analysis shows that $ds^2$ extends to the whole $\Sigma$ as a locally flat metric with conical singularities at $p_1,\dots, p_{2g-2}$ with cone angles $4\pi$. If $\sigma$ is present, then $ds^2$ is invariant since $\vphi$ and $ds_0^2$ were invariant. Finally, replacing $\vphi$ with $\vphi+c$ for a suitable $c\in \RR$ we can make $ds^2$ to have a unit volume.

  Let us show that $ds^2$ with the above mentioned properties is unique. Given any such $ds^2_1$, define the function $\vphi_1$ by the equation $ds^2_1 = e^{\vphi_1} ds_0^2$. Then $\vphi_1$ must satisfy~\eqref{eq:moS2}, hence $\vphi_1 = \vphi+\cst$ by the uniqueness part of Lemma~\ref{lemma:potential}. But $\cst = 0$ due to the volume normalization.

  It remains to construct the $(1,0)$-form $\omega_0$ and $(0,1)$-form $\alpha_0$ required by the proposition. Note that the local holonomy of $ds^2$ is trivial (that is, the parallel transport along any contractible loop acts trivially on the tangent space). It follows that the holonomy of $ds^2$ along a loop $\gamma$ depends on the homology class of $\gamma$ only. Therefore, we can find a real harmonic differential $u$ such that the holonomy of $ds^2$ along any loop $\gamma$ is given by multiplication by $\exp(i\int_\gamma u)$. The differential $u$ can be taken such that $\sigma^*u = -u$ if $\sigma$ is present. Define
  \[
    \alpha_0 = iu^{0,1}.
  \]
  Then $\alpha_0$ satisfies all the properties from the items~\ref{item:existence_of_metric_on_Sigma2},~\ref{item:existence_of_metric_on_Sigma3} of the proposition.

  To define $\omega_0$, fix a point $p_0\in \Sigma$ and pick a cotangent vector $v_0\in T^\ast_{p_0}\Sigma$ with length 1 with respect to $ds^2$. Given a point $p\in \Sigma$ close to $p_0$ we can apply a parallel transportation to $v_0$ along any short path connecting $p_0$ with $p$ to obtain a cotangent vector at $p$. In this way we obtain a $(1,0)$-form $\omega$ in a small vicinity of $p_0$. The fact that $ds^2$ is locally flat with conical singularities of cone angles $4\pi$ implies that $\omega$ is holomorphic, and can be extended to the whole $\Sigma$ as a multivalued holomorphic differential with the multiplicative monodromy $\exp(-2i\int_\gamma \Im\alpha_0)$ along each non-trivial loop $\gamma$, and we have $ds^2 = |\omega|^2$ everywhere. We now can set
  \[
    \omega_0(p) = \exp(2i\int_{p_0}^p\Im \alpha_0)\omega.
  \]
  Replacing $\omega$ with $\eta\omega$ for some $|\eta| = 1$ and choosing $p_0\in \partial \Sigma_0$ if the involution $\sigma$ is present we can ensure that $\sigma^\omega_0 = \bar{\omega}_0$ so that all the properties of $\omega_0$ required in items~\ref{item:existence_of_metric_on_Sigma1}--\ref{item:existence_of_metric_on_Sigma3} of the proposition are fulfilled.
\end{proof}

\subsection{Properties of the kernel \texorpdfstring{$\Dd_\alpha^{-1}$}{Dalpha-1}}
\label{subsec:asymptotics_Ddalpha}

The goal of this subsection is to fill in the details omitted in Section~\ref{sec:The Riemann surface: continuous setting}. We begin by proving Lemma~\ref{lemma:q0_is_quadratic_form}:

\begin{proof}[Proof of Lemma~\ref{lemma:q0_is_quadratic_form}]
  We will be using the notation from Section~\ref{subsec:spin_line_bundles}. To prove that $q_0$ is a quadratic form we will construct a spin structure $\xi_0$ on $\Sigma$ such that $q_0$ corresponds to $\xi_0$ via Theorem~\ref{thmas:Johnson_spin_structures}. Let $\Sigma' = \Sigma\smm\{ p_1,\dots, p_{2g-2} \}$ and $\pi:UT\Sigma'\to \Sigma'$ be the projection. Define $\mu\in H^1(UT\Sigma',\ZZ/2\ZZ)$ by
  \[
    \mu(\tilde\gamma) = \pi(\tilde\gamma)\cdot (\gamma_1+\ldots+\gamma_{g-1})\mod 2.
  \]
  By evaluating the $(1,0)$-form $\omega_0$ at tangent vectors we obtain a non-vanishing function on $UT\Sigma'$. Denote this function by $\vphi_{\omega_0}$. Then it is easy to see that $\xi_0\in H^1(UT\Sigma',\ZZ/2\ZZ)$ defined by
  \[
    \xi_0(\tilde\gamma) = \frac{1}{2\pi}\Im \int_{\tilde\gamma}d\log \vphi_{\omega_0} + \mu(\tilde\gamma)\mod 2
  \]
  depends only on the homology class of $\tilde\gamma$ in $H_1(UT\Sigma,\ZZ/2\ZZ)$ and defines a spin structure on $\Sigma$.

  Given a smooth loop $\gamma$ on $\Sigma'$ denote by $\tilde\gamma$ its lift to $UT\Sigma'$ given by its tangent vector. We have
  \begin{equation}
    \label{eq:qqf}
    \wind(\gamma,\omega_0) = \Im \int_{\tilde\gamma}d\log\vphi_{\omega_0}.
  \end{equation}
  It follows from the definition~\eqref{eq:def_of_q0} of $q_0$ and~\eqref{eq:qqf} that
  \[
    q_0(\gamma) = \xi_0(\tilde\gamma) + 1\mod 2.
  \]
  Hence, by Theorem~\ref{thmas:Johnson_spin_structures}, $q_0$ is a quadratic form.
\end{proof}

We now prove Proposition~\ref{prop:def_of_S}.

\begin{proof}[Proof of Proposition~\ref{prop:def_of_S}]
  The fact that $\Dd_\alpha^{-1}(p,q)$ satisfies equations~\eqref{eq:equation_on_q_of_Salpha} and~\eqref{eq:equation_on_p_of_Salpha} follows from the formula for $\Dd_\alpha^{-1}(p,q)$ and Proposition~\ref{prop:existence_of_metric_on_Sigma}, item~\ref{item:existence_of_metric_on_Sigma2}. It follows from the properties of theta functions (Proposition~\ref{prop:of_theta}) that the function $\Dd_\alpha^{-1}(p,q)$ extends to 
\[
  \Bigl((\Sigma\smm\{ p_1,\dots,p_{2g-2} \})\times (\Sigma\smm\{ p_1,\dots,p_{2g-2} \})\Bigr)\smm\mathrm{Diagonal}
\]
  as a multivalued function. We are left to verify that $\Dd_\alpha^{-1}$ has the correct monodromy. 

  Let $q_-$ be the spin structure corresponding to the odd spin line bundle $\Ff_-$ used in the construction of the prime form $\pf$. By Lemma~\ref{lemma:wind_lift} and Theorem~\ref{thmas:Johnson_spin_structures} we have  
  \begin{equation}
    \label{eq:doS1}
    q_-(\gamma) = (2\pi)^{-1}\wind(\gamma,\omega_-) + 1\mod 2
  \end{equation}
  for any smooth simple loop $\gamma$ on $\Sigma$. Recall that $\varsigma$ is the smooth function on $\Sigma\smm\{ p_1,\dots, p_{2g-2} \}$ defined by
  \[
    \omega_- = \varsigma \omega_0.
  \]
  It follows that 
  \begin{equation}
    \label{eq:doS2}
    \wind(\gamma, \omega_-) = \wind(\gamma, \omega_0) + \Im\int_\gamma d\log \varsigma.
  \end{equation}
  Combining~\eqref{eq:doS1} and~\eqref{eq:doS2} with the definition~\eqref{eq:def_of_q0} we get
  \begin{equation}
    \label{eq:doS3}
    q_-(\gamma) - q_0(\gamma) = \frac{1}{2\pi}\Im\int_\gamma d\log \varsigma + \gamma\cdot (\gamma_1 + \ldots + \gamma_{g-1}) \mod 2.
  \end{equation}
  Recall that $[a^0, b^0]$ is the characteristics of $q_0$ and $[a^-, b^-]$ is the characteristics of $q_-$. Using the definition of the characteristics we can rewrite~\eqref{eq:doS3} as
  \begin{equation}
    \label{eq:doS4}
    \frac{1}{2\pi}\Im\int_\gamma d\log \varsigma = \sum_{i = 1}^g(2(a^0_i - a^-_i)\gamma\cdot B_i + 2(b^0_i - b^-_i) \gamma\cdot A_i)) + \gamma\cdot (\gamma_1 + \ldots + \gamma_{g-1}) \mod 2.
  \end{equation}
  The equality~\eqref{eq:doS4} holds for any smooth loop $\gamma$ on $\Sigma\smm\{ p_1,\dots, p_{2g-2} \}$ such that $\varsigma$ does not vanish along $\gamma$. This determines the monodromy of $\sqrt{\varsigma}$ uniquely. Combining this with the properties of theta function from Proposition~\ref{prop:of_theta} we determine the monodromy of $\pf(p,q)\sqrt{\omega_0(p)}\sqrt{\omega_0(q)}$, and also the monodromy of all other terms in the definition of $\Dd_\alpha^{-1}$. Direct verification shows that the monodromy of $\Dd_\alpha^{-1}$ is as stated in the proposition.
\end{proof}

The proofs of Lemma~\ref{lemma:diagonal_expansion_of_Salpha} and Lemma~\ref{lemma:bw_near_singularity_expansion_of_Salpha} come from direct computations which we leave to the reader. We finish the subsection by proving Lemma~\ref{lemma:derivative_of_theta}.

\begin{proof}[Proof of Lemma~\ref{lemma:derivative_of_theta}]
  Recall the permutation matrix $J$ introduced in~\eqref{eq:def_of_J}. Recall that $\theta[\alpha](z)$ is defined by~\eqref{eq:def_of_theta_alpha}. From~\eqref{eq:Omega_under_involution}, the definition~\eqref{eq:def_of_theta} of the theta function and symmetries of $\alpha_t$ and $\alpha_G$ with respect to $\sigma$ we get
  \begin{equation}
    \label{eq:doth1}
    \theta[\alpha_{h,t} + \alpha_G](z) = \overline{\theta[-\alpha_{h,t} + \alpha_G](-J\bar{z})}.
  \end{equation}
  In particular, $\theta[\alpha_{h,t}+\alpha_G](0)\neq 0$ implies $\theta[-\alpha_{h,t}+\alpha_G](0)\neq 0$. Define
  \begin{equation}
    \label{eq:def_of_Walpha}
    W_{\alpha_{h,t} + \alpha_G}(q) = d_q \log \theta[\alpha_{h,t} + \alpha_G](\Aa(p-q))\vert_{p=q},
  \end{equation}
  where $\Aa$ is the Abel map~\eqref{eq:def_of_Abel_map}. By~\eqref{eq:doth1} and~\eqref{eq:symmetries_of_normalized_differentials} we have
  \begin{equation}
    \label{eq:Walpha_is_symmetric}
    \sigma^*W_{\alpha_{h,t} + \alpha_G} = \overline{W}_{-\alpha_{h,t} + \alpha_G}.
  \end{equation}
  By the definition of $r$ given in Lemma~\ref{lemma:diagonal_expansion_of_Salpha} we have
  \begin{equation}
    \label{eq:doth1.4}
    r_{\alpha_t + \alpha_G}\omega_0 = \frac{1}{\pi i} W_{\alpha_{h,t} + \alpha_G} - \frac{2}{\pi i}\partial \Re \vphi_t.
  \end{equation}
  Using that $\sigma^\ast\alpha_t = -\bar\alpha_t$ and~\eqref{eq:Walpha_is_symmetric} we conclude that
  \begin{equation}
    \label{eq:ralpha_is_symmetric}
    \sigma^\ast (r_{\alpha_t + \alpha_G}\omega_0) = -\overline{r_{-\alpha_t+\alpha_G}\omega_0}.
  \end{equation}
  We conclude that
  \begin{multline}
    \label{eq:doth1.5}
    -\frac{1}{4}\int_\Sigma \left( r_{\alpha_t+\alpha_G}\omega_0\wedge \dot{\alpha}_t - \overline{r_{-\alpha_t + \alpha_G}\omega_0\wedge \dot{\alpha}_t} \right) =\\
    = -\frac{1}{4}\int_\Sigma\left( r_{\alpha_t + \alpha_G}\omega_0\wedge \dot{\alpha}_t - \sigma^\ast\left( r_{-\alpha_t + \alpha_G}\omega_0\wedge \dot{\alpha}_t \right) \right) = -\frac{1}{2}\int_\Sigma r_{\alpha_t + \alpha_G}\omega_0\wedge \dot{\alpha}_t.
  \end{multline}
  Let us now substitute~\eqref{eq:doth1.4} into the right-hand side of~\eqref{eq:doth1.5}. Using Riemann bilinear relations~\eqref{eq:Riemann bilinear relations} we obtain
  \begin{equation}
    \label{eq:doth2}
    \begin{split}
      &-\frac{1}{4}\int_\Sigma \left( r_{\alpha_t+\alpha_G}\omega_0\wedge \dot{\alpha}_t - \overline{r_{-\alpha_t + \alpha_G}\omega_0\wedge \dot{\alpha}_t} \right) = -\frac{1}{2}\int_\Sigma r_{\alpha_t + \alpha_G}\omega_0\wedge \dot{\alpha}_t =\\
      &= -\frac{1}{2\pi i}\int_\Sigma W_{\alpha_{h,t} + \alpha_G}\wedge \dot{\alpha}_{h,t} + \frac{1}{\pi }\Im \int_\Sigma \partial \Re\vphi\wedge \dbar\dot{\vphi}_t =\\
      &= -\frac{1}{\pi}\int_\Sigma W_{\alpha_{h,t} + \alpha_G}\wedge \Im\dot{\alpha}_{h,t} - \frac{1}{\pi }\Re \int_{\Sigma_0} \Delta \Re\vphi_t\cdot \dot{\vphi}_t\, ds^2 =\\
      &= \sum_{j = 1}^g \frac{\theta[\alpha_{h,t} + \alpha_G]_j(0)}{\theta[\alpha_{h,t} + \alpha_G](0)}(\Omega a(\dot{\alpha}_{h,t}) - b(\dot{\alpha}_{h,t}))_j - \frac{d}{dt}\frac{1}{2\pi}\int_{\Sigma_0} \Delta\Re\vphi_t\Re\vphi_t\,ds^2
    \end{split}
  \end{equation}
  where $\theta[\alpha_{h,t} + \alpha_G]_j(0)$ denotes the partial derivative of $\theta[\alpha_{h,t} + \alpha_G](z)$ by $z_j$ at $z = 0$.

  On the other hand, differentiating the series defining the theta function we get
  \begin{equation}
    \label{eq:doth3}
    \frac{d}{dt}\log \theta[\alpha_{h,t} + \alpha_G](0) = \sum_{j = 1}^g \frac{\theta[\alpha_{h,t} + \alpha_G]_j(0)}{\theta[\alpha_{h,t} + \alpha_G](0)}(\Omega a(\dot{\alpha}_{h,t}) - b(\dot{\alpha}_{h,t}))_j -2\pi i a(\dot{\alpha}_{h,t})\cdot b(\alpha_G).
  \end{equation}
  Combining~\eqref{eq:doth2} and~\eqref{eq:doth3} we get the result.
\end{proof}

\subsection{Teichm\"uller space and the space of Torelli marked Riemann surfaces}
\label{subsec:teichmuller_space}

In order to formulate the setup for main theorems (see Section~\ref{subsec:intro_main_results}) we need some basic facts about the Teichm\"uller space. We address the reader to~\cite{ImayoshiTaniguchi} for a detailed exposition of the subject. Let $\Sigma_\rf$ be a fixed closed Riemann surface of genus $g$. As a set, Teichm\"uller space is defined as
\[
  \Tch_g = \{ (\Sigma, f)\ \mid\ f:\Sigma_\rf\to \Sigma\text{ orientation preserving homeomorphism} \}/_\sim
\]
where $(\Sigma_1,f_1)\sim (\Sigma_2,f_2)$ if and only if $f_2\circ f_1^{-1}$ is homotopic to a conformal map between $\Sigma_1$ and $\Sigma_2$. One of the ways to describe the topology on $\Tch_g$, originally referred to Teichm\"uller, is to consider extremal mappings in the homotopy class of $f_2\circ f_1^{-1}$. Teichm\"uller's theorem asserts that for any two points $(\Sigma_1,f_1)$ and $(\Sigma_2,f_2)$ in $\Tch_g$ there exists a unique quasiconformal mapping $h:\Sigma_1\to \Sigma_2$ called the `Teichm\"uller map' homotopic to $f_2\circ f_1^{-1}$ and minimizing the $L^\infty$ norm of the Beltrami coefficient among other quasiconformal mappings. The Beltrami coefficient of $h$ appears to be of the form $k\frac{u}{|u|}$, where $u$ is some holomorphic quadratic differential on $\Sigma_1$ and $k\in [0,1)$ is a constant. In particular, $h$ is locally affine outside zeros of $u$. The Teichm\"uller distance between $(\Sigma_1,f_1)$ and $(\Sigma_2,f_2)$ is defined to be the logarithm of the maximal dilatation of $h$, that is, $\log\frac{k+1}{k-1}$. 

Using Teichm\"uller maps we can define the topology of $\mC^2$ convergence on the space of diffeomorphisms between Riemann surfaces. Fix a finite open cover $U_1\cup \ldots\cup U_n$ of $\Sigma_\rf$. Assume that $(\Sigma_k, f_k)$ is a sequence of points in $\Tch_g$ converging to a point $(\Sigma, f)$. Without loss of generality we assume that $f,f_1,f_2,\dots$ are all the corresponding Teichm\"uller maps. Assume that for each $k$ we are given a  diffeomorphism $\xi_k:\Sigma_k\to \Sigma$. We say that the sequence $\xi_k$ converges to identity in $\mC^2$ topology if the following holds:
\begin{enumerate}
  \item For each $j = 1,\dots, n$ and any compact $K\subset U_j$ we have $\xi_k(f_k(K))\subset f(U_j)$ for $k$ large enough depending on $K$.
  \item For each $j = 1,\dots, n$ there exist holomorphic coordinates $z_j: f(U_j)\to \CC$ and $z_j^{(k)}: f_k(U_j)\to \CC$ such that the functions $z_j^{(k)}\circ\xi_k^{-1}\circ z_j^{-1}$ converge to identity uniformly in $\mC^2$ topology on any compact of $z_j(f(U_j))$.
\end{enumerate}
Note that if the sequence $\xi_k$ converges to identity in $\mC^2$ topology, then maximum dilatations of $\xi_k$'s considered as quasiconformal mappings converges to $1$.

The topological space $\Tch_g$ can be equipped with a structure of a complex manifold. Recall that the mapping class group $\Mod(\Sigma_\rf)$ is defined as 
\[
  \Mod(\Sigma_\rf) = \frac{\Diff(\Sigma_\rf)}{\Diff_0(\Sigma_\rf)},
\]
where $\Diff_0(\Sigma_\rf)$ is the group of diffeomorphisms of $\Sigma_\rf$ homotopic to identity. The group $\Mod(\Sigma_\rf)$ acts on $\Tch_g$ properly discontinuously, each point of $\Tch_g$ has finite stabilizer, and the quotient $\Mm_g = \Tch_g/\Mod(\Sigma_\rf)$ is a smooth complex orbifold called~\emph{the moduli space of genus $g$ Riemann surfaces}.

There is a natural homomorphism of groups $\Mod(\Sigma_\rf)\to \Aut(H_1(\Sigma_\rf, \ZZ))$, where by $\Aut$ we denote the set of symplectic automorphisms. The kernel of this homomorphism is called the \emph{Torelli group}; we denote it by $\Ii_g(\Sigma_\rf)$. The quotient 
\begin{equation}
  \label{eq:def_of_t}
  \Mm_g^t = \frac{\Tch_g}{\Ii_g(\Sigma_\rf)}
\end{equation}
is called the \emph{moduli space of Torelli marked curves}. One can show that the action of $\Ii(\Sigma_\rf)$ has no fixed points on $\Tch_g$, hence $\Mm_g^t$ is a smooth complex manifold.

Choose a symplectic basis $A_1^\rf,\dots, A_g^\rf, B_1^\rf,\dots, B_g^\rf$ in $H_1(\Sigma_\rf,\ZZ)$. Then for each $[(\Sigma, f)]\in \Mm_g^t$ there is a natural choice of a symplectic basis in $H_1(\Sigma, \ZZ)$ given by
\[
  A_i = f_*A_i^\rf, \qquad B_i = f_*B_i^\rf.
\]
One can show that this defines a bijection between the set of points of $\Mm_g^t$ and the set of isomorphism classes of Riemann surfaces of genus $g$ with a fixed symplectic basis in the first homologies. This basis is usually called a \emph{Torelli marking} of the surface.

Finally, let us introduce the moduli space $\Mm_g^{t,(0,1)}$ of Torelli marked Riemann surfaces of genus $g$ with a fixed anti-holomorphic $(0,1)$-form. Set first $\Mm_g^{t,(0,1)} = \CC^g\times \Mm_g^t$. Given $(z_1,\dots, z_g)\in \CC^g$ define the $(0,1)$-form on $\Sigma$ by
\[
  \alpha = \sum_{i = 1}^g z_i \bar{\omega}_i
\]
where $\omega_1,\dots, \omega_g$ is the set of normalized differentials on $\Sigma$ associated with the Torelli marking, see Section~\ref{subsec:simplicial_basis}. Using this identification we can interpret $\Mm_g^{t,(0,1)}$ as the moduli space of tuples $(\Sigma, A, B, \alpha)$, where $\Sigma$ is a Riemann surface of genus $g$, $(A,B)$ is a simplicial basis in $H_1(\Sigma, \ZZ)$, and $\alpha$ is an anti-holomorphic $(0,1)$-form.

\end{appendix}

\printbibliography

\end{document}